\newif\ifshort
\shortfalse
\shorttrue 

\documentclass[acmsmall]{acmart}
\AtBeginDocument{%
  }

\setcopyright{rightsretained}
\acmDOI{10.1145/3704852}
\acmYear{2025}
\acmJournal{PACMPL}
\acmVolume{9}
\acmNumber{POPL}
\acmArticle{16}
\acmMonth{1}
\received{2024-07-08}
\received[accepted]{2024-11-07}

\AtEndPreamble{\newtheorem{remark}[theorem]{Remark}}%
\AtBeginDocument{\let\oldRemark\remark\def\remark{\oldRemark\normalfont}}
\AtEndPreamble{}%
\AtEndPreamble{}%
\AtEndPreamble{\theoremstyle{acmdefinition}
\newtheorem{counterexample}[theorem]{Counter example}%
\theoremstyle{acmplain}}%
\usepackage {calculus,eqntabular,renumber,galois,strut}
\usepackage{enumitem,interval,multicol,relsize}
\ifshort
\usepackage[bibliography=common]{apxproof}
\else
\usepackage[appendix=inline]{apxproof}
\fi
\font\bboldfontten=bbold12 at 10pt
\font\bboldfonteight=bbold12 at 8pt
\font\bboldfontsix=bbold12 at 6pt
\renewcommand{\mathbb}[1]{\mathchoice{\mbox{\bboldfontten#1}}{\mbox{\bboldfontten#1}}{\mbox{\scriptsize\bboldfonteight#1}}{\mbox{\tiny\bboldfontsix#1}}}
\usepackage{MnSymbol}
\newcommand{\ltuple}[1]{\langle\nobreak#1,\allowbreak}
\newcommand{\mtuple}[1]{\:#1\nobreak,\allowbreak}
\newcommand{\rtuple}[1]{\:#1\nobreak\rangle}

\newcommand{\pair}[2]{\ltuple{#1}\allowbreak\rtuple{#2}}
\newcommand{\triple}[3]{\ltuple{#1}\mtuple{#2}\rtuple{#3}}
\newcommand{\quadruple}[4]{\ltuple{#1}\mtuple{#2}\mtuple{#3}\rtuple{#4}}
\newcommand{\quintuple}[5]{\ltuple{#1}\mtuple{#2}\mtuple{#3}\mtuple{#4}\rtuple{#5}}
\newcommand{\sextuple}[6]{\ltuple{#1}\mtuple{#2}\mtuple{#3}\mtuple{#4}\mtuple{#5}\rtuple{#6}}

\newcommand{\decatuple}[9]
{\ltuple{#1}\mtuple{#2}\mtuple{#3}\mtuple{#4}\mtuple{#5}\mtuple{#6}\mtuple{#7}\mtuple{#8}\mtuple{#9}\rtuple}
\makeatletter
\newcommand{\undecatuple}[9]
{\ltuple{#1}\mtuple{#2}\mtuple{#3}\mtuple{#4}\mtuple{#5}\mtuple{#6}\mtuple{#7}\mtuple{#8}\mtuple{#9}\@undecatuple}
\newcommand{\@undecatuple}[2]{\mtuple{#1}\rtuple{#2}}
\newcommand{\dodecatuple}[9]
{\ltuple{#1}\mtuple{#2}\mtuple{#3}\mtuple{#4}\mtuple{#5}\mtuple{#6}\mtuple{#7}\mtuple{#8}\mtuple{#9}\@dodecatuple}
\newcommand{\@dodecatuple}[3]{\mtuple{#1}\mtuple{#2}\rtuple{#3}}
\newcommand{\tridecatuple}[9]
{\ltuple{#1}\mtuple{#2}\mtuple{#3}\mtuple{#4}\mtuple{#5}\mtuple{#6}\mtuple{#7}\mtuple{#8}\mtuple{#9}\@tridecatuple}
\newcommand{\@tridecatuple}[4]{\mtuple{#1}\mtuple{#2}\mtuple{#3}\rtuple{#4}}
\newcommand{\tetradecatuple}[9]
{\ltuple{#1}\mtuple{#2}\mtuple{#3}\mtuple{#4}\mtuple{#5}\mtuple{#6}\mtuple{#7}\mtuple{#8}\mtuple{#9}\@tetradecatuple}
\newcommand{\@tetradecatuple}[5]{\mtuple{#1}\mtuple{#2}\mtuple{#3}\mtuple{#4}\rtuple{#5}}
\newcommand{\pentadecatuple}[9]
{\ltuple{#1}\mtuple{#2}\mtuple{#3}\mtuple{#4}\mtuple{#5}\mtuple{#6}\mtuple{#7}\mtuple{#8}\mtuple{#9}\@pentadecatuple}
\newcommand{\@pentadecatuple}[6]{\mtuple{#1}\mtuple{#2}\mtuple{#3}\mtuple{#4}\mtuple{#5}\rtuple{#6}}
\newcommand{\hexadecatuple}[9]
{\ltuple{#1}\mtuple{#2}\mtuple{#3}\mtuple{#4}\mtuple{#5}\mtuple{#6}\mtuple{#7}\mtuple{#8}\mtuple{#9}\@hexadecatuple}
\newcommand{\@hexadecatuple}[7]{\mtuple{#1}\mtuple{#2}\mtuple{#3}\mtuple{#4}\mtuple{#5}\mtuple{#6}\rtuple{#7}}
\makeatletter
\newcommand{\sqb}[1]{\ensuremath{\def\@paramsqb{#1}\ifx\@paramsqb\@empty\else\llbracket\@paramsqb\rrbracket\fi}}
\newcommand{\osqb}[1]{\ensuremath{\def\@paramsqb{#1}\ifx\@paramsqb\@empty\else\mathopen{\bar{\llbracket}}\@paramsqb\mathclose{\bar{\rrbracket}}\fi}}
\makeatother 

\newcommand{\Lfp}[1]{\ensuremath{\textsf{\upshape lfp}^{\scriptscriptstyle\mskip2mu #1}\,}}
\newcommand{\Gfp}[1]{\ensuremath{\textsf{\upshape gfp}^{\scriptscriptstyle\mskip2mu #1}\,}}

\usepackage{relsize}
\makeatletter
\def\@LAMBDAoperator{\text{\boldmath$\lambda$}}%
\def\@LAMBDApoint{%
\mathchoice%
{\,\mbox{\relsize{2}\bf\raisebox{0.3ex}{.}}\,}%
{\,\mbox{\relsize{2}\bf\raisebox{0.3ex}{.}}\,}%
{\,\mbox{\relsize{1}\bf\raisebox{0.3ex}{.}}\,}%
{\,\mbox{\bf\raisebox{0.3ex}{.}}\,}%
}
\def\LAMBDA#1{\@ifnextchar[{{\@@LAMBDA@IN{#1}}}{{\@@LAMBDA{#1}}}}
\def\@@LAMBDA#1{\@LAMBDAoperator{#1}{\@LAMBDApoint}}
\def\@@LAMBDA@IN#1[#2]{\@LAMBDAoperator{#1}\,{\in}\,{#2}{\@LAMBDApoint}}
\def\LAMBDAoperator{{\ensuremath{\@LAMBDAoperator}}}
\makeatother 
\usepackage{accents}
\usepackage{rotating}
\makeatletter
\DeclareRobustCommand{\cev}[1]{%
  \mathpalette\do@cev{#1}%
}
\newcommand{\do@cev}[2]{%
  \fix@cev{#1}{+}%
  \reflectbox{$\m@th#1\vec{\reflectbox{$\fix@cev{#1}{-}\m@th#1#2\fix@cev{#1}{+}$}}$}%
  \fix@cev{#1}{-}%
}
\newcommand{\fix@cev}[2]{%
  \ifx#1\displaystyle
    \mkern#23mu
  \else
    \ifx#1\textstyle
      \mkern#23mu
    \else
      \ifx#1\scriptstyle
        \mkern#22mu
      \else
        \mkern#22mu
      \fi
    \fi
  \fi
}
\DeclareRobustCommand{\Cev}[1]{\cev{\mskip2mu#1\mskip-2mu}\mskip2mu}

\makeatother
\newcommand{\functionto}{\ensuremath{\rightarrow}}
\newcommand{\increasingfunctionto}{\ensuremath{\!\stackrel{\raisebox{-0.75ex}[0pt][0pt]{\tiny$\mskip6mu\nearrow$}\:}{\longrightarrow}}\!}
\newcommand{\joinmorphismto}{\mathrel{\raisebox{0.85ex}{\rlap{\hskip1ex\tiny$\sqcup$}}{\longrightarrow}}}
\newcommand{\qef}{\hfill \ensuremath{\blacksquare}}
\AtBeginDocument{\let\oldendexample\endexample\def\endexample{\qef\oldendexample}}
\AtBeginDocument{\let\oldendremark\endremark\def\endremark{\qef\oldendremark}}
\AtBeginDocument{\let\oldendcounterexample\endcounterexample\def\endcounterexample{\qef\oldendcounterexample}}
\newbox\bsqcup
\newcommand{\msqcup}[1]{\mathchoice{\setbox\bsqcup=\hbox{$\displaystyle\sqcup$}\mathbin{\rlap{\hbox to \wd\bsqcup{\hfill\raisebox{0.3ex}{{\relsize{-3}{$#1$}}}\hfill}}{\displaystyle\sqcup}}}{\setbox\bsqcup=\hbox{$\textstyle\sqcup$}\mathbin{\rlap{\hbox to \wd\bsqcup{\hfill\raisebox{0.3ex}{{\relsize{-3}{$#1$}}}\hfill}}{\textstyle\sqcup}}}{\setbox\bsqcup=\hbox{$\scriptstyle\sqcup$}\mathbin{\rlap{\hbox to \wd\bsqcup{\hfill\raisebox{0.3ex}{{\relsize{-4}{$#1$}}}\hfill}}{\scriptstyle\sqcup}}}{\setbox\bsqcup=\hbox{$\scriptscriptstyle\sqcup$}\mathbin{\rlap{\hbox to \wd\bsqcup{\hfill\raisebox{0.3ex}{{\relsize{-5}{$#1$}}}\hfill}}}{\scriptscriptstyle\sqcup}}}

\newcommand{\bigmsqcup}[1]{\mathop{\mathchoice{\setbox\bsqcup=\hbox{$\displaystyle\bigsqcup$}{\rlap{\hbox to \wd\bsqcup{\hfill\raisebox{0.3ex}{\relsize{-3}{$#1$}}\hfill}}{\displaystyle\bigsqcup}}}{\setbox\bsqcup=\hbox{$\textstyle\bigsqcup$}{\rlap{\hbox to \wd\bsqcup{\hfill\raisebox{0.3ex}{\relsize{-3}{$#1$}}\hfill}}{\textstyle\bigsqcup}}}{\setbox\bsqcup=\hbox{$\scriptstyle\bigsqcup$}{\rlap{\hbox to \wd\bsqcup{\hfill\raisebox{0.3ex}{\relsize{-4}{$#1$}}\hfill}}{\scriptstyle\bigsqcup}}}{\setbox\bsqcup=\hbox{$\scriptscriptstyle\bigsqcup$}{\rlap{\hbox to \wd\bsqcup{\hfill\raisebox{0.3ex}{\relsize{-5}}{$#1$}\hfill}}}{\scriptscriptstyle\bigsqcup}}}\limits}
\newcommand{\maccent}[2]{\smash{\mathchoice{\begin{tabular}[b]{@{}c@{}}\relsize{-3}{$#2$}\\[-1.6ex]$\displaystyle#1$\end{tabular}}{\begin{tabular}[b]{@{}c@{}}\relsize{-3}{$\textstyle#2$}\\[-1.6ex]$\textstyle#1$\end{tabular}}{\begin{tabular}[b]{@{}c@{}}\relsize{-4}{$#2$}\\[-1.875ex]$\scriptstyle#1$\end{tabular}}{\begin{tabular}[b]{@{}c@{}}\relsize{-5}{$#2$}\\[-2.1ex]$\scriptscriptstyle#1$\end{tabular}}}}
\ifshort
\newcounter{refappendix}
\makeatletter
\newcommand{\proofinapxsymbol}{\rlap{\hskip0.8mm\raisebox{0.5pt}[0pt][0pt]{\hskip-0.25pt\textup{\textbf{\scriptsize A}}}}$\bigcirc$}
\newcommand{\proofinapx}{%
\addtocounter{refappendix}{1}%
\immediate\write\axp@proofsfile{\noexpand\hypertarget\string{apx\arabic{refappendix}\string}\string{\string~\string~\string}}%
\hyperlink{apx\arabic{refappendix}}{\ \proofinapxsymbol}}%
\makeatother
\else
\newcommand{\proofinapx}{\relax}
\fi
\newcommand{\llbrace}{\mathopen{{\{}\mskip-4mu{|}}}
\newcommand{\rrbrace}{\mathopen{{|}\mskip-4mu{\}}}}
\newcommand{\postSemantics}{\textsf{\textup{post}}^\sharp}
\newcommand{\adhocHyperSemantics}{\overline{\textsf{\textup{Post}}}^\sharp}
\newcommand{\Hpost}{\textsf{\textup{Post}}^\sharp}

\newcommand{\hyperlogicup}[3]
{\overline{\llbrace}\,#1\,\overline{\rrbrace}\,\texttt{#2}\,\overline{\llbrace}\,#3\,\overline{\rrbrace}}
\newcommand{\hoarelogicup}[3]
{\:{\overline{\{}\,#1\overline{\}}\,\texttt{#2}\,\overline{\{}#3\overline{\}}}\:}
\newcommand{\hoarelogicdown}[3]
{\:{\underline{\{}\,#1\underline{\}}\,\texttt{#2}\,\underline{\{}#3\underline{\}}}\:}
\begin{document}
\title{Calculational Design of Hyperlogics by Abstract Interpretation}

\author{Patrick Cousot}
\email{pcousot@cims.nyu.edu}
\orcid{0003-0101-9953}
\author{Jeffery Wang}
\email{cw3736@nyu.edu}
\orcid{0009-0003-4885-8268}
\affiliation{%
  \institution{CS, Courant Institute of Mathematical Studies, NYU}
  \streetaddress{60 Vth Ave}
  \city{New York}
  \state{NY}
  \country{USA}
  \postcode{10011-8868}
}

\renewcommand{\shortauthors}{P.\ Cousot and J.\ Wang}

\begin{abstract}
We design various logics for proving hyper properties of iterative programs by application of 
abstract interpretation principles.

In part I, we design a generic, structural, fixpoint abstract interpreter parameterized by an algebraic abstract domain describing finite and infinite computations that can be instantiated for various operational, denotational, or relational program semantics. Considering  semantics as program properties, we define a \textsf{post} algebraic transformer for execution properties (e.g.\ sets of traces) and a \textsf{Post} algebraic transformer for semantic (hyper) properties (e.g.\ sets of sets of traces), we provide corresponding calculuses as instances of the generic abstract interpreter, and we derive under and over approximation hyperlogics.

In part II, we define exact and approximate semantic abstractions, and show that they preserve the mathematical structure of the algebraic semantics, the collecting semantics \textsf{post}, the hyper collecting semantics \textsf{Post}, and the hyperlogics.

Since proofs by sound and complete hyperlogics require an exact characterization of the program semantics within the proof, we consider in part III abstractions of the  (hyper) semantic properties that yield simplified proof rules. These abstractions include the join, the homomorphic, the elimination, the principal ideal, the order ideal, the frontier order ideal, and the chain limit algebraic abstractions, as well as their combinations, that lead to new algebraic generalizations of hyperlogics, including the ${\forall}{\exists}^\ast$, ${\forall}{\forall}^\ast$, and $\exists\forall^\ast$ hyperlogics.
\end{abstract}


%
\maketitle

\nocite{DBLP:journals/pacmpl/AntonopoulosKLNNN23}

\section{Introduction}
Program (hyper) logics provide methods for reasoning about (sets of) program executions as defined by a semantics. For example, 
hyperproperties were defined by Michael Clarkson and Fred Schneider on execution traces \cite{DBLP:journals/jcs/ClarksonS10} but more recent proposals consider relational logics. We aim at designing program (hyper) logics independently of a specific program semantics, and, more precisely, independently of the formal representation of program executions used by these semantics. 

In \hyperlink{PARTI}{part I}, we recall elements of set and order theories (sect.\@ \ref{sect:Order:Theory}) and then define a structural fixpoint \emph{algebraic program semantics} (sect.\@ \ref{sec:definition:abstract:semantics}) which is an abstract interpreter parameterized by an \emph{algebraic abstract domain}
(sect.\@ \ref{sec:Algebraic-Abstract-Domain}) defined axiomatically. The abstract domain includes terminating and nonterminating executions and can be instantiated to various data and execution models such as the classic relational semantics
(sect.\@ \ref{sect:RelationalSemantics}) or the trace semantics corresponding to the original definition of hyperproperties \cite{DBLP:journals/jcs/ClarksonS10} (sect.\@ \ref{sect:Trace-Semantics} in the appendix) . Then in sect.\@ \ref{sec:Algebraic-Program-Execution-Properties}, we define an \emph{execution collecting semantics} (e.g.\ sets of traces i.e.\ trace properties) and introduce a sound and complete calculus \textsf{post} of execution properties. In sect.\@ \ref{sec:Calculus:Hyper:Properties}, we define a \emph{semantic collecting semantics} (e.g.\ sets of sets of traces i.e.\ hyperproperties) and introduce a structural, fixpoint, sound, and complete calculus \textsf{Post} of semantics properties. In sect.\@ \ref{sec:Abstract-Logic-Semantic-Properties}, we define \emph{upper and lower semantic logics} (e.g.\ a logic for trace hyperproperties) and derive over and under \emph{sound and complete proof systems} by calculational design.

In \hyperlink{PARTII}{part II}, we define the abstraction of the structural algebraic program semantics (sect.\@ \ref{sec:AbstractionAbstractSemantics}) and show that
it induces an abstraction of the algebraic execution collecting semantics (sect.\@ \ref{sec:AbstractionExecutionTransformer}), the algebraic semantic collecting semantics (sect.\@ \ref{sec:AbstractionSemanticTransformer}), and the algebraic upper and lower logics (sect.\@ \ref{sec:AbstractionAbstractLogics}). Such abstractions preserve the mathematical 
structure of the algebraic semantic, collecting semantics, and logics in the abstract. This shows that the algebraic semantics, collecting semantics, and logics can be instantiated to any one in the \emph{hierarchies of semantics} considered e.g.\ in \cite{DBLP:journals/jacm/AptP86,DBLP:journals/tcs/Cousot02,DBLP:journals/tcs/GiacobazziM05}.

Hyperlogics are under or over approximations of semantic properties that is sets of semantics. A program semantics satisfies a hyperproperty if and only if it appears \emph{exactly} in the hyperproperty.  It follows that proofs by semantic logics (for hyperproperties) require, for completeness, to describe the program semantics exactly in the proof. By analogy with Hoare logic, this would require the loop invariants to be the strongest, which is an extreme requirement. 

This is why, in part III, we  consider abstractions of semantic properties, which are less general, but otherwise offer adequate representations of semantic properties and/or allow for much simplified proof rules, closer to the tradition of classic program execution logics, and complete for well identified classes of \emph{abstract semantic properties}. The classic \emph{join abstraction} (sect.\@ \ref{sec:SemanticExecutionPropertyAbstraction}), \emph{homomorphic abstraction} (sect.\@ \ref{sec:HomomorphicSemanticAbstraction}), and \emph{intersection abstraction} (sect.\@ \ref{sec:ExecutionPropertyElimination})
yield simplified proof rules for hyperlogics. The \emph{principal ideal} (sect.\@ \ref{sec:Principal-Ideal-Abstraction}), \emph{order ideal} (sect.\@ \ref{sec:Order-Ideal-Abstraction}), \emph{frontiers} (sect.\@  \ref{sec:FrontierAbstraction}), \emph{chain limit} (sect.\@ \ref{ChainLimitAbstraction}), \emph{chain limit order ideal} (sect.\@ \ref{sec:ChainLimitOrderIdealAbstraction}) abstractions are more
specific to hyperproperties. They are compared in sect.\@ \ref{sect:ComparingAbstractions}. These abstraction generalize known hyperlogics for the algebraic semantics and allow us to provide new
sound and complete proof rules, including for ${\forall}{\exists}$ (sect.\@ \ref{sec:GeneralizedForallExists}), ${\forall}{\forall}$ (sect.\@ \ref{sec:generalized:forall:forall}), and $\exists\forall$ (sect.\@ \ref{GeneralizedExistsForall}) (hyper)properties.. This last case is based on conjunctive abstractions (i.e. conjunctions in logics or reduced products in static analysis) studied in sect.\@ \ref{sec:ConjunctiveAbstraction} of the appendix).

We finally briefly refer to the related works (already cited extensively in the text) in sect.\@ \ref{sect:RelatedWork} and summarize our contributions in the conclusion which also proposes future work (sect.\@ \ref{sect:Conclusion}). \ifshort When clickable, the symbol \proofinapxsymbol\ links to proofs and additional developments in the appendix. The paper together with its appendix is available in the auxiliary material.\fi

\begin{flushleft}
\bigskip
\hypertarget{PARTI}{\Large\textbf{\textsc{Part I:\ Algebraic Semantics, Execution Properties, Semantic}}}\\
{\Large\textbf{\textsc{\phantom{Part I:\ }(Hyper) Properties, Calculi, and Logics}}}
\end{flushleft}
\section{Elements of Set and Order Theories}\label{sect:Order:Theory}
\subsection{Partially Ordered Sets}
\begin{definition}[Properties of posets]\label{def:Properties:posets}
Let $\pair{\mathbb{L}}{\sqsubseteq}$ be a poset with  partially defined least upper bound (lub or join) $\sqcup$, greatest lower bound (glb or meet) $\sqcap$, infimum $\bot$, and supremum $\top$, if any.  \cite{DBLP:books/daglib/0023601}.
\begin{enumerate}[leftmargin=*,itemsep=3pt,label={\itshape \roman*}.,ref=\ref{def:Properties:posets}.{\roman*},labelsep=0.75em]
\item\label{def:Properties:semilattice} $\triple{L}{\sqsubseteq}{\sqcup}$ is a \emph{join semilattice} when the least upper bound (lub, join) $\sqcup S$ exists for any non-empty finite subset  $S\in\wp(L)\setminus\{\emptyset\}$ of $L$. If it exists, the infimum is $\bot=\sqcup\emptyset$. The dual 
is a \emph{meet semilattice} with greatest lower bound (glb, meet) $\sqcap$ and supremum $\top=\sqcup L$, if it exists. A \emph{lattice} is both a join and meet semilattice.  By \emph{limit} we mean either the join or the meet.

\item\label{def:Properties:chain-complete} A poset is \emph{increasing chain complete} if and only if every nonempty increasing chain of $L$ has a lub.  It is \emph{decreasing chain complete} if and only if every nonempty decreasing chain of $L$ has a glb\footnote{We do not respectively use the classic \emph{CPO} and \emph{dual CPO} for which chains are usually restricted to be of length $\omega$.}. It is \emph{chain complete} if both 
increasing and decreasing chain complete.

\item\label{def:Properties:complete-lattice} A poset is a \emph{complete lattice} if and only if any subset, including the empty set, has a lub (hence a glb and the infimum and supremum do exist). \end{enumerate}
\end{definition}
Observe that (\ref{def:Properties:semilattice}) and (\ref{def:Properties:chain-complete}) are independent (i.e.\ none implies the other). We often use them simultaneously. For example, in a \emph{increasing chain-complete join semilattice}, lubs exist for non-empty finite sets and non-empty increasing chains. 

\subsection{Ordinals}
We let $\mathbb{O}$  = $\{0,1,2,\ldots,$ $\omega,\omega+1,\omega+2,\ldots,$ $\omega\times 2,\omega\times 2+1,\omega\times 2+2,\ldots,$ $\omega\times 3,\ldots,$ $\omega\times\omega=\omega^2,\ldots,$ $\omega^\omega,\ldots,$ $\omega^{\omega^\omega},$ $\ldots,$ $\omega^{\left.\begin{array}{@{}c@{}}\scriptstyle\omega^{\udots^\omega}\end{array}\right\}\omega\mbox{\tiny\ times}},\ldots\}$ be the class of ordinals where $\omega$ is the first infinite limit ordinal \cite{Monk-Set-Theory}. $\pair{\mathbb{O}}{\leqslant}$  extends the order on the naturals $\pair{\mathbb{N}}{\leqslant}$ into the infinite. Ordinals yield typical examples of well-orderings (such that any two elements are comparable and any $<$-strictly decreasing chain is finite). Any well-ordering is order-isomorphic to an ordinal (called its rank e.g. $\omega$ for $\mathbb{N}$), \cite[th.\ 13.10 \& 13.11]{Monk-Set-Theory}. We use Von Neumann definition of ordinals \cite[ch.\ 2]{Monk-Set-Theory} with $0=\emptyset$, the successor is $\delta+1=\delta\cup\{\delta\}$, $<$ is $\in$, $\lambda=\bigcup_{\beta<\lambda}\beta$ for infinite  limit ordinals $\lambda$ (which are not a successor ordinal such as $\omega$, $\omega^2$, etc), and the corresponding transfinite induction \cite[Sec. 10]{Monk-Set-Theory}, $P(0)$, $\forall \delta\in\mathbb{O}\mathrel{.}P(\delta)\Rightarrow P(\delta+1)$, and for all limit ordinals $\lambda\in\mathbb{O}$, $(\forall \beta<\lambda\mathrel{.}P(\beta))\Rightarrow P(\lambda)$ implies $\forall\delta\in\mathbb{O}\mathrel{.}P(\delta)$.
 
\subsection{Functions on Partially Ordered Sets}

\begin{definition}[Properties of functions on posets]\label{def:Properties:functions:posets}
Let $\pair{L}{\sqsubseteq}$ be a poset and $f\in L\rightarrow L$. 
\begin{enumerate}[leftmargin=*,itemsep=3pt,label={\itshape \roman*}.,ref=\ref{def:Properties:functions:posets}.{\roman*},labelsep=0.75em]

\item\label{def:Properties:increasing} $f$ is \emph{increasing} (sometimes referred to as \emph{monotone} or \emph{isotone}) means that $\forall x,y\in L\mathrel{.}(x\sqsubseteq y)\Rightarrow (f(x)\sqsubseteq f(y))$. ``Increasing'' is order self-dual.  \emph{Decreasing} (or \emph{antitone}) is $\forall x,y\in L\mathrel{.}(x\sqsubseteq y)\Rightarrow (f(y)\sqsubseteq f(x))$;

\hskip1.5em For example, a sequence $\pair{X^\delta\in L}{\delta<\lambda}$ for ordinals $\delta,\lambda\in\mathbb{O}$ is an {increasing chain} means that $\forall \delta\leqslant\delta'<\lambda\mathrel{.}X^\delta\sqsubseteq X^{\delta'}$. A {decreasing chain} has $\forall \delta\leqslant\delta'<\lambda\mathrel{.}X^{\delta'}\sqsubseteq X^{\delta}$;

\item\label{def:Properties:finite-preserving} Function $f$ is \emph{existing finite join-preserving} (also written \emph{existing finite $\sqcup$-preserving}) if and only if for any non-empty finite set $S\in\wp_f(L)\setminus\{\emptyset\}$ such that $\sqcup S$ exists in $L$ then $\sqcup f(S)$
exists in $L$ and $f(\sqcup S)=\sqcup f(S)$  with $f(S)=\{f(x)\mid x\in S\}$, and dually for meets. $f$ is \emph{existing finite  limit-preserving} if and only if it is both existing finite join and meet preserving. ``Existing'' can be omitted in a lattice; 

\item\label{def:Properties:continuous}$f$ is \emph{upper-continuous} (or existing increasing chain join-preserving) if and only if
for any non-empty increasing chain $S\in\wp_f(L)$ such that $\sqcup S$ exists in $L$, then $\sqcup f(S)$ exists in $L$ such that
$f(\sqcup S)=\sqcup f(S)$. The dual is \emph{lower-continuous} for existing decreasing chain meet-preserving, and \emph{continuous} means both lower and upper continuous. By Scott-Kleene theorem, continuity ensures that functions reach fixpoints iteratively at $\omega$ \cite[th.\@ 15.36]{Cousot-PAI-2021}. This condition for \emph{convergence at $\omega$} is sufficient but not necessary e.g.\ \cite[th.\@ 15.21]{Cousot-PAI-2021};

\item\label{def:Properties:limit-preserving} $f$ is \emph{existing join-preserving} (also written \emph{existing $\sqcup$-preserving}) if and only if for any non-empty  set $S\in\wp(L)\setminus\{\emptyset\}$ such that $\sqcup S$ exists in $L$, then $\sqcup f(S)$ exists in $L$ such that
$f(\sqcup S)=\sqcup f(S)$, and dually for meets. $f$ is \emph{existing  limit-preserving} if and only if it is both existing join and meet preserving. ``Existing'' can be omitted in a complete lattice; 

\item\label{def:Properties:increasing:both:parameters} The definitions \ref{def:Properties:finite-preserving} to \ref{def:Properties:limit-preserving} are extended to $f\in (L\times L)\rightarrow L$ by $f$ has \emph{left limit property} if and only if $\forall y\in L\mathrel{.}\LAMBDA{x}f(x,y)$ has that limit property and $f$ has \emph{right limit property} whenever
$\forall x\in L\mathrel{.}\LAMBDA{y}f(x,y)$ has that limit property. $f$ has that the limit  property \emph{in both parameters} if and only if $f$ has both of the left and right limit properties;

\item\label{def:Properties:increasing:strict} When extending the definitions  \ref{def:Properties:finite-preserving} to \ref{def:Properties:increasing:both:parameters} to empty sets or chains, the function $f$ is then said to be \emph{lower strict}, dually \emph{upper strict}, and \emph{strict} for both cases.

\end{enumerate}
\end{definition}
Observe that \ref{def:Properties:increasing} $\Leftarrow$ \ref{def:Properties:finite-preserving} $\Leftarrow$ \ref{def:Properties:continuous}  $\Leftarrow$ \ref{def:Properties:limit-preserving}.

\subsection{Fixpoints}

Let $f\in \mathbb{L}\increasingfunctionto \mathbb{L}$ be an increasing function on a poset $\pair{\mathbb{L}}{\sqsubseteq}$. There are essentially two classic characterizations of the least fixpoint $\Lfp{\sqsubseteq}f$ of $f$ (we also use their order duals). 
\begin{proposition}[Fixpoint]\label{prop:Tarski}$\Lfp{\sqsubseteq}f=\bigsqcap\{x\mid f(x)\sqsubseteq x\}$ by \textup{\cite{Tarski-fixpoint}} on complete lattices which also holds on increasing chain complete posets \textup{\cite{Escardo03-TarskiDCPO}}.
\end{proposition}
\begin{proposition}[Iteration to fixpoint]\label{prop:Tarski:constructive}
\hskip1em If\/ $\quadruple{\mathbb{L}}{\sqsubseteq}{\bot}{\sqcup}$ is a poset with infimum $\bot$ and partially defined join $\sqcup$ then the iterates $\pair{X^{\delta}}{\delta\in\mathbb{O}}$ of $f$ are partially defined as $X^{\delta+1}\triangleq f(X^{\delta})$, and $X^{\lambda}\triangleq\bigsqcup_{\beta<\lambda}X^{\beta}$ for limit ordinals $\lambda$ (hence $X^0=\bigsqcup\emptyset=\bot$ for limit ordinal $0$). They are well defined when $f$ is increasing (hence when it is finite join preserving, upper-continuous or existing join-preserving) and $\quadruple{\mathbb{L}}{\sqsubseteq}{\bot}{\sqcup}$ is an increasing chain complete poset (hence when it is a complete lattice) in which case they form an increasing chain (i.e.\ $\forall \beta<\delta\in\mathbb{O}\mathrel{.}{X^{\beta}}\sqsubseteq{X^{\delta}}$) ultimately stationary at the limit\/ $\exists\epsilon\mathrel{.}\forall\beta\geqslant\epsilon\mathrel{.}X^\beta=\Lfp{\sqsubseteq}f$ \textup{\cite{CousotCousot-PJM-82-1-1979}}. In case  $f$ is upper-continuous (hence when preserving existing joins), the iterates are stationary at $\epsilon=\omega$ so that the iterates may be restricted to $\mathbb{N}$ and $\Lfp{\sqsubseteq}f=\bigsqcup_{n\in\mathbb{N}}X^{n}$ \textup{\cite[page 305]{Tarski-fixpoint}}.
\end{proposition}

\subsection{Galois Connections, Retractions, and Isomorphisms}\label{sec:Galois:Connections:Retractions}
Galois connections are used throughout the paper either to formalize correspondances between transformers or to formalize exact or approximate abstractions. 
%
Formally,
a Galois connection $\pair{C}{\sqsubseteq}\galois{\alpha}{\gamma}\pair{A}{\preceq}$ is a pair $\pair{\alpha}{\gamma}$ of functions between posets 
$\pair{C}{\sqsubseteq}$ and $\pair{A}{\preceq}$ satisfying $\forall x\in C\mathrel{.}\forall y\in A\mathrel{.}\alpha(x)\preceq y\Leftrightarrow x\sqsubseteq\gamma(y)$. We use a double headed arrow \rlap{$\mskip9.5mu{\rightarrow}$}$\longrightarrow$ to indicate surjection in Galois retractions and $\GaloiS{}{}$ for bijections. We use classic properties of Galois connections which proofs are found in \cite{DeneckeErneWismath-GC-03}.

\subsection{Closures}
We let $\mathbb{1}$ be the identity function. An upper closure operator $\rho$ on ${\mathbb{L}}$ is increasing, extensive and idempotent so
$\pair{\mathbb{L}}{\sqsubseteq}\galoiS{\rho}{\mathbb{1}}\pair{\rho(\mathbb{L})}{\sqsubseteq}$ where $\rho(X)\triangleq\{\rho(x)\mid x\in X\}$ is the post image (dually, a lower closure operator is reductive). It follows that $\rho$ preserves existing arbitrary joins so
if $\quadruple{\mathbb{L}}{\sqsubseteq}{\bot}{\sqcup}$ is an increasing chain complete poset (respectively complete lattice $\sextuple{\mathbb{L}}{\sqsubseteq}{\bot}{\top}{\sqcup}{\sqcap}$) then $\pair{\rho(\mathbb{L})}{\sqsubseteq}$ has the same structure with 
infimum $\rho(\bot)$, join $\LAMBDA{X}\rho(\sqcup X)$, meet $\sqcap$ and top $\top$, if any. In case of a complete lattice this is  Morgan Ward's \cite[th.\@ 4.1]{Ward42}. If $\rho_1$ and $\rho_2$ are upper closures on $\mathbb{L}$ then $\rho_1\comp\rho_2$ and $\rho_2\comp\rho_1$ are upper closure operators on $\mathbb{L}$ if and only if $\rho_1$ and $\rho_2$ are commuting (i.e.\ $\rho_1\comp\rho_2=\rho_2\comp\rho_1$) in which case
$\rho_1\comp\rho_2(\mathbb{L})=\rho_2\comp\rho_1(\mathbb{L})=\rho_1(\mathbb{L})\cap\rho_2(\mathbb{L})$ \cite[p. 525]{Ore-Combinations-1943}.

\section{Algebraic Semantics}\label{sect:Algebraic-Semantics}
We introduce the syntax and algebraic semantics of a simple iterative language based on an abstract domain that generalizes \cite[Ch.\ 21]{Cousot-PAI-2021} to include infinite program behaviors. The algebraic semantics  is reminiscent of \cite{DBLP:journals/cacm/IanovF58,DBLP:journals/jacm/GoguenTWW77,DBLP:conf/mfcs/CourcelleN78,DBLP:conf/ershov/Ershov79,DBLP:conf/ifip/Nivat80,DBLP:journals/toplas/BroyWP87,DBLP:journals/cacm/HoareHJMRSSSS87,DBLP:conf/mfcs/Guessarian78,DBLP:conf/birthday/Hoare13a,DBLP:journals/scp/HoareS14} and others. Such algebraic semantics are a basis for studying a hierarchy of program properties independently of the data manipulated by programs.

\subsection{Syntax}
We consider an imperative language $\mathbb{S}$ with assignments, sequential composition, conditionals, and conditional iteration with breaks. The syntax is $\texttt{S}\in\mathbb{S}\mathbin{{:}{:}{=}}\texttt{x = A}
\mid\texttt{x = [$a$,$b$]}
\mid\texttt{skip}
\mid\texttt{S;S}
\mid\texttt{if (B) S else S}
\mid\texttt{while (B) S}
\mid\texttt{break}
$. \texttt{A} is an arithmetic expression.  The nondeterministic assignment \texttt{x = [$a$, $b$]} with $a\in\mathbb{Z}\cup\{-\infty\}$ and $b\in\mathbb{Z}\cup\{\infty\}$, $-\infty-1=-\infty$, $\infty+1=\infty$ (or any, possibly unbounded, order isomorphic set). 
The Boolean expressions \texttt{B} include the negation $\neg\texttt{B}$. A \texttt{break} exits the closest enclosing loop (which existence is to be checked syntactically).

\subsection{Structural Definitions}\label{sec:Structural-Definitions}
Let $\lhd$ be the ``immediate strict syntactic component'' well-founded partial order on statements $\mathbb{S}$ such that $\texttt{S$_1$} \lhd \texttt{S$_1$;S$_2$}$,\ \ 
$\texttt{S$_2$} \lhd \texttt{S$_1$;S$_2$}$,\ \ 
$\texttt{S$_1$} \lhd \texttt{if (B) S$_1$ else S$_2$}$,\ \ 
$\texttt{S$_2$} \lhd \texttt{if (B) S$_1$ else S$_2$}$,\ \ 
$\texttt{S} \lhd \texttt{while (B) S}$, and is otherwise false. 

Given a nonempty set $\mathcal{V}$, the function $f\in\mathbb{S}\rightarrow\mathcal{V}$ has a structural definition if and only if
$f(\texttt{S})\in\mathcal{V}$ for basic commands (defined as minimal elements of $\lhd$) and, otherwise, is of the form
$f(\texttt{S})=F_{\texttt{\footnotesize S}}(\{\pair{\texttt{S}'}{f(\texttt{S}')}\mid\texttt{S}'\lhd\texttt{S}\})$
where $F_{\texttt{\footnotesize S}}\in \{\pair{\texttt{S}'}{v'}\mid\texttt{S}'\lhd\texttt{S}\wedge v'\in \mathcal{V}\} \rightarrow \mathcal{V}$ is a total function.
Denotational semantics, Hoare logic, predicate transformers, and the abstract semantics of sect.\@ \ref{sec:definition:abstract:semantics} all have structural definitions (called ``compositional'' in denotational semantics).

\subsection{Algebraic Computational Domain}\label{sec:Algebraic-Abstract-Domain}

We consider computational domains $\mathbb{D}^{\sharp}_{+}$ and $\mathbb{D}^{\sharp}_{\infty}$ to be abstract domains respectively abstracting the finite and infinite computations of statements and partially ordered by the respective computational orderings ${\sqsubseteq_{+}^{\sharp}}$ and ${\sqsubseteq_{\infty}^{\sharp}}$, as follows (${{\fatsemi}^{\sharp}}$ is polymorphic).
\begin{eqntabular}[fl]{r@{\hskip1ex}c@{\hskip1ex}l@{\quad}}
\mathbb{D}^{\sharp}_{+}&\triangleq&\undecatuple{\mathbb{L}^{\sharp}_{+}}{\sqsubseteq_{+}^{\sharp}}{\bot_{+}^{\sharp}}{\sqcup_{+}^{\sharp}}{\textsf{init}^{\sharp}}{\textsf{assign}^{\sharp}\sqb{\texttt{x},\texttt{A}}}{\textsf{rassign}^{\sharp}\sqb{\texttt{x},a,b}}{\textsf{test}^{\sharp}\sqb{\texttt{B}}}{\textsf{break}^{\sharp}}{\textsf{skip}^{\sharp}}{{\fatsemi}^{\sharp}}\label{eq:def:generic:finite:abstract-domain}
\\
\mathbb{D}^{\sharp}_{\infty}&\triangleq&\quintuple{\mathbb{L}^{\sharp}_{\infty}}{\sqsubseteq_{\infty}^{\sharp}}{\top_{\infty}^{\sharp}}{\sqcap_{\infty}^{\sharp}}{{\fatsemi}^{\sharp}}\label{eq:def:generic:infinite:abstract-domain}
\end{eqntabular}
\begin{example}\label{ex:bi-inductive-trace-seamntics}Bi-inductive definitions \cite{DBLP:conf/popl/CousotC92} are used in \cite{DBLP:journals/tcs/Cousot02} 
to define a trace semantics on states $\Sigma$ which can be isomorphically decomposed into the domain of finite traces $\quadruple{\mathbb{L}^{\sharp}_{+}}{\sqsubseteq_{+}^{\sharp}}{\bot_{+}^{\sharp}}{\sqcup_{+}^{\sharp}}$
= $\quadruple{\wp(\Sigma^{\ast})}{\subseteq}{\emptyset}{\cup}$ (where $\cup$ is the lub of increasing chains starting form $\emptyset$ for least fixpoints) and the domain of infinite traces $\quadruple{\mathbb{L}^{\sharp}_{\infty}}{\sqsubseteq_{\infty}^{\sharp}}{\top_{\infty}^{\sharp}}{\sqcap_{\infty}^{\sharp}}$ = $\quadruple{\wp(\Sigma^{\omega})}{\subseteq}{\Sigma^{\omega}}{\cap}$ (where $\cap$ is the glb of decreasing chains starting form ${\Sigma^{\omega}}$ for greatest fixpoints),
which abstractions yield a hierarchy of classic semantics, including Hoare logic. 

Our objective in \hyperlink{PARTI}{part I} is to study hyperlogics abstracting away from a particular semantics thus allowing for multiple instantiations (such as traces in sect.\@ \ref{sect:Trace-Semantics}) and, in \hyperlink{PARTII}{part II}, for multiple  abstractions (which include Hoare logic).

A single domain $\mathbb{D}^{\sharp}\triangleq\mathbb{D}^{\sharp}_{+}\cup\mathbb{D}^{\sharp}_{\infty}$ is used in denotational semantics 
\cite{ScottStrachey71-PRG6,DBLP:journals/siamcomp/Plotkin76} but this is not always possible e.g.\ when $\mathbb{D}^{\sharp}_{+}\cap\mathbb{D}^{\sharp}_{\infty}\neq\emptyset$. Moreover the separation into two different domains for finite and infinite executions allows e.g.\ for the use of input-output relations for finite behaviors and traces for infinite behaviors. (see also the discussion in remark \ref{rem:least-versus-greatest-in-semantics} in the appendix.)
\end{example}

\begin{definition}[Abstract domain well-definedness]\label{def:abstract:domain:well:def}We say that ${\mathbb{D}^{\sharp}}\triangleq\pair{\mathbb{D}^{\sharp}_{+}}{\mathbb{D}^{\sharp}_{\infty}}$ is a well-defined chain-complete lattice (respectively complete lattice) with increasing (respectively  finite limit-preserving, continuous, and existing limit-preserving) composition, if and only if 
\begin{enumerate}[leftmargin=*,itemsep=3pt,label={\Alph*}.,ref=\ref{def:abstract:domain:well:def}.{\Alph*},labelsep=0.75em]

\item \label{def:abstract:domain:well:def:finite:domain}The finitary calculational domain
$\quadruple{\mathbb{L}^{\sharp}_{+}}{\sqsubseteq_{+}^{\sharp}}{\bot_{+}^{\sharp}}{\sqcup_{+}^{\sharp}}$ is an increasing chain-complete join semilattice with infimum, (respectively $\sextuple{\mathbb{L}^{\sharp}_{+}}{\sqsubseteq_{+}^{\sharp}}{\bot_{+}^{\sharp}}{\top_{+}^{\sharp}}{\sqcup_{+}^{\sharp}}{\sqcap_{+}^{\sharp}}$ is a complete lattice);

\item  \label{def:abstract:domain:well:def:abstract:operators}
${\textsf{init}^{\sharp}}$, ${\textsf{break}^{\sharp}}$, ${\textsf{skip}^{\sharp}}$ $\in$ ${\mathbb{L}^{\sharp}_{+}}$, ${\textsf{assign}^{\sharp}\sqb{\texttt{x},\texttt{A}}}$, ${\textsf{rassign}^{\sharp}\sqb{\texttt{x},a,b}}$, ${\textsf{test}^{\sharp}\sqb{\texttt{B}}}$ $\in$ ${\mathbb{L}^{\sharp}_{+}}$ are well-defined in ${\mathbb{L}^{\sharp}_{+}}$;

\item \label{def:abstract:domain:well:def:infinite:domain}The infinitary calculational domain
$\quintuple{\mathbb{L}^{\sharp}_{\infty}}{\sqsubseteq_{\infty}^{\sharp}}{\top_{\infty}^{\sharp}}{\sqcup_{\infty}^{\sharp}}{\sqcap_{\infty}^{\sharp}}$ is a decreasing chain-complete join lattice with supremum
(respectively $\sextuple{\mathbb{L}^{\sharp}_{\infty}}{\sqsubseteq_{\infty}^{\sharp}}{\bot_{\infty}^{\sharp}}{\top_{\infty}^{\sharp}}{\sqcup_{\infty}^{\sharp}}{\sqcap_{\infty}^{\sharp}}$ is a complete lattice);

\item \label{def:abstract:domain:well:def:operators}
The sequential composition $\mathbin{{\fatsemi}^{\sharp}}\in \bigl({\mathbb{L}^{\sharp}_{+}}\times{\mathbb{L}^{\sharp}_{+}}\functionto{\mathbb{L}^{\sharp}_{+}}\bigr)\cup
\bigl(
(({\mathbb{L}^{\sharp}_{+}}\times{\mathbb{L}^{\sharp}_{\infty}})
\cup
({\mathbb{L}^{\sharp}_{\infty}}\times{\mathbb{L}^{\sharp}_{+}})
\cup
({\mathbb{L}^{\sharp}_{\infty}}\times{\mathbb{L}^{\sharp}_{\infty}})
)
\functionto{\mathbb{L}^{\sharp}_{\infty}}\bigr)$ is associative and satisfies the following conditions (where $\sextuple{\mathbb{L}^{\sharp}_{x}}{\sqsubseteq_{x}^{\sharp}}{\bot_{x}^{\sharp}}{\top_{x}^{\sharp}}{\sqcup_{x}^{\sharp}}{\sqcap_{x}^{\sharp}}$, $x\in\{+,\infty\}$
designates $\sextuple{\mathbb{L}^{\sharp}_{+}}{\sqsubseteq_{+}^{\sharp}}{\bot_{+}^{\sharp}}{\top_{+}^{\sharp}}{\sqcup_{+}^{\sharp}}{\sqcap_{+}^{\sharp}}$ when $x=+$ and $\sextuple{\mathbb{L}^{\sharp}_{\infty}}{\sqsubseteq_{\infty}^{\sharp}}{\bot_{\infty}^{\sharp}}{\top_{\infty}^{\sharp}}{\sqcup_{\infty}^{\sharp}}{\sqcap_{\infty}^{\sharp}}$ when $x=\infty$).
\vskip5pt%
\begin{enumerate}[leftmargin=*,itemsep=3pt,label={\alph*}.,ref=\ref{def:abstract:domain:well:def:operators}.{\alph*},labelsep=0.75em]

\item \label{def:abstract:domain:well:def:init:neutral}
$\forall S\in{\mathbb{L}^{\sharp}_{+}}\mathrel{.} S\mathbin{{\fatsemi}^{\sharp}}{\textsf{init}^{\sharp}}={\textsf{init}^{\sharp}}\mathbin{{\fatsemi}^{\sharp}}S=S$;

\item \label{def:abstract:domain:well:def:bot:absorbent}
$\forall S\in{\mathbb{L}^{\sharp}_{+}}\mathrel{.} S\mathbin{{\fatsemi}^{\sharp}}{\bot_{+}^{\sharp}}={\bot_{+}^{\sharp}}$ and $\forall S\in{\mathbb{L}^{\sharp}_{x}}\mathrel{.} {\bot_{+}^{\sharp}} \mathbin{{\fatsemi}^{\sharp}}S={\bot_{+}^{\sharp}}$ (same for ${\mathbb{L}^{\sharp}_{\infty}}$ when ${\bot_{\infty}^{\sharp}}$ exists);

\item \label{def:abstract:domain:well:def:oo:absorbent}
$\forall S\in{\mathbb{L}^{\sharp}_{\infty}}\mathrel{.} \forall S'\in{\mathbb{L}^{\sharp}_{x}}\mathrel{.}S\mathbin{{\fatsemi}^{\sharp}}S'=S$;


\item \label{def:abstract:domain:well:def:join:additive} In its left, right, or both parameters, the sequential composition $\mathbin{{\fatsemi}^{\sharp}}$ is either
\begin{enumerate}[leftmargin=*,itemsep=3pt,label={\roman*}.,ref=\ref{def:abstract:domain:well:def:join:additive}.{\roman*},labelsep=0.75em]

\item\label{def:abstract:domain:well:def:increasing}
{increasing} for ${\sqsubseteq_{+}^{\sharp}}$ and/or ${\sqsubseteq_{\infty}^{\sharp}}$;

\item\label{def:abstract:domain:well:finite-join-preserving}
{finite join preserving} for ${\sqcup_{+}^{\sharp}}$ and/or ${\sqcup_{\infty}^{\sharp}}$;

\item \label{def:abstract:domain:well:continuous} in addition to \ref{def:abstract:domain:well:finite-join-preserving}, is lower continuous for ${\sqcap_{+}^{\sharp}}$  and/or ${\sqcap_{\infty}^{\sharp}}$ and/or upper continuous for ${\sqcup_{+}^{\sharp}}$ and/or ${\sqcup_{\infty}^{\sharp}}$;

\item \label{def:abstract:domain:well:limit-preserving}
existing arbitrary ${\sqcup_{+}^{\sharp}}$-preserving and/or existing arbitrary ${\sqcap_{\infty}^{\sharp}}$-preserving.

\end{enumerate}
\end{enumerate}
\end{enumerate}
\end{definition}
\begin{remark}\label{rem:bi-inductive-definition}In case ${\mathbb{L}^{\sharp}_{+}}\cap{\mathbb{L}^{\sharp}_{\infty}}=\emptyset$, we can define $\mathbb{L}^{\sharp}\triangleq{\mathbb{L}^{\sharp}_{+}}\cup{\mathbb{L}^{\sharp}_{\infty}}$ with $X^{+}\triangleq X\cap{\mathbb{L}^{\sharp}_{+}}$, $X^{\infty}\triangleq X\cap{\mathbb{L}^{\sharp}_{\infty}}$, and
$X\mathrel{{\sqsubseteq}^{\sharp}} Y\triangleq X^{+}\mathrel{{\sqsubseteq}^{\sharp}_{+}} Y^{+}\wedge X^{\infty}\mathrel{{\sqsubseteq}^{\sharp}_{\infty}} Y^{\infty}$ which corresponds to the bi-inductive definitions \cite{DBLP:conf/popl/CousotC92} mentioned in example \ref{ex:bi-inductive-trace-seamntics}.
\end{remark}
\begin{remark}\label{rem:composition-properties}Hypotheses \ref{def:abstract:domain:well:def:abstract:operators}, \ref{def:abstract:domain:well:def:increasing} and \ref{def:abstract:domain:well:finite-join-preserving} determine the precision of the semantic of basic commands, composition, choices, conditionals, and iteration in the algebraic semantics. These hypotheses as well as \ref{def:abstract:domain:well:continuous} and \ref{def:abstract:domain:well:limit-preserving} determine whether fixpoint iterations should be infinite or transfinite (see  proposition \ref{prop:Tarski:constructive}).
\end{remark}

\subsection{Definition of the Algebraic Semantics}\label{sec:definition:abstract:semantics}
The algebraic semantics of statements $\texttt{S}\in\mathbb{S}$ is an abstract property of executions.
The basic commands \texttt{S} are assignment, random assignment, \texttt{break} out of the immediately enclosing loop, and \texttt{skip}, with the following $\sqb{\texttt{S}}_{e}^{\sharp}$ and break $\sqb{\texttt{S}}_{b}^{\sharp}$ finite/ending/terminating semantics in ${\mathbb{L}^{\sharp}_{+}}$ as well as infinite/nonterminating  $\sqb{\texttt{S}}_{\bot}^{\sharp}$ abstract semantics in ${\mathbb{L}^{\sharp}_{\infty}}$. 
\subsubsection{Basic Statements}
\bgroup\arraycolsep=0.475\arraycolsep\begin{eqntabular}[fl]{lcl@{\hskip1.75em}lcl@{\hskip1.75em}lcl}
\sqb{\texttt{x = A}}_{e}^{\sharp}&\triangleq&{\textsf{assign}^{\sharp}\sqb{\texttt{x},\texttt{A}}}
&
\sqb{\texttt{x = A}}_{b}^{\sharp}&\triangleq&{\bot_{+}^{\sharp}}
&
\sqb{\texttt{x = A}}_{\bot}^{\sharp}&\triangleq&{\bot_{\infty}^{\sharp}}
\nonumber\\
\sqb{\texttt{x = [$a$, $b$]}}_{e}^{\sharp}&\triangleq&{\textsf{rassign}^{\sharp}\sqb{\texttt{x},a,b}}
&
\sqb{\texttt{x = [$a$, $b$]}}_{b}^{\sharp}&\triangleq&{\bot_{+}^{\sharp}}
&
\sqb{\texttt{x = [$a$, $b$]}}_{\bot}^{\sharp}&\triangleq&{\bot_{\infty}^{\sharp}}
\nonumber\\
\sqb{\texttt{break}}_{e}^{\sharp}&\triangleq&{\bot_{+}^{\sharp}}
&
\sqb{\texttt{break}}_{b}^{\sharp}&\triangleq&{\textsf{break}^{\sharp}}
&
\sqb{\texttt{break}}_{\bot}^{\sharp}&\triangleq&{\bot_{\infty}^{\sharp}}
\label{eq:def:sem:abstract:basis}\\
\sqb{\texttt{skip}}_{e}^{\sharp}&\triangleq&{\textsf{skip}^{\sharp}}
&
\sqb{\texttt{skip}}_{b}^{\sharp}&\triangleq&{\bot_{+}^{\sharp}}
&
\sqb{\texttt{skip}}_{\bot}^{\sharp}&\triangleq& {\bot_{\infty}^{\sharp}}
\nonumber\\
\sqb{\texttt{B}}_{e}^{\sharp}&\triangleq&{\textsf{test}^{\sharp}}\sqb{\texttt{B}}
&
\sqb{\texttt{B}}_{b}^{\sharp}&\triangleq&{\bot_{+}^{\sharp}}
&
\sqb{\texttt{B}}_{\bot}^{\sharp}&\triangleq& {\bot_{\infty}^{\sharp}}
\nonumber
\end{eqntabular}\egroup
For the assignment \texttt{x = A}, the abstract semantics ${\textsf{assign}^{\sharp}\sqb{\texttt{x},\texttt{A}}}$ is specified by the abstract domain, and so, is well-defined by \ref{def:abstract:domain:well:def:abstract:operators}. $\sqb{\texttt{x = A}}_{b}^{\sharp}={\bot_{+}^{\sharp}}$ because the assignment cannot break. $\sqb{\texttt{x = A}}_{\bot}^{\sharp}={\bot_{\infty}^{\sharp}}$ since the assignment always terminates.
The algebraic semantics of the other primitives is similar, except for
the \texttt{break} statement. $\sqb{\texttt{break}}_{e}^{\sharp}={\bot_{+}^{\sharp}}$ since the \texttt{break} cannot continue in sequence. 
The semantics $\sqb{\texttt{break}}_{b}^{\sharp}$ of the \texttt{break} is given by the abstract domain primitive ${\textsf{break}^{\sharp}}$ which is finite and well-defined. $\sqb{\texttt{break}}_{\bot}^{\sharp}={\bot_{\infty}^{\sharp}}$ since a \texttt{break} always terminates.

\subsubsection{Structural Statements}
For the sequential composition and the conditional where $\sqb{\texttt{B;S}}_{x}^{\sharp} \triangleq {\textsf{test}^{\sharp}\sqb{\texttt{B}}}\mathbin{\fatsemi^{\sharp}}\sqb{\texttt{S}}_{x}^{\sharp}$, $x\in\{e,b,\bot\}$, we  define
\begin{eqntabular}[fl]{@{}lcl@{\quad}lcl}
\sqb{\texttt{S$_1$;S$_2$}}_{e}^{\sharp}&\triangleq&\sqb{\texttt{S$_1$}}_{e}^{\sharp}\mathbin{\,\fatsemi^\sharp\,}\sqb{\texttt{S$_2$}}_{e}^{\sharp}
&
\sqb{\texttt{if (B) S$_1$ else S$_2$}}_{e}^{\sharp}&\triangleq&\sqb{\texttt{B;S}_1}_{e}^{\sharp}\mathbin{\,\sqcup_{+}^{\sharp}\,}\sqb{\neg\texttt{B;S}_2}_{e}^{\sharp}
\nonumber
\\
\sqb{\texttt{S$_1$;S$_2$}}_{b}^{\sharp}&\triangleq&\sqb{\texttt{S$_1$}}_{b}^{\sharp}\mathbin{\,\sqcup_{+}^{\sharp}\,}(\sqb{\texttt{S$_1$}}_{e}^{\sharp}\mathbin{\fatsemi^{\sharp}}\sqb{\texttt{S$_2$}}_{b}^{\sharp})
&
\sqb{\texttt{if (B) S$_1$ else S$_2$}}_{b}^{\sharp}&\triangleq&\sqb{\texttt{B;S}_1}_{b}^{\sharp}\mathbin{\,\sqcup_{+}^{\sharp}\,}\sqb{\neg\texttt{B;S}_2}_{b}^{\sharp}
\label{eq:def:abstract:sem:seq}\label{eq:def:abstract:sem:if}
\\
\sqb{\texttt{S$_1$;S$_2$}}_{\bot}^{\sharp}&\triangleq&
\sqb{\texttt{S$_1$}}_{\bot}^{\sharp} \mathbin{\,\sqcup_{\infty}^{\sharp}\,}(\sqb{\texttt{S$_1$}}_{e}^{\sharp}\mathbin{\fatsemi^{\sharp}}\sqb{\texttt{S$_2$}}_{\bot}^{\sharp})
&
\sqb{\texttt{if (B) S$_1$ else S$_2$}}_{\bot}^{\sharp}&\triangleq&
\sqb{\texttt{B;S}_1}_{\bot}^{\sharp}\mathbin{\sqcup_{\infty}^{\sharp}}\sqb{\neg\texttt{B;S}_2}_{\bot}^{\sharp}
\nonumber\end{eqntabular}
The semantics of the composition and conditional are well-defined by \ref{def:abstract:domain:well:def:operators} for $\fatsemi^\sharp$
and \ref{def:abstract:domain:well:def:finite:domain} and \ref{def:abstract:domain:well:def:infinite:domain} which ensure the existence of the finite and infinite joins.

\texttt{S$_1$;S$_2$} terminates if \texttt{S$_1$} terminates and is followed by \texttt{S$_2$} that terminates. \texttt{S$_1$;S$_2$} breaks
(resp.\ nonterminates) if either \texttt{S$_1$} breaks (resp.\ nonterminates) or \texttt{S$_1$} terminates and is followed by \texttt{S$_2$} that 
breaks (resp.\ nonterminates). 

For a given execution of the conditional \texttt{if (B) S$_1$ else S$_2$} only one branch is taken, so the semantics of the other one will be empty by definition (\ref{eq:def:sem:abstract:basis}) of $\sqb{\texttt{B}}_{e}^{\sharp}$ that should return ${\bot_{+}^{\sharp}}$\footnote{unless the semantics of Boolean expressions is to be very exotic.} and \ref{def:abstract:domain:well:def:bot:absorbent}.
\begin{example}Assume that \texttt{S$_1$} never terminates in that $\sqb{\texttt{S$_1$}}_{\bot}^{\sharp}=\top^{\sharp}_{\infty}$ (sometimes named ``chaos'' modelling all possible nonterminating behaviors). Then, by (\ref{eq:def:abstract:sem:if}),
$\sqb{\texttt{S$_1$;S$_2$}}_{\bot}^{\sharp}$
$\triangleq$
$\sqb{\texttt{S$_1$}}_{\bot}^{\sharp} \mathbin{\,\sqcup_{\infty}^{\sharp}\,}(\sqb{\texttt{S$_1$}}_{e}^{\sharp}\mathbin{\fatsemi^{\sharp}}\sqb{\texttt{S$_2$}}_{\bot}^{\sharp})$
= 
$\top^{\sharp}_{\infty}\mathbin{\,\sqcup_{\infty}^{\sharp}\,}(\sqb{\texttt{S$_1$}}_{e}^{\sharp}\mathbin{\fatsemi^{\sharp}}\sqb{\texttt{S$_2$}}_{\bot}^{\sharp})$
=
$\top^{\sharp}_{\infty}$ meaning that \texttt{S$_1$;S$_2$} never terminates either in chaos.

For the conditional, assume \texttt{B} is always true and \texttt{S$_1$} never terminates in that $\sqb{\texttt{S$_1$}}_{\bot}^{\sharp}=\top^{\sharp}_{\infty}$. Then the false branch is never taken so that $\sqb{\neg\texttt{B;S}_2}_{\bot}^{\sharp}=\bot^{\sharp}_{\infty}$. It follows, by (\ref{eq:def:abstract:sem:if}), that $\sqb{\texttt{if (B) S$_1$ else S$_2$}}_{\bot}^{\sharp}$
$\triangleq$
$\sqb{\texttt{B;S}_1}_{\bot}^{\sharp}\mathbin{\sqcup_{\infty}^{\sharp}}\sqb{\neg\texttt{B;S}_2}_{\bot}^{\sharp}$
=
$\top^{\sharp}_{\infty}\mathbin{\sqcup_{\infty}^{\sharp}}\bot^{\sharp}_{\infty}$
=
$\top^{\sharp}_{\infty}$ so that the conditional \texttt{if (B) S$_1$ else S$_2$} never terminates.
\end{example}
\subsubsection{Iteration}\label{sec:abstract:Iteration}
For iteration \texttt{while (B) S}, we define the transformers
\begin{eqntabular}{L@{\qquad\qquad}rcl}
backward&{\Cev{F}_{e}^{\sharp}}&\triangleq&\LAMBDA{X\in{\mathbb{L}^{\sharp}_{+}}}{\textsf{init}^{\sharp}} \mathbin{\sqcup_{+}^{\sharp}} (\sqb{\texttt{B;S}}_{e}^{\sharp}\mathbin{{\fatsemi}^{\sharp}} X)\label{eq:natural-transformer-finite-backward}\\
forward&{\vec{F}_{e}^{\sharp}}&\triangleq&\LAMBDA{X\in{\mathbb{L}^{\sharp}_{+}}}{\textsf{init}^{\sharp}} \mathbin{\sqcup_{+}^{\sharp}} (X\mathbin{{\fatsemi}^{\sharp}} \sqb{\texttt{B;S}}_{e}^{\sharp})\label{eq:natural-transformer-finite-forward}\\
infinite&{F_{\bot}^{\sharp}}&\triangleq&\LAMBDA{X\in{\mathbb{L}^{\sharp}_{\infty}}}\sqb{\texttt{B;S}}_{e}^{\sharp}\mathbin{{\fatsemi}^{\sharp}} X\label{eq:trace-transformer-infinite}
\end{eqntabular}
\begin{lemma}[Finite fixpoints well-definedness]\label{lem:Fesharp-welldefined}\proofinapx\quad If\/ $\mathbb{D}^{\sharp}_{+}$ is a well-defined increasing chain complete join semilattice and $\mathbin{{\fatsemi}^{\sharp}}$ left satisfies any one of the \ref{def:abstract:domain:well:def:increasing}, \ref{def:abstract:domain:well:finite-join-preserving}, \ref{def:abstract:domain:well:continuous}, or
\ref{def:abstract:domain:well:limit-preserving} properties for $\mathbb{D}^{\sharp}_{+}$ then ${\Cev{F}_{e}^{\sharp}}$ satisfy the same property and its least fixpoint deso exist (and similarly for ${\vec{F}_{e}^{\sharp}}$ when $\mathbin{{\fatsemi}^{\sharp}}$ right satisfies any one of the properties listed in \ref{def:abstract:domain:well:def:join:additive}).
\end{lemma}
\begin{toappendix}
\begin{proof}[Proof of lemma \ref{lem:Fesharp-welldefined}]
By definition (\ref{eq:natural-transformer-finite-backward}), ${\Cev{F}_{e}^{\sharp}}$ is the composition of constants ${\textsf{init}^{\sharp}}$ and $\sqb{\texttt{B;S}}_{e}^{\sharp}\mathbin{{\fatsemi}^{\sharp}}$, the lub $\mathbin{\sqcup_{+}^{\sharp}}$ in a join semilattice (which satisfies all properties of definition \ref{def:Properties:functions:posets}), and sequential composition $\mathbin{{\fatsemi}^{\sharp}}$. Therefore, depending on which property
\ref{def:abstract:domain:well:def:increasing}, \ref{def:abstract:domain:well:finite-join-preserving}, \ref{def:abstract:domain:well:continuous}, or
\ref{def:abstract:domain:well:limit-preserving} does satisfy, ${\Cev{F}_{e}^{\sharp}}$ satisfies the same property. It follows by
\ref{def:abstract:domain:well:def:finite:domain} that the iterates of ${\Cev{F}_{e}^{\sharp}}$ do exist, so that, by proposition \ref{prop:Tarski:constructive}, $\Lfp{\sqsubseteq_{+}^{\sharp}}{\Cev{F}_{e}^{\sharp}}$ does exists. The same way $\Lfp{\sqsubseteq_{+}^{\sharp}}{\vec{F}_{e}^{\sharp}}$ does exist by (\ref{eq:natural-transformer-finite-forward}).
\end{proof}
\end{toappendix}
Let us show that $\Lfp{\sqsubseteq_{+}^{\sharp}}{\Cev{F}_{e}^{\sharp}}$ = $\Lfp{\sqsubseteq_{+}^{\sharp}}{\vec{F}_{e}^{\sharp}}$ inductively defines the set of finite executions reaching the entry of the iteration \texttt{while(B) S} after zero or more terminating body iterations. To see that, we define

\begin{eqntabular}[fl]{@{\qquad}Pt{0.85\textwidth}}
the powers $\pair{X^\delta}{\delta\in\mathbb{O}}$ of $X\in{\mathbb{L}^{\sharp}_{+}}$ are $X^0$ $\triangleq$ ${\textsf{\textup{init}}^{\sharp}}$, 
$X^{\delta+1}$ $\triangleq$ $X\mathbin{{\fatsemi}^{\sharp}}X^{\delta}$ for successor ordinals, and $X^\lambda\triangleq\mathop{\bigsqcup_{+}^{\sharp}}_{\beta<\lambda}X^\beta$ for limit ordinals. 
\label{eq:def:abstract:powers}
\end{eqntabular}
We now characterize the executions of iterations in terms of the fixpoints of the execution transformers \ref{eq:natural-transformer-finite-backward}---\ref{eq:natural-transformer-finite-forward}. We show that $\Lfp{\sqsubseteq_{+}^{\sharp}}{\Cev{F}_{e}^{\sharp}}$ = $\Lfp{\sqsubseteq_{+}^{\sharp}}{\vec{F}_{e}^{\sharp}}$ inductively characterize 0 or more finite iterations of the loop body for which the loop condition holds and the loop body terminates.
\begin{lemma}[Commutativity]\label{lem:abstract:powers:commute}\proofinapx\quad If\/ $\mathbb{D}^{\sharp}_{+}$ is a well-defined complete lattice (resp.\ increasing chain-complete poset) with 
right existing ${\sqcup_{+}^{\sharp}}$-preserving (resp.\ right upper continuous) composition ${{\fatsemi}^{\sharp}}$ and $X\in{\mathbb{L}^{\sharp}_{+}}$  then
$\forall \delta\in\mathbb{O}\mathrel{.}X\mathbin{{\fatsemi}^{\sharp}}X^{\delta} = X^{\delta}\mathbin{{\fatsemi}^{\sharp}}X$ (resp.\ if $\pair{X^{\delta}}{\delta\in\mathbb{O}}$ is an increasing chain).
\end{lemma}
\begin{toappendix}
\begin{proof}[Proof of lemma \ref{lem:abstract:powers:commute}]The proof is by transfinite induction on $\delta$. 
\begin{itemize}[leftmargin=*,nosep]

\item For $\delta=0$, we have
$X\mathbin{{\fatsemi}^{\sharp}}X^{0}$
=
$X\mathbin{{\fatsemi}^{\sharp}}{\textsf{\textup{init}}^{\sharp}}$
=
${\textsf{\textup{init}}^{\sharp}}\mathbin{{\fatsemi}^{\sharp}} X$
=
$X^{0}\mathbin{{\fatsemi}^{\sharp}} X$ by definition \ref{def:abstract:domain:well:def:init:neutral} and definition (\ref{eq:def:abstract:powers}) of the powers.

\item If $X\mathbin{{\fatsemi}^{\sharp}}X^{\delta}$ = $X^{\delta}\mathbin{{\fatsemi}^{\sharp}}X$  by induction hypothesis, then
$X\mathbin{{\fatsemi}^{\sharp}}X^{\delta+1}$ 
=
$X\mathbin{{\fatsemi}^{\sharp}}(X\mathbin{{\fatsemi}^{\sharp}}X^{\delta})$ 
=
$X\mathbin{{\fatsemi}^{\sharp}}(X^{\delta}\mathbin{{\fatsemi}^{\sharp}}X)$ 
=
$(X\mathbin{{\fatsemi}^{\sharp}}X^{\delta})\mathbin{{\fatsemi}^{\sharp}}X$ 
= 
$X^{\delta+1}\mathbin{{\fatsemi}^{\sharp}}X$
by def.\ (\ref{eq:def:abstract:powers}) of the iterates, induction hypothesis, associativity \ref{def:abstract:domain:well:def:operators}, and (\ref{eq:def:abstract:powers}).

\item If $\lambda$ is a limit ordinal and $\forall\beta<\lambda\mathrel{.}X\mathbin{{\fatsemi}^{\sharp}}X^{\beta}$ = $X^{\beta}\mathbin{{\fatsemi}^{\sharp}}X$  by induction hypothesis, then
$X\mathbin{{\fatsemi}^{\sharp}}X^{\lambda}$
=
$X\mathbin{{\fatsemi}^{\sharp}}(\mathop{\bigsqcup_{+}^{\sharp}}_{\beta<\lambda}X^\beta)$
=
$\mathop{\bigsqcup_{+}^{\sharp}}_{\beta<\lambda} (X\mathbin{{\fatsemi}^{\sharp}}X^\beta)$
=
$\mathop{\bigsqcup_{+}^{\sharp}}_{\beta<\lambda} X^{\beta+1}$
by (\ref{eq:def:abstract:powers}), right existing ${\sqcup_{+}^{\sharp}}$-preserving ${{\fatsemi}^{\sharp}}$ \ref{def:abstract:domain:well:limit-preserving} (resp.\ right upper continuity when $\pair{X^{\delta}}{\delta\in\mathbb{O}}$ is an increasing chain \ref{def:abstract:domain:well:continuous}).\qed
\end{itemize}
\let\qed\relax
\end{proof}
\end{toappendix}
\begin{lemma}[Finite body iterations]\label{lem:abstract:Lfp-subseteq-F-e}\proofinapx\quad
If\/ $\mathbb{D}^{\sharp}_{+}$ is a well-defined increasing chain-complete join semilattice with right upper continuous composition ${{\fatsemi}^{\sharp}}$ then
$\Lfp{\sqsubseteq_{+}^{\sharp}}{\Cev{F}_{e}^{\sharp}}$ = $\smash{\mathop{\bigsqcup_{+}^{\sharp}}\limits_{\delta\in\mathbb{O}}(\sqb{\texttt{B;S}}_{e}^{\sharp})^{\delta}}$.
\end{lemma}
\begin{toappendix}
\begin{proof}[Proof of lemma \ref{lem:abstract:Lfp-subseteq-F-e}] By lemma \ref{lem:Fesharp-welldefined}, if\/ $\mathbb{D}^{\sharp}_{+}$ is a well-defined incraesing chain-complete join semilattice with right upper continuous composition then ${\Cev{F}_{e}^{\sharp}}$ in (\ref{eq:natural-transformer-finite-forward}) is upper continuous hence increasing since continuous functions are increasing and the
composition of increasing functions is increasing. It follows, by proposition \ref{prop:Tarski:constructive}, that the least fixpoint $\Lfp{\sqsubseteq_{+}^{\sharp}}{\Cev{F}_{e}^{\sharp}}$ exists and is the limit of the increasing iterates $\pair{X^\delta}{\delta\in\mathbb{O}}$ of ${\Cev{F}_{e}^{\sharp}}$ from the infimum ${\bot_{+}^{\sharp}}$ (which exists in a chain-complete lattice).

Let us prove that $X^{\delta}$ = ${\bigsqcup_{+}^{\sharp}}_{\beta<\delta}(\sqb{\texttt{B;S}}_{e}^{\sharp})^{\beta}$ by transfinite induction on $\delta$. 
\begin{itemize}[leftmargin=*,nosep]

\item For $\delta=0$, we have 
$X^{0}$ = ${\bot_{+}^{\sharp}}$ = ${\bigsqcup_{+}^{\sharp}}\emptyset$ = ${\bigsqcup_{+}^{\sharp}}_{\beta<0}(\sqb{\texttt{B;S}}_{e}^{\sharp})^{\beta}$ by definition of the iterates and the infimum.

\item Assume by induction hypothesis that $X^{\delta}$ = ${\bigsqcup_{+}^{\sharp}}_{\beta<\delta}(\sqb{\texttt{B;S}}_{e}^{\sharp})^{\beta}$. Then 
$X^{\delta+1}$
=
${\Cev{F}_{e}^{\sharp}}(X^{\delta})$ 
=
${\textsf{init}^{\sharp}} \mathbin{\sqcup_{+}^{\sharp}}(\sqb{\texttt{B;S}}_{e}^{\sharp}\mathbin{{\fatsemi}^{\sharp}} X^{\delta})$
=
${\textsf{init}^{\sharp}} \mathbin{\sqcup_{+}^{\sharp}}( \sqb{\texttt{B;S}}_{e}^{\sharp} \mathbin{{\fatsemi}^{\sharp}}({\bigsqcup_{+}^{\sharp}}_{\beta<\delta}(\sqb{\texttt{B;S}}_{e}^{\sharp})^{\beta}) )$
=
${\textsf{init}^{\sharp}} \mathbin{\sqcup_{+}^{\sharp}}{\bigsqcup_{+}^{\sharp}}_{\beta<\delta}((\sqb{\texttt{B;S}}_{e}^{\sharp}\mathbin{{\fatsemi}^{\sharp}}\sqb{\texttt{B;S}}_{e}^{\sharp})^{\beta} )$
=
$(\sqb{\texttt{B;S}}_{e}^{\sharp})^0\mathbin{\sqcup_{+}^{\sharp}}{\bigsqcup_{+}^{\sharp}}_{\beta<\delta}(\sqb{\texttt{B;S}}_{e}^{\sharp})^{\beta+1} )$
=
${\bigsqcup_{+}^{\sharp}}_{\beta<\delta+1}(\sqb{\texttt{B;S}}_{e}^{\sharp})^{\beta}$
by definition of iterates,
definition (\ref{eq:natural-transformer-finite-backward}) of ${\Cev{F}_{e}^{\sharp}}$,
induction hypothesis,
definition \ref{def:abstract:domain:well:def:join:additive},
definition of the powers,
grouping terms in the join.

\item Assume that $\lambda$ is a limit ordinal and that, by induction hypothesis, $\forall\delta<\lambda\mathrel{.}X^{\delta}$ = ${\bigsqcup_{+}^{\sharp}}_{\beta<\delta}(\sqb{\texttt{B;S}}_{e}^{\sharp})^{\beta}$. Then we have
$X^{\lambda}$
=
${\bigsqcup_{+}^{\sharp}}_{\delta<\lambda} X^{\delta}$
=
${\bigsqcup_{+}^{\sharp}}_{\delta<\lambda} {\bigsqcup_{+}^{\sharp}}_{\beta<\delta}(\sqb{\texttt{B;S}}_{e}^{\sharp})^{\beta}$
=
${\bigsqcup_{+}^{\sharp}}_{\beta<\lambda} (\sqb{\texttt{B;S}}_{e}^{\sharp})^{\beta}$
by definition of the iterates, induction hypothesis, and definition of the join ${\sqcup_{+}^{\sharp}}$ (which exists since the iterates are increasing.
\end{itemize}
We conclude by proposition \ref{prop:Tarski:constructive} that $\Lfp{\sqsubseteq_{+}^{\sharp}}{\Cev{F}_{e}^{\sharp}}$ = $\mathop{\bigsqcup_{+}^{\sharp}}\limits_{\delta\in\mathbb{O}}(\sqb{\texttt{B;S}}_{e}^{\sharp})^{\delta}$. 
\end{proof}
\end{toappendix}
\begin{lemma}[Forward versus backward]\label{lem:abstract:forward=backward}\proofinapx\quad
If\/ $\mathbb{D}^{\sharp}$ is a well-defined increasing chain-complete join semilattice with right upper continuous sequential composition $\mathbin{{\fatsemi}^{\sharp}}$ then
$\Lfp{\sqsubseteq_{+}^{\sharp}}{\Cev{F}_{e}^{\sharp}}$ = $\Lfp{\sqsubseteq_{+}^{\sharp}}{\vec{F}_{e}^{\sharp}}$.
\end{lemma}
\begin{toappendix}
\begin{proof}[Proof of lemma \ref{lem:abstract:forward=backward}]The proof is similar to that of lemma \ref{lem:abstract:Lfp-subseteq-F-e}. Let $\pair{X^{\delta}}{\delta\in\mathbb{O}}$ be the iterates
of ${\vec{F}_{e}^{\sharp}}$. For the basis, $X^{0}$ = ${\bot_{+}^{\sharp}}$. 
For the successor induction step,
$X^{\delta+1}$
=
${\vec{F}_{e}^{\sharp}}(X^{\delta})$ 
=
${\textsf{init}^{\sharp}} \mathbin{\sqcup_{+}^{\sharp}}(X^{\delta}\mathbin{{\fatsemi}^{\sharp}} \sqb{\texttt{B;S}}_{e}^{\sharp})$
=
${\textsf{init}^{\sharp}} \mathbin{\sqcup_{+}^{\sharp}}(({\bigsqcup_{+}^{\sharp}}_{\beta<\delta}(\sqb{\texttt{B;S}}_{e}^{\sharp})^{\beta}) \mathbin{{\fatsemi}^{\sharp}} \sqb{\texttt{B;S}}_{e}^{\sharp})$
=
${\textsf{init}^{\sharp}} \mathbin{\sqcup_{+}^{\sharp}}{\bigsqcup_{+}^{\sharp}}_{\beta<\delta}((\sqb{\texttt{B;S}}_{e}^{\sharp})^{\beta}\mathbin{{\fatsemi}^{\sharp}} \sqb{\texttt{B;S}}_{e}^{\sharp})$
=
${\textsf{init}^{\sharp}} \mathbin{\sqcup_{+}^{\sharp}}{\bigsqcup_{+}^{\sharp}}_{\beta<\delta}((\sqb{\texttt{B;S}}_{e}^{\sharp}\mathbin{{\fatsemi}^{\sharp}}\sqb{\texttt{B;S}}_{e}^{\sharp})^{\beta} )$
=
$(\sqb{\texttt{B;S}}_{e}^{\sharp})^0\mathbin{\sqcup_{+}^{\sharp}}{\bigsqcup_{+}^{\sharp}}_{\beta<\delta}(\sqb{\texttt{B;S}}_{e}^{\sharp})^{\beta+1} )$
${\bigsqcup_{+}^{\sharp}}_{\beta<\delta+1}(\sqb{\texttt{B;S}}_{e}^{\sharp})^{\beta}$
by definition of iterates,
definition (\ref{eq:natural-transformer-finite-forward}) of ${\vec{F}_{e}^{\sharp}}$,
induction hypothesis,
definition \ref{def:abstract:domain:well:def:join:additive},
lemma \ref{lem:abstract:powers:commute},
definition of the powers,
grouping terms in the join.
For the limit induction step, $X^{\lambda}$
=
${\bigsqcup_{+}^{\sharp}}_{\delta<\lambda} X^{\delta}$
=
${\bigsqcup_{+}^{\sharp}}_{\delta<\lambda} {\bigsqcup_{+}^{\sharp}}_{\beta<\delta}(\sqb{\texttt{B;S}}_{e}^{\sharp})^{\beta}$
=
${\bigsqcup_{+}^{\sharp}}_{\beta<\lambda} (\sqb{\texttt{B;S}}_{e}^{\sharp})^{\beta}$
by definition of the iterates, induction hypothesis, and definition of the join. We conclude  that $\Lfp{\sqsubseteq_{+}^{\sharp}}{\Cev{F}_{e}^{\sharp}}$ = $\smash{\mathop{\bigsqcup_{+}^{\sharp}}\limits_{\delta\in\mathbb{O}}(\sqb{\texttt{B;S}}_{e}^{\sharp})^{\delta}}$ = $\Lfp{\sqsubseteq_{+}^{\sharp}}{\Cev{F}_{e}^{\sharp}}$ by proposition \ref{prop:Tarski:constructive} and lemma \ref{lem:abstract:Lfp-subseteq-F-e}.
\end{proof}
\end{toappendix}
\begin{example}Assume that the test \texttt{B} of the iteration \texttt{while (B) S} is always false, that is ${\textsf{test}^{\sharp}}\sqb{\texttt{B}}={\bot_{\infty}^{\sharp}}$. Then, by (\ref{eq:natural-transformer-finite-backward}), (\ref{eq:natural-transformer-finite-forward}), (\ref{def:abstract:domain:well:def:bot:absorbent}), and def.\ lub, ${\Cev{F}_{e}^{\sharp}}$ = ${\vec{F}_{e}^{\sharp}}$ = $\LAMBDA{X\in{\mathbb{L}^{\sharp}_{+}}}{\textsf{init}^{\sharp}}$. It follows that $\Lfp{\sqsubseteq_{+}^{\sharp}}{\Cev{F}_{e}^{\sharp}}$ = $\Lfp{\sqsubseteq_{+}^{\sharp}}{\vec{F}_{e}^{\sharp}}$ = ${\textsf{init}^{\sharp}}$ meaning that the loop is never entered. The semantics of the loop after 0 or more iterations is therefore that after 0 iterations.
\end{example}
\begin{lemma}[Infinite fixpoint well-definedness]\label{lem:Fbotsharp-welldefined}\proofinapx\quad If\/ $\mathbb{D}^{\sharp}_{\infty}$ is a well-defined decreasing chain complete poset and $\mathbin{{\fatsemi}^{\sharp}}$ right satisfies any one of the \ref{def:abstract:domain:well:def:increasing}, \ref{def:abstract:domain:well:finite-join-preserving}, \ref{def:abstract:domain:well:continuous}, or
\ref{def:abstract:domain:well:limit-preserving} properties for $\mathbb{D}^{\sharp}_{\infty}$ then ${F_{\bot}^{\sharp}}$ satisfies the same property and $\Gfp{\sqsubseteq_{\infty}^{\sharp}}{F_{\bot}^{\sharp}}$ does exist.
\end{lemma}
\begin{toappendix}
\begin{proof}[Proof of lemma \ref{lem:Fbotsharp-welldefined}]
If $\mathbin{{\fatsemi}^{\sharp}}$ satisfies any one of the \ref{def:abstract:domain:well:def:increasing}, \ref{def:abstract:domain:well:finite-join-preserving}, \ref{def:abstract:domain:well:continuous}, or
\ref{def:abstract:domain:well:limit-preserving} properties for $\mathbb{D}^{\sharp}_{\infty}$ then, by (\ref{eq:trace-transformer-infinite}), 
${F_{\bot}^{\sharp}}=\LAMBDA{X\in{\mathbb{L}^{\sharp}_{\infty}}}\sqb{\texttt{B;S}}_{e}^{\sharp}\mathbin{{\fatsemi}^{\sharp}} X$ satisfies the 
same property since $\sqb{\texttt{B;S}}_{e}^{\sharp}$ is constant. By the dual of proposition \ref{prop:Tarski:constructive}, $\Gfp{\sqsubseteq_{\infty}^{\sharp}}{F_{\bot}^{\sharp}}$ exists in a decreasing chain complete poset.
\end{proof}
\end{toappendix}
We now show that $\Gfp{\sqsubseteq_{\infty}^{\sharp}}{F_{\bot}^{\sharp}}$ coinductively characterizes  the infinite executions of the iteration \texttt{while (B) S} after infinitely many terminating  iterations of the body \texttt{S} with condition \texttt{B} always true. 
\begin{lemma}[Infinite body iterations]\label{lem:abstract:Lfp-subseteq-F-oo}\proofinapx\quad If\/ $\mathbb{D}^{\sharp}$ is a well-defined decreasing chain-complete poset and $\mathbin{{\fatsemi}^{\sharp}}$ is right increasing for $\sqsubseteq_{\infty}^{\sharp}$ in \ref{def:abstract:domain:well:def:increasing} then
$\Gfp{\sqsubseteq_{\infty}^{\sharp}}{F_{\bot}^{\sharp}}$ = 
$\mathop{\bigsqcap_{\infty}^{\sharp}}_{\delta\in\mathbb{O}}((\sqb{\texttt{B;S}}_{e}^{\sharp})^{\delta} \mathbin{{\fatsemi}^{\sharp}} {\top_{\infty}^{\sharp}})$.
\end{lemma}
\begin{toappendix}
\begin{proof}[Proof of lemma \ref{lem:abstract:Lfp-subseteq-F-oo}] $\mathbin{{\fatsemi}^{\sharp}}$ is increasing for $\sqsubseteq_{\infty}^{\sharp}$ so that, by lemma \ref{lem:Fbotsharp-welldefined}, ${F_{\bot}^{\sharp}}$ is 
is increasing for $\sqsubseteq_{\infty}^{\sharp}$. Since  $\mathbb{D}^{\sharp}$ is a decreasing chain-complete poset, the iterates
$\pair{X^{\delta}}{\delta\in\mathbb{O}}$  of ${F_{\bot}^{\sharp}}$ from the supremum ${\top_{\infty}^{\sharp}}$ are well-defined, so that, by the dual of proposition \ref{prop:Tarski:constructive}, $\Gfp{\sqsubseteq_{\infty}^{\sharp}}{F_{\bot}^{\sharp}}$ exists and is the limit of these iterates. These iterates are $X^0={\top_{\infty}^{\sharp}}$, 
$X^1=F_{\bot}^{\sharp}(X^0)$
=
$\sqb{\texttt{B;S}}_{e}^{\sharp}\mathbin{{\fatsemi}^{\sharp}} {\top_{\infty}^{\sharp}}$. 
Assume that
$X^{\delta}=(\sqb{\texttt{B;S}}_{e}^{\sharp})^{\delta}\mathbin{{\fatsemi}^{\sharp}} {\top_{\infty}^{\sharp}}$ by induction hypothesis so that
$X^{\delta+1}$
=
$\sqb{\texttt{B;S}}_{e}^{\sharp}\mathbin{{\fatsemi}^{\sharp}} X^{\delta}$
=
$\sqb{\texttt{B;S}}_{e}^{\sharp}\mathbin{{\fatsemi}^{\sharp}}(\sqb{\texttt{B;S}}_{e}^{\sharp})^{\delta}\mathbin{{\fatsemi}^{\sharp}} {\top_{\infty}^{\sharp}}$
= 
$(\sqb{\texttt{B;S}}_{e}^{\sharp})^{(\delta+1)-1}\mathbin{{\fatsemi}^{\sharp}} {\top_{\infty}^{\sharp}}$ by associativity, def.\ (\ref{eq:def:abstract:powers}) of the powers, and def.\ of the iterates in prop.\ \ref{prop:Tarski:constructive}. $X^{\delta+1}$ is of the form of the recurrence hypothesis proving that it holds for all iterates.
Passing to the limit, we have
$\Gfp{\subseteq}{F_{\bot}^{\sharp}}$
=
$\mathop{\bigsqcap_{\infty}^{\sharp}}_{\delta\in\mathbb{O}}X^{\delta}$
=
$\mathop{\bigsqcap_{\infty}^{\sharp}}_{\delta\in\mathbb{O}}((\sqb{\texttt{B;S}}_{e}^{\sharp})^{\delta}\mathbin{{\fatsemi}^{\sharp}} {\top_{\infty}^{\sharp}})$.
\end{proof}
\end{toappendix}
The abstract semantics of iteration is defined as
\begin{eqntabular}{l@{\hskip1ex}c@{\hskip1ex}l@{\qquad}l@{\hskip1ex}c@{\hskip1ex}l@{\qquad}}
\sqb{\texttt{while (B) S}}_{e}^{\sharp}&\triangleq&
(\Lfp{\sqsubseteq_{+}^{\sharp}}{\Cev{F}_{e}^{\sharp}})\mathbin{{\fatsemi}^{\sharp}}(\sqb{\neg\texttt{B}}_{e}^{\sharp}\mathbin{\sqcup_{e}^{\sharp}}\sqb{\texttt{B;S}}_{b}^{\sharp})\label{eq:sem:abstract:finite}
&
\sqb{\texttt{while (B) S}}_{b}^{\sharp}&\triangleq&\bot_{+}^{\sharp}\label{eq:sem:abstract:break}\\
\sqb{\texttt{while (B) S}}_{bi}^{\sharp}&\triangleq&(\Lfp{\sqsubseteq_{+}^{\sharp}}{\Cev{F}_{e}^{\sharp}})\mathbin{{\fatsemi}^{\sharp}}\sqb{\texttt{B;S}}_{\bot}^{\sharp}\label{eq:sem:abstract:body-infinite}
&
\sqb{\texttt{while (B) S}}_{li}^{\sharp}&\triangleq&\Gfp{\sqsubseteq_{\infty}^{\sharp}}{F_{\bot}^{\sharp}}\label{eq:sem:abstract:loop-infinite}
\\
\sqb{\texttt{while (B) S}}_{\bot}^{\sharp}&\triangleq&\sqb{\texttt{while (B) S}}_{bi}^{\sharp}\mathbin{\sqcup_{\infty}^{\sharp}}\sqb{\texttt{while (B) S}}_{li}^{\sharp}\label{eq:sem:abstract:loop-nontermination}
\end{eqntabular}
The least fixpoint $\Lfp{\sqsubseteq_{+}^{\sharp}}{\Cev{F}_{e}^{\sharp}}$ defines executions reaching the loop entry point after zero or finitely many iterations. Then (\ref{eq:sem:abstract:finite}) defines the finite executions of the loop when, after 0 or more iterations, the iteration condition \texttt{B} is false, or a break is executed in the body which exists the loop. By (\ref{eq:sem:abstract:break}) the \texttt{break} is from the closest enclosing loop (which existence must be checked syntactically). The loop nontermination in (\ref{eq:sem:abstract:loop-nontermination}) can happen either because, after zero or finitely many iterations, the next execution of the iteration body never terminates (\ref{eq:sem:abstract:body-infinite}), or results in (\ref{eq:sem:abstract:loop-infinite}) from infinitely many finite iterations, as defined by the greatest fixpoint $\Gfp{\sqsubseteq_{\infty}^{\sharp}}{F_{\bot}^{\sharp}}$, and obtained as the limit of iterations of ${F_{\bot}^{\sharp}}$ from $\top_{\infty}^{\sharp}$. These fixpoints in (\ref{eq:sem:abstract:finite}) and (\ref{eq:sem:abstract:loop-infinite}) do exist by lemmas \ref{lem:Fesharp-welldefined} and \ref{lem:Fbotsharp-welldefined}. 
\begin{theorem}\label{th:abstract:well:defined}\proofinapx\quad If\/ ${\mathbb{D}^{\sharp}}$ is well-defined then for all $\texttt{S}\in\mathbb{S}$, $\sqb{\texttt{S}}_{e}^{\sharp}$, $\sqb{\texttt{S}}_{b}^{\sharp}$, and $\sqb{\texttt{S}}_{\bot}^{\sharp}$ are well-defined.
\end{theorem}
\begin{toappendix}
\begin{proof}[Proof of theorem \ref{th:abstract:well:defined}]\proofinapx\quad The proof is by structural induction, observing that all operators hence their compositions are well-defined, including $\mathbin{\sqcup_{+}^{\sharp}}$, $\mathbin{\sqcup_{\infty}^{\sharp}}$, and $\mathbin{{\fatsemi}^{\sharp}}$. Lemmas \ref{lem:Fesharp-welldefined} and \ref{lem:Fbotsharp-welldefined} show that the transformers ${\Cev{F}_{e}^{\sharp}}$, ${\vec{F}_{e}^{\sharp}}$, ${F_{\bot}^{\sharp}}$ are increasingso that their fixpoints do exist.
\end{proof}
\end{toappendix}

\subsection{Algebraic Abstract Semantic Domain and Abstract Semantics}\label{sec:abstract:semantics}
The components of the abstract semantics can be recorded in a triple with named components, ordered componentwise by ${\sqsubseteq}^{\sharp}$, as follows
\bgroup\abovedisplayskip5pt\belowdisplayskip5pt\begin{eqntabular}{rcl}
{\mathbb{L}^{\sharp}}&\triangleq&(e:{\mathbb{L}^{\sharp}_{+}}\times\bot:{\mathbb{L}^{\sharp}_{\infty}}\times{br:\mathbb{L}^{\sharp}_{+}})\label{eq:def:Abstract-Semantic-Domain-Semantics}\\
\sqb{\texttt{S}}^\sharp&\triangleq&\triple{e:\sqb{\texttt{S}}_{e}^{\sharp}}{\bot:\sqb{\texttt{S}}_{\bot}^{\sharp}}{br:\sqb{\texttt{S}}_{b}^{\sharp}}\nonumber
\end{eqntabular}\egroup
If $T=\triple{e:F}{\bot:I}{br:B}\in{\mathbb{L}^{\sharp}}$, then we select the individual components of the Cartesian product $T$ using the field selectors $e$, $br$, and $\bot$, as follows 
\bgroup\abovedisplayskip=5pt\belowdisplayskip-3pt\begin{eqntabular}{@{\qquad\qquad}C}
$T_{+}=F$,\quad $T_{\infty}=I$,\quad and\quad $T_{br}=B$.
\label{eq:abstract:shorthands}
\end{eqntabular}\egroup
By convention, 
\bgroup\abovedisplayskip=5pt\belowdisplayskip5pt\begin{eqntabular}[fl]{@{\quad}p{0.9\textwidth}}
The shorthand $F$ denotes $\triple{e:F}{\bot:\bot^{\sharp}_{\infty}}{br:\bot^{\sharp}_{+}}$ and similarly for other unique nonempty components.
\label{eq:unique-nonempty-component-shorthand}
\end{eqntabular}\egroup
The abstract semantics $\sqb{\texttt{S}}^\sharp\in{\mathbb{L}^{\sharp}}$ records three components $\sqb{\texttt{S}}_{e}^{\sharp}$, $\sqb{\texttt{S}}_{\bot}^{\sharp}$, and $\sqb{\texttt{S}}_{b}^{\sharp}$ of the definition of the algebraic semantics of statements $\texttt{S}$ in sect.\@ \ref{sec:definition:abstract:semantics}. 
\begin{lemma}\label{lem:abstract-domain-sharp}\proofinapx\quad If\/ ${\mathbb{D}^{\sharp}}$ is a well-defined chain-complete join semilattice (respectively complete lattice) with sequential composition $\mathbin{{\fatsemi}^{\sharp}}$ satisfying any one of the hypotheses \ref{def:abstract:domain:well:def:join:additive} then $\pair{{\mathbb{L}^{\sharp}}}{{\sqsubseteq^{\sharp}}}$ has the same structure, componentwise.
\end{lemma}
\begin{toappendix}
\begin{proof}[Proof of lemma \ref{lem:abstract-domain-sharp}]
Lemma \ref{lem:abstract-domain-sharp} follows from the fact that the Cartesian product of complete lattices (respectively,  a chain-complete join semilattice) is a complete lattice \cite[p.\ 33]{DBLP:books/daglib/0023601} (resp., is a chain-complete join semilattice \cite[p.\ 55]{DBLP:books/daglib/0023601}). 
\end{proof}
\end{toappendix}
All semantic definitions are extended componentwise. For ${\,\fatsemi^\sharp\,}\in\; {\mathbb{L}^{\sharp}}\times{\mathbb{L}^{\sharp}}\functionto{\mathbb{L}^{\sharp}}$, we define
\bgroup\abovedisplayskip5pt\belowdisplayskip0pt\begin{eqntabular}[fl]{l}
\langle{ok{\mskip1mu:\mskip1mu}\langle{e{\mskip1mu:\mskip1mu}F_1},{\bot{\mskip1mu:\mskip1mu}I_1}\rangle},{b{\mskip1mu:\mskip1mu}B_1}\rangle
\mathbin{\fatsemi^{\sharp}}
\langle{ok{\mskip1mu:\mskip1mu}\langle{e{\mskip1mu:\mskip1mu}F_2},{\bot{\mskip1mu:\mskip1mu}I_2}\rangle},{b{\mskip1mu:\mskip1mu}B_2}\rangle
\triangleq
\langle{ok{\mskip1mu:\mskip1mu}\langle{e{\mskip1mu:\mskip1mu}F_1{\,\fatsemi^\sharp\,}F_2},{\bot{\mskip1mu:\mskip1mu}I_{1} {\,\sqcup_{\infty}^{\sharp}\,}(F_1{\,\fatsemi^{\sharp}\,}I_2)
}\rangle},{b{\mskip1mu:\mskip1mu}B_1{\,\sqcup_{+}^{\sharp}\,}(F_1{\,\fatsemi^{\sharp}\,}B_2)}\rangle\nonumber\\
\textrm{so that, by (\ref{eq:def:abstract:sem:seq}), }\sqb{\texttt{S$_1$;S$_2$}}^{\sharp}=\sqb{\texttt{S$_1$}}^{\sharp}\mathbin{\fatsemi^{\sharp}}\sqb{\texttt{S$_2$}}^{\sharp}.\label{eq:def:abstract:sem:group:seq}
\end{eqntabular}\egroup
\begin{remark}The semantic domain of our algebraic semantics is much more refined than traditional ones such as 
\cite{DBLP:conf/birthday/Hoare13a} where, the computational and logical ordering are subset inclusion and, following the denotational semantics \cite{ScottStrachey71-PRG6} approach, ``Nontermination has to be represented by a fictitious “state at infinity” that can be “reached” only by a non-terminating program. Also, if the fictitious state is in the image of a state, then that image is universal.'' \cite{DBLP:journals/cacm/HoareHJMRSSSS87}.  This can be achieved by instantiation e.g. to a trace semantics followed by an abstraction (mapping infinite traces to the ``fictitious “state at infinity”'').

Moreover, we do not specify the algebraic semantics by ``laws'' (or axioms) but in structural fixpoint form, which is known to be equivalent, according to the generalization \cite{DBLP:conf/cav/CousotC95} of Peter Aczel correspondance \cite{Aczel:1977:inductive-definitions} between deductive/proof systems and fixpoint definitions. The ``laws'' for basic statements are the definitions (\ref{eq:def:sem:abstract:basis}). The other ``laws'' for structured statements and iteration are theorems following from the definition \ref{def:abstract:domain:well:def} of an abstract domain and fixpoint induction principles \cite{DBLP:conf/lopstr/Cousot19} following from propositions \ref{prop:Tarski} and
\ref{prop:Tarski:constructive}.
\end{remark}
All semantics in \cite{DBLP:journals/jacm/AptP86,DBLP:journals/tcs/Cousot02,DBLP:journals/tcs/GiacobazziM05} can be instantiated to the
algebraic abstract semantics of sect.\@ \ref{sec:abstract:semantics}. There are obviously others, such as symbolic execution \cite{DBLP:journals/cacm/King76} (extended to infinite
behaviors). For semantics defined by transformations such as compilation, the transformation is an instance of the
algebraic abstract semantics, but the semantics of the transformed program is not, because of a different syntax, although it can certainly be also defined in an algebraic style.

The original definition of hyperproperties \cite{DBLP:journals/jcs/ClarksonS10} was relative to a trace (or path) semantics $\sqb{\texttt{S}}^{\pi}$ which, as shown in the appendix \proofinapx, is an instance of the algebraic abstract semantics $\sqb{\texttt{S}}^{\sharp}$ where the  domain $\mathbb{D}^{\sharp}_{+}$ is the complete lattice $\mathbb{D}^{\pi}_{+}$ of sets of finite traces and the  domain $\mathbb{D}^{\sharp}_{\infty}$ is the complete lattice $\mathbb{D}^{\pi}_{\infty}$ of sets of infinite traces where traces account for the successive values taken by variables during execution, as recorded in states. All operators preserve arbitrary joins. For lower continuity, see counterexample \ref{cex:not:lower:continuous} for infinite traces and the following lower continuity proof for finite traces.

\begin{toappendix}
\section{Trace Semantics}\label{sect:Trace-Semantics}
\subsection{The Trace Semantics Domain}

\subsubsection{States}
States $\sigma\in\Sigma\triangleq\mathbb{X}\rightarrow \mathbb{V}$ (also called environments) map variables  $\texttt{x}\in\mathbb{X}$ to their values $\sigma(\texttt{x})$ in $\mathbb{V}$ including integers, $\mathbb{Z}\subseteq \mathbb{V}$. 

\subsubsection{Finite Traces}
We let $\pi=\pi_0\pi_1\ldots\pi_{n-1}\in\Sigma^n\triangleq\interval[open right]{0}{n}\rightarrow\Sigma$ be the nonempty finite traces of length $|\pi|=n$, $n\geqslant1$ over states $\pi_i\in\Sigma$, $i\in\interval[open right]{0}{n}$, $\Sigma^+\triangleq\bigcup_{n\geqslant1}\Sigma^n$. 
The empty trace $\epsilon$ is in $\Sigma^0=\{\epsilon\}$. $\Sigma^\ast\triangleq\Sigma^+\cup\Sigma^0$ is the set of possibly empty traces. A set of finite  traces defines a property of finite executions in extension.
\begin{eqntabular}{rcl}
\sextuple{\mathbb{L}^{\pi}_{+}}{\sqsubseteq_{+}^{\pi}}{\bot_{+}^{\pi}}{\top_{+}^{\pi}}{\sqcup_{+}^{\pi}}{\sqcap_{+}^{\pi}}
&\triangleq&
\sextuple{\wp(\Sigma^{+})}{\subseteq}{\emptyset}{\Sigma^{+}}{\cup}{\cap}
\label{eq:def:finite-traces-domain}
\end{eqntabular}

\subsubsection{Infinite Traces}
The infinite traces  $\pi=\pi_0\pi_1\ldots\pi_{n}\ldots\in\Sigma^\infty\triangleq\interval[open right]{0}{\infty}\rightarrow\Sigma$ have  length $|\pi|=\infty$ over states $\pi_i\in\Sigma$, $i\in\interval[open right]{0}{\infty}$. We let
$\Sigma^{+\infty}\triangleq\Sigma^+\cup\Sigma^\infty$ and $\Sigma^{\ast\infty}\triangleq\Sigma^\ast\cup\Sigma^\infty$. 

A trace $\sigma\pi\in\Sigma^{+\infty}$ has first state $\sigma\in\Sigma$. A trace of the form $\pi\sigma$ is necessarily finite with last state $\sigma$ and $\pi\in\Sigma^{\ast}$. If $0\leqslant i\leqslant j<n$ and $\pi\in\Sigma^n$ then $\pi_{\interval{i}{j}}\triangleq\pi_i\pi_{i+1}\ldots\pi_{j}$ is the subtrace of $\pi$ stating at $i$ and ending at $j$.
 A set of infinite traces defines a property of nonterminating executions in extension.
\begin{eqntabular}{rcl}
\sextuple{\mathbb{L}^{\pi}_{\infty}}{\sqsubseteq_{\infty}^{\pi}}{\bot_{\infty}^{\pi}}{\top_{\infty}
^{\pi}}{\sqcap_{\infty}^{\sharp}}{\sqcup_{\infty}^{\sharp}}
&\triangleq&
\sextuple{\wp(\Sigma^{\infty})}{\subseteq}{\emptyset}{\Sigma^{\infty}}{\cap}{\cup}
\label{eq:def:infinite-traces-domain}
\end{eqntabular}
Notice that $\Gfp{\sqsubseteq_{\infty}^{\pi}}{F_{\bot}^{\pi}}=\Gfp{\subseteq}{F_{\bot}^{\pi}}$ so that infinite execution traces are defined co-inductively.

\subsubsection{Traces Operators}
\begin{eqntabular}[fl]{r@{\ \ }c@{\ \ }l@{\ \ }r@{\ \ }c@{\ \ }l}
{\textsf{init}^{\pi}}&\triangleq&\Sigma^1
&
{\textsf{test}^{\pi}\sqb{\texttt{B}}}&\triangleq&\{\sigma\mid\sigma\in\mathcal{B}\sqb{\texttt{B}}\}
\nonumber\\
{\textsf{assign}^{\pi}\sqb{\texttt{x},\texttt{A}}}&\triangleq&\{{\sigma}{\sigma[\texttt{x}\leftarrow\mathcal{A}\sqb{\texttt{A}}\sigma]}\in\Sigma^2\mid \sigma\in\Sigma\}
&
{\textsf{break}^{\pi}}&\triangleq&\{\sigma\,\textsf{break-to}(\sigma)\mid\sigma\in\Sigma\}
\label{eq:trace-break}\\
{\textsf{rassign}^{\pi}\sqb{\texttt{x},a,b}}&\triangleq&\{{\sigma}{\sigma[\texttt{x}\leftarrow i]}\in\Sigma^2\mid a-1 <i< b+1\}
&
{\textsf{skip}^{\pi}}&\triangleq&\{\sigma\sigma\mid\sigma\in\Sigma\}\nonumber
\end{eqntabular}
See \cite[page 43]{Cousot-PAI-2021} for a definition of \textsf{break-to} (exiting the enclosing loop with variables unchanged). We deliberately leave unspecified the syntax and semantics of arithmetic expressions $\mathcal{A}\sqb{\texttt{A}}\in\Sigma\rightarrow \mathbb{V}$ and Boolean expressions  $\mathcal{B}\sqb{\texttt{B}}\in\wp(\Sigma)\simeq\Sigma\rightarrow \{\textsf{\upshape true},\textsf{\upshape false}\}$. The only assumption on expressions is the absence of side effects.

We let ${\fatsemi}^{\pi}$ be the concatenation of sets of traces  $T\in\wp(\Sigma^{\ast\infty})$ and $T'\in\wp(\Sigma^{\ast\infty})$ such that
 \begin{eqntabular*}{rcl}
T\mathrel{{\fatsemi}^{\pi}}T'&\triangleq&\{\pi'\in T'\mid\epsilon\in T\}\cup \{\pi\in T\mid \epsilon\in T'\}\cup{}
(T\cap\Sigma^{\infty}) \cup\{\pi\sigma\pi'\mid\pi\sigma\in T\wedge \sigma\pi'\in T'\}
\end{eqntabular*}
The powers of a set $T\in\wp(\Sigma^{\ast\infty})$ of traces are $\{\epsilon\}^n=\{\epsilon\}$ and otherwise  $T^0=\Sigma^1$and $T^{n+1}$ = $T^{n}\fatsemi T$ = $T\fatsemi T^{n}$ for all $n\geqslant0$. We denote
$T^\infty\in\wp(\Sigma^\infty)$ the set of infinite traces obtained by concatenation of traces of $T$. Notice that $\fatsemi$ is right increasing but not right lower continuous on infinite traces $\wp(\Sigma^{\infty})$.
\begin{counterexample}\label{cex:not:lower:continuous}Let $r=\{\sigma_1, \sigma_1\sigma_2,\ldots, \sigma_1\ldots\sigma_n,\ldots\}$ be the prefix closure of the infinite trace $\sigma_1\sigma_2\sigma_3\ldots$. Define $X_i=\{\sigma_i\sigma_{i+1}\sigma_{i+2}\ldots,\sigma_{i+1}\sigma_{i+2}\ldots,\sigma_{i+2}\ldots,, \ldots\}$ be the suffix closure of the infinite trace $\sigma_i\sigma_{i+1}\sigma_{i+2}\sigma_{i+3}\ldots$ so that $\pair{X_i}{i\in\mathbb{N}}$ is a decreasing chain. Then $r\,\mathbin{{\fatsemi}^{\pi}}\,\bigcap_{i\in\mathbb{N}}X_i$ = $r\,\mathbin{{\fatsemi}^{\pi}}\,\emptyset$ = $\emptyset$, while
$\bigcap_{i\in\mathbb{N}}(r\,\mathbin{{\fatsemi}^{\pi}}\,X_i)$ = $\bigcap_{i\in\mathbb{N}}\{\sigma_1\sigma_{2}\sigma_{3}\ldots\}$
= $\{\sigma_1\sigma_{2}\sigma_{3}\ldots\}$. 
\end{counterexample}
However, $\fatsemi$ is right lower continuous on finite traces $\wp(\Sigma^{+})$.
\begin{proof}[Proof of  right lower continuity of $\fatsemi$ for finite traces]
\newcommand{\traceconcat}{\mathbin{\wideparen{\raisebox{1ex}[0pt][0pt]{.}}}}
Let $\pair{X^i\in \wp(\Sigma^{+})}{i\in\mathbb{N}}$ be a $\subseteq$-decreasing chain of sets of finite traces. We must prove that
$r\,\mathbin{{\fatsemi}^{\pi}}\,(\bigcap_{i\in\mathbb{N}}X^i)$ = $\bigcap_{i\in\mathbb{N}}(r\,\mathbin{{\fatsemi}^{\pi}}\,X^i)$. The inclusion $\subseteq$ is trivial. Conversely, let $\pi\in\bigcap_{i\in\mathbb{N}}(r\,\mathbin{{\fatsemi}^{\pi}}\,X^i)\subseteq\wp(\Sigma^{+})$ then there exists $\bar{\pi}_0\in X^0$, $\bar{\pi}_1\in X^1$, \ldots,
$\bar{\pi}_i\in X^i$, \ldots\ and ${\pi}_0$, ${\pi}_1$, \ldots,
${\pi}_i$, \ldots\ $\in$ $r$ such that $\bar{\pi}_0\mathrel{{\leqslant}_s}\bar{\pi}_1\mathrel{{\leqslant}_s} \ldots \mathrel{{\leqslant}_s}\bar{\pi}_i\mathrel{{\leqslant}_s}\ldots$ and $\pi$ = ${\pi}_0\traceconcat\bar{\pi}_0$ = ${\pi}_1\traceconcat\bar{\pi}_1$ = \ldots\ = ${\pi}_i\traceconcat\bar{\pi}_i$ = \ldots\ where
$\mathrel{{\leqslant}_s}$ is the suffix ordering on traces and ${}\traceconcat{}$ is trace concatenation. The length
of the $\pair{\bar{\pi}_i}{i\in\mathbb{N}}$ is ultimately stationary at some $k\in\mathbb{N}$. This means that there exists $\bar{\pi}_k$ such that $\forall i\in\mathbb{N}\mathrel{.}\bar{\pi}_k\in X^i$. As a result,
$\pi$ = ${\pi}_k\traceconcat\bar{\pi}_k$ $\in$ $r\,\mathbin{{\fatsemi}^{\pi}}\,(\bigcap_{i\in\mathbb{N}}X_i)$.
\end{proof}

\subsection{Structural Trace Semantics}\label{sec:Structural:trace:semantics}

$\Lfp{\sqsubseteq_{+}^{\pi}}{F_{e}^{\pi}}=\Lfp{\subseteq}{F_{e}^{\pi}}$ is the set of finite traces reaching the entry of the iteration \texttt{while (B) S} after zero or more terminating body iterations . 
\begin{lemma}\label{lem:Lfp-subseteq-F-e}
$\Lfp{\subseteq}{F_{e}^{\pi}}$ = $\bigcup_{n\in\mathbb{N}}(\sqb{\texttt{B$\,\mathbin{{\fatsemi}^{\pi}}\,$S}}_{e}^{\pi})^n$. 
\end{lemma}
\begin{proof}[Proof of lemma \ref{lem:Lfp-subseteq-F-e}] An instance of lemma \ref{lem:abstract:Lfp-subseteq-F-e} for $\mathbb{D}^{\pi}_{+}$.
\end{proof}
$\Gfp{\sqsubseteq_{\infty}^{\pi}}{F_{\bot}^{\pi}}=\Gfp{\subseteq}{F_{\bot}^{\pi}}$ is the set of infinite traces of the iteration \texttt{while (B) S} after infinitely many terminating body iterations . 
\begin{lemma}\label{lem:Gfp-subseteq-F-e}
$\Gfp{\subseteq}{F_{\bot}^{\pi}}$ = $(\sqb{\texttt{B$\,\mathbin{{\fatsemi}^{\pi}}\,$S}}_{e}^{\pi})^\infty$. 
\end{lemma}
\begin{proof}[Proof of lemma \ref{lem:Gfp-subseteq-F-e}] An instance of lemma \ref{lem:abstract:Lfp-subseteq-F-oo} for $\mathbb{D}^{\pi}_{\infty}$. Moreover, $\mathop{\bigsqcap_{\infty}^{\sharp}}_{n\in\mathbb{N}}((\sqb{\texttt{B;S}}_{e}^{\sharp})^n \mathbin{{\fatsemi}^{\sharp}} {\bot_{\infty}^{\sharp}})$ becomes 
$\bigcap_{n\in\mathbb{N}}((\sqb{\texttt{B;S}}_{e}^{\pi})^n \mathbin{{\fatsemi}^{\pi}} \Sigma^{\infty})$
=
$(\sqb{\texttt{B$\,\mathbin{{\fatsemi}^{\pi}}\,$S}}_{e}^{\pi})^\infty$
since all traces in $(\sqb{\texttt{B$\,\mathbin{{\fatsemi}^{\pi}}\,$S}}_{e}^{\pi})^\infty$ belong to $(\sqb{\texttt{B$\,\mathbin{{\fatsemi}^{\pi}}\,$S}}_{e}^{\pi})^{n}\mathbin{{\fatsemi}^{\pi}} \Sigma^{\infty}$, $n\geqslant 0$ while
any trace not of that form must be $\pi\pi'\pi''$ with $\pi\in(\sqb{\texttt{B$\,\mathbin{{\fatsemi}^{\pi}}\,$S}}_{e}^{\pi})^{n}$, $\pi'\not\in \sqb{\texttt{B$\,\mathbin{{\fatsemi}^{\pi}}\,$S}}_{e}^{\pi}$, and $\pi''\in\Sigma^\infty$  for some $n\in\mathbb{N}$ and
so does not belong to $X^{n+2}$ hence not to the intersection.
\end{proof}
\begin{example}\label{ex:semantics:S}Consider $\texttt{S} \triangleq \texttt{while (x!=2) if (x==1) then break else x=x+2}$. It's trace semantics is 
\arraycolsep0.5\arraycolsep
\begin{eqntabular}{rcl}
\sqb{\texttt{S}}_{e}^{\pi}&=&\{\texttt{x}:-2k;\texttt{x}:-2k+2;\ldots;\texttt{x}:0;\texttt{x}:2\mid k\geqslant -1\}{}\cup{}\{\texttt{x}:-2k+1;\texttt{x}:-2k+3;\ldots;\texttt{x}:1\mid k\geqslant 0\}\nonumber\\
\sqb{\texttt{S}}_{b}^{\pi}&=&\emptyset\label{eq:ex:semantics:S}\\
\sqb{\texttt{S}}_{\bot}^{\pi}&=&\{\texttt{x}:n;\ldots;\texttt{x}:n+2k;\ldots\mid n>2\}\ .\renumber{\mbox{\qef}}
\end{eqntabular}
\let\qef\relax
\end{example}
\begin{proof}[Proof of (\ref{eq:ex:semantics:S})]
We have 
$\sqb{\texttt{(x!=2)}\mathbin{{\fatsemi}^{\pi}}\texttt{if (x==1) then break else x=x+2}}$ = $\pair{ok:\{\texttt{x}:n;\texttt{x}:n+2\mid n\notin\{1,2\}\}}{br:\{\texttt{x}:1\}}$ so that
${F_{e}^{\pi}}(X)=\{\texttt{x}:n\mid n\in\mathbb{Z}\}\cup\{\texttt{x}:n;\texttt{x}:n+2;\pi\mid
n\notin\{1,2\}\wedge \texttt{x}:n+2;\pi\in X^+\}$ for the finite traces reaching the loop head. 

The iterates are ${F_{e}^{\pi}}^0=\emptyset$, ${F_{e}^{\pi}}^1=\{\texttt{x}:n\mid n\in\mathbb{Z}\}$, ${F_{e}^{\pi}}^2=\{\texttt{x}:n\mid n\in\mathbb{Z}\}\cup\{\texttt{x}:n;\texttt{x}:n+2\mid n\notin\{1,2\}\}$, ${F_{e}^{\pi}}^3=\{\texttt{x}:n\mid n\in\mathbb{Z}\}\cup\{\texttt{x}:n;\texttt{x}:n+2\mid n\notin\{1,2\}\}\cup\{\texttt{x}:n;\texttt{x}:n+2;\texttt{x}:n+4\mid n\notin\{-1,0,1,2\}\}$,
so that ${F_{e}^{\pi}}^k=\{\texttt{x}:n\mid n\in\mathbb{Z}\}\cup\bigcup_{j=1}^{k-1}\{\texttt{x}:n;\ldots;\texttt{x}:n+2j;\mid n\notin\interval{3-2j}{2}\}$ by induction hypothesis. For the induction step
\begin{calculus}[=\ \ ]
\formula{{F_{e}^{\pi}}^{k+1}}\\
=
\formula{{F_{e}^{\pi}}({F_{e}^{\pi}}^{k})}\\
=
\formulaexplanation{\{\texttt{x}:n\mid n\in\mathbb{Z}\}\cup\{\texttt{x}:n;\texttt{x}:n+2;\pi\mid
n\notin\{1,2\}\wedge \texttt{x}:n+2;\pi\in {F_{e}^{\pi}}^{k}\}}{def.\ $F_{e}^{\pi}$}\\
=
\formulaexplanation{\{\texttt{x}:n\mid n\in\mathbb{Z}\}\cup\{\texttt{x}:n;\texttt{x}:n+2;\pi\mid
n\notin\{1,2\}\wedge \texttt{x}:n+2;\pi\in \{\texttt{x}:n\mid n\in\mathbb{Z}\}\cup\bigcup_{j=1}^{k-1}\{\texttt{x}:n;\ldots;\texttt{x}:n+2j\mid n\notin\interval{3-2j}{2}\}\}}{induction hypothesis}\\
=
\formulaexplanation{\{\texttt{x}:n\mid n\in\mathbb{Z}\}\cup\{\texttt{x}:n;\texttt{x}:n+2;\pi\mid
n\notin\{1,2\}\wedge \texttt{x}:n+2;\pi\in \{\texttt{x}:n\mid n\in\mathbb{Z}\}\}
\cup\bigcup_{j=1}^{k-1}\{\texttt{x}:n;\texttt{x}:n+2;\pi\mid\texttt{x}:n+2;\pi\in 
\{\texttt{x}:n;\ldots;\texttt{x}:n+2j\mid n\notin\interval{3-2j}{2}\}\}}{def.\ $\cup$}\\
=
\formulaexplanation{\{\texttt{x}:n\mid n\in\mathbb{Z}\}\cup\{\texttt{x}:n;\texttt{x}:n+2\mid
n\notin\{1,2\}\}
\cup\bigcup_{j=1}^{k-1}\{\texttt{x}:n;\texttt{x}:n+2;\pi\mid\texttt{x}:n+2;\pi\in 
\{\texttt{x}:n+2;\ldots;\texttt{x}:n+2+2j\mid n+2\notin\interval{3-2j}{2}\}\}}{simplification and renaming}\\
=
\formula{\{\texttt{x}:n\mid n\in\mathbb{Z}\}\cup\{\texttt{x}:n;\texttt{x}:n+2\mid
n\notin\{1,2\}\}
\cup\bigcup_{j=1}^{k-1}\{\texttt{x}:n;\texttt{x}:n+2;\texttt{x}:n+4;\ldots;\texttt{x}:n+2(j+1)\mid n\notin\interval{1-2j}{0}\}}\\[-0.75ex]\rightexplanation{def.\ $\in$}\\
=
\formulaexplanation{\{\texttt{x}:n\mid n\in\mathbb{Z}\}\cup\{\texttt{x}:n;\texttt{x}:n+2\mid
n\notin\{1,2\}\}
\cup\bigcup_{j'=2}^{k}\{\texttt{x}:n;\texttt{x}:n+2;\texttt{x}:n+4;\ldots;\texttt{x}:n+2j'\mid n\notin\interval{1-2(j'-1)}{0}\}}{def.\ $j=j'-1$ so $j'=j+1$ }\\
=
\formula{\{\texttt{x}:n\mid n\in\mathbb{Z}\}
\cup\bigcup_{j'=1}^{k}\{\texttt{x}:n;\texttt{x}:n+2;\ldots;\texttt{x}:n+2j'\mid n\notin\interval{3-2j'}{2}\}}\\[-0.75ex]\rightexplanation{incorporating the term $\{\texttt{x}:n;\texttt{x}:n+2\mid
n\notin\{1,2\}\}$ in the join for $j'=1$}
\end{calculus}

\smallskip

\noindent
This shows that all iterates of ${F_{e}^{\pi}}$ have the form ${F_{e}^{\pi}}^{k}$. Since ${F_{e}^{\pi}}$ preserves joins, we have, by Tarski's fixpoint iteration theorem \cite[page 305]{Tarski-fixpoint}, that 
\begin{calculus}[=\ \ ]
\formula{\Lfp{\subseteq}{F_{e}^{\pi}}}\\
=
\formula{\bigcup_{k\in\mathbb{N}}{F_{e}^{\pi}}^{k}}\\
=
\formula{\bigcup_{k\in\mathbb{N}}\Bigl(\{\texttt{x}:n\mid n\in\mathbb{Z}\}\cup\bigcup_{j=1}^{k-1}\{\texttt{x}:n;\ldots;\texttt{x}:n+2j;\mid n\notin\interval{3-2j}{2}\}\Bigr)}\\
=
\formula{\{\texttt{x}:n\mid n\in\mathbb{Z}\}\cup\bigcup_{j\geqslant1}\{\texttt{x}:n;\ldots;\texttt{x}:n+2j;\mid n\notin\interval{3-2j}{2}\}}\\
=
\formula{\bigcup_{j\in\mathbb{N}}\{\texttt{x}:n;\ldots;\texttt{x}:n+2j;\mid n\notin\interval{3-2j}{2}\}}\\
\explanation{since for $j=0$, we have $n\notin\interval{3-2j}{2}$ which is $n\notin\interval{3}{2}$ that is 
$n\notin\emptyset$ or $n\in\mathbb{Z}$ with $\texttt{x}:n;\ldots;\texttt{x}:n+2j$ = $\texttt{x}:n;\ldots;\texttt{x}:n$ = $\texttt{x}:n$}
\end{calculus}

\smallskip

For the infinite traces, we have
\begin{calculus}[=\ \ ]
\formula{F^{\bot}(X),\quad X\in\wp(\Sigma^{+\infty})}\\
=
\formula{\sqb{\texttt{B$\,\mathbin{{\fatsemi}^{\pi}}\,$S}}_{e}^{\pi}\mathbin{{\fatsemi}^{\pi}} X^{\infty}}\\
=
\formula{\{\texttt{x}:n;\texttt{x}:n+2;\pi\mid
n\notin\{1,2\}\wedge \texttt{x}:n+2;\pi\in X^{\infty}\}}
\end{calculus}

\smallskip

\noindent The iterates of $F^{\bot}$ are 
${F^{\bot}}^0=\Sigma^{\infty}$, 
${F^{\bot}}^1$ =
$\{\texttt{x}:n;\texttt{x}:n+2;\pi\mid
n\notin\{1,2\}\wedge \texttt{x}:n+2;\pi\in \Sigma^{\infty}\}$ =
$\{\texttt{x}:n;\texttt{x}:n+2;\pi\mid
n\notin\{1,2\}\wedge\pi\in \Sigma^{\infty}\}$, 
 ${F^{\bot}}^2$ 
 =
 $\{\texttt{x}:n;\texttt{x}:n+2;\pi\mid
n\notin\{1,2\}\wedge \texttt{x}:n+2;\pi\in \{\texttt{x}:n;\texttt{x}:n+2;\pi\mid
n\notin\{1,2\}\wedge\pi\in \Sigma^{\infty}\}\}$
=
$\{\texttt{x}:n;\texttt{x}:n+2;\pi\mid
n\notin\{1,2\}\wedge \texttt{x}:n+2;\pi\in \{\texttt{x}:n+2;\texttt{x}:n+4;\pi'\mid
n+2\notin\{1,2\}\wedge\pi'\in \Sigma^{\infty}\}\}$
=
$\{\texttt{x}:n;\texttt{x}:n+2;\texttt{x}:n+4;\pi'\mid
n\notin\{-1,0,1,2\}\wedge\pi'\in \Sigma^{\infty}\}$ which leads to the induction hypothesis
${F^{\bot}}^k$ 
=
$\{\texttt{x}:n;\ldots;\texttt{x}:n+2k;\pi\mid n\notin\interval{3-2k}{2}\wedge\pi\in \Sigma^{\infty}\}$. For the induction step,
\begin{calculus}[=\ \ ]
\formula{{F^{\bot}}^{k+1}}\\
=
\formula{{F^{\bot}}({F^{\bot}}^{k})}\\
=
\formula{\{\texttt{x}:n;\texttt{x}:n+2;\pi\mid
n\notin\{1,2\}\wedge \texttt{x}:n+2;\pi\in \{\texttt{x}:n;\ldots;\texttt{x}:n+2k;\pi\mid n\notin\interval{3-2k}{2}\wedge\pi\in \Sigma^{\infty}\}\}}\\
=
\formula{\{\texttt{x}:n;\texttt{x}:n+2;\pi\mid
n\notin\{1,2\}\wedge \texttt{x}:n+2;\pi\in \{\texttt{x}:n+2;\ldots;\texttt{x}:n+2+2k;\pi'\mid n+2\notin\interval{1-2k}{0}\wedge\pi'\in \Sigma^{\infty}\}\}}\\
=
\formula{\{\texttt{x}:n;\texttt{x}:n+2;\ldots;\texttt{x}:n+2+2k;\pi'\mid
n\notin\{1,2\}\wedge  n\notin\interval{1-2k}{0}\wedge\pi'\in \Sigma^{\infty}\}}\\
=
\formula{\{\texttt{x}:n;\ldots;\texttt{x}:n+2(k+1);\pi'\mid
n\notin\interval{3-2(k+1)}{2}\wedge\pi'\in \Sigma^{\infty}\}}
\end{calculus}

\smallskip

This shows that all iterates of ${F^{\bot}}$ have the form ${F^{\bot}}^{k}$. Since ${F^{\bot}}$ preserves meets, we have, by the dual of Tarski's fixpoint iteration theorem \cite[page 305]{Tarski-fixpoint}, that \
\begin{calculus}[=\ \ ]
\formula{\Gfp{\subseteq}{F^{\bot}}}\\
=
\formula{\bigcap_{k\in\mathbb{N}}{F^{\bot}}^{k}}\\
=
\formula{\bigcap_{k\in\mathbb{N}}\{\texttt{x}:n;\ldots;\texttt{x}:n+2k;\pi\mid n\notin\interval{3-2k}{2}\wedge\pi\in \Sigma^{\infty}\}}\\
=
\formula{\{\texttt{x}:n;\ldots;\texttt{x}:n+2k;\ldots\mid n>2\}}
\end{calculus}

\smallskip

\noindent since all infinite traces of the form $\texttt{x}:n;\ldots;\texttt{x}:n+2k;\ldots$ with $n>2$ belong to all iterates ${F^{\bot}}^{k}$ hence to their intersection while, conversely, all other traces start with $\texttt{x}:n;\ldots$ and $n\leqslant 2$ so do not belong to the ${F^{\bot}}^{k}$, $k\geqslant 1$ so don't belong to their intersection, or else, start with $n>2$, but have the form $\texttt{x}:n;\ldots; \texttt{x}:n+2k+1;\ldots$ and so do not belong to ${F^{\bot}}^{k}$, hence to the intersection.

\smallskip

The trace semantics of $\texttt{S} \triangleq \texttt{while (x!=2) if (x==1) then break else x=x+2}$
is therefore
\begin{calculus}[=\ \ ]
\hyphen{5}\formula{\sqb{\texttt{S}}_{e}^{\pi}}\\
$\triangleq$
\formulaexplanation{\Lfp{\subseteq}{F_{e}^{\pi}}\mathbin{{\fatsemi}^{\pi}}(\sqb{\neg\texttt{(x!=2) }}\cup\sqb{\texttt{(x!=2)$\,\mathbin{{\fatsemi}^{\pi}}\,$if (x==1) then break else x=x+2}}_{b}^{\pi})}{by (\ref{eq:sem:abstract:finite})}\\
=
\formula{\Lfp{\subseteq}{F_{e}^{\pi}}\mathbin{{\fatsemi}^{\pi}}(\{\texttt{x}:2; \texttt{x}:2\}\cup\{\texttt{x}:1\})}\\[-0.75ex]
\explanation{$\sqb{\neg\texttt{(x!=2)}}=\{\texttt{x}:2; \texttt{x}:2\}$ and $\sqb{\texttt{(x!=2)}\mathbin{{\fatsemi}^{\pi}}\texttt{if (x==1) then break else x=x+2}}$ = $\pair{ok:\{\texttt{x}:n;\texttt{x}:n+2\mid n\notin\{1,2\}\}}{br:\{\texttt{x}:1\}}$}\\
=
\formula{\bigcup_{j\in\mathbb{N}}\{\texttt{x}:n;\ldots;\texttt{x}:n+2j;\mid n\notin\interval{3-2j}{2}\}\mathbin{{\fatsemi}^{\pi}}(\{\texttt{x}:2; \texttt{x}:2,\texttt{x}:1\})}\\
=
\formula{\{\texttt{x}:-2k;\texttt{x}:-2k+2;\ldots;\texttt{x}:0;\texttt{x}:2\mid k\geqslant -1\}\cup\{\texttt{x}:-2k+1;\texttt{x}:-2k+3;\ldots;\texttt{x}:1\mid k\geqslant 0\}}
\end{calculus}

\smallskip

\noindent since, by definition of $\mathbin{{\fatsemi}^{\pi}}$, we have only two possible cases.
\begin{itemize}
\item Either $n+2j=2$, $j\in\mathbb{N}$, $n\notin\interval{3-2j}{2}$ so $n=-2k$ with $j=1+k\geqslant 0$ that is $k\geqslant -1$ which implies  $n\notin\interval{3-2j}{2}=\interval{3-(2-n)}{2}=\interval{n+1}{2}$;
\item Or $n+2j=1$, $j\in\mathbb{N}$, $n\notin\interval{3-2j}{2}$ so $n=-2k+1$ with $j=k\geqslant 0$ which implies
$n\notin\interval{3-2j}{2}$ since $n = -2k+1 < 3-2k$.
\end{itemize}

\begin{calculus}[=\ \ ]
\hyphen{5}\formulaexplanation{\sqb{\texttt{S}}_{b}^{\pi}\colsep{\triangleq}
\emptyset}{by (\ref{eq:trace-break})}\\

\hyphen{5}\formula{\sqb{\texttt{S}}_{\bot}^{\pi}}\\
$\triangleq$
\formulaexplanation{\Lfp{\subseteq}{F_{e}^{\pi}}\mathbin{{\fatsemi}^{\pi}}\sqb{\texttt{(x!=2)$\,\mathbin{{\fatsemi}^{\pi}}\,$if (x==1) then break else x=x+2}}_{\bot}^{\pi}\cup\Gfp{\subseteq}{F_{\bot}^{\pi}}}{by (\ref{eq:sem:abstract:loop-nontermination})}\\\
=
\formulaexplanation{\Gfp{\subseteq}{F_{\bot}^{\pi}}}{\texttt{(x!=2)$\,\mathbin{{\fatsemi}^{\pi}}\,$if (x==1) then break else x=x+2} always terminates}\\
=
\lastformula{\{\texttt{x}:n;\ldots;\texttt{x}:n+2k;\ldots\mid n>2\}}{\mbox{\qed}}
\end{calculus}
\let\qed\relax
\end{proof}

\begin{remark}\label{rem:least-versus-greatest-in-semantics}We follow \cite{DBLP:conf/popl/CousotC92} by using least fixpoints for finite traces and greatest fixpoints for infinite traces. We could, equivalently, definite finite traces by a greatest fixpoint as in \cite{DBLP:journals/iandc/LeroyG09}, since the least and greatest fixpoints are equal
$\Lfp{\subseteq}{F_{e}^{\pi}}=\Gfp{\subseteq}{F_{e}^{\pi}}$, which would look more uniform. However, the induction principles for least and greatest fixpoints are not the same. This would require proofs relative to finite executions to be done coinductively instead of the usual inductive reasonings by induction on the length of traces. A related problem is
that the abstraction theorems for least and greatest fixpoints are not the same \cite[Chapter 18]{Cousot-PAI-2021}. The abstraction of a least fixpoint is, in general, more precise than that of a greatest one. So if finite traces had been defined by a greatest fixpoint, it would be necessary to prove that it is equal to the least fixpoint before applying the appropriate abstractions. Then the greatest fixpoint characterization of the finite traces becomes useless. Least and greatest fixpoints can also be merged using the bi-inductive order of  \cite{DBLP:conf/popl/CousotC92} (which abstractions yield Egli-Milner and Scott order \cite{DBLP:journals/tcs/Cousot02}).
\end{remark}

\subsection{Bi-inductive Trace Semantics}\label{sec:Biinductive-Trace-Semantic-Domain}
The  trace semantics instantiation of (\ref{eq:def:Abstract-Semantic-Domain-Semantics}) is
\bgroup\abovedisplayskip3pt\belowdisplayskip3pt\begin{eqntabular}{rcl}
\sqb{\texttt{S}}^{\pi}&\triangleq&\triple{e:\sqb{\texttt{S}}_{e}^{\pi}}{\bot:\sqb{\texttt{S}}_{\bot}^{\pi}}{br:\sqb{\texttt{S}}_{b}^{\pi}}
\end{eqntabular}\egroup
belonging to the Cartesian product $:(e:\wp(\Sigma^{+})\times\bot:\wp(\Sigma^{\infty})\times{br:\wp(\Sigma^{+})})$
with named selectors $e$, $\bot$, and $br$. Since $\sqb{\texttt{S}}_{e}^{\pi}$ and $\sqb{\texttt{S}}_{\bot}^{\pi}$ are disjoint they can be put together as follows.
\bgroup\belowdisplayskip3pt\begin{eqntabular}{rcl}
\sqb{\texttt{S}}^{\pi}&\triangleq&\pair{ok:\sqb{\texttt{S}}_{e}^{\pi}\cup\sqb{\texttt{S}}_{\bot}^{\pi}}{br:\sqb{\texttt{S}}_{b}^{\pi}}
\end{eqntabular}\egroup
belonging to the Cartesian product
$ok:\wp(\Sigma^{+\infty})\times br:\wp(\Sigma^{+})$ with named selectors $ok$ and $br$. We can recover
$\sqb{\texttt{S}}_{e}^{\pi}=(\sqb{\texttt{S}}_{ok}^{\pi})\cap\Sigma^+$ and $\sqb{\texttt{S}}_{\bot}^{\pi}=(\sqb{\texttt{S}}_{ok}^{\pi})\cap\Sigma^\infty$.
Moreover, if $T=\pair{ok:Q}{br:B}\in ok:\wp(\Sigma^{+\infty})\times br:\wp(\Sigma^{+})$, then we define the shorthands 
\begin{eqntabular}{C}
$T_{ok}=Q$,\quad $T_{+}=Q\cap\Sigma^+$,\quad $T_{\infty}=Q\cap\Sigma^\infty$,\quad and\quad $T_{br}=B$.
\label{eq:shorthands}
\end{eqntabular}
Then the pairwise order on $(e:\wp(\Sigma^{+})\times\bot:\wp(\Sigma^{\infty}))$ becomes the computational ordering of \cite{DBLP:conf/popl/CousotC92,DBLP:journals/iandc/CousotC09} defined on $\sqb{\texttt{S}}_{e}^{\pi}\cup\sqb{\texttt{S}}_{\bot}^{\pi}$ as $X\sqsubseteq Y\triangleq (X\cap\Sigma^+\subseteq Y\cap\Sigma^+)\wedge (X\cap\Sigma^\infty\supseteq Y\cap\Sigma^\infty)$.
\end{toappendix}

Notice that the algebraic semantics can be instantiated to semantics of probabilistic and quantum programs. In this cases the hyperlogics developed in this paper, which differentiate between computational and approximation orders, apply to probabilistic programs \cite{DBLP:journals/entcs/RandZ15,DBLP:journals/ijfcs/HartogV02} and to quantum programs \cite{DBLP:journals/toplas/Ying11,DBLP:journals/pacmpl/YanJY22,DBLP:journals/iandc/FengL23}

\section{Structural Fixpoint Natural Relational Semantics}\label{sect:RelationalSemantics}

The structural fixpoint natural relational semantics of \cite[sect.\@ II.1]{DBLP:journals/pacmpl/Cousot24} is an instance
of the algebraic semantics of sect.\@ \ref{sect:Algebraic-Semantics}. Given states $\Sigma$, $\bot\not\in\Sigma$ denoting
nontermination, and $\Sigma_\bot\triangleq\Sigma\cup\{\bot\}$, the finitary domain 
${\mathbb{L}^{\varrho}_{+}}\triangleq\pair{\wp(\Sigma\times\Sigma)}{\subseteq}$ in \ref{def:abstract:domain:well:def:finite:domain}
and the infinitary domain ${\mathbb{L}^{\varrho}_{\infty}}\triangleq\pair{\wp(\Sigma\times\{\bot\})}{\subseteq}$ in \ref{def:abstract:domain:well:def:infinite:domain} 
are both complete lattices for set inclusion $\subseteq$ so ${\bot_{+}^{\varrho}}=\emptyset$. We let \mbox{$\mathbb{1}$ be the identity function}. The primitives \ref{def:abstract:domain:well:def:abstract:operators} are well-defined.
\bgroup\belowdisplayskip0pt\begin{eqntabular}{rcl@{\qquad}rcl@{\qquad}}
{\textsf{assign}^{\varrho}\sqb{\texttt{x},\texttt{A}}}&\triangleq& \{\pair{\sigma}{\sigma[\texttt{x}\leftarrow\mathcal{A}\sqb{\texttt{A}}\sigma]}\mid \sigma\in\Sigma\}
&
{\textsf{init}^{\varrho}}&\triangleq& \mathbb{1}\nonumber\\
{\textsf{rassign}^{\varrho}\sqb{\texttt{x},a,b}}&\triangleq& \{\pair{\sigma}{\sigma[\texttt{x}\leftarrow i]}\mid \sigma\in\Sigma\wedge a-1 <i< b+1\}
&
{\textsf{break}^{\varrho}}&\triangleq& \mathbb{1}\label{eq:relational-semantics-primitives}\stepcounter{equation}\renumber{\raisebox{-0.66em}[0pt][0pt]{(\ref{eq:relational-semantics-primitives})}}\\
{\textsf{test}^{\varrho}\sqb{\texttt{B}}}&\triangleq& \{\pair{\sigma}{\sigma}\mid\sigma\in\mathcal{B}\sqb{\texttt{B}}\}
&
{\textsf{skip}^{\varrho}}&\triangleq& \mathbb{1}\nonumber\\
r\mathbin{{\fatsemi}^{\varrho}}r'&\triangleq&\rlap{$\{\pair{x}{\bot}\mid \pair{x}{\bot}\in r\}\cup\{\pair{x}{y}\mid \exists z\in\Sigma\mathrel{.}\pair{x}{z}\in r\wedge \pair{z}{y}\in r'\}$}
\nonumber
\end{eqntabular}%
\abovedisplayskip-3pt%
\begin{eqntabular}[fl]{p{0.925\textwidth}@{\qquad}}
$\mathbin{{\fatsemi}^{\varrho}}$ left preserves arbitrary joins $\cup$ on $\wp(\Sigma\times\Sigma_{\bot})$.
$\mathbin{{\fatsemi}^{\varrho}}$ right preserves non empty joins $\cup$ on $\wp(\Sigma\times\Sigma_{\bot})$. $\mathbin{{\fatsemi}^{\varrho}}$ is right increasing (but not necessarily lower continuous for the finitary and infinitary domains)\proofinapx.\label{eq:fatsemi-varrho-additive}
\end{eqntabular}\egroup
\begin{toappendix}
\begin{proof}[Proof of (\ref{eq:fatsemi-varrho-additive})]\

\hyphen{5}\quad Let $\pair{X_i}{i\in\Delta}$ be a possibly empty family of elements of ${\wp(\Sigma\times\Sigma_{\bot})}$.
\begin{calculus}
\formula{(\bigcup_{i\in\Delta}X_i)\mathbin{{\fatsemi}^{\varrho}}r'}\\
=
\formulaexplanation{((\bigcup_{i\in\Delta}X_i\cap\wp(\Sigma\times\Sigma))\cup (\bigcup_{i\in\Delta}X_i\cap\wp(\Sigma\times\{\bot\})))\mathbin{{\fatsemi}^{\varrho}}r'}{def.\ $\wp(\Sigma\times\Sigma_{\bot})$}\\
=
\formulaexplanation{\{\pair{x}{\bot}\mid \pair{x}{\bot}\in (\bigcup_{i\in\Delta}X_i\cap\wp(\Sigma\times\{\bot\}))\}\cup\{\pair{x}{y}\mid \exists z\in\Sigma\mathrel{.}\pair{x}{z}\in (\bigcup_{i\in\Delta}X_i\cap\wp(\Sigma\times\Sigma))\wedge \pair{z}{y}\in r'\}}{def.\ $\mathbin{{\fatsemi}^{\varrho}}$, 
$\forall x\in\Sigma\mathrel{.}\pair{x}{\bot}\not\in\Sigma\times\Sigma$, and $\forall z\in\Sigma\mathrel{.}\pair{x}{z}\not\in \Sigma\times\{\bot\}$ since $\bot\not\in\Sigma$}\\[0.75ex]
=
\formula{\bigcup_{i\in\Delta}(\{\pair{x}{\bot}\mid \pair{x}{\bot}\in (X_i\cap\wp(\Sigma\times\{\bot\}))\}\cup\{\pair{x}{y}\mid \exists z\in\Sigma\mathrel{.}\pair{x}{z}\in (X_i\cap\wp(\Sigma\times\Sigma))\wedge \pair{z}{y}\in r'\})}\\[-1.5ex]\rightexplanation{def.\ $\cup$}\\
=
\formulaexplanation{\bigcup_{i\in\Delta}(\{\pair{x}{\bot}\mid \pair{x}{\bot}\in (X_i\cap\wp(\Sigma\times\Sigma))\cup(X_i\cap\wp(\Sigma\times\{\bot\}))\}\cup\{\pair{x}{y}\mid \exists z\in\Sigma\mathrel{.}\pair{x}{z}\in  (X_i\cap\wp(\Sigma\times\Sigma))\cup (\bigcup_{i\in\Delta}X_i\cap\wp(\Sigma\times\{\bot\}))\wedge \pair{z}{y}\in r'\})}{$\bot\not\in\Sigma$}\\
=
\formulaexplanation{\bigcup_{i\in\Delta}(\{\pair{x}{\bot}\mid \pair{x}{\bot}\in X_i\}\cup\{\pair{x}{y}\mid \exists z\in\Sigma\mathrel{.}\pair{x}{z}\in  X_i\wedge \pair{z}{y}\in r'\})}{def.\ $\wp(\Sigma\times\Sigma_{\bot})$}\\
=
\formulaexplanation{\bigcup_{i\in\Delta}(X_i\mathbin{{\fatsemi}^{\varrho}}r')}{def.\ $\mathbin{{\fatsemi}^{\varrho}}$, Q.E.D.}
\end{calculus}

\smallskip

\noindent Notice that if $\Delta=\emptyset$ then $\emptyset\mathbin{{\fatsemi}^{\varrho}}r'=\emptyset$.

\medskip

\hyphen{5}\quad Let $\pair{X_i}{i\in\Delta}$ be a nonempty family of elements of ${\wp(\Sigma\times\Sigma_{\bot})}\setminus\{\emptyset\}$.
\begin{calculus}
\formula{r\mathbin{{\fatsemi}^{\varrho}}(\bigcup_{i\in\Delta}X_i)}\\
=
\formulaexplanation{\{\pair{x}{\bot}\mid \pair{x}{\bot}\in r\}\cup\{\pair{x}{y}\mid \exists z\in\Sigma\mathrel{.}\pair{x}{z}\in r\wedge \pair{z}{y}\in (\bigcup_{i\in\Delta}X_i)\}}{def.\ $\mathbin{{\fatsemi}^{\varrho}}$}\\
=
\formulaexplanation{\bigcup_{i\in\Delta}(\{\pair{x}{\bot}\mid \pair{x}{\bot}\in r\}\cup\{\pair{x}{y}\mid \exists z\in\Sigma\mathrel{.}\pair{x}{z}\in r\wedge \pair{z}{y}\in X_i\})}{def.\ $\cup$}\\

=
\formulaexplanation{\bigcup_{i\in\Delta}(r\mathbin{{\fatsemi}^{\varrho}}X_i)}{def.\ $\mathbin{{\fatsemi}^{\varrho}}$, Q.E.D.}
\end{calculus}
\smallskip

\noindent If $\Delta=\emptyset$ then $r\mathbin{{\fatsemi}^{\varrho}}(\bigcup_{i\in\Delta}X_i)=r\mathbin{{\fatsemi}^{\varrho}}\emptyset=\{\pair{x}{\bot}\mid \pair{x}{\bot}\in r\}$ which, in general is not empty, while
$\bigcup_{i\in\Delta}(r\mathbin{{\fatsemi}^{\varrho}}X_i)=\emptyset$.

\medskip

\hyphen{5}\quad The following counter example shows that if $\pair{X^i\in\wp(\Sigma\times\Sigma)}{i\in\mathbb{N}}$ is a decreasing
chain and $r\in\wp(\Sigma\times\Sigma)$, we may have $r\mathbin{{\fatsemi}^{\varrho}}(\bigcap_{i\in\mathbb{N}}X^i)$ $\neq$
$\bigcap_{i\in\mathbb{N}}(r\mathbin{{\fatsemi}^{\varrho}}X^i)$. 

Take $r\triangleq\{\bar{\sigma}\}\times\Sigma$ and $X^i=\{\pair{\sigma_j}{\bar{\sigma}}\mid j\geqslant i\}$ (that is $X^0=\{\pair{\sigma_0}{\bar{\sigma}},\pair{\sigma_1}{\bar{\sigma}},
\pair{\sigma_2}{\bar{\sigma}},\ldots\}$, $X^1=\{\pair{\sigma_1}{\bar{\sigma}},
\pair{\sigma_2}{\bar{\sigma}},\ldots\}$, $X^2=\{
\pair{\sigma_2}{\bar{\sigma}},\ldots\}$, etc). Then $r\mathbin{{\fatsemi}^{\varrho}}(\bigcap_{i\in\mathbb{N}}X^i)=r\mathbin{{\fatsemi}^{\varrho}}\emptyset=\emptyset$ while
$\bigcap_{i\in\mathbb{N}}(r\mathbin{{\fatsemi}^{\varrho}}X^i)=\bigcap_{i\in\mathbb{N}}\{\pair{\bar{\sigma}}{\bar{\sigma}}\}=\{\pair{\bar{\sigma}}{\bar{\sigma}}\}$.
\end{proof}
\end{toappendix}
\begin{example}\label{ex:example1-relational}Define $\texttt{S}_1\triangleq\texttt{while (y!=0) y=y-1;}$ with relational semantics
\begin{eqntabular*}{rcl}
\sqb{\texttt{S}_1}^\varrho &=&\triple{e:{\{\pair{\sigma}{\sigma[\texttt{y}\leftarrow 0]}
\mid\sigma(\texttt{y})\geqslant 0\}}}{\bot:\{\pair{\sigma}{\bot}\mid  \sigma(\texttt{y})<0\}}{br:\emptyset}
\end{eqntabular*}
meaning that $\texttt{S}_1$ terminates with $\texttt{y}=0$ when \texttt{y} is initially positive and otherwise does not terminate. 

 Define $\texttt{S}_2$ $\triangleq$ \texttt{y=[-oo,oo]; S}$_1$ with relational semantics
 \begin{eqntabular*}{rcl}
\sqb{\texttt{S}_2}^\varrho
&=&
\triple{e:{\{\pair{\sigma}{\sigma[\texttt{y}\leftarrow 0]}
\mid \sigma\in\Sigma\}}}{\bot:\{\pair{\sigma}{\bot}\mid  \sigma\in\Sigma\}}{br:\emptyset}
\end{eqntabular*}
meaning that either \texttt{S}$_2$
terminates with \texttt{y}=0 or does not terminate \proofinapx.
\end{example}
\begin{toappendix}
\begin{proof}[Proof of example \ref{ex:example1-relational}]
\begin{calculus}
\hyphen{5}\formula{\sqb{\texttt{y!=0;y=y-1;}}_{e}^{\varrho}}\\
=
\formulaexplanation{\sqb{\texttt{y!=0}}_{e}^{\varrho}\mathbin{\,\fatsemi^\varrho\,}\sqb{\texttt{y=y-1;}}_{e}^{\varrho}}{(\ref{eq:def:abstract:sem:seq})}\\
=
\formulaexplanation{\{\pair{\sigma}{\sigma}\mid\sigma(\texttt{y})\neq0\}\mathbin{\,\fatsemi^\varrho\,}\{\pair{\sigma}{\sigma[\texttt{y}\leftarrow\sigma(\texttt{y})-1]}\mid \sigma\in\Sigma\}}{(\ref{eq:def:sem:abstract:basis}) and (\ref{eq:relational-semantics-primitives})}\\
=
\formulaexplanation{\{\pair{\sigma}{\sigma[\texttt{y}\leftarrow\sigma(\texttt{y})-1]}\mid \sigma(\texttt{y})\neq0\}}{def.\ (\ref{eq:relational-semantics-primitives}) of $\fatsemi^\varrho$}\\[1ex]

\hyphen{5}\formulaexplanation{\Cev{F}_{e}^{\varrho}}{for $\texttt{S}_1$ = \texttt{while (y!=0) y=y-1;}}\\
$\triangleq$
\formulaexplanation{\LAMBDA{X\in\wp(\Sigma\times\Sigma)}{\textsf{init}^{\varrho}} \mathbin{\sqcup_{+}^{\varrho}} (\sqb{\texttt{y!=0;y=y-1;}}_{e}^{\varrho}\mathbin{{\fatsemi}^{\varrho}} X)}{(\ref{eq:natural-transformer-finite-backward})}\\
=
\formulaexplanation{\LAMBDA{X\in\wp(\Sigma\times\Sigma)}\{\pair{\sigma}{\sigma}\mid\sigma\in\Sigma\} \cup (\{\pair{\sigma}{\sigma[\texttt{y}\leftarrow\sigma(\texttt{y})-1]}\mid \sigma(\texttt{y})\neq0\}\mathbin{{\fatsemi}^{\varrho}} X)}{(\ref{eq:relational-semantics-primitives})}\\
=
\formulaexplanation{\LAMBDA{X\in\wp(\Sigma\times\Sigma)}\{\pair{\sigma}{\sigma}\mid\sigma\in\Sigma\} \cup \{\pair{\sigma}{\sigma'}\mid
\sigma(\texttt{y})\neq0\wedge\pair{\sigma[\texttt{y}\leftarrow\sigma(\texttt{y})-1]}{\sigma'}\in X\}}{(\ref{eq:relational-semantics-primitives})}\\[-0.75ex]
\end{calculus}

\hyphen{5}\ \ By (\ref{eq:fatsemi-varrho-additive}) and (\ref{eq:natural-transformer-finite-backward}), ${\Cev{F}_{e}^{\varrho}}$  for
$\texttt{S}_1$ = \texttt{while (y!=0) y=y-1;}  preserves
nonempty joins $\cup$ so that the infinite  iterates $\pair{X^i}{i\leqslant\omega}$ of $\Lfp{\subseteq}{\Cev{F}_{e}^{\varrho}}$ are as follows
\begin{calculus}[$X^{n+1}$\ =\ ]
$X^0$ = \formula{\emptyset}\\
$X^1$ = \formula{\{\pair{\sigma}{\sigma}\mid\sigma\in\Sigma\} }\\
$X^2$ = \formulaexplanation{\{\pair{\sigma}{\sigma}\mid\sigma\in\Sigma\} \cup \{\pair{\sigma}{\sigma'}\mid
\sigma(\texttt{y})\neq0\wedge\pair{\sigma[\texttt{y}\leftarrow\sigma(\texttt{y})-1]}{\sigma'}\in X^1\}}{def.\ iterates}\\
\phantom{$X^2$} = \formulaexplanation{\{\pair{\sigma}{\sigma}\mid\sigma\in\Sigma\} \cup \{\pair{\sigma}{\sigma[\texttt{y}\leftarrow\sigma(\texttt{y})-1]}\mid
\sigma(\texttt{y})\neq0\}}{def. $X^1$}\\

$X^3$ = \formulaexplanation{\{\pair{\sigma}{\sigma}\mid\sigma\in\Sigma\} \cup \{\pair{\sigma}{\sigma'}\mid
\sigma(\texttt{y})\neq0\wedge\pair{\sigma[\texttt{y}\leftarrow\sigma(\texttt{y})-1]}{\sigma'}\in X^2\}}{def.\ iterates}\\

\phantom{$X^3$} = \formulaexplanation{\{\pair{\sigma}{\sigma}\mid\sigma\in\Sigma\} \cup \{\pair{\sigma}{\sigma'}\mid
\sigma(\texttt{y})\neq0\wedge\pair{\sigma[\texttt{y}\leftarrow\sigma(\texttt{y})-1]}{\sigma'}\in (\{\pair{\sigma'}{\sigma'}\mid\sigma'\in\Sigma\} \cup \{\pair{\sigma'}{\sigma'[\texttt{y}\leftarrow\sigma'(\texttt{y})-1]}\mid
\sigma'(\texttt{y})\neq0\})\}}{def. $X^2$}\\

\phantom{$X^3$} = \formulaexplanation{\{\pair{\sigma}{\sigma}\mid\sigma\in\Sigma\} 
\cup \{\pair{\sigma}{\sigma[\texttt{y}\leftarrow\sigma(\texttt{y})-1]}\mid\sigma(\texttt{y})\neq0\}
\cup 
\{\pair{\sigma}{\sigma''}\mid
\sigma(\texttt{y})\neq0\wedge\pair{\sigma[\texttt{y}\leftarrow\sigma(\texttt{y})-1]}{\sigma''}\in \{\pair{\sigma'}{\sigma'[\texttt{y}\leftarrow\sigma'(\texttt{y})-1]}\mid
\sigma'(\texttt{y})\neq0\}\}
}{def.\ $\cup$}\\

\phantom{$X^3$} = \formulaexplanation{\{\pair{\sigma}{\sigma}\mid\sigma\in\Sigma\} 
\cup \{\pair{\sigma}{\sigma[\texttt{y}\leftarrow\sigma(\texttt{y})-1]}\mid\sigma(\texttt{y})\neq0\}
\cup 
\{\pair{\sigma}{\sigma''}\mid
\sigma(\texttt{y})\neq0\wedge\pair{\sigma[\texttt{y}\leftarrow\sigma(\texttt{y})-1]}{\sigma''}\in \{\pair{{\sigma[\texttt{y}\leftarrow\sigma(\texttt{y})-1]}}{{\sigma[\texttt{y}\leftarrow\sigma(\texttt{y})-1]}[\texttt{y}\leftarrow{\sigma[\texttt{y}\leftarrow\sigma(\texttt{y})-1]}(\texttt{y})-1]}\mid
{\sigma[\texttt{y}\leftarrow\sigma(\texttt{y})-1]}(\texttt{y})\neq0\}\}
}{def.\ $\in$}\\

\phantom{$X^3$} = \formulaexplanation{\{\pair{\sigma}{\sigma}\mid\sigma\in\Sigma\} 
\cup \{\pair{\sigma}{\sigma[\texttt{y}\leftarrow\sigma(\texttt{y})-1]}\mid\sigma(\texttt{y})\neq0\}
\cup 
\{\pair{\sigma}{{\sigma[\texttt{y}\leftarrow\sigma(\texttt{y})-1]}[\texttt{y}\leftarrow{\sigma[\texttt{y}\leftarrow\sigma(\texttt{y})-1]}(\texttt{y})-1]}\mid
\sigma(\texttt{y})\neq0\wedge
{\sigma[\texttt{y}\leftarrow\sigma(\texttt{y})-1]}(\texttt{y})\neq0\}
}{def.\ $\in$}\\

\phantom{$X^3$} = \formulaexplanation{\{\pair{\sigma}{\sigma}\mid\sigma\in\Sigma\} 
\cup \{\pair{\sigma}{\sigma[\texttt{y}\leftarrow\sigma(\texttt{y})-1]}\mid\sigma(\texttt{y})\neq0\}
\cup 
\{\pair{\sigma}{{\sigma[\texttt{y}\leftarrow\sigma(\texttt{y})-2]}}\mid
\sigma(\texttt{y})\neq0\wedge\sigma(\texttt{y})\neq1\}
}{simplification}\\
{$X^n$} = \formulaexplanation{\{\pair{\sigma}{\sigma}\mid\sigma\in\Sigma\} 
\cup \bigcup_{i=1}^{n-1}\{\pair{\sigma}{{\sigma[\texttt{y}\leftarrow\sigma(\texttt{y})-i]}}\mid
\bigwedge_{j=0}^{i-1}\sigma(\texttt{y})\neq j\}}{induction hypothesis}\\

$X^{n+1}$ = \formulaexplanation{\{\pair{\sigma}{\sigma}\mid\sigma\in\Sigma\} \cup \{\pair{\sigma}{\sigma'}\mid
\sigma(\texttt{y})\neq0\wedge\pair{\sigma[\texttt{y}\leftarrow\sigma(\texttt{y})-1]}{\sigma'}\in X^{n}\}}{def.\ iterates}\\

\phantom{$X^{n+1}$} = \formulaexplanation{\{\pair{\sigma}{\sigma}\mid\sigma\in\Sigma\} \cup \{\pair{\sigma}{\sigma'}\mid
\sigma(\texttt{y})\neq0\wedge\pair{\sigma[\texttt{y}\leftarrow\sigma(\texttt{y})-1]}{\sigma'}\in (\{\pair{\sigma}{\sigma}\mid\sigma\in\Sigma\} 
\cup \bigcup_{i=1}^{n-1}\{\pair{\sigma}{{\sigma[\texttt{y}\leftarrow\sigma(\texttt{y})-i]}}\mid
\bigwedge_{j=0}^{i-1}\sigma(\texttt{y})\neq j\})\}}{def.\ $X^n$}\\

\phantom{$X^{n+1}$} = \formulaexplanation{\{\pair{\sigma}{\sigma}\mid\sigma\in\Sigma\} \cup \{\pair{\sigma}{\sigma[\texttt{y}\leftarrow\sigma(\texttt{y})-1]}\mid
\sigma(\texttt{y})\neq0\}
\cup
\{\pair{\sigma}{\sigma'}\mid
\sigma(\texttt{y})\neq0\wedge\pair{\sigma[\texttt{y}\leftarrow\sigma(\texttt{y})-1]}{\sigma'}\in \bigcup_{i=1}^{n-1}\{\pair{\sigma''}{{\sigma''[\texttt{y}\leftarrow\sigma''(\texttt{y})-i]}}\mid
\bigwedge_{j=0}^{i-1}\sigma''(\texttt{y})\neq j\}\}
}{def.\ $\cup$, renaming}\\

\phantom{$X^{n+1}$} = \formulaexplanation{\{\pair{\sigma}{\sigma}\mid\sigma\in\Sigma\} \cup \{\pair{\sigma}{\sigma[\texttt{y}\leftarrow\sigma(\texttt{y})-1]}\mid
\sigma(\texttt{y})\neq0\}
\cup
\bigcup_{i=1}^{n-1}\{\pair{\sigma}{\sigma'}\mid
\sigma(\texttt{y})\neq0\wedge\pair{\sigma[\texttt{y}\leftarrow\sigma(\texttt{y})-1]}{\sigma'}\in \{\pair{\sigma''}{{\sigma''[\texttt{y}\leftarrow\sigma''(\texttt{y})-i]}}\mid
\bigwedge_{j=0}^{i-1}\sigma''(\texttt{y})\neq j\}\}
}{def.\ $\cup$}\\

\phantom{$X^{n+1}$} = \formula{\{\pair{\sigma}{\sigma}\mid\sigma\in\Sigma\} \cup \{\pair{\sigma}{\sigma[\texttt{y}\leftarrow\sigma(\texttt{y})-1]}\mid
\sigma(\texttt{y})\neq0\}
\cup
\bigcup_{i=1}^{n-1}\{\pair{\sigma}{\sigma[\texttt{y}\leftarrow\sigma(\texttt{y})-(i+1)]}\mid
\sigma(\texttt{y})\neq0\wedge
\bigwedge_{j=0}^{i-1}{\sigma[\texttt{y}\leftarrow\sigma(\texttt{y})-1]}(\texttt{y})\neq j\}
}\\
\explanation{def.\ $\in$ so $\sigma''={\sigma[\texttt{y}\leftarrow\sigma(\texttt{y})-1]}$ and $\sigma'={{\sigma''[\texttt{y}\leftarrow\sigma''(\texttt{y})-i]}}={{{\sigma[\texttt{y}\leftarrow\sigma(\texttt{y})-1]}[\texttt{y}\leftarrow{\sigma[\texttt{y}\leftarrow\sigma(\texttt{y})-1]}(\texttt{y})-i]}}={\sigma[\texttt{y}\leftarrow\sigma(\texttt{y})-(i+1)]}$}\\

\phantom{$X^{n+1}$} = \formulaexplanation{\{\pair{\sigma}{\sigma}\mid\sigma\in\Sigma\} \cup \{\pair{\sigma}{\sigma[\texttt{y}\leftarrow\sigma(\texttt{y})-1]}\mid
\sigma(\texttt{y})\neq0\}
\cup
\bigcup_{i=1}^{n-1}\{\pair{\sigma}{\sigma[\texttt{y}\leftarrow\sigma(\texttt{y})-(i+1)]}\mid
\sigma(\texttt{y})\neq0\wedge
\bigwedge_{j=0}^{i-1}{\sigma(\texttt{y})}\neq j+1\}
}{simplification}\\

\phantom{$X^{n+1}$} = \formula{\{\pair{\sigma}{\sigma}\mid\sigma\in\Sigma\} \cup \{\pair{\sigma}{\sigma[\texttt{y}\leftarrow\sigma(\texttt{y})-1]}\mid
\sigma(\texttt{y})\neq0\}
\cup
\bigcup_{i=1}^{n-1}\{\pair{\sigma}{\sigma[\texttt{y}\leftarrow\sigma(\texttt{y})-(i+1)]}\mid
\bigwedge_{j=0}^{i}{\sigma(\texttt{y})}\neq j\}
}\\
\rightexplanation{change of dummy variable and incorporation of $\sigma(\texttt{y})\neq0$ in the conjunction for $j=0$}\\

\phantom{$X^{n+1}$} = \formula{\{\pair{\sigma}{\sigma}\mid\sigma\in\Sigma\} \
\cup
\bigcup_{i=0}^{n-1}\{\pair{\sigma}{\sigma[\texttt{y}\leftarrow\sigma(\texttt{y})-(i+1)]}\mid
\bigwedge_{j=0}^{i}{\sigma(\texttt{y})}\neq j\}\}
}\\
\rightexplanation{incorporation of $\{\pair{\sigma}{\sigma[\texttt{y}\leftarrow\sigma(\texttt{y})-1]}\mid \sigma(\texttt{y})\neq0\}$ in the union for $i=0$}\\

\phantom{$X^{n+1}$} = \formula{\{\pair{\sigma}{\sigma}\mid\sigma\in\Sigma\} \
\cup
\bigcup_{i=1}^{(n+1)-1}\{\pair{\sigma}{\sigma[\texttt{y}\leftarrow\sigma(\texttt{y})-i]}\mid
\bigwedge_{j=0}^{i-1}{\sigma(\texttt{y})}\neq j\}
}\\[-0.75ex]\rightexplanation{change of dummy variables}\\[-0.5ex]
\end{calculus}
\hyphen{5}\ \ By recurrence, $X^n=\{\pair{\sigma}{\sigma}\mid\sigma\in\Sigma\} 
\cup \bigcup_{i=1}^{n-1}\{\pair{\sigma}{{\sigma[\texttt{y}\leftarrow\sigma(\texttt{y})-i]}}\mid
\bigwedge_{j=0}^{i-1}\sigma(\texttt{y})\neq j\}$, so that the least fixpoint of ${\Cev{F}_{e}^{\varrho}}$ for
$\texttt{S}_1$ = \texttt{while (y!=0) y=y-1;}  is
\begin{calculus}
\formula{\Lfp{\subseteq}{\Cev{F}_{e}^{\varrho}}}\\
=
\formulaexplanation{\bigcup_{n\in\mathbb{N}}X^n}{def.\ iterates}\\
=
\formula{\{\pair{\sigma}{\sigma}\mid\sigma\in\Sigma\} \
\cup
\bigcup_{n\in\mathbb{N}}\bigcup_{i=1}^{n}\{\pair{\sigma}{\sigma[\texttt{y}\leftarrow\sigma(\texttt{y})-i]}\mid
\bigwedge_{j=0}^{i-1}{\sigma(\texttt{y})}\neq j\}
}\\
=
\formula{\{\pair{\sigma}{\sigma}\mid\sigma\in\Sigma\} \
\cup
\bigcup_{i>0}\{\pair{\sigma}{\sigma[\texttt{y}\leftarrow\sigma(\texttt{y})-i]}\mid
{\sigma(\texttt{y})}\not\in\interval{0}{i-1}\}
}\\[-0.5ex]
\end{calculus}
\hyphen{5}\quad It follows that for $\texttt{S}_1\triangleq\texttt{while (y!=0) y=y-1;}$, we have
\begin{calculus}
\formula{\sqb{\texttt{S}_1}_{e}^{\varrho}}\\
=
\formulaexplanation{
\Lfp{\subseteq}{\Cev{F}_{e}^{\varrho}}\mathbin{{\fatsemi}^{\varrho}}(\sqb{\neg\texttt{B}}_{e}^{\varrho}\cup\sqb{\texttt{B;S}}_{b}^{\varrho})}{by (\ref{eq:sem:abstract:finite}) with \texttt{B} = \texttt{(y!=0)}, $\neg$\texttt{B} = \texttt{(y=0)}, and \texttt{S} =  \texttt{y=y-1;}}\\
=
\formula{(\{\pair{\sigma}{\sigma}\mid\sigma\in\Sigma\} \
\cup
\bigcup_{i>0}\{\pair{\sigma}{\sigma[\texttt{y}\leftarrow\sigma(\texttt{y})-i]}\mid
{\sigma(\texttt{y})}\not\in\interval{0}{i-1}\}
)\mathbin{{\fatsemi}^{\varrho}}(\{\pair{\sigma}{\sigma}\mid\sigma(\texttt{y})=0\}\cup\emptyset)}\\
=
\formulaexplanation{\{\pair{\sigma}{\sigma}\mid\sigma(\texttt{y})=0\} \
\cup
\bigcup_{i>0}\{\pair{\sigma}{\sigma[\texttt{y}\leftarrow\sigma(\texttt{y})-i]}\mid
{\sigma(\texttt{y})}\not\in\interval{0}{i-1}\wedge{\sigma[\texttt{y}\leftarrow\sigma(\texttt{y})-i]}(\texttt{y})=0\}}{(\ref{eq:relational-semantics-primitives})}\\
=
\formula{\{\pair{\sigma}{\sigma}\mid\sigma(\texttt{y})=0\} \
\cup
\bigcup_{i>0}\{\pair{\sigma}{\sigma[\texttt{y}\leftarrow\sigma(\texttt{y})-i]}\mid
{\sigma(\texttt{y})}\not\in\interval{0}{i-1}\wedge\sigma(\texttt{y})=i\}}\\\rightexplanation{function application}\\
=
\formula{\{\pair{\sigma}{\sigma[\texttt{y}\leftarrow 0]}\mid\sigma(\texttt{y})=0\} \
\cup
\bigcup_{i>0}\{\pair{\sigma}{\sigma[\texttt{y}\leftarrow 0]}\mid
\sigma(\texttt{y})=i\}}\\\rightexplanation{substitution $\sigma(\texttt{y})=i$}\\
=
\formulaexplanation{\{\pair{\sigma}{\sigma[\texttt{y}\leftarrow 0]}\mid
\sigma(\texttt{y})\geqslant0\}}{joining cases}
\end{calculus}

\medskip

\hyphen{5}\quad It follows that for  $\texttt{S}_2$ = \texttt{y=[-oo,oo]; S}$_1$, we have
\begin{calculus}
\formula{\sqb{\texttt{S}_2}_{e}^{\varrho}}\\
=
\formulaexplanation{\sqb{\texttt{y=[-oo,oo];}}_{e}^{\varrho}\mathbin{{\fatsemi}^{\varrho}}\sqb{\texttt{S}_1}_{e}^{\varrho}}{(\ref{eq:def:abstract:sem:seq})}\\
=
\formulaexplanation{\{\pair{\sigma}{\sigma[\texttt{y}\leftarrow n]}\mid n\in\mathbb{N}\}\mathbin{{\fatsemi}^{\varrho}}\{\pair{\sigma}{\sigma[\texttt{y}\leftarrow 0]}\mid
\sigma(\texttt{y})\geqslant0\}}{(\ref{eq:def:sem:abstract:basis}) and as previously shown}\\
=
\formulaexplanation{\{\pair{\sigma}{\sigma[\texttt{y}\leftarrow n][\texttt{y}\leftarrow 0]}
\mid n\in\mathbb{N}\wedge
\sigma[\texttt{y}\leftarrow n](\texttt{y})\geqslant0\}}{(\ref{eq:relational-semantics-primitives})}\\
=
\formulaexplanation{\{\pair{\sigma}{\sigma[\texttt{y}\leftarrow 0]}
\mid \sigma\in\Sigma\}}{simplification}
\end{calculus}

\medskip

\hyphen{5}\quad By (\ref{eq:trace-transformer-infinite}), we have
\begin{calculus}
\formulaexplanation{{F_{\bot}^{\varrho}}}{for $\texttt{S}_1$ = \texttt{while (y!=0) y=y-1;}}\\
=
\formula{\LAMBDA{X\in{\mathbb{L}^{\varrho}_{\infty}}}{\sqb{\texttt{y!=0;y=y-1;}}_{e}^{\varrho}}\mathbin{{\fatsemi}^{\varrho}} X}\\
=
\formula{\LAMBDA{X\in\wp(\Sigma\times\{\bot\})}\{\pair{\sigma}{\sigma[\texttt{y}\leftarrow\sigma(\texttt{y})-1]}\mid \sigma(\texttt{y})\neq0\}\mathbin{{\fatsemi}^{\varrho}} X}\\
=
\formulaexplanation{\{\pair{x}{y}\mid \exists z\in\Sigma\mathrel{.}\pair{x}{z}\in \{\pair{\sigma}{\sigma[\texttt{y}\leftarrow\sigma(\texttt{y})-1]}\mid \sigma(\texttt{y})\neq0\}\wedge \pair{z}{y}\in X\}}{(\ref{eq:relational-semantics-primitives})}\\
=
\formulaexplanation{\{\pair{\sigma}{y}\mid \sigma(\texttt{y})\neq0\wedge \pair{\sigma[\texttt{y}\leftarrow\sigma(\texttt{y})-1]}{y}\in X\}}{def.\ $\in$}\\
=
\formulaexplanation{\{\pair{\sigma}{\bot}\mid \sigma(\texttt{y})\neq0\wedge \pair{\sigma[\texttt{y}\leftarrow\sigma(\texttt{y})-1]}{\bot}\in X\}}{$X\in\wp(\Sigma\times\{\bot\})$}
\end{calculus}

\medskip

\hyphen{5}\ \ By (\ref{eq:fatsemi-varrho-additive}) and (\ref{eq:natural-transformer-finite-backward}), ${{F_{\bot}^{\varrho}}}$ for $\texttt{S}_1$ = \texttt{while (y!=0) y=y-1;} converges at $\omega$ so that the infinite  iterates $\pair{X^i}{i\leqslant\omega}$ of $\sqb{\texttt{S}_1}_{li}^{\varrho}$
=
$\Gfp{\subseteq}{{F_{\bot}^{\varrho}}}$ are as follows
\begin{calculus}[$X^{n+1}$\ =\ ]
$X^0$ = \formula{\Sigma\times\{\bot\}}\\

$X^1$ = \formulaexplanation{\{\pair{\sigma}{\bot}\mid \sigma(\texttt{y})\neq0\wedge \pair{\sigma[\texttt{y}\leftarrow\sigma(\texttt{y})-1]}{\bot}\in \Sigma\times\{\bot\}\}}{def.\ iterates}\\
\phantom{$X^1$} = \formulaexplanation{\{\pair{\sigma}{\bot}\mid \sigma(\texttt{y})\neq0\}}{simplification}\\

$X^2$ = \formulaexplanation{\{\pair{\sigma}{\bot}\mid \sigma(\texttt{y})\neq0\wedge \pair{\sigma[\texttt{y}\leftarrow\sigma(\texttt{y})-1]}{\bot}\in \{\pair{\sigma'}{\bot}\mid \sigma'(\texttt{y})\neq0\}\}}{def.\ iterates}\\
\phantom{$X^2$} = \formulaexplanation{\{\pair{\sigma}{\bot}\mid \sigma(\texttt{y})\neq0\wedge {\sigma[\texttt{y}\leftarrow\sigma(\texttt{y})-1]}(\texttt{y})\neq0\}}{def.\ $\in$}\\
\phantom{$X^2$} = \formulaexplanation{\{\pair{\sigma}{\bot}\mid \sigma(\texttt{y})\neq0\wedge \sigma(\texttt{y})\neq1\}}{function application}\\

$X^n$ = \formulaexplanation{\{\pair{\sigma}{\bot}\mid \bigwedge_{i=0}^{n-1}\sigma(\texttt{y})\neq i\}}{induction hypothesis}\\[-1ex]

$X^{n+1}$ = \formulaexplanation{\{\pair{\sigma}{\bot}\mid \sigma(\texttt{y})\neq0\wedge \pair{\sigma[\texttt{y}\leftarrow\sigma(\texttt{y})-1]}{\bot}\in \{\pair{\sigma}{\bot}\mid \bigwedge_{i=0}^{n-1}\sigma(\texttt{y})\neq i\}\}}{def.\ iterates}\\
\phantom{$X^{n+1}$} = \formulaexplanation{\{\pair{\sigma}{\bot}\mid \sigma(\texttt{y})\neq0\wedge \bigwedge_{i=0}^{n-1}{\sigma[\texttt{y}\leftarrow\sigma(\texttt{y})-1]}(\texttt{y})\neq i\}}{def.\ $\in$}\\
\phantom{$X^{n+1}$} = \formulaexplanation{\{\pair{\sigma}{\bot}\mid \sigma(\texttt{y})\neq0\wedge \bigwedge_{i=0}^{n-1}\sigma(\texttt{y})\neq i+1\}}
{function application}\\
\phantom{$X^{n+1}$} = \formulaexplanation{\{\pair{\sigma}{\bot}\mid \sigma(\texttt{y})\neq0\wedge \bigwedge_{j=1}^{n}\sigma(\texttt{y})\neq j\}}
{change of dummy variables $j=i+1$}\\
\phantom{$X^{n+1}$} = \formulaexplanation{\{\pair{\sigma}{\bot}\mid \bigwedge_{i=0}^{(n+1)-1}\sigma(\texttt{y})\neq i\}}
{grouping terms and renaming}\\[-0.5ex]
\end{calculus}
\hyphen{5}\ \ By recurrence, $X^n=\{\pair{\sigma}{\bot}\mid \bigwedge_{i=0}^{n-1}\sigma(\texttt{y})\neq i\}$, so that,  
by convergence at $\omega$, the greatest fixpoint is
\begin{calculus}
\formula{\sqb{\texttt{S}_1}_{li}^{\varrho}\colsep{=}\Gfp{\subseteq}{F_{\bot}^{\varrho}}}\\
=
\formulaexplanation{\bigcap_{n\in\mathbb{N}}X^n}{def.\ iterates}\\
=
\formula{\bigcap_{n\in\mathbb{N}}\{\pair{\sigma}{\bot}\mid \bigwedge_{i=0}^{n-1}\sigma(\texttt{y})\neq i\}}\\
=
\formulaexplanation{\{\pair{\sigma}{\bot}\mid \sigma(\texttt{y})<0\}}{$\Sigma=\{\texttt{x},\texttt{y}\}\rightarrow\mathbb{Z}$}\\[-1ex]
\end{calculus}
\hyphen{5}\ \ Obviously $\sqb{\texttt{(y!=0); y=y-1;}}_{\bot}^{\varrho}=\emptyset$ since the body always terminates, so that, by (\ref{eq:sem:abstract:body-infinite}), we have
 $\sqb{\texttt{while (y!=0) y=y-1;}}_{bi}^{\varrho}$
 $\triangleq$
$\Lfp{\sqsubseteq_{+}^{\varrho}}{\Cev{F}_{e}^{\varrho}}\mathbin{{\fatsemi}^{\varrho}}\sqb{\texttt{(y!=0); y=y-1;}}_{\bot}^{\varrho}$ 
=
$\Lfp{\subseteq}{\Cev{F}_{e}^{\varrho}}\mathbin{{\fatsemi}^{\varrho}}\emptyset$
=
$\emptyset$. By (\ref{eq:sem:abstract:loop-nontermination}), we have
$\sqb{\texttt{while (y!=0) y=y-1;}}_{\bot}^{\varrho}$
$\triangleq$
$\sqb{\texttt{while (y!=0) y=y-1;}}_{bi}^{\varrho}\cup\sqb{\texttt{while (y!=0) y=y-1;}}_{li}^{\varrho}$
=
$\{\pair{\sigma}{\bot}\mid \sigma(\texttt{y})<0\}$. 

\smallskip

\hyphen{5}\ \ It follows that
\begin{calculus}
\formula{\sqb{\texttt{S}_2}_{\bot}^{\varrho}}\\
=
\formula{\sqb{\texttt{y=[-oo,oo]; S$_1$}}_{\bot}^{\varrho}}\\
=
\formulaexplanation{\sqb{\texttt{y=[-oo,oo];}}_{\bot}^{\varrho} \cup(\sqb{\texttt{y=[-oo,oo];}}_{e}^{\varrho}\mathbin{\fatsemi^{\varrho}}\sqb{\texttt{S}_1}_{\bot}^{\varrho})}{(\ref{eq:def:abstract:sem:seq})}\\
=
\formulaexplanation{\emptyset \cup(\{\pair{\sigma}{\sigma[\texttt{y}\leftarrow i]}\mid \sigma\in\Sigma\wedge i\in\mathbb{N}\}\mathbin{\fatsemi^{\varrho}}\{\pair{\sigma}{\bot}\mid \sigma(\texttt{y})<0\})}{(\ref{eq:def:abstract:sem:seq}), (\ref{eq:def:sem:abstract:basis}) and (\ref{eq:relational-semantics-primitives})}\\
=
\formulaexplanation{\{\pair{x}{\bot}\mid \pair{x}{\bot}\in \{\pair{\sigma}{\sigma[\texttt{y}\leftarrow i]}\mid \sigma\in\Sigma\wedge i\in\mathbb{N}\}\}\cup\{\pair{x}{y}\mid \exists z\in\Sigma\mathrel{.}\pair{x}{z}\in \{\pair{\sigma}{\sigma[\texttt{y}\leftarrow i]}\mid \sigma\in\Sigma\wedge i\in\mathbb{N}\}\wedge \pair{z}{y}\in \{\pair{\sigma'}{\bot}\mid \sigma'(\texttt{y})<0\}\}}{(\ref{eq:relational-semantics-primitives})}\\
=
\formulaexplanation{\{\pair{\sigma}{\bot}\mid \exists i\mathrel{.}\sigma\in\Sigma\wedge i\in\mathbb{N}\wedge {\sigma[\texttt{y}\leftarrow i]}(\texttt{y})<0\}}{def.\ $\in$}\\
=
\formulaexplanation{\{\pair{\sigma}{\bot}\mid  \exists i\mathrel{.}\sigma\in\Sigma\wedge i\in\mathbb{N}\wedge i<0\}}{function application}\\
=
\formulaexplanation{\{\pair{\sigma}{\bot}\mid  \sigma\in\Sigma\}}{simplification}\\[-0.75ex]
\end{calculus}
\hyphen{5}\ \ By (\ref {eq:def:Abstract-Semantic-Domain-Semantics}) and (\ref{eq:sem:abstract:body-infinite}), we get
$\sqb{\texttt{S}_1}^\varrho$ = $\triple{e:{\{\pair{\sigma}{\sigma[\texttt{y}\leftarrow 0]}
\mid\sigma(\texttt{y})\geqslant 0\}}}{\bot:\{\pair{\sigma}{\bot}\mid  \sigma(\texttt{y})<0\}}{br:\emptyset}$ 
and
$\sqb{\texttt{S}_2}^\varrho$
$\triangleq$
$\triple{e:{\{\pair{\sigma}{\sigma[\texttt{y}\leftarrow 0]}
\mid \sigma\in\Sigma\}}}{\bot:\{\pair{\sigma}{\bot}\mid  \sigma\in\Sigma\}}{br:\emptyset}$.
\end{proof}
\end{toappendix}
\begin{example}\label{ex:examplee-relational}
Define \texttt{S}$_3$ $\triangleq$ \texttt{while (x!=0) \{ S$_2$ x=x-1; \}} with relational semantics
\begin{eqntabular*}{r@{\ }c@{\ }l}
\sqb{\texttt{S}_3}^\sharp
&=&
\triple{e:{\{\pair{\sigma}{\sigma}\mid\sigma(\texttt{x})=0\}
\cup
\{\pair{\sigma}{{\sigma[\texttt{y}\leftarrow 0][\texttt{x}\leftarrow 0]}}\mid
\sigma(\texttt{x})>0\}}}{\bot:{\{\pair{\sigma}{\bot}\mid  \sigma(\texttt{x})\neq 0\}}}{br:\emptyset}
\end{eqntabular*}
meaning that \texttt{S}$_3$ terminates because either the loop is not entered or it is entered with $\texttt{x}>0$ and \texttt{S}$_2$ terminates at each iteration setting $\texttt{y}$ to $0$. \texttt{S}$_3$ does not terminate when the loop is entered and either its body does not terminate or $\texttt{x}<0$.

Define \texttt{S}$_4$ $\triangleq$ \texttt{x=[-oo,oo]; S}$_3$ with relational semantics
\begin{eqntabular*}{r@{\ \ }c@{\ \ }l}
\sqb{\texttt{S}_4}^\sharp
&=&
\triple{e:{\{\pair{\sigma}{\sigma[\texttt{x}\leftarrow 0]}\mid \sigma\in\Sigma\}
\cup
\{\pair{\sigma}{{\sigma[\texttt{y}\leftarrow 0][\texttt{x}\leftarrow 0]}}\mid
\sigma\in\Sigma\}}}{\bot:{\{\pair{\sigma}{\bot}\mid\sigma\in\Sigma\}}}{br:\emptyset}
\end{eqntabular*}
meaning  either 
termination with \texttt{x}=0 (when  \texttt{x} is randomly assigned $0$) or with \texttt{x}=0 and \texttt{y}=0 (when  \texttt{x} is randomly assigned a positive number while \texttt{x} is randomly assigned a positive number or zero) or nontermination (when \texttt{x} is randomly assigned a negative number or \texttt{x} is randomly assigned a positive number and \texttt{y} are randomly assigned a negative number). \proofinapx. In this example, the fixpoint iterations are infinite
but would be transfinite for a transition semantics (corresponding to the lexicographic ordering for the nested loops) \cite{DBLP:journals/tcs/Cousot02}.
\end{example}
\begin{toappendix}
\begin{proof}[Proof of example \ref{ex:examplee-relational}]
\begin{calculus}
\hyphen{5}\formula{\sqb{\texttt{x!=0; S$_2$ x=x-1;}}_{e}^{\varrho}}\\
=
\formulaexplanation{\sqb{\texttt{x!=0}}_{e}^{\varrho}\mathbin{\,\fatsemi^\varrho\,}\sqb{\texttt{S}_2}_{e}^{\varrho}\mathbin{\,\fatsemi^\varrho\,}\sqb{\texttt{x=x-1;}}_{e}^{\varrho}}{(\ref{eq:def:abstract:sem:seq})}\\
=
\formulaexplanation{\{\pair{\sigma}{\sigma}\mid\sigma(\texttt{x})\neq0\}
\mathbin{\,\fatsemi^\varrho\,}
{\{\pair{\sigma}{\sigma[\texttt{y}\leftarrow 0]}\mid \sigma\in\Sigma\}}
\mathbin{\,\fatsemi^\varrho\,}
\{\pair{\sigma}{\sigma[\texttt{x}\leftarrow\sigma(\texttt{x})-1]}\mid \sigma\in\Sigma\}}{(\ref{eq:def:sem:abstract:basis}) and (\ref{eq:relational-semantics-primitives})}\\
=
\formulaexplanation{\{\pair{\sigma}{\sigma[\texttt{y}\leftarrow 0][\texttt{x}\leftarrow\sigma[\texttt{y}\leftarrow 0](\texttt{x})-1]}\mid \sigma(\texttt{x})\neq0\}}{def.\ (\ref{eq:relational-semantics-primitives}) of $\fatsemi^\varrho$}\\
=
\formulaexplanation{\{\pair{\sigma}{\sigma[\texttt{y}\leftarrow 0][\texttt{x}\leftarrow\sigma(\texttt{x})-1]}\mid \sigma(\texttt{x})\neq0\}}{$\texttt{x}\neq\texttt{y}$}\\[1ex]

\hyphen{5}\formulaexplanation{\Cev{F}_{e}^{\varrho}}{for \texttt{S}$_3$ = \texttt{while (x!=0) \{ S$_2$ x=x-1; \}}}\\
$\triangleq$
\formulaexplanation{\LAMBDA{X\in\wp(\Sigma\times\Sigma)}{\textsf{init}^{\varrho}} \mathbin{\sqcup_{+}^{\varrho}} (\sqb{\texttt{x!=0; S$_2$ x=x-1;}}_{e}^{\varrho}\mathbin{{\fatsemi}^{\varrho}} X)}{(\ref{eq:natural-transformer-finite-backward})}\\
=
\formulaexplanation{\LAMBDA{X\in\wp(\Sigma\times\Sigma)}\{\pair{\sigma}{\sigma}\mid\sigma\in\Sigma\} \cup (\{\pair{\sigma}{\sigma[\texttt{y}\leftarrow 0][\texttt{x}\leftarrow\sigma(\texttt{x})-1]}\mid \sigma(\texttt{x})\neq0\}\mathbin{{\fatsemi}^{\varrho}} X)}{(\ref{eq:relational-semantics-primitives})}\\
=
\formulaexplanation{\LAMBDA{X\in\wp(\Sigma\times\Sigma)}\{\pair{\sigma}{\sigma}\mid\sigma\in\Sigma\} \cup \{\pair{\sigma}{\sigma'}\mid
\sigma(\texttt{x})\neq0\wedge\pair{\sigma[\texttt{y}\leftarrow 0][\texttt{x}\leftarrow\sigma(\texttt{x})-1]}{\sigma'}\in X\}}{(\ref{eq:relational-semantics-primitives})}\\[-0.75ex]
\end{calculus}

\hyphen{5}\ \ By (\ref{eq:fatsemi-varrho-additive}) and (\ref{eq:natural-transformer-finite-backward}), ${\Cev{F}_{e}^{\varrho}}$  for \texttt{S}$_3$  = 
\texttt{while (x!=0) \{ S$_2$ x=x-1; \}}  preserves
nonempty joins $\cup$ so that the infinite  iterates $\pair{X^i}{i\leqslant\omega}$ of $\Lfp{\subseteq}{\Cev{F}_{e}^{\varrho}}$ are as follows
\begin{calculus}[$X^{n+1}$\ =\ ]
$X^0$ = \formula{\emptyset}\\
$X^1$ = \formula{\{\pair{\sigma}{\sigma}\mid\sigma\in\Sigma\} }\\
$X^2$ = \formulaexplanation{\{\pair{\sigma}{\sigma}\mid\sigma\in\Sigma\} \cup \{\pair{\sigma}{\sigma'}\mid
\sigma(\texttt{x})\neq0\wedge\pair{\sigma[\texttt{y}\leftarrow 0][\texttt{x}\leftarrow\sigma(\texttt{x})-1]}{\sigma'}\in X^1\}}{def.\ iterates}\\
\phantom{$X^2$} = \formulaexplanation{\{\pair{\sigma}{\sigma}\mid\sigma\in\Sigma\} \cup \{\pair{\sigma}{\sigma[\texttt{y}\leftarrow 0][\texttt{x}\leftarrow\sigma(\texttt{x})-1]}\mid
\sigma(\texttt{x})\neq0\}}{def. $X^1$}\\
{$X^n$} = \formulaexplanation{\{\pair{\sigma}{\sigma}\mid\sigma\in\Sigma\} 
\cup \bigcup_{i=1}^{n-1}\{\pair{\sigma}{{\sigma[\texttt{y}\leftarrow 0][\texttt{x}\leftarrow\sigma(\texttt{x})-i]}}\mid
\bigwedge_{j=0}^{i-1}\sigma(\texttt{x})\neq j\}}{induction hypothesis}\\

$X^{n+1}$ = \formulaexplanation{\{\pair{\sigma}{\sigma}\mid\sigma\in\Sigma\} \cup \{\pair{\sigma}{\sigma'}\mid
\sigma(\texttt{x})\neq0\wedge\pair{\sigma[\texttt{y}\leftarrow 0][\texttt{x}\leftarrow\sigma(\texttt{x})-1]}{\sigma'}\in X^{n}\}}{def.\ iterates}\\
\phantom{$X^{n+1}$} = 
\formulaexplanation{\{\pair{\sigma}{\sigma}\mid\sigma\in\Sigma\} \cup \{\pair{\sigma}{\sigma'}\mid
\sigma(\texttt{x})\neq0\wedge\pair{\sigma[\texttt{y}\leftarrow 0][\texttt{x}\leftarrow\sigma(\texttt{x})-1]}{\sigma'}\in (\{\pair{\sigma}{\sigma}\mid\sigma\in\Sigma\} 
\cup \bigcup_{i=1}^{n-1}\{\pair{\sigma}{{\sigma[\texttt{y}\leftarrow 0][\texttt{x}\leftarrow\sigma(\texttt{x})-i]}}\mid
\bigwedge_{j=0}^{i-1}\sigma(\texttt{x})\neq j\})\}}{def.\ $X^{n}$}\\

\phantom{$X^{n+1}$} = 
\formulaexplanation{\{\pair{\sigma}{\sigma}\mid\sigma\in\Sigma\} \cup 
\{\pair{\sigma}{\sigma'}\mid
\sigma(\texttt{x})\neq0\wedge\pair{\sigma[\texttt{y}\leftarrow 0][\texttt{x}\leftarrow\sigma(\texttt{x})-1]}{\sigma'}\in \{\pair{\sigma''}{\sigma''}\mid\sigma''\in\Sigma\} \}
\cup
 \bigcup_{i=1}^{n-1}\{\pair{\sigma}{\sigma'}\mid
\sigma(\texttt{x})\neq0\wedge\pair{\sigma[\texttt{y}\leftarrow 0][\texttt{x}\leftarrow\sigma(\texttt{x})-1]}{\sigma'}\in\{\pair{\sigma''}{{\sigma''[\texttt{y}\leftarrow 0][\texttt{x}\leftarrow\sigma''(\texttt{x})-i]}}\mid
\bigwedge_{j=0}^{i-1}\sigma''(\texttt{x})\neq j\})\}}{def.\ $\in$ and $\cup$, renaming}\\

\phantom{$X^{n+1}$} = 
\formulaexplanation{\{\pair{\sigma}{\sigma}\mid\sigma\in\Sigma\} 
\cup 
\{\pair{\sigma}{\sigma[\texttt{y}\leftarrow 0][\texttt{x}\leftarrow\sigma(\texttt{x})-1]}\mid
\sigma(\texttt{x})\neq0 \}
\cup
 \bigcup_{i=1}^{n-1}\{\pair{\sigma}{\sigma'}\mid\exists \sigma''\mathrel{.}
 \sigma''={\sigma[\texttt{y}\leftarrow 0][\texttt{x}\leftarrow\sigma(\texttt{x})-1]}
 \wedge
 {{\sigma''[\texttt{y}\leftarrow 0][\texttt{x}\leftarrow\sigma''(\texttt{x})-i]}}=\sigma'
\wedge
\bigwedge_{j=0}^{i-1}\sigma''(\texttt{x})\neq j)\}}{def.\ $\in$}\\

\phantom{$X^{n+1}$} = 
\formula{\{\pair{\sigma}{\sigma}\mid\sigma\in\Sigma\} 
\cup 
\{\pair{\sigma}{\sigma[\texttt{y}\leftarrow 0][\texttt{x}\leftarrow\sigma(\texttt{x})-1]}\mid
\sigma(\texttt{x})\neq0 \}
\cup
 \bigcup_{i=1}^{n-1}\{\pair{\sigma}{\sigma'}\mid\exists \sigma''\mathrel{.}
 \sigma''={\sigma[\texttt{y}\leftarrow 0][\texttt{x}\leftarrow\sigma(\texttt{x})-1]}
 \wedge
 {{\sigma[\texttt{y}\leftarrow 0][\texttt{x}\leftarrow\sigma(\texttt{x})-(i+1)]}}=\sigma'
\wedge
\bigwedge_{j=0}^{i-1}\sigma(\texttt{x})\neq (j+1))\}}
\\\explanation{function application with $\sigma''(\texttt{x})=\sigma(\texttt{x})-1$
and ${{\sigma''[\texttt{y}\leftarrow 0][\texttt{x}\leftarrow\sigma(\texttt{x})-(i+1)]}}={{\sigma[\texttt{y}\leftarrow 0][\texttt{x}\leftarrow\sigma(\texttt{x})-(i+1)]}}$}\\

\phantom{$X^{n+1}$} = 
\formulaexplanation{\{\pair{\sigma}{\sigma}\mid\sigma\in\Sigma\} 
\cup 
\{\pair{\sigma}{\sigma[\texttt{y}\leftarrow 0][\texttt{x}\leftarrow\sigma(\texttt{x})-1]}\mid
\sigma(\texttt{x})\neq0 \}
\cup
 \bigcup_{i=1}^{n-1}\{\pair{\sigma}{{\sigma[\texttt{y}\leftarrow 0][\texttt{x}\leftarrow\sigma(\texttt{x})-(i+1)]}}\mid
\bigwedge_{j=0}^{i-1}\sigma(\texttt{x})\neq (j+1)j)\}}{simplification}\\

\phantom{$X^{n+1}$} = 
\formulaexplanation{\{\pair{\sigma}{\sigma}\mid\sigma\in\Sigma\} 
\cup 
\{\pair{\sigma}{\sigma[\texttt{y}\leftarrow 0][\texttt{x}\leftarrow\sigma(\texttt{x})-1]}\mid
\sigma(\texttt{x})\neq0 \}
\cup
 \bigcup_{i'=2}^{n}\{\pair{\sigma}{{\sigma[\texttt{y}\leftarrow 0][\texttt{x}\leftarrow\sigma(\texttt{x})-i']}}\mid
\bigwedge_{j=0}^{i'-2}\sigma(\texttt{x})\neq (j+1))\}}{change of variable $i'=i+1$}\\

\phantom{$X^{n+1}$} = 
\formulaexplanation{\{\pair{\sigma}{\sigma}\mid\sigma\in\Sigma\} 
\cup 
\{\pair{\sigma}{\sigma[\texttt{y}\leftarrow 0][\texttt{x}\leftarrow\sigma(\texttt{x})-1]}\mid
\sigma(\texttt{x})\neq0 \}
\cup
 \bigcup_{i'=2}^{n}\{\pair{\sigma}{{\sigma[\texttt{y}\leftarrow 0][\texttt{x}\leftarrow\sigma(\texttt{x})-i']}}\mid
\bigwedge_{j'=1}^{i'-1}\sigma(\texttt{x})\neq j')\}}{change of variable $j'=j+1$}\\

\phantom{$X^{n+1}$} = 
\formula{\{\pair{\sigma}{\sigma}\mid\sigma\in\Sigma\} 
\cup 
 \bigcup_{i=1}^{(n+1)-1}\{\pair{\sigma}{\sigma[\texttt{y}\leftarrow 0][\texttt{x}\leftarrow\sigma(\texttt{x})-i]}\mid
\bigwedge_{j=0}^{i-1}\sigma(\texttt{x})\neq j\})\}}\\\rightexplanation{grouping terms for $i=1$}\\[-1.5em]
\end{calculus}
which is the induction hypothesis for $n+1$.

\medskip

\hyphen{5}\ \ By recurrence, $X^n=\{\pair{\sigma}{\sigma}\mid\sigma\in\Sigma\} 
\cup \bigcup_{i=1}^{n-1}\{\pair{\sigma}{{\sigma[\texttt{y}\leftarrow 0][\texttt{x}\leftarrow\sigma(\texttt{x})-i]}}\mid
\bigwedge_{j=0}^{i-1}\sigma(\texttt{x})\neq j\}$, so that the least fixpoint of ${\Cev{F}_{e}^{\varrho}}$ for \texttt{S}$_3$  = 
\texttt{while (x!=0) \{ S$_2$ x=x-1; \}}  is
\begin{calculus}
\formulaexplanation{\Lfp{\subseteq}{\Cev{F}_{e}^{\varrho}}}{for \texttt{S}$_3$  = \texttt{while (x!=0) \{ S$_2$ x=x-1; \}}}\\
=
\formulaexplanation{\bigcup_{n\in\mathbb{N}}X^n}{def.\ iterates}\\
=
\formulaexplanation{\{\pair{\sigma}{\sigma}\mid\sigma\in\Sigma\} 
\cup \bigcup_{n\in\mathbb{N}}\bigcup_{i=1}^{n-1}\{\pair{\sigma}{{\sigma[\texttt{y}\leftarrow 0][\texttt{x}\leftarrow\sigma(\texttt{x})-i]}}\mid
\bigwedge_{j=0}^{i-1}\sigma(\texttt{x})\neq j\}}{def.\ $\cup$}\\
=
\formula{\{\pair{\sigma}{\sigma}\mid\sigma\in\Sigma\} 
\cup\bigcup_{i>0}\{\pair{\sigma}{{\sigma[\texttt{y}\leftarrow 0][\texttt{x}\leftarrow\sigma(\texttt{x})-i]}}\mid
\sigma(\texttt{x})\not\in\interval{0}{i-1}\}}
\end{calculus}
\hyphen{5}\quad It follows that for $\texttt{S}_3\triangleq\texttt{while (x!=0) \{ S$_2$ x=x-1; \}}$, we have
\begin{calculus}
\formula{\sqb{\texttt{S}_3}_{e}^{\varrho}}\\
=
\formulaexplanation{
\Lfp{\subseteq}{\Cev{F}_{e}^{\varrho}}\mathbin{{\fatsemi}^{\varrho}}(\sqb{\neg\texttt{B}}_{e}^{\varrho}\cup\sqb{\texttt{B;S}}_{b}^{\varrho})}{by (\ref{eq:sem:abstract:finite}) with \texttt{B} = \texttt{(x!=0)}, $\neg$\texttt{B} = \texttt{(x=0)}, and \texttt{S} =  S$_2$ \texttt{x=x-1;}}\\
=
\formula{(\{\pair{\sigma}{\sigma}\mid\sigma\in\Sigma\} 
\cup\bigcup_{i>0}\{\pair{\sigma}{{\sigma[\texttt{y}\leftarrow 0][\texttt{x}\leftarrow\sigma(\texttt{x})-i]}}\mid
\sigma(\texttt{x})\not\in\interval{0}{i-1}\})\mathbin{{\fatsemi}^{\varrho}}(\{\pair{\sigma}{\sigma}\mid\sigma(\texttt{x})=0\}\cup\emptyset)}\\
=
\formulaexplanation{(\{\pair{\sigma}{\sigma}\mid\sigma\in\Sigma\} \mathbin{{\fatsemi}^{\varrho}}\{\pair{\sigma}{\sigma}\mid\sigma(\texttt{x})=0\})
\cup
(\bigcup_{i>0}\{\pair{\sigma}{{\sigma[\texttt{y}\leftarrow 0][\texttt{x}\leftarrow\sigma(\texttt{x})-i]}}\mid
\sigma(\texttt{x})\not\in\interval{0}{i-1}\}\mathbin{{\fatsemi}^{\varrho}}(\{\pair{\sigma}{\sigma}\mid\sigma(\texttt{x})=0\}))}{by (\ref{eq:fatsemi-varrho-additive}), $\mathbin{{\fatsemi}^{\varrho}}$ left preserves joins}\\[0.75ex]
=
\formula{\{\pair{\sigma}{\sigma}\mid\sigma(\texttt{x})=0\}
\cup
(\bigcup_{i>0}\{\pair{\sigma}{{\sigma[\texttt{y}\leftarrow 0][\texttt{x}\leftarrow\sigma(\texttt{x})-i]}}\mid
\sigma(\texttt{x})\not\in\interval{0}{i-1}\wedge\sigma(\texttt{x})-i=0\})}\\[-1.5ex]\rightexplanation{def.\ (\ref{eq:relational-semantics-primitives}) of $\mathbin{{\fatsemi}^{\varrho}}$}\\
=
\formulaexplanation{\{\pair{\sigma}{\sigma}\mid\sigma(\texttt{x})=0\}
\cup
\{\pair{\sigma}{{\sigma[\texttt{y}\leftarrow 0][\texttt{x}\leftarrow 0]}}\mid
\sigma(\texttt{x})>0\}}{simplification}
\end{calculus}
\medskip

\hyphen{5}\quad Then, for  $\texttt{S}_4$ = \texttt{x=[-oo,oo]; S}$_3$, we have

\begin{calculus}
\formula{\sqb{\texttt{S}_4}_{e}^{\varrho}}\\
 =
\formulaexplanation{\sqb{\texttt{x=[-oo,oo];}}_{e}^{\varrho}\mathbin{{\fatsemi}^{\varrho}}\sqb{\texttt{S}_3}_{e}^{\varrho}}{(\ref{eq:def:abstract:sem:seq})}\\
=
\formula{\{\pair{\sigma}{\sigma[\texttt{x}\leftarrow n]}\mid n\in\mathbb{N}\}\mathbin{{\fatsemi}^{\varrho}}(\{\pair{\sigma}{\sigma}\mid\sigma(\texttt{x})=0\}
\cup
\{\pair{\sigma}{{\sigma[\texttt{y}\leftarrow 0][\texttt{x}\leftarrow 0]}}\mid
\sigma(\texttt{x})>0\})}\\\rightexplanation{(\ref{eq:def:sem:abstract:basis}) and as previously shown}\\
=
\formulaexplanation{\{\pair{\sigma}{\sigma[\texttt{x}\leftarrow n]}\mid n\in\mathbb{N}\}\mathbin{{\fatsemi}^{\varrho}}\{\pair{\sigma}{\sigma}\mid\sigma(\texttt{x})=0\}
\cup
\{\pair{\sigma}{\sigma[\texttt{x}\leftarrow n]}\mid n\in\mathbb{N}\}\mathbin{{\fatsemi}^{\varrho}}\{\pair{\sigma}{{\sigma[\texttt{y}\leftarrow 0][\texttt{x}\leftarrow 0]}}\mid
\sigma(\texttt{x})>0\}}{by (\ref{eq:fatsemi-varrho-additive}), $\mathbin{{\fatsemi}^{\varrho}}$ left preserves joins}\\[0.75ex]
=
\formula{\{\pair{\sigma}{\sigma[\texttt{x}\leftarrow n]}\mid\sigma[\texttt{x}\leftarrow n](\texttt{x})=0\}
\cup
\{\pair{\sigma}{{\sigma[\texttt{x}\leftarrow n][\texttt{y}\leftarrow 0][\texttt{x}\leftarrow 0]}}\mid
\sigma[\texttt{x}\leftarrow n](\texttt{x})>0\}}\\[-0.5ex]
\rightexplanation{def.\ (\ref{eq:relational-semantics-primitives}) of $\fatsemi^\varrho$}\\
=
\formulaexplanation{\{\pair{\sigma}{\sigma[\texttt{x}\leftarrow n]}\mid n=0\}
\cup
\{\pair{\sigma}{{\sigma[\texttt{y}\leftarrow 0][\texttt{x}\leftarrow 0]}}\mid
n>0\}}{function application}\\
=
\formulaexplanation{\{\pair{\sigma}{\sigma[\texttt{x}\leftarrow 0]}\mid \sigma\in\Sigma\}
\cup
\{\pair{\sigma}{{\sigma[\texttt{y}\leftarrow 0][\texttt{x}\leftarrow 0]}}\mid
\sigma\in\Sigma\}}{simplification}
\end{calculus}

\medskip

\hyphen{5}\quad The iteration  \texttt{S}$_3$ = \texttt{while (x!=0) \{ S$_2$ x=x-1; \}} may iterate for ever. To show this, we have, by (\ref{eq:trace-transformer-infinite}), that
\begin{calculus}
\formulaexplanation{{F_{\bot}^{\varrho}}}{for $\texttt{S}_3$ = \texttt{while (x!=0) \{ S$_2$ x=x-1; \}}}\\
=
\formula{\LAMBDA{X\in{\mathbb{L}^{\varrho}_{\infty}}}{\sqb{\texttt{x!=0;S$_2$ x=x-1;}}_{e}^{\varrho}}\mathbin{{\fatsemi}^{\varrho}} X}\\
=
\formula{\LAMBDA{X\in{\mathbb{L}^{\varrho}_{\infty}}}
\{\pair{\sigma}{\sigma}\mid\sigma(\texttt{x})\neq 0\}
\mathbin{{\fatsemi}^{\varrho}}
{\{\pair{\sigma}{\sigma[\texttt{y}\leftarrow 0]}
\mid \sigma\in\Sigma\}}
\mathbin{{\fatsemi}^{\varrho}}
\{\pair{\sigma}{\sigma[\texttt{x}\leftarrow \sigma(\texttt{x})-1]}
\mid \sigma\in\Sigma\}
\mathbin{{\fatsemi}^{\varrho}} X
}\\[-0.5ex]\rightexplanation{(\ref{eq:def:abstract:sem:seq}), (\ref{eq:def:sem:abstract:basis}), def.\ $\sqb{\texttt{S$_2$}}_{e}^{\varrho}$ in ex.\ \ref{ex:example1-relational}}\\
=
\formulaexplanation{\LAMBDA{X\in{\mathbb{L}^{\varrho}_{\infty}}}
\{\pair{\sigma}{\sigma[\texttt{y}\leftarrow 0][\texttt{x}\leftarrow \sigma[\texttt{y}\leftarrow 0](\texttt{x})-1]}
\mid \sigma(\texttt{x})\neq 0\}
\mathbin{{\fatsemi}^{\varrho}} X
}{def.\ (\ref{eq:relational-semantics-primitives}) of $\fatsemi^\varrho$}\\
=
\formulaexplanation{\LAMBDA{X\in{\mathbb{L}^{\varrho}_{\infty}}}
\{\pair{\sigma}{\sigma[\texttt{y}\leftarrow 0][\texttt{x}\leftarrow \sigma(\texttt{x})-1]}
\mid \sigma(\texttt{x})\neq 0\}
\mathbin{{\fatsemi}^{\varrho}} X
}{simplification since \texttt{x} $\neq$ \texttt{y}}
\end{calculus}

\medskip

\hyphen{5}\ \ By (\ref{eq:fatsemi-varrho-additive}) and (\ref{eq:natural-transformer-finite-backward}), ${{F_{\bot}^{\varrho}}}$ for $\texttt{S}_3$ = \texttt{while (x!=0) \{ S$_2$ x=x-1; \}} converges at $\omega$ so that the infinite  iterates $\pair{X^i}{i\leqslant\omega}$ of $\sqb{\texttt{S}_3}_{li}^{\varrho}$
=
$\Gfp{\subseteq}{{F_{\bot}^{\varrho}}}$ are as follows
\begin{calculus}[$X^{n+1}$\ =\ ]
$X^0$ = \formula{\Sigma\times\{\bot\}}\\

$X^1$ = \formulaexplanation{\{\pair{\sigma}{\sigma[\texttt{y}\leftarrow 0][\texttt{x}\leftarrow \sigma(\texttt{x})-1]}
\mid \sigma(\texttt{x})\neq 0\}
\mathbin{{\fatsemi}^{\varrho}} X^0}{def.\ iterates and ${{F_{\bot}^{\varrho}}}$}\\
\phantom{$X^1$} = \formulaexplanation{\{\pair{\sigma}{\bot}\mid \sigma(\texttt{x})\neq0\}}{def.\ (\ref{eq:relational-semantics-primitives}) of $\fatsemi^\varrho$ and $X^0$}\\

$X^2$ = \formulaexplanation{\{\pair{\sigma}{\sigma[\texttt{y}\leftarrow 0][\texttt{x}\leftarrow \sigma(\texttt{x})-1]}
\mid \sigma(\texttt{x})\neq 0\}
\mathbin{{\fatsemi}^{\varrho}} X^1}{def.\ iterates and ${{F_{\bot}^{\varrho}}}$}\\

\phantom{$X^2$} = \formulaexplanation{\{\pair{\sigma}{\bot}
\mid \sigma(\texttt{x})\neq 0\wedge{\sigma[\texttt{y}\leftarrow 0][\texttt{x}\leftarrow \sigma(\texttt{x})-1]}(\texttt{x})\neq 0\}}{def.\ (\ref{eq:relational-semantics-primitives}) of $\fatsemi^\varrho$ and $X^1$}\\

\phantom{$X^2$} = \formulaexplanation{\{\pair{\sigma}{\bot}
\mid \sigma(\texttt{x})\neq 0\wedge\sigma(\texttt{x})\neq 1\}}{function application and simplification}\\

$X^n$ = \formulaexplanation{\{\pair{\sigma}{\bot}\mid \bigwedge_{i=0}^{n-1}\sigma(\texttt{x})\neq i\}}{induction hypothesis}\\

$X^{n+1}$ = \formulaexplanation{\{\pair{\sigma}{\sigma[\texttt{y}\leftarrow 0][\texttt{x}\leftarrow \sigma(\texttt{x})-1]}
\mid \sigma(\texttt{x})\neq 0\}
\mathbin{{\fatsemi}^{\varrho}} X^n}{def.\ iterates and ${{F_{\bot}^{\varrho}}}$}\\

\phantom{$X^{n+1}$} = \formulaexplanation{\{\pair{\sigma}{\bot}
\mid \sigma(\texttt{x})\neq 0\wedge \bigwedge_{i=0}^{n-1}{\sigma[\texttt{y}\leftarrow 0][\texttt{x}\leftarrow \sigma(\texttt{x})-1]}(\texttt{x})\neq i\}}{def.\ (\ref{eq:relational-semantics-primitives}) of $\fatsemi^\varrho$ and $X^n$}\\

\phantom{$X^{n+1}$} = \formulaexplanation{\{\pair{\sigma}{\bot}\mid \bigwedge_{i=0}^{(n+1)-1}\sigma(\texttt{x})\neq i\}}{simplification}\\[-0.5ex]
\end{calculus}

\hyphen{5}\ \ By recurrence, $X^n=\{\pair{\sigma}{\bot}\mid \bigwedge_{i=0}^{n-1}\sigma(\texttt{x})\neq i\}$, so that,
by convergence at $\omega$, the greatest fixpoint is
\begin{calculus}
\formula{\sqb{\texttt{S}_3}_{li}^{\varrho}\colsep{=}\Gfp{\subseteq}{F_{\bot}^{\varrho}}}\\
=
\formulaexplanation{\bigcap_{n\in\mathbb{N}}X^n}{def.\ iterates}\\
=
\formula{\bigcap_{n\in\mathbb{N}}\{\pair{\sigma}{\bot}\mid \bigwedge_{i=0}^{n-1}\sigma(\texttt{x})\neq i\}}\\
=
\formulaexplanation{\{\pair{\sigma}{\bot}\mid \sigma(\texttt{x})<0\}}{$\Sigma=\{\texttt{x},\texttt{y}\}\rightarrow\mathbb{Z}$}\\[-1ex]
\end{calculus}

\medskip

\hyphen{5}\quad The iteration  \texttt{S}$_3$ = \texttt{while (x!=0) \{ S$_2$ x=x-1; \}} may also not terminate because of the nontermination of S$_2$ in its body. The loop body \texttt{S$_2$ x=x-1;} may not terminate, as follows.
\begin{calculus}
\formula{\sqb{\texttt{S$_2$ x=x-1;}}_{\bot}^{\varrho}}\\
=
\formulaexplanation{\sqb{\texttt{S$_2$}}_{\bot}^{\varrho} \cup(\sqb{\texttt{S$_2$}}_{e}^{\varrho}\mathbin{\fatsemi^{\varrho}}\sqb{\texttt{x=x-1;}}_{\bot}^{\varrho})}{(\ref{eq:def:abstract:sem:seq})}\\
=
\formulaexplanation{\sqb{\texttt{S$_2$}}_{\bot}^{\varrho}}{def.\ (\ref{eq:relational-semantics-primitives}) of $\fatsemi^\varrho$ and (\ref{eq:def:sem:abstract:basis}) so that $\sqb{\texttt{x=x-1;}}_{\bot}^{\varrho}=\emptyset$}\\
=
\formulaexplanation{\{\pair{\sigma}{\bot}\mid  \sigma\in\Sigma\}}{by example \ref{ex:example1-relational}}
\end{calculus}

\medskip

\noindent This implies that
\begin{calculus}
\formula{\sqb{\texttt{x!=0; S$_2$ x=x-1;}}_{\bot}^{\varrho}}\\
=
\formula{\{\pair{x}{\bot}\mid \pair{x}{\bot}\in {\sqb{\texttt{x!=0;}}_{\bot}^{\varrho}}\}\cup\{\pair{x}{y}\mid \exists z\in\Sigma\mathrel{.}\pair{x}{z}\in {\sqb{\texttt{x!=0;}}_{e}^{\varrho}}\wedge \pair{z}{y}\in {\sqb{\texttt{S$_2$ x=x-1;}}_{\bot}^{\varrho}}\}}\\[-0.5ex]
\rightexplanation{(\ref{eq:def:abstract:sem:seq}) and def.\  of (\ref{eq:relational-semantics-primitives})}\\
=
\formulaexplanation{\{\pair{\sigma}{\bot}\mid  \sigma(\texttt{x})\neq 0\}}{(\ref{eq:def:sem:abstract:basis})}
\end{calculus}

\medskip

\noindent It follows that
\begin{calculus}
\formula{\sqb{\texttt{S}_3}_{bi}^{\sharp}}\\
$\triangleq$
\formulaexplanation{(\Lfp{\sqsubseteq_{+}^{\sharp}}{\Cev{F}_{e}^{\sharp}})\mathbin{{\fatsemi}^{\sharp}}\sqb{\texttt{B;S}}_{\bot}}{(\ref{eq:sem:abstract:body-infinite}) with \texttt{B} = \texttt{x!=0} and \texttt{S} = \texttt{S$_2$ x=x-1;}}\\
=
\formula{({\{\pair{\sigma}{\sigma}\mid\sigma\in\Sigma\} 
\cup\bigcup_{i>0}\{\pair{\sigma}{{\sigma[\texttt{y}\leftarrow 0][\texttt{x}\leftarrow\sigma(\texttt{x})-i]}}\mid
\sigma(\texttt{x})\not\in\interval{0}{i-1}\}}) \mathbin{{\fatsemi}^{\sharp}}\{\pair{\sigma}{\bot}\mid  \sigma(\texttt{x})\neq 0\}}\\[-1ex]
\rightexplanation{previous evaluation of $\Lfp{\sqsubseteq_{+}^{\sharp}}{\Cev{F}_{e}^{\sharp}}$ and ${\sqb{\texttt{x!=0; S$_2$ x=x-1;}}_{\bot}^{\varrho}}$}\\
=
\formulaexplanation{({\{\pair{\sigma}{\sigma}\mid\sigma\in\Sigma\} \mathbin{{\fatsemi}^{\sharp}}\{\pair{\sigma}{\bot}\mid  \sigma(\texttt{x})\neq 0\})
\cup
\bigcup_{i>0}(\{\pair{\sigma}{{\sigma[\texttt{y}\leftarrow 0][\texttt{x}\leftarrow\sigma(\texttt{x})-i]}}\mid
\sigma(\texttt{x})\not\in\interval{0}{i-1}\}}\mathbin{{\fatsemi}^{\sharp}}\{\pair{\sigma}{\bot}\mid  \sigma(\texttt{x})\neq 0\})}{$\mathbin{{\fatsemi}^{\varrho}}$ left preserves joins $\cup$ on $\wp(\Sigma\times\Sigma_{\bot})$}\\[0.75ex]
=
\formula{\{\pair{\sigma}{\bot}\mid  \sigma(\texttt{x})\neq 0\}\cup \bigcup_{i>0}(\{\pair{\sigma}{\bot}\mid
\sigma(\texttt{x})\not\in\interval{0}{i-1}\wedge {\sigma[\texttt{y}\leftarrow 0][\texttt{x}\leftarrow\sigma(\texttt{x})-i]}(\texttt{x})\neq 0\})}\\[-0.75ex]

\rightexplanation{def.\ (\ref{eq:relational-semantics-primitives}) of $\fatsemi^\varrho$}\\
=
\formulaexplanation{\{\pair{\sigma}{\bot}\mid  \sigma(\texttt{x})\neq 0\}\cup \bigcup_{i>0}(\{\pair{\sigma}{\bot}\mid
\sigma(\texttt{x})\not\in\interval{0}{i-1}\wedge \sigma(\texttt{x})\neq i\})}{function application}\\
=
\formulaexplanation{\{\pair{\sigma}{\bot}\mid  \sigma(\texttt{x})\neq 0\}}{simplification}
\end{calculus}

\medskip

\hyphen{5}\quad The nonterminating behavior $\sqb{\texttt{S}_3}_{\bot}^{\sharp}$ of the iteration  \texttt{S}$_3$ = \texttt{while (x!=0) \{ S$_2$ x=x-1; \}} is defined, by (\ref{eq:sem:abstract:loop-nontermination}), to be either due to the nontermination $\sqb{\texttt{S}_3}_{bi}^{\sharp}$ of its body or infinite iteration $\sqb{\texttt{S}_3}_{li}^{\sharp}$.
\begin{calculus}
\formula{\sqb{\texttt{S}_3}_{\bot}^{\sharp}}\\
$\triangleq$
\formulaexplanation{\sqb{\texttt{S}_3}_{bi}^{\sharp}\cup\sqb{\texttt{S}_3}_{li}^{\sharp}}{(\ref{eq:sem:abstract:loop-nontermination}) and ${\mathbin{\sqcup_{\infty}^{\sharp}}}={\cup}$}\\
=
\formulaexplanation{{\{\pair{\sigma}{\bot}\mid  \sigma(\texttt{x})\neq 0\}}\cup{\{\pair{\sigma}{\bot}\mid \sigma(\texttt{x})<0\}}}{as previously shown}\\
=
\formulaexplanation{{\{\pair{\sigma}{\bot}\mid  \sigma(\texttt{x})\neq 0\}}}{simplification}
\end{calculus}

\medskip

\hyphen{5}\quad The nonterminating behavior $\sqb{\texttt{S}_4}_{\bot}^{\sharp}$ of the iteration  \texttt{S}$_4$ = $\triangleq$ \texttt{x=[-oo,oo]; S}$_3$ is now
\begin{calculus}
\formula{\sqb{\texttt{S}_4}_{\bot}^{\sharp}}\\
=
\formulaexplanation{\sqb{\texttt{x=[-oo,oo];}}_{\bot}^{\sharp}\cup(\sqb{\texttt{x=[-oo,oo];}}_{e}^{\sharp}\mathbin{\fatsemi^{\sharp}}\sqb{\texttt{S$_3$}}_{\bot}^{\sharp})}{(\ref{eq:def:abstract:sem:seq})}\\
=
\formulaexplanation{\{\pair{\sigma}{\sigma[\texttt{x}\leftarrow n]}\mid n\in\mathbb{N}\}\mathbin{\fatsemi^{\sharp}}{\{\pair{\sigma}{\bot}\mid  \sigma(\texttt{x})\neq 0\}}}{(\ref{eq:def:sem:abstract:basis})}\\
=
\formulaexplanation{\{\pair{x}{y}\mid \exists z\in\Sigma\mathrel{.}\pair{x}{z}\in \{\pair{\sigma}{\sigma[\texttt{x}\leftarrow n]}\mid n\in\mathbb{N}\}\wedge \pair{z}{y}\in \{\pair{\sigma'}{\bot}\mid  \sigma'(\texttt{x})\neq 0\}\}}{def.\ (\ref{eq:relational-semantics-primitives}) of $\fatsemi^\varrho$}\\
=
\formulaexplanation{\{\pair{\sigma}{\bot}\mid \exists n\in\mathbb{N}\mathrel{.} {\sigma[\texttt{x}\leftarrow n]}(\texttt{x})\neq 0\}}{$z={\sigma[\texttt{x}\leftarrow n]}$}\\
=
\formulaexplanation{\{\pair{\sigma}{\bot}\mid\sigma\in\Sigma\}}{simplification}
\end{calculus}

\medskip

Grouping all cases together according to (\ref{eq:def:Abstract-Semantic-Domain-Semantics}), we get
$\sqb{\texttt{S}_3}^\sharp$
=
$\triple{e:\sqb{\texttt{S}_3}_{e}^{\sharp}}{\bot:\sqb{\texttt{S}_3}_{\bot}^{\sharp}}{br:\sqb{\texttt{S}_3}_{b}^{\sharp}}$
=
$\triple{e:{\{\pair{\sigma}{\sigma}\mid\sigma(\texttt{x})=0\}
\cup
\{\pair{\sigma}{{\sigma[\texttt{y}\leftarrow 0][\texttt{x}\leftarrow 0]}}\mid
\sigma(\texttt{x})>0\}}}{\bot:{\{\pair{\sigma}{\bot}\mid  \sigma(\texttt{x})\neq 0\}}}{br:\emptyset}$
and
$\sqb{\texttt{S}_4}^\sharp$
=
$\triple{e:\sqb{\texttt{S}_4}_{e}^{\sharp}}{\bot:\sqb{\texttt{S}_4}_{\bot}^{\sharp}}{br:\sqb{\texttt{S}_4}_{b}^{\sharp}}$
=
$\triple{e:{\{\pair{\sigma}{\sigma[\texttt{x}\leftarrow 0]}\mid \sigma\in\Sigma\}
\cup
\{\pair{\sigma}{{\sigma[\texttt{y}\leftarrow 0][\texttt{x}\leftarrow 0]}}\mid
\sigma\in\Sigma\}}}{\bot:{\{\pair{\sigma}{\bot}\mid\sigma\in\Sigma\}}}{br:\emptyset}$, proving example \ref{ex:examplee-relational}.
\end{proof}
\end{toappendix}

\section{Algebraic Program Execution Properties}\label{sec:Algebraic-Program-Execution-Properties}

\subsection{Algebraic Execution Properties}\label{sect:Execution-Properties}

Traditionally, logics involve two formal languages, one to express programs and another one to express properties of the program executions. The syntax and semantics of these programming and logic languages are considered to be different. Therefore, 
in addition to the program syntax and semantics, this traditional  approach requires to define the syntax and semantics of the logic expressing program properties.

A semantics $\sqb{\texttt{S}}^\sharp\in{\mathbb{L}^{\sharp}}$ in (\ref{eq:def:Abstract-Semantic-Domain-Semantics}) is an abstraction of a property of the executions of the statement ${\texttt{S}}$. Therefore ${\mathbb{L}^{\sharp}}$ will be the domain
of execution properties whether used to describe the semantics or logic properties of programs executions. This will avoid us  the necessary traditional distinction between programs semantics and program properties. 

This idea follows \cite{DBLP:journals/scp/Hehner90,DBLP:books/daglib/0070829,DBLP:journals/scp/Hehner99}'s slogan that ``Programs are predicates'' and define properties of program executions as programs (which semantics is already defined). It is also found in Dexter Kozen's Kleene algebra with tests \cite{DBLP:journals/toplas/Kozen97,DBLP:journals/tocl/Kozen00,DBLP:journals/corr/abs-2312-09662}. Therefore, from an abstract point of view, program execution specification and verification need nothing more than programs and an associated calculus $\textsf{post}^\sharp$ on programs.
\subsection{The Algebraic Program Execution Property Transformer}\label{sect:Transformer-Program-Execution-Properties}
Let us define the transformer $\textsf{post}^\sharp\in{\mathbb{L}^{\sharp}}\increasingfunctionto{\mathbb{L}^{\sharp}}\increasingfunctionto{\mathbb{L}^{\sharp}}$ such that
\begin{eqntabular}{rcl}
\textsf{post}^\sharp(S)P&\triangleq&P\mathbin{\fatsemi^{\sharp}}S
\label{eq:def:abstract:transformer:post}
\end{eqntabular}
where $S$ is a semantics in ${\mathbb{L}^{\sharp}}$ as defined by  (\ref{eq:def:Abstract-Semantic-Domain-Semantics}) and $\mathbin{\fatsemi^{\sharp}}$ is defined by (\ref{eq:def:abstract:sem:group:seq}). If $P$ is a precondition when at \texttt{S} then $\textsf{post}^\sharp\sqb{\texttt{S}}^{\sharp}P$ is the postcondition after \texttt{S} (including when breaking out of \texttt{S}).

For example, using the shorthand (\ref{eq:unique-nonempty-component-shorthand}), $\textsf{post}^\sharp(S){\textsf{init}^{\sharp}}=S$ by \ref{def:abstract:domain:well:def:init:neutral} and $\textsf{post}^\sharp(S){P}=P$ for all $P\in{\mathbb{L}^{\sharp}_{\infty}}$ by \ref{def:abstract:domain:well:def:oo:absorbent}.

In definition (\ref{eq:def:abstract:transformer:post}) of ``predicate transformers'' the meaning of ``predicates'' about programs executions is abstracted away as programs specifying executions. Further abstractions will yield the classic understanding of ``predicates'', ``abstract property'', etc. The classic Galois connections $\textsf{post}$--$\widetilde{\textsf{pre}}$ \cite[(12.22)]{Cousot-PAI-2021} and  $\textsf{post}$--$\textsf{post}^{-1}$  \cite[(12.6)]{Cousot-PAI-2021} are still valid with this different definition of \textsf{post}.

The following lemmas show that the $\textsf{\textup{post}}$ transformer inherits the properties of sequential composition. It applies e.g.\ to $\pair{\mathbb{L}^{\sharp}_{+}}{\sqsubseteq_{+}^{\sharp}}$ in \ref{def:abstract:domain:well:def:finite:domain}, 
$\pair{\mathbb{L}^{\sharp}_{\infty}}{\sqsubseteq_{\infty}^{\sharp}}$ in \ref{def:abstract:domain:well:def:infinite:domain},
or $\pair{\mathbb{L}^{\sharp}}{{\sqsubseteq}^{\sharp}}$ in (\ref{eq:def:Abstract-Semantic-Domain-Semantics}).
\begin{lemma}\label{lem:post-S}\proofinapx\quad 
Let $\triple{L}{\sqsubseteq}{\sqcup}$ be a poset with partially defined join $\sqcup$. Let
$\fatsemi$ be the sequential composition on $L$. If\/ ${\fatsemi}$ left-satisfies
any one of the properties of definition \ref{def:Properties:functions:posets} or their dual then for all 
$S\in{\mathbb{L}}$, ${\textsf{\textup{post}}(S)}$ satisfies the same property.
\end{lemma}
\begin{toappendix}
\begin{proof}[Proof of lemma \ref{lem:post-S}] Let $\pair{P_i}{i\in\Delta}$ be a family of elements of $L$ such that $\Delta=\{0,1\}$ with $P_0\mathrel{\sqsubseteq} P_1$ for the left-increasingness hypothesis (def.\ \ref{def:Properties:increasing}), $\Delta$ is finite
for the existing finite $\sqcup$ left preserving hypothesis (def.\ \ref{def:Properties:finite-preserving}), $\Delta\in\mathbb{O}$ and $\pair{P_i}{i\in\Delta}$ is an increasing chain for the left upper-continuity hypothesis (def.\ \ref{def:Properties:continuous}), $\Delta$ is an arbitrary set for the existing join left preservation property (def.\ \ref{def:Properties:limit-preserving}), possibly empty in case of left strictness (def.\ \ref{def:Properties:increasing:strict}). The proof is similar in all of these cases, as follows
\begin{calculus}[$\Leftrightarrow$\ \ ]
\formula{\textsf{post}(S)(\bigsqcup\nolimits_{i\in\Delta}P_i)}\\
$\Leftrightarrow$
\formulaexplanation{(\bigsqcup\nolimits_{i\in\Delta}P_i)\mathbin{\fatsemi}S}{def.\ (\ref{eq:def:abstract:transformer:post}) of $\textsf{post}$}\\
$\Leftrightarrow$
\formulaexplanation{\bigsqcup\nolimits_{i\in\Delta}(P_i\mathbin{\fatsemi}S)}{by the left preservation hypothesis for $\mathbin{\fatsemi}$}\\
$\Leftrightarrow$
\lastformulaexplanation{\bigsqcup\nolimits_{i\in\Delta}\textsf{post}(S)P_i}{def.\ (\ref{eq:def:abstract:transformer:post}) of $\textsf{post}$}{\mbox{\qed}}
\end{calculus}\let\qed\relax
\end{proof}
\end{toappendix}
The following Galois connection shows the equivalence of forward/deductive and backward/ab\-duc\-tive reasonings on the program semantics.
\begin{lemma}\label{lem:GC:post-S}\proofinapx\quad
If $\triple{L}{\sqsubseteq}{\sqcup}$ is a poset and the sequential composition $\;{\fatsemi}\;$ is
existing $\sqcup$ left preserving then we have the Galois connection
\begin{eqntabular}{C}
$\forall S\in{\mathbb{L}}\mathrel{.}\pair{\mathbb{L}}{\sqsubseteq}\galois{\textsf{\textup{post}}(S)}{\widetilde{\textsf{\textup{pre}}}(S)}\pair{\mathbb{L}}{\sqsubseteq}$\quad where\quad
$\widetilde{\textsf{\textup{pre}}}(S)Q\,\triangleq\,\bigsqcup\{P\in{\mathbb{L}}\mid {\textsf{\textup{post}}(S)}P\mathrel{\sqsubseteq}Q\}\bigr)$.
\label{eq:GC:post-S}
\end{eqntabular}
\end{lemma}
\begin{toappendix}
\begin{proof}[Proof of lemma \ref{eq:GC:post-S}]
By lemma \ref{lem:post-S}, $\textsf{post}(S)$ preserves existing joins. It is the therefore the lower adjoint of a Galois connection \cite[exercise 11.39]{Cousot-PAI-2021}.  $\widetilde{\textsf{\textup{pre}}}(S)$ is its unique upper adjoint \cite[exercise 11.39]{Cousot-PAI-2021}.
\end{proof}
\end{toappendix}
\begin{lemma}\label{lem:post}\proofinapx\quad Let $\triple{L}{\sqsubseteq}{\sqcup}$ be a poset with partially defined join $\sqcup$. Let
$\fatsemi$ be the sequential composition on $L$. If\/ ${\fatsemi}$ right-satisfies
any one of the properties of definition \ref{def:Properties:functions:posets} or their dual then 
${\textsf{\textup{post}}}$ satisfies the same property.
\end{lemma}
\begin{toappendix}
\begin{proof}[Proof of lemma \ref{lem:post}]Let $\pair{P_i}{i\in\Delta}$ be a family of elements of $L$ such that $\Delta=\{0,1\}$ with $P_0\mathrel{\sqsubseteq} P_1$ for the right-increasingness hypothesis (def.\ \ref{def:Properties:increasing}), $\Delta$ is finite
for the existing finite $\sqcup$ right preserving hypothesis (def.\ \ref{def:Properties:finite-preserving}), $\Delta\in\mathbb{O}$ and $\pair{P_i}{i\in\Delta}$ is an increasing chain for the right upper-continuity hypothesis (def.\ \ref{def:Properties:continuous}), $\Delta$ is an arbitrary set for the existing join right preservation property (def.\ \ref{def:Properties:limit-preserving}), possibly empty in case of right strictness (def.\ \ref{def:Properties:increasing:strict}). The proof is similar in all of these cases, as follows
\begin{calculus}[=\ \ ]
\formula{\textsf{post}(\bigsqcup\nolimits_{i\in\Delta}S_i)}\\
=
\formulaexplanation{\LAMBDA{P}\textsf{post}(\bigsqcup\nolimits_{i\in\Delta}S_i)P}{function application}\\
=
\formulaexplanation{\LAMBDA{P}P\mathbin{\fatsemi}(\bigsqcup\nolimits_{i\in\Delta}S_i)}{def.\ (\ref{eq:def:abstract:transformer:post}) of $\textsf{post}$}\\
=
\formulaexplanation{\LAMBDA{P}\bigsqcup\nolimits_{i\in\Delta}(P\mathbin{\fatsemi}S_i)}{by the right preservation hypothesis for $\mathbin{\fatsemi}$}\\
=
\formulaexplanation{\LAMBDA{P}\bigsqcup\nolimits_{i\in\Delta}\textsf{post}(S_i)P}{def.\ (\ref{eq:def:abstract:transformer:post}) of $\textsf{post}$}\\
=
\lastformulaexplanation{\mathop{\dot{\bigsqcup}}\nolimits_{i\in\Delta}\textsf{post}(S_i)}{pointwise def.\ of ${\dot{\bigsqcup}}$}{\mbox{\qed}}
\end{calculus}\let\qed\relax
\end{proof}
\end{toappendix}
The following Galois connection formalizes Dijkstra's program inversion \cite{DBLP:conf/pc/Dijkstra78g}.
\begin{lemma}\label{lem:GC:post}\proofinapx\quad
If $\triple{L}{\sqsubseteq}{\sqcup}$ is a poset and the sequential composition $\;{\fatsemi}\;$ is
existing $\sqcup$ right preserving then we have the following Galois connection \textup{(}${\mathbb{L}\joinmorphismto\mathbb{L}}$ is the set of existing join preserving operators on $\mathbb{L}$ and $\dot{\sqsubseteq}$ is the pointwise extension of\/ $\sqsubseteq$\textup{)}
\abovedisplayskip0pt\belowdisplayskip0pt\begin{eqntabular}{C}
$\pair{\mathbb{L}}{{\sqsubseteq}}\galois{\textsf{\textup{post}}}{\ulstrut{\textsf{\textup{post}}}^{-1}}\pair{\mathbb{L}\joinmorphismto\mathbb{L}}{{\dot{\sqsubseteq}}}$\quad where\quad ${{\textsf{\textup{post}}}^{-1}}(T)=\mathop{\dot{\bigsqcup}}\{S\in\mathbb{L}\mid
{\textsf{\textup{post}}}(S)\mathrel{\dot{\sqsubseteq}}T\}$.
\label{eq:GC:post}
\end{eqntabular}
\end{lemma}
\begin{toappendix}
\begin{proof}[Proof of lemma \ref{lem:GC:post}]By lemma \ref{lem:post}, $\textsf{post}$ preserves existing joins. It is the therefore the lower adjoint of a Galois connection \cite[exercise 11.39]{Cousot-PAI-2021}.  ${\textsf{\textup{post}}}$ is its unique upper adjoint \cite[exercise 11.39]{Cousot-PAI-2021}.
\end{proof}
\end{toappendix}
\subsection{A Calculus of Algebraic Program Execution Properties}\label{sect:Calculus-Program-Execution-Properties}
We derive the sound and complete $\textsf{post}^\sharp$ calculus by calculational design, as follows.
\begin{theorem}[Program execution property calculus]\label{th:Program:execution:properties:calculus}\proofinapx\quad
If\/ $\mathbb{D}^{\sharp}$ is a well-defined increasing and decreasing chain-complete join semilattice with right upper continuous sequential composition $\mathbin{{\fatsemi}^{\sharp}}$ then
\belowdisplayskip0pt\abovedisplayskip0pt\jot=4pt
\begin{eqntabular}{rcl}
\textsf{\textup{post}}^\sharp\sqb{\texttt{x = A}}^{\sharp} P
&=&
\triple{e:P_{+}\mathbin{\fatsemi^{\sharp}}{\textsf{\textup{assign}}^{\sharp}\sqb{\texttt{x},\texttt{A}}}}{\bot:P_{\infty}}{br:P_{br}}
\label{eq:post:abstract:assignment}\\
\textsf{\textup{post}}^\sharp\sqb{\texttt{x = [$a$, $b$]}}^{\sharp} P
&=&
\triple{e:P_{+}\mathbin{\fatsemi^{\sharp}}{\textsf{\textup{rassign}}^{\sharp}\sqb{\texttt{x},a,b}}}{\bot:P_{\infty}}{br:P_{br}}
\label{eq:post:abstract:random:assignment}\\
\textsf{\textup{post}}^\sharp\sqb{\texttt{skip}}^{\sharp}P
&=&
\triple{e:P_{+}\mathbin{\fatsemi^{\sharp}}{\textsf{\textup{skip}}^{\sharp}}}{\bot:P_{\infty}}{br:P_{br}}
\label{eq:post:abstract:skip}\\
\textsf{\textup{post}}^\sharp\sqb{\texttt{B}}^{\sharp} P
&=&
\triple{e:P_{+}\mathbin{\fatsemi^{\sharp}}{\textsf{\textup{test}}^{\sharp}}\sqb{\texttt{B}}}{\bot:P_{\infty}}{br:P_{br}}
\label{eq:post:abstract:B}\\
\textsf{\textup{post}}^\sharp\sqb{\texttt{break}}^{\sharp} P
&=&
\triple{e:{\bot_{+}^{\sharp}}}{\bot:P_{\infty}}{br:P_{br}\mathbin{\sqcup_{+}^{\sharp}}(P_e\mathbin{\fatsemi^{\sharp}}{\textsf{\textup{break}}^{\sharp}})}
\label{eq:post:abstract:break}\\
{\textsf{\textup{post}}^\sharp\sqb{\texttt{S}_1\texttt{;}\texttt{S}_2}^{\sharp} P}
&=&
{\textsf{\textup{post}}^\sharp\sqb{\texttt{S$_2$}}^{\sharp}(\textsf{\textup{post}}^\sharp\sqb{\texttt{S$_1$}}^{\sharp} P)}
\label{eq:post:abstract:seq}\\
\textsf{\textup{post}}^\sharp\sqb{\texttt{if (B) S$_1$ else S$_2$}}^{\sharp} P
&=&
\textsf{\textup{post}}^\sharp\sqb{\texttt{B;S}_1}^{\sharp}P\mathbin{\,\sqcup^{\sharp}\,}\textsf{\textup{post}}^\sharp\sqb{\neg\texttt{B;S}_2}^{\sharp}P
\label{eq:post:abstract:if}\\
{\vec{F}_{pe}^{\sharp}}&\triangleq&\LAMBDA{P}\LAMBDA{X}{\textsf{\textup{post}}^\sharp({\textsf{\textup{init}}^{\sharp}})P \mathbin{\sqcup_{+}^{\sharp}} \textsf{\textup{post}}^\sharp(\sqb{\texttt{B;S}}_{e}^{\sharp})(X)}
\label{eq:def:F-pe-sharp}\\
{F_{p\bot}^{\sharp}}&\triangleq&\LAMBDA{X}\textsf{\textup{post}}^\sharp (X)
(\sqb{\texttt{B;S}}_{e}^{\sharp})\label{eq:def:F-p-bot-sharp}\\
\textsf{\textup{post}}^\sharp\sqb{\texttt{while (B) S}}^{\sharp}P&=&
\langle ok:\langle{e:{\textsf{\textup{post}}^\sharp(\sqb{\neg\texttt{B}}_{e}^{\sharp}\mathbin{\sqcup_{e}^{\sharp}}\sqb{\texttt{B;S}}_{b}^{\sharp})(\Lfp{{\sqsubseteq}_{+}^{\sharp}}({\vec{F}_{pe}^{\sharp}}(P)))}},\,\label{eq:post:abstract:while}\\[-0.25ex]
&&\phantom{\langle ok:\langle}{\bot:{\textsf{\textup{post}}^\sharp(\sqb{\texttt{B;S}}_{\bot}^{\sharp})(\Lfp{{\sqsubseteq}_{+}^{\sharp}}({\vec{F}_{pe}^{\sharp}}(P)))}\mathbin{{\sqcup}_{\infty}^{\sharp}}{}}\nonumber\\[-0.5ex]
&&\renumber{$\textsf{\textup{post}}^\sharp(\Gfp{{\sqsubseteq}_{\infty}^{\sharp}}{F_{p\bot}^{\sharp}})P$ $\rangle,$\qquad}\\[-0.75ex]
&&\phantom{\langle}{br:P_{br}}\rangle
\nonumber
\end{eqntabular}
is sound and complete.
\end{theorem}
\begin{toappendix}
To prove theorem \ref{th:Program:execution:properties:calculus}, we need preliminary lemmas.
\begin{lemma}\label{lem:fixpoint:post:+}If\/ $\mathbb{D}^{\sharp}_{+}$ is a well-defined increasing chain-complete join semilattice with sequential composition ${{\fatsemi}^{\sharp}}$ that is existing $\sqcup$ right preserving and upper continuous in both arguments then
$\textsf{\textup{post}}^\sharp(\Lfp{\sqsubseteq_{+}^{\sharp}}{\vec{F}_{e}^{\sharp}})P$ $=$ $\Lfp{{\sqsubseteq}_{+}^{\sharp}}({\vec{F}_{pe}^{\sharp}}(P))$.
\end{lemma}
\begin{proof}[Proof of \ref{lem:fixpoint:post:+}]By lemma \ref{lem:Fesharp-welldefined}, ${\vec{F}_{e}^{\sharp}}$ is increasing 
so that the transfinite iterates $\pair{X^\delta}{\delta\in\mathbb{O}}$ of
${\vec{F}_{e}^{\sharp}}$ from ${\bot_{+}^{\sharp}}$ from an
increasing chain which is ultimately stationary at rank $\epsilon$ so that $\Lfp{\sqsubseteq_{+}^{\sharp}}{\vec{F}_{e}^{\sharp}}=X^{\epsilon}$ \cite{CousotCousot-PJM-82-1-1979}.

\smallskip

\hyphen{6} We have $\textsf{post}^\sharp(X^0)$ = $\textsf{post}^\sharp({\bot_{+}^{\sharp}})$ = $\LAMBDA{P}P\mathbin{\fatsemi^{\sharp}}{\bot_{+}^{\sharp}}$ = $\LAMBDA{P}{\bot_{+}^{\sharp}}$ by def.\ (\ref{eq:def:abstract:transformer:post}) of $\textsf{post}^\sharp$ and $\forall S\in{\mathbb{L}^{\sharp}_{+}}\mathrel{.} S\mathbin{{\fatsemi}^{\sharp}}{\bot_{+}^{\sharp}}={\bot_{+}^{\sharp}}$ in definition \ref{def:abstract:domain:well:def:init:neutral}.

\smallskip

\hyphen{6} Let us prove commutation of ${\vec{F}_{e}^{\sharp}}$ and $\LAMBDA{P}{\vec{F}_{pe}^{\sharp}}(P)$ for the abstraction $\textsf{post}^\sharp$ of the iterates.
\begin{calculus}[=\ \ ]
\formula{\textsf{post}^\sharp({\vec{F}_{e}^{\sharp}}(X^{\delta}))}\\
=
\formulaexplanation{\LAMBDA{P}\textsf{post}^\sharp({\vec{F}_{e}^{\sharp}}(X^{\delta}))P}{def.\ function application}\\
=
\formulaexplanation{\LAMBDA{P}\textsf{post}^\sharp({\textsf{init}^{\sharp}} \mathbin{\sqcup_{+}^{\sharp}} (X^{\delta}\mathbin{{\fatsemi}^{\sharp}} \sqb{\texttt{B;S}}_{e}^{\sharp}))P}{def.\ (\ref{eq:natural-transformer-finite-forward}) of ${\vec{F}_{e}^{\sharp}}$}\\
=
\formula{\LAMBDA{P}\textsf{post}^\sharp({\textsf{init}^{\sharp}})P \mathbin{\sqcup_{+}^{\sharp}} \textsf{post}^\sharp(\sqb{\texttt{B;S}}_{e}^{\sharp})(\textsf{post}^\sharp(X^{\delta})P)}\\[-0.5ex]
\rightexplanation{$\textsf{\textup{post}}^\sharp$ is existing join preserving by hypothesis on $\mathbin{{\fatsemi}^{\sharp}}$ and lemma \ref{lem:post-S}}\\
=
\formulaexplanation{\LAMBDA{P}{\vec{F}_{pe}^{\sharp}}(P)(\textsf{post}^\sharp(X^{\delta}))}{def.\ (\ref{eq:def:F-pe-sharp}) of ${\vec{F}_{pe}^{\sharp}}$}
\end{calculus}

\medskip

\noindent We conclude by continuity and \cite[th.\@ 18.26]{Cousot-PAI-2021}.
\end{proof}
\noindent Note that if $\textsf{\textup{post}}^\sharp$ is simply increasing, we have an over approximation.
\begin{lemma}\label{lem:fixpoint:post:bot}
If\/ $\mathbb{D}^{\sharp}$ is well-defined decreasing chain-complete lattice and the sequential composition ${{\fatsemi}^{\sharp}}$ is right lower continuous then $\textsf{\textup{post}}^\sharp({\Gfp{{\sqsubseteq}_{\infty}^{\sharp}}{F_{\bot}^{\sharp}}})$
= $\textsf{\textup{post}}^\sharp(\Gfp{{\sqsubseteq}_{\infty}^{\sharp}}{F_{p\bot}^{\sharp}})$.
\end{lemma}
\begin{proof}[Proof of (\ref{lem:fixpoint:post:bot})]Let us prove commutation for the iterates  $\pair{X^\delta}{\delta\in\mathbb{O}}$  of ${\Gfp{\dot{\sqsubseteq}_{\infty}^{\sharp}}{F_{\bot}^{\sharp}}}$.
\begin{calculus}[=\ \ ]
\formula{\textsf{post}^\sharp({F_{\bot}^{\sharp}}(X^\delta))}\\
=
\formulaexplanation{\LAMBDA{P}\textsf{post}^\sharp({F_{\bot}^{\sharp}}(X^\delta))P}{function application}\\
=
\formulaexplanation{\LAMBDA{P}\textsf{post}^\sharp(\sqb{\texttt{B;S}}_{e}^{\sharp}\mathbin{{\fatsemi}^{\sharp}} X^\delta)P}{def.\ (\ref{eq:trace-transformer-infinite}) of ${F_{\bot}^{\sharp}}$}\\
=
\formulaexplanation{\LAMBDA{P}P\mathbin{\fatsemi^{\sharp}}(\sqb{\texttt{B;S}}_{e}^{\sharp}\mathbin{{\fatsemi}^{\sharp}} X^\delta)}{def.\ (\ref{eq:def:abstract:transformer:post}) of $\textsf{post}^\sharp$}\\
=
\formulaexplanation{\LAMBDA{P}(P\mathbin{\fatsemi^{\sharp}}\sqb{\texttt{B;S}}_{e}^{\sharp})\mathbin{{\fatsemi}^{\sharp}} X^\delta}{$\mathbin{{\fatsemi}^{\sharp}}$ associative by definition \ref{def:abstract:domain:well:def:operators}}\\
=
\formulaexplanation{\LAMBDA{P}\textsf{post}^\sharp(X^\delta)(P\mathbin{\fatsemi^{\sharp}}\sqb{\texttt{B;S}}_{e}^{\sharp})}
{def.\ (\ref{eq:def:abstract:transformer:post}) of $\textsf{post}^\sharp$}\\
=
\formulaexplanation{\LAMBDA{P}\textsf{post}^\sharp(X^\delta)(\textsf{post}^\sharp(\sqb{\texttt{B;S}}_{e}^{\sharp})P)}
{def.\ (\ref{eq:def:abstract:transformer:post}) of $\textsf{post}^\sharp$}\\
=
\formulaexplanation{\LAMBDA{P}{F_{p\bot}^{\sharp}}(\textsf{post}^\sharp(X^\delta))P}{def.\ (\ref{eq:def:F-p-bot-sharp}) of ${F_{p\bot}^{\sharp}}$}\\
=
\formulaexplanation{{F_{p\bot}^{\sharp}}(\textsf{post}^\sharp(X^\delta))}{function application}
\end{calculus}

\medskip

\noindent By hypothesis, the sequential composition ${{\fatsemi}^{\sharp}}$ is right lower continuous, so that by lemma
\ref{lem:GC:post}, $\textsf{post}^\sharp$ is lower continuous. By commutativity, we conclude by the dual of \cite[th.\@ 18.26]{Cousot-PAI-2021}.
\end{proof}
\begin{proof}[Proof of theorem \ref{th:Program:execution:properties:calculus}]
The proof is by structural induction on the statement syntax.\par
\begin{calculus}[=\ \ ]
\hyphen{5}\formula{\textsf{post}^\sharp\sqb{\texttt{x = A}}^{\sharp} P}\\
=
\formulaexplanation{P\mathbin{\fatsemi^{\sharp}}\sqb{\texttt{x = A}}^\sharp}{def.\ (\ref{eq:def:abstract:transformer:post}) of $\textsf{post}^\sharp$}\\
=
\formulaexplanation{P\mathbin{\fatsemi^{\sharp}}\triple{e:{\textsf{assign}^{\sharp}\sqb{\texttt{x},\texttt{A}}}}{\bot:{\bot_{\infty}^{\sharp}}}{br:{\bot_{+}^{\sharp}}}}{(\ref{eq:def:Abstract-Semantic-Domain-Semantics}) and (\ref{eq:def:sem:abstract:basis})}\\
=
\formulaexplanation{\triple{e:P_{+}\mathbin{\fatsemi^{\sharp}}{\textsf{assign}^{\sharp}\sqb{\texttt{x},\texttt{A}}}}{\bot:P_{\infty} \mathbin{{\sqcup}^{\sharp}_{\infty}}(P_{+}\mathbin{\fatsemi^{\sharp}}{\bot_{\infty}^{\sharp}})}{br:P_{br}\mathbin{\sqcup_{+}^{\sharp}}(P_{+}\mathbin{\fatsemi^{\sharp}}{\bot_{+}^{\sharp}})}}{def.\ (\ref{eq:def:abstract:sem:group:seq}) of $\mathbin{\fatsemi^{\sharp}}$}\\
=
\formula{\triple{e:P_{+}\mathbin{\fatsemi^{\sharp}}{\textsf{assign}^{\sharp}\sqb{\texttt{x},\texttt{A}}}}{\bot:P_{\infty} \mathbin{{\sqcup}^{\sharp}_{\infty}}{\bot_{\infty}^{\sharp}}}{br:P_{br}\mathbin{\sqcup_{+}^{\sharp}}{\bot_{+}^{\sharp}}}}\\
\rightexplanation{${\bot_{\infty}^{\sharp}}$ and ${\bot_{+}^{\sharp}}$  absorbent for $\mathbin{\fatsemi^{\sharp}}$ by definition \ref{def:abstract:domain:well:def:oo:absorbent} }\\
=
\formulaexplanation{\triple{e:P_{+}\mathbin{\fatsemi^{\sharp}}{\textsf{assign}^{\sharp}\sqb{\texttt{x},\texttt{A}}}}{\bot:P_{\infty}}{br:P_{br}}}{def.\ lub}\\[1ex]

\hyphen{5}\discussion{The $\textsf{post}^\sharp$ transformers (\ref{eq:post:abstract:random:assignment}) for \texttt{x = [$a$, $b$]}, (\ref{eq:post:abstract:skip}) for \texttt{x = skip}, and (\ref{eq:post:abstract:B}) for \texttt{B} are similar.}\\[1em]

\hyphen{5}\formula{\textsf{post}^\sharp\sqb{\texttt{break}}^{\sharp} P}\\
=
\formulaexplanation{P\mathbin{\fatsemi^{\sharp}}\sqb{\texttt{break}}^\sharp}{def.\ (\ref{eq:def:abstract:transformer:post}) of $\textsf{post}^\sharp$}\\
=
\formulaexplanation{P\mathbin{\fatsemi^{\sharp}}
\triple{e:{\bot_{+}^{\sharp}}}{\bot:{\bot_{\infty}^{\sharp}}}{br:{\textsf{break}^{\sharp}}}}{(\ref{eq:def:Abstract-Semantic-Domain-Semantics}) and (\ref{eq:def:sem:abstract:basis})}\\
=
\formulaexplanation{
\triple{e:P_{+}\mathbin{\fatsemi^{\sharp}}{\bot_{+}^{\sharp}}}{\bot:P_{\infty}\mathbin{\fatsemi^{\sharp}}{\bot_{\infty}^{\sharp}}}{br:P_{br}\mathbin{\sqcup_{+}^{\sharp}}(P_{e}\mathbin{\fatsemi^{\sharp}}{\textsf{break}^{\sharp}})}}
{def.\ (\ref{eq:def:abstract:sem:group:seq}) of $\mathbin{\fatsemi^{\sharp}}$}\\
=
\formulaexplanation{
\triple{e:{\bot_{+}^{\sharp}}}{\bot:P_{\infty}}{br:P_{br}\mathbin{\sqcup_{+}^{\sharp}}(P_{e}\mathbin{\fatsemi^{\sharp}}{\textsf{break}^{\sharp}})}}{definitions \ref{def:abstract:domain:well:def:oo:absorbent} and  \ref{def:abstract:domain:well:def:init:neutral}}\\[1ex]

\hyphen{5}\formula{\textsf{post}^\sharp\sqb{\texttt{S}_1\texttt{;}\texttt{S}_2}^{\sharp} P}\\
=
\formulaexplanation{P\mathbin{\fatsemi^{\sharp}}(\sqb{\texttt{S}_1\texttt{;}\texttt{S}_2}^\sharp)}{def.\ (\ref{eq:def:abstract:transformer:post}) of $\textsf{post}^\sharp$}\\
=
\formulaexplanation{P\mathbin{\fatsemi^{\sharp}}(\sqb{\texttt{S}_1}^\sharp\mathbin{\fatsemi^{\sharp}}\sqb{\texttt{S}_2}^\sharp)}{def.\ (\ref{eq:def:abstract:sem:group:seq}) of $\mathbin{\fatsemi^{\sharp}}$}\\
=
\formulaexplanation{(P\mathbin{\fatsemi^{\sharp}}\sqb{\texttt{S}_1}^\sharp)\mathbin{\fatsemi^{\sharp}}\sqb{\texttt{S}_2}^\sharp}{$\mathbin{\fatsemi^{\sharp}}$ associative by definition \ref{def:abstract:domain:well:def}\ref{def:abstract:domain:well:def:operators}}\\
=
\formulaexplanation{\textsf{post}^\sharp\sqb{\texttt{S}_2}^{\sharp}(P\mathbin{\fatsemi^{\sharp}}\sqb{\texttt{S}_1}^\sharp)}{def.\ (\ref{eq:def:abstract:transformer:post}) of $\textsf{post}^\sharp\sqb{\texttt{S}_2}^{\sharp}Q\triangleq Q\mathbin{\fatsemi^{\sharp}}\sqb{\texttt{S}_2}^\sharp$}\\
=
\formulaexplanation{\textsf{post}^\sharp\sqb{\texttt{S}_2}^{\sharp}(\textsf{post}^\sharp\sqb{\texttt{S}_1}^{\sharp}P)}{def.\ (\ref{eq:def:abstract:transformer:post}) of $\textsf{post}^\sharp\sqb{\texttt{S}_1}^{\sharp}P\triangleq P\mathbin{\fatsemi^{\sharp}}\sqb{\texttt{S}_1}^\sharp$}\\[1ex]

\hyphen{5}\formula{\textsf{post}^\sharp\sqb{\texttt{if (B) S$_1$ else S$_2$}}^{\sharp} P}\\
=
\formulaexplanation{P\mathbin{\fatsemi^{\sharp}}\sqb{\texttt{if (B) S$_1$ else S$_2$}}^\sharp}{def.\ (\ref{eq:def:abstract:transformer:post}) of $\textsf{post}^\sharp$}\\
=
\formulaexplanation{P\mathbin{\fatsemi^{\sharp}}(\sqb{\texttt{B;S}_1}^{\sharp}\mathbin{\,\sqcup^{\sharp}\,}\sqb{\neg\texttt{B;S}_2}^{\sharp})}{(\ref{eq:def:abstract:sem:if}) and (\ref{eq:def:Abstract-Semantic-Domain-Semantics})}\\
=
\formula{(P\mathbin{\fatsemi^{\sharp}}\sqb{\texttt{B;S}_1}^{\sharp})\mathbin{\,\sqcup^{\sharp}\,}(P\mathbin{\fatsemi^{\sharp}}\sqb{\neg\texttt{B;S}_2}^{\sharp})}\\[-0.5ex]\rightexplanation{binary (hence finite) join preservation is definition \ref{def:abstract:domain:well:def:join:additive}, lemma \ref{lem:GC:post-S}, and (\ref{eq:def:Abstract-Semantic-Domain-Semantics})}\\
=
\formulaexplanation{\textsf{post}^\sharp\sqb{\texttt{B;S}_1}^{\sharp}P\mathbin{\,\sqcup^{\sharp}\,}\textsf{post}^\sharp\sqb{\neg\texttt{B;S}_2}^{\sharp}P}{def.\ (\ref{eq:def:abstract:transformer:post}) of $\textsf{post}^\sharp$}\\[1ex]

\hyphen{5}\ \ For $\textsf{post}^\sharp\sqb{\texttt{while (B) S}}^\sharp P$, we proceed by cases.\\[1ex]

\hyphen{6} \formula{\textsf{post}^\sharp\sqb{\texttt{while (B) S}}_{e}^{\sharp} P}\\
=
\formulaexplanation{\textsf{post}^\sharp(\Lfp{\sqsubseteq_{+}^{\sharp}}{\Cev{F}_{e}^{\sharp}}\mathbin{{\fatsemi}^{\sharp}}(\sqb{\neg\texttt{B}}_{e}^{\sharp}\mathbin{\sqcup_{e}^{\sharp}}\sqb{\texttt{B;S}}_{b}^{\sharp})P)}{(\ref{eq:sem:abstract:finite})}\\
=
\formulaexplanation{\textsf{post}^\sharp(\sqb{\neg\texttt{B}}_{e}^{\sharp}\mathbin{\sqcup_{e}^{\sharp}}\sqb{\texttt{B;S}}_{b}^{\sharp})(\textsf{post}^\sharp(\Lfp{\sqsubseteq_{+}^{\sharp}}{\Cev{F}_{e}^{\sharp}}\mathbin{{\fatsemi}^{\sharp}})P)}
{(\ref{eq:post:abstract:seq})}\\
=
\numberedformulaexplanation{\textsf{post}^\sharp(\sqb{\neg\texttt{B}}_{e}^{\sharp}\mathbin{\sqcup_{e}^{\sharp}}\sqb{\texttt{B;S}}_{b}^{\sharp})(\Lfp{{\sqsubseteq}_{+}^{\sharp}}({\vec{F}_{pe}^{\sharp}}(P)))}{lemma \ref{lem:fixpoint:post:+}\label{val:post:+:while}}\\[1ex]

\hyphen{6}\ \ Similarly, for case (\ref{eq:sem:abstract:body-infinite}), we get\\
\formula{\textsf{post}^\sharp\sqb{\texttt{while (B) S}}_{bi}^{\sharp}\,P}\\
=
\formula{\textsf{post}^\sharp(\sqb{\texttt{B;S}}_{\bot}^{\sharp})(\Lfp{{\sqsubseteq}_{+}^{\sharp}}({\vec{F}_{pe}^{\sharp}}(P)))}\\[1ex]

\hyphen{6} \formula{\textsf{post}^\sharp\sqb{\texttt{while (B) S}}_{b}^{\sharp}P}\\
=
\formulaexplanation{P\mathbin{\fatsemi^{\sharp}}\sqb{\texttt{while (B) S}}_{b}^{\sharp}}{def.\ (\ref{eq:def:abstract:transformer:post}) of $\textsf{post}^\sharp$}\\
=
\formulaexplanation{P\mathbin{\fatsemi^{\sharp}}\bot_{+}^{\sharp}}{(\ref{eq:sem:abstract:break})}\\
=
\formulaexplanation{\bot_{+}^{\sharp}}{$\bot_{+}^{\sharp}$ absorbent for $\mathbin{\fatsemi^{\sharp}}$ in definition \ref{def:abstract:domain:well:def:init:neutral}}\\[1ex]

\hyphen{6} 
\formula{\textsf{post}^\sharp\sqb{\texttt{while (B) S}}_{li}^{\sharp}}\\
=
\formulaexplanation{\textsf{post}^\sharp(\Gfp{\sqsubseteq_{\infty}^{\sharp}}{F_{\bot}^{\sharp}})}{(\ref{eq:sem:abstract:loop-infinite})}\\
=
\numberedformulaexplanation{\textsf{\textup{post}}^\sharp(\Gfp{{\sqsubseteq}_{\infty}^{\sharp}}{F_{p\bot}^{\sharp}})}{lemma \ref{lem:fixpoint:post:bot}}\label{val:post:oo:while}\\[1ex]

\hyphen{6} \formula{\textsf{post}^\sharp(\sqb{\texttt{while (B) S}}_{\bot}^{\sharp})}\\
=
\formulaexplanation{\textsf{post}^\sharp(\sqb{\texttt{while (B) S}}_{bi}^{\sharp}\mathbin{\sqcup_{\infty}^{\sharp}}\sqb{\texttt{while (B) S}}_{li}^{\sharp})}{(\ref{eq:sem:abstract:loop-nontermination})}\\
=
\formulaexplanation{\textsf{post}^\sharp(\sqb{\texttt{while (B) S}}_{bi}^{\sharp})\mathbin{\dot{\sqcup}_{\infty}^{\sharp}}\textsf{post}^\sharp(\sqb{\texttt{while (B) S}}_{li}^{\sharp})}{binary (hence finite) join preservation and (\ref{lem:GC:post-S})}\\
=
\formulaexplanation{\LAMBDA{P}\textsf{post}^\sharp(\sqb{\texttt{while (B) S}}_{bi}^{\sharp})P\mathbin{{\sqcup}_{\infty}^{\sharp}}\textsf{post}^\sharp(\sqb{\texttt{while (B) S}}_{li}^{\sharp})P}{pointwise def.\ $\mathbin{\dot{\sqcap}_{\infty}}$}\\
=
\formulaexplanation{\LAMBDA{P}{\textsf{post}^\sharp(\sqb{\texttt{B;S}}_{\bot}^{\sharp})(\Lfp{{\sqsubseteq}_{+}^{\sharp}}{\vec{F}_{pe}^{\sharp}}(P))}
\mathbin{{\sqcup}_{\infty}^{\sharp}}
\textsf{\textup{post}}^\sharp(\Gfp{{\sqsubseteq}_{\infty}^{\sharp}}{F_{p\bot}^{\sharp}})P}{as proved in (\ref{val:post:+:while}) and (\ref{val:post:oo:while})}\\[1ex]

\hyphen{6}\ \ Grouping all cases together, we get\\
\formula{\textsf{post}^\sharp\sqb{\texttt{while (B) S}}^{\sharp} P}\\
=
\formulaexplanation{P\mathbin{\fatsemi^{\sharp}}\sqb{\texttt{while (B) S}}^{\sharp}}{def.\ (\ref{eq:def:abstract:transformer:post}) of $\textsf{post}^\sharp$}\\
=
\formulaexplanation{P\mathbin{\fatsemi^{\sharp}}\triple{e:\sqb{\texttt{while (B) S}}_{e}^{\sharp}}{\bot:\sqb{\texttt{while (B) S}}_{\bot}^{\sharp}}{br:\sqb{\texttt{while (B) S}}_{b}^{\sharp}}}{(\ref{eq:def:Abstract-Semantic-Domain-Semantics})}\\
=
\formulaexplanation{\triple{e:P_{ok}^{+}\mathbin{\fatsemi^{\sharp}}\sqb{\texttt{while (B) S}}_{e}^{\sharp}}{\bot:P_{ok}^{\infty}\mathbin{{\sqcup}^{\sharp}_{\infty}}P_{ok}^{+}\mathbin{\fatsemi^{\sharp}}\sqb{\texttt{while (B) S}}_{\bot}^{\sharp}}{br:P_{br}\mathbin{{\sqcup}^{\sharp}_{+}}P_{ok}^{+}\mathbin{\fatsemi^{\sharp}}\sqb{\texttt{while (B) S}}_{b}^{\sharp}}}{def.\ (\ref{eq:def:abstract:sem:group:seq}) of $\mathbin{\fatsemi^{\sharp}}$}\\
=
\formula{\triple{e:\textsf{post}^\sharp\sqb{\texttt{while (B) S}}_{e}^{\sharp}P}{\bot:\textsf{post}^\sharp\sqb{\texttt{while (B) S}}_{\bot}^{\sharp}P}{br:P_{br}}}\\[-0.5ex]\rightexplanation{def.\ (\ref{eq:def:abstract:transformer:post}) of $\textsf{post}^\sharp$, $\sqb{\texttt{while (B) S}}_{b}^{\sharp}\triangleq\bot_{+}^{\sharp}$ by (\ref{eq:sem:abstract:break}), $P_{ok}^{+}\mathbin{\fatsemi^{\sharp}}\bot_{+}^{\sharp}=\bot_{+}^{\sharp}$ by \ref{def:abstract:domain:well:def:bot:absorbent}, and $\bot_{+}^{\sharp}$ infimum by \ref{def:abstract:domain:well:def:finite:domain}}\\
=
\lastformulaexplanation{\triple{e:{\textsf{post}^\sharp(\sqb{\neg\texttt{B}}_{e}^{\sharp}\mathbin{\sqcup_{e}^{\sharp}}\sqb{\texttt{B;S}}_{b}^{\sharp})(\Lfp{{\sqsubseteq}_{+}^{\sharp}}{\vec{F}_{pe}^{\sharp}}(P))}}{\bot:{\textsf{post}^\sharp(\sqb{\texttt{B;S}}_{\bot}^{\sharp})(\Lfp{{\sqsubseteq}_{+}^{\sharp}}{\vec{F}_{pe}^{\sharp}}(P))}\mathbin{{\sqcap}_{\infty}^{\sharp}}{\textsf{\textup{post}}^\sharp(\Gfp{{\sqsubseteq}_{\infty}^{\sharp}}{F_{p\bot}^{\sharp}})P}}{br:P_{br}}}{as previously proved for each case, proving (\ref{eq:post:abstract:while}).}{\mbox{\qed}}
\end{calculus}
\let\qed\relax
\end{proof}
\end{toappendix}
\begin{remark}\label{rem:psot:instance:abstract:semantics}By defining the appropriate primitives, the \textsf{post} program execution calculus (\ref{eq:post:abstract:assignment}) --- (\ref{eq:post:abstract:while}) of theorem \ref{th:Program:execution:properties:calculus} is an instance of the generic abstract semantics (\ref{eq:def:Abstract-Semantic-Domain-Semantics}).
\end{remark}
\begin{example}[Finitary powerset deterministic calculational domain]\label{ex:powerset-deterministic-domain-post}In \cite{DBLP:conf/popl/AssafNSTT17}, the while language is deterministic and has no \texttt{break}s so the random assignment and \texttt{break}s
 have to be eliminated in (\ref{eq:def:sem:abstract:basis}). The denotational semantics is $\sqb{\texttt{S}}\in (\Sigma\times\Sigma)_{\bot}\functionto (\Sigma\times\Sigma)_{\bot}$ where $(\Sigma\times\Sigma)_{\bot}$ is the domain of relations between states extended by $\bot$ to denote nontermination with Scott flat ordering $\sqsubseteq$. 

Anticipating on the abstractions of \hyperlink{PARTII}{part II}, this is an abstraction \cite[sect.\@ 8.2]{DBLP:journals/tcs/Cousot02} of the trace semantics of sect.\@ \ref{sect:Trace-Semantics}. Then a semantic abstraction \ref{def:exact:abstraction} gets rid of nontermination 
\cite[sect.\@ 8.1.6]{DBLP:journals/tcs/Cousot02} and another one  \cite[sect.\@ 9.1]{DBLP:journals/tcs/Cousot02} abstracts relations to transformers to yield the collecting semantics \cite[p.\ 876]{DBLP:conf/popl/AssafNSTT17}. 

Skipping these abstractions of the trace semantics, we can directly instantiate the generic abstract semantics of sect.\@ \ref{sect:Algebraic-Semantics} to a finitary relational semantics such as $\sqb{S}^e$ in \cite{DBLP:journals/pacmpl/Cousot24}. Then $\textsf{post}^\sharp$
in (\ref{eq:def:abstract:transformer:post}) becomes $\textsf{post}^\sharp(S)P$ $=$ $\{\pair{\sigma}{\sigma''}\mid\exists\sigma'\in\Sigma\mathrel{.}\pair{\sigma}{\sigma'}\in P\wedge\pair{\sigma'}{\sigma''}\in S\}$,  
which is a specification of the collecting semantics postulated in \cite[p.\ 876]{DBLP:conf/popl/AssafNSTT17}.
$\textsf{post}^\sharp(S)$ preserves arbitrary unions so, in absence of \texttt{break}s and ignoring nontermination, together with $\sqb{\texttt{B}}_{e}^{\sharp}\comp\sqb{\texttt{B}}_{e}^{\sharp}=\sqb{\texttt{B}}_{e}^{\sharp}$, $\sqb{\texttt{B}}_{e}^{\sharp}\comp\sqb{\neg\texttt{B}}_{e}^{\sharp}=\emptyset$, and$\sqb{\texttt{skip}}_{e}^{\sharp}={\textsf{init}^{\sharp}}$ by \ref{def:abstract:domain:well:def:init:neutral}, (\ref{eq:post:abstract:while}) in theorem \ref{th:Program:execution:properties:calculus} simplifies to
\begin{eqntabular*}{rcl}
\textsf{\textup{post}}^\sharp\sqb{\texttt{while (B) S}}^{\sharp} P
&=&\textsf{\textup{post}}^\sharp(\sqb{\neg\texttt{B}}_{e}^{\sharp})(\Lfp{\subseteq}\LAMBDA{X}P \cup \textsf{\textup{post}}^\sharp(\sqb{\texttt{\texttt{if (B) S else skip}}}_{e}^{\sharp})(X)
\end{eqntabular*}
which is precisely the data-independent abstraction of the collecting semantics of \cite[p.\ 876]{DBLP:conf/popl/AssafNSTT17}.
\end{example}
%
\subsection{Algebraic Logics of Program Execution Properties}\label{sec:Algebraic-Logics-Program-Execution-Properties}
By defining $\overline{\{}\,P\,\overline{\}}\,\texttt{S}\,\overline{\{}\,Q\,\overline{\}}\triangleq(\pair{P}{Q}\in{\maccent{\alpha}{\filledtriangleup}}(\sqb{\texttt{S}}^\sharp))$ with   ${\maccent{\alpha}{\filledtriangleup}}(S)$ $\triangleq$ $\{\pair{P}{Q}\mid\textsf{post}^\sharp(S) P\mathrel{{\sqsubseteq^{\sharp}}} Q\}$ and dually 
$\underline{\{}\,P\,\underline{\}}\,\texttt{S}\,\underline{\{}\,Q\,\underline{\}}\triangleq(\pair{P}{Q}\in\maccent{\alpha}{\filledtriangledown}(\sqb{\texttt{S}}^\sharp ))$ with
$\maccent{\alpha}{\filledtriangledown}(S)$ $\triangleq$ $\{\pair{P}{Q}\mid Q \mathrel{{\sqsubseteq^{\sharp}}} \textsf{post}^\sharp(S) P\}$, we respectively get the abstract version \cite[chapter 26]{Cousot-PAI-2021} of Hoare logic \cite{DBLP:journals/cacm/Hoare69} and that of reverse/incorrectness logic \cite{DBLP:conf/sefm/VriesK11,DBLP:journals/pacmpl/OHearn20} (extended to loops breaks and nontermination \cite{DBLP:journals/acta/MannaP74,DBLP:journals/pacmpl/Cousot24}). This is now classic and will be used but not be further detailled.

\section{A Calculus of Algebraic Program Semantic (Hyper) Properties}\label{sec:Calculus:Hyper:Properties}
We now study proof methods for semantic properties, that is properties of the semantics, that we define in extension. This is called hyperproperties when the semantics is a set of traces \cite{DBLP:journals/jcs/ClarksonS10,DBLP:conf/post/ClarksonFKMRS14}, and by extension, for their abstractions, in particular to relational semantics.

\subsection{Algebraic Semantic (Hyper) Properties}\label{sect:HyperProperties}

Defined in extension, program semantic properties are in $\wp({\mathbb{L}^{\sharp}})$. 
\begin{example}[Algebraic noninterference]\label{ex:Algebraic-noninterference} Noninterference \cite{DBLP:conf/mfcs/GoguenM77}, can be generalized
to semantic (hyper) properties of algebraic semantics, as follows.  The precondition $R_i\in\wp({\mathbb{L}^{\sharp}_{+}}\times{\mathbb{L}^{\sharp}_{+}})$ is a relation between prelude executions extended to $\mathbb{L}^{\sharp}$ by (\ref{eq:unique-nonempty-component-shorthand}). The postcondition $R_f\in\wp({\mathbb{L}^{\sharp}}\times{\mathbb{L}^{\sharp}})$ is a relation between terminated or infinite executions. Then algebraic noninterference is
$\textsf{ANI}\triangleq\{\mathcal{P}\in\wp(\mathbb{L}^{\sharp})\mid\forall S_1,S_2\in \mathcal{P}\mathrel{.}\forall P_1,P_2\in {\mathbb{L}^{\sharp}_{+}}\mathrel{.}\pair{P_1}{P_2}\in R_i\implies
\pair{\textsf{post}^\sharp(S_1)P_1}{\textsf{post}^\sharp(S_2)P_2}\in R_f\}$. An instance is algebraic abstract noninterference
$\textsf{AANI}\triangleq\{\mathcal{P}\in\wp(\mathbb{L}^{\sharp})\mid\forall S_1,S_2\in \mathcal{P}\mathrel{.}\forall P_1,P_2\in {\mathbb{L}^{\sharp}_{+}}\mathrel{.}\alpha_1({P_1})=\alpha_1({P_2})\implies
\alpha_2({\textsf{post}^\sharp(S_1)P_1})=\alpha_2({\textsf{post}^\sharp(S_2)P_2})\}$ for abstractions
$\alpha_1\in\mathbb{L}^{\sharp}\rightarrow A_1$ and $\alpha_2\in\mathbb{L}^{\sharp}\rightarrow A_2$ with special case
$\alpha_1=\alpha_2$ to characterize abstract domain completeness in abstract interpretation \cite{DBLP:journals/tissec/GiacobazziM18,DBLP:conf/sas/MastroeniP23,DBLP:journals/toplas/GiacobazziMP24}.
After \cite{DBLP:journals/jcs/ClarksonS10}, the generalized algebraic noninterference is
$\textsf{GANI}\triangleq\{\mathcal{P}\in\wp(\mathbb{L}^{\sharp})\mid\forall S_1,S_2\in \mathcal{P}\mathrel{.}\exists\bar{S}\in\mathcal{P}\mathrel{.}\forall P_1,P_2\in {\mathbb{L}^{\sharp}_{+}}\mathrel{.}\forall\bar{P}\in\bar{S}\mathrel{.}\pair{\bar{P}}{P_1}\in R_i\implies
\pair{\textsf{post}^\sharp(S_1)\bar{P}}{\textsf{post}^\sharp(S_2)P_2}\in R_f\}$.
\end{example}
\subsection{The Algebraic Program Semantic (Hyper) Properties Transformer}\label{sect:Transformer-HyperProperties}

When considering semantic properties in extension, the traditional view of transformers is that they now belong to  $\wp({\mathbb{L}^{\sharp}})\functionto\wp({\mathbb{L}^{\sharp}})$ with
\begin{eqntabular}{rcl}
\textsf{Post}^\sharp&\in&{\mathbb{L}^{\sharp}}\functionto\wp({\mathbb{L}^{\sharp}})\increasingfunctionto\wp({\mathbb{L}^{\sharp}})\nonumber\\
\textsf{Post}^\sharp(S)\mathcal{P}&\triangleq&\{\textsf{post}^\sharp(S)P\mid P\in\mathcal{P}\}
\label{eq:def:Post}
\end{eqntabular}
\cite{DBLP:conf/popl/AssafNSTT17,DBLP:conf/sas/MastroeniP18,DBLP:journals/afp/Dardinier23a,DBLP:conf/pldi/DardinierM24} are all instances of this definition.  The advantage is that logical implication is the traditional $\subseteq$.
But the classic  structural definition (see sect.\@ \ref{sec:Structural-Definitions}) of the transformer $\textsf{Post}^\sharp$ fails (unless restrictions are placed on the considered hyperproperties). For the conditional
\begin{calculus}[=\ \ ]
\formula{\textsf{Post}^\sharp\sqb{\texttt{if (B) S$_1$ else S$_2$}}^{\sharp}\mathcal{P}}\\
=
\unnumberedformulaexplanation{\{\textsf{post}^\sharp\sqb{\texttt{if (B) S$_1$ else S$_2$}}^{\sharp}P\mid P\in\mathcal{P}\}}{def.\ (\ref{eq:def:Post}) of $\textsf{Post}^\sharp(S)$}\\
=
\numberedformulaexplanation{\{\textsf{post}^\sharp\sqb{\texttt{B;S}_1}^{\sharp}P\mathbin{\,\sqcup^{\sharp}\,}\textsf{post}^\sharp\sqb{\neg\texttt{B;S}_2}^{\sharp}P\mid  P\in\mathcal{P}\}}{(\ref{eq:post:abstract:if})}\label{cal:Post:if:exact}\\
$\subseteq$
\numberedformulaexplanation{\{\textsf{post}^\sharp\sqb{\texttt{B;S}_1}^{\sharp}P_1\mathbin{\,\sqcup^{\sharp}\,}\textsf{post}^\sharp\sqb{\neg\texttt{B;S}_2}^{\sharp}P_2\mid P_1\in\mathcal{P}\wedge P_2\in\mathcal{P}\}}{def.\ $\subseteq$}\label{cal:Post:if:approximate}\\
=
\unnumberedformulaexplanation{\{Q_1\mathbin{\sqcup^{\sharp}}Q_2\mid Q_1{\mskip2mu\in\mskip2mu}\{\textsf{post}^\sharp\sqb{\texttt{B;S}_1}^{\sharp}P_1\mid P_1{\mskip2mu\in\mskip2mu}\mathcal{P}\}\wedge Q_2{\mskip2mu\in\mskip2mu}\{\textsf{post}^\sharp\sqb{\neg\texttt{B;S}_2}^{\sharp}P_2 \mid P_2{\mskip2mu\in\mskip2mu}\mathcal{P}\}\}}{def.\ $\in$}\\
=
\unnumberedformulaexplanation{\{Q_1\mathbin{\sqcup^{\sharp}}Q_2\mid Q_1{\mskip2mu\in\mskip2mu}\textsf{Post}^\sharp\sqb{\texttt{B;S}_1}^{\sharp}\mathcal{P}
\wedge 
Q_2{\mskip2mu\in\mskip2mu}\textsf{Post}^\sharp\sqb{\neg\texttt{B;S}_2}^{\sharp}\mathcal{P}\}}{def.\ (\ref{eq:def:Post}) of $\textsf{Post}^\sharp(S)$}
\end{calculus}

\medskip

\noindent The problem is that in (\ref{cal:Post:if:exact}) the two possible executions of the conditional are tight together, whereas, by necessity of traditional independent structural induction on both branches of the conditional, this link is lost in (\ref{cal:Post:if:approximate}). So the hypercollecting semantics of \cite[p.\ 877]{DBLP:conf/popl/AssafNSTT17} is incomplete (the inclusion (\ref{cal:Post:if:approximate}) may be strict).
\ifshort\else

\fi
A solution to preserve structurality is to observe that
\begin{eqntabular}{rcl}
\{\textsf{post}^\sharp(S)P\} &=& \textsf{Post}^\sharp(S)\{P\}
\label{eq:Post::post}
\end{eqntabular}
so that the calculation goes on at (\ref{cal:Post:if:exact})
\begin{calculus}[=\ \ ]
=
\formulaexplanation{\{Q_1\mathbin{\,\sqcup^{\sharp}\,}Q_2\mid  
Q_1\in\{\textsf{post}^\sharp\sqb{\texttt{B;S}_1}^{\sharp}P\}
\wedge
Q_2\in\{\textsf{post}^\sharp\sqb{\neg\texttt{B;S}_2}^{\sharp}P\}
\wedge
P\in\mathcal{P}\}}{def.\ singleton and $\in$}\\
=
\formulaexplanation{\{Q_1\mathbin{\,\sqcup^{\sharp}\,}Q_2\mid  
Q_1\in\textsf{Post}^\sharp\sqb{\texttt{B;S}_1}^{\sharp}\{P\}
\wedge
Q_2\in\textsf{Post}^\sharp\sqb{\neg\texttt{B;S}_2}^{\sharp}\{P\} 
\wedge
P\in\mathcal{P}\}}{def.\ (\ref{eq:def:Post}) of $\textsf{Post}^\sharp(S)$}
\end{calculus}
\noindent\uLstrut so that $\textsf{Post}^\sharp\sqb{\texttt{if (B) S$_1$ else S$_2$}}^{\sharp}$ is exactly defined structurally as a function of the components $\textsf{Post}^\sharp\sqb{\texttt{B;S}_1}^{\sharp}$ and $\textsf{Post}^\sharp\sqb{\neg\texttt{B;S}_2}^{\sharp}$.

Of course, this element wise reasoning may be considered inelegant. Its necessity becomes more clear when considering the trace semantics of sect.\@ \ref{sect:Trace-Semantics}. When reasoning on paths e.g.\ in an iteration statement, the same paths must be considered consistently at each iteration. This requirement may be lifted after abstraction, for example with invariants which forget about computation history.
\ifshort\else

\fi
For backward reasonings, we define \textsf{Pre} such that for all $S\in {\mathbb{L}^{\sharp}}$, we have\proofinapx\par
\noindent\begin{minipage}{0.46\textwidth}
\bgroup\abovedisplayskip3pt\belowdisplayskip0pt\begin{eqntabular}{rcl}
\textsf{Pre}(S)\mathcal{Q}&\triangleq&\{P\mid \textsf{post}^\sharp(S)P\in\mathcal{Q}\}
\label{eq:def:Pre}
\end{eqntabular}\egroup
\end{minipage}%
\hfill
\begin{minipage}{0.4\textwidth}
\bgroup\abovedisplayskip3pt\belowdisplayskip0pt\begin{eqntabular}{c@{\qquad}}
\pair{\wp({\mathbb{L}^{\sharp}})}{\subseteq}\galois{\textsf{Post}^\sharp(S)}{{\textsf{Pre}}(S)}\pair{\wp({\mathbb{L}^{\sharp}})}{\subseteq}\label{eq:GC:abstract:transformers}
\end{eqntabular}\egroup
\end{minipage}\par
\begin{toappendix}
\begin{proof}[Proof of (\ref{eq:GC:abstract:transformers})]
\begin{calculus}[$\Leftrightarrow$\ \ ]
\formula{\textsf{Post}^\sharp(S)\mathcal{P}\subseteq \mathcal{Q}}\\
$\Leftrightarrow$
\formulaexplanation{\{\textsf{post}^\sharp(S)P\mid P\in\mathcal{P}\}\subseteq \mathcal{Q}}{def.\ (\ref{eq:def:Post}) of $\textsf{Post}^\sharp$}\\
$\Leftrightarrow$
\formulaexplanation{\forall P\in\mathcal{P}\mathrel{.}\textsf{post}^\sharp(S)P\in\mathcal{Q}}{def.\ $\subseteq$}\\
$\Leftrightarrow$
\formulaexplanation{\mathcal{P}\subseteq\{P\mid\textsf{post}^\sharp(S)P\in\mathcal{Q}\}}{def.\ $\subseteq$}\\
$\Leftrightarrow$
\lastformulaexplanation{\mathcal{P}\subseteq\textsf{Pre}(S)\mathcal{Q}}{def.\ \ref{eq:GC:abstract:transformers}) of $\textsf{Pre}$}{\mbox{\qed}}
\end{calculus}
\let\qed\relax
\end{proof}
\end{toappendix}
If ${\mathbb{D}^{\sharp}}$ is a well-defined chain-complete lattice with right finite $\bigmsqcup{x}$ preservation composition $\mathbin{{\fatsemi}^{\sharp}}$ then we have ($\bigmsqcup{x}$, $x\in\{{+},{\infty}\}$, stands for ${\sqcup_{+}^{\sharp}}$ in definition \ref{def:abstract:domain:well:def:finite:domain} when $x={+}$ and for ${\sqcup_{\infty}^{\sharp}}$ in  definition \ref{def:abstract:domain:well:def:infinite:domain} when $x={\infty}$)\proofinapx
\bgroup\abovedisplayskip3pt\belowdisplayskip0pt\begin{eqntabular}{L@{\qquad}rcl}
&\textsf{Post}^\sharp(S_1\msqcup{x}S_2)\mathcal{P}&=&(\textsf{Post}^\sharp(S_1)\mathbin{\dot{\msqcup{x}}}\textsf{Post}^\sharp(S_2))
\mathcal{P}\label{eq:Post:abstract:dot-sqcup-sharp}\\
where&(S_1
\mathbin{\dot{\msqcup{x}}}
S_2)\mathcal{P}&\triangleq&\{Q_1\msqcup{x}Q_2\mid  
Q_1\in S_1\{P\}
\wedge
Q_2\in S_2\{P\} 
\wedge
P\in\mathcal{P}\}\nonumber
\end{eqntabular}\egroup
\begin{toappendix}
\begin{proof}[Proof of (\ref{eq:Post:abstract:dot-sqcup-sharp})]
\begin{calculus}[=\ \ ]
\formula{\textsf{Post}^\sharp(S_1\msqcup{x}S_2)\mathcal{P}}\\
=
\formulaexplanation{\{\textsf{post}^\sharp(S_1\msqcup{x}S_2)P\mid P\in\mathcal{P}\}}{(\ref{eq:Post::post})}\\
=
\formulaexplanation{\{P\mathbin{\fatsemi^{\sharp}}(S_1\msqcup{x}S_2)\mid P\in\mathcal{P}\}}{def.\ (\ref{eq:def:abstract:transformer:post}) of $\textsf{post}^\sharp$}\\
=
\formulaexplanation{\{(P\mathbin{\fatsemi^{\sharp}}S_1)\msqcup{x}(P\mathbin{\fatsemi^{\sharp}}S_2)\mid P\in\mathcal{P}\}}{right finite $\bigmsqcup{x}$ preservation in definition \ref{def:abstract:domain:well:def:join:additive}}\\
=
\formulaexplanation{\{\textsf{post}^\sharp(S_1)P\msqcup{x}\textsf{post}^\sharp(S_2)P\mid P\in\mathcal{P}\}}{def.\ (\ref{eq:def:abstract:transformer:post}) of $\textsf{post}^\sharp$}\\
=
\formulaexplanation{\{Q_1\msqcup{x}Q_2\mid Q_1\in \{\textsf{post}^\sharp(S_1)P\}\wedge Q_2\in\{\textsf{post}^\sharp(S_2)P\}\wedge P\in\mathcal{P}\}}{def.\ $\in$ and singleton}\\
=
\formulaexplanation{\{Q_1\msqcup{x}Q_2\mid Q_1\in \textsf{Post}^\sharp(S_1)\{P\}\wedge Q_2\in\textsf{Post}^\sharp(S_2)\{P\}\wedge P\in\mathcal{P}\}}{(\ref{eq:Post::post})}\\
=
\lastformulaexplanation{\textsf{Post}^\sharp(S_1)\mathbin{\dot{\msqcup{x}}}\textsf{Post}^\sharp(S_2)
\mathcal{P}}{def.\ $\mathbin{\dot{\msqcup{x}}}$ in (\ref{eq:Post:abstract:dot-sqcup-sharp})}{\mbox{\qed}}
\end{calculus}
\let\qed\relax
\end{proof}
\end{toappendix}
\begin{remark}\label{rem:Post-non-preservation}Contrary to join preservation lemma \ref{lem:post-S} for \textsf{post}, \textsf{Post} may not preserve existing joins and meets so that, in general, $\bigsqcup\nolimits\limits_{i\in\Delta}{\textsf{Post}^\sharp(S_i)}\neq{\textsf{Post}^\sharp(\bigsqcup\nolimits\limits_{i\in\Delta}S_i)}$ and dually. 
For example, let $\mathcal{P}$ be a semantic property. By (\ref{eq:def:Post}), $\mathop{{\bigsqcup}^{\sharp}_{+}}\limits_{n\in\mathbb{N}}{\textsf{\textup{Post}}^\sharp((\sqb{\texttt{B}\fatsemi\texttt{S}}^{\sharp})^n)\mathcal{P}}$ = $\mathop{{\bigsqcup}^{\sharp}_{+}}\limits_{n\in\mathbb{N}}\{\textsf{post}^\sharp((\sqb{\texttt{B}\fatsemi\texttt{S}}^{\sharp})^n)P\mid P\in\mathcal{P}\}$ is the set of finite executions, for every precondition $P\in\mathcal{P}$, reaching the entry of the iteration \texttt{while(B) S} after exactly $n$ terminating body iterations, for all $n\in\mathbb{N}$. On the contrary ${\textsf{\textup{Post}}^\sharp(\mathop{{\bigsqcup}^{\sharp}_{+}}\limits_{n\in\mathbb{N}}(\sqb{\texttt{B}\fatsemi\texttt{S}}^{\sharp})^n)\mathcal{P}}$ = $\{\textsf{post}^\sharp(\mathop{{\bigsqcup}^{\sharp}_{+}}\limits_{n\in\mathbb{N}}(\sqb{\texttt{B}\fatsemi\texttt{S}}^{\sharp})^n)P\mid P\in\mathcal{P}\}$ =  $\{\mathop{{\bigsqcup}^{\sharp}_{+}}\limits_{n\in\mathbb{N}}\textsf{post}^\sharp((\sqb{\texttt{B}\fatsemi\texttt{S}}^{\sharp})^n)P\mid P\in\mathcal{P}\}$
is the set of finite executions, for every precondition $P\in\mathcal{P}$, reaching the entry of the iteration \texttt{while(B) S} after any number of terminating body iterations.
\end{remark}
\subsection{A Calculus of Algebraic Semantic (Hyper) Properties}\label{sect:Calculus-HyperProperties}
In the calculational design of the $\textsf{Post}^\sharp$, we will need the following trivial proposition.
\begin{proposition}[Singleton fixpoint]\label{prop:prop:singletonization}
There is an obvious isomorphism between a poset $\quadruple{L}{\sqsubseteq}{\bot}{\sqcup}$ and its singletons
$\quadruple{\breve{L}}{\breve{\sqsubseteq}}{\breve{\bot}}{\breve{\sqcup}}$ with ${\breve{L}}\triangleq{\{\{x\}\mid x\in L\}}$, $\{x\}\mathbin{\breve{\sqsubseteq}}\{y\}\triangleq x\sqsubseteq y$, ${\breve{\bot}}\triangleq\{\bot\}$, $\{x\}{\breve{\sqcup}}\{y\}\triangleq \{x \sqcup y\}$, so that, for a increasing chain complete poset we have $\{\Lfp{\sqsubseteq}F\}=\{\bigsqcup_{\delta\in \mathbb{O}}F^{\delta}\}=\breve{\bigsqcup}_{\delta\in \mathbb{O}}\{F^{\delta}\}=\Lfp{\breve{\sqsubseteq}}\breve{F}$ where $\pair{F^{\delta}}{\delta\in\mathbb{O}}$ are the transfinite iterates of $F$ from $\bot$ and $\breve{F}(\{x\})\triangleq\{F(x)\}$. Dually for greatest fixpoints.
\end{proposition}
We derive the sound and complete $\textsf{Post}^\sharp$ calculus by calculational design, as follows.

\smallskip

\begin{theorem}[Program semantic (hyper) property calculus]\label{th:upper-abstract-hyper-properties}\proofinapx\quad
If\/ $\mathbb{D}^{\sharp}$ is a well-defined increasing and decreasing chain-complete join semilattice with right upper continuous sequential composition $\mathbin{{\fatsemi}^{\sharp}}$ then
\arraycolsep0.5\arraycolsep\jot3pt
\begin{eqntabular}[fl]{rcl}
\textsf{\textup{Post}}^\sharp\sqb{\texttt{x = A}}^{\sharp}\mathcal{P}
&=&
\{\triple{e:P_{+}\mathbin{\fatsemi^{\sharp}}{\textsf{\textup{assign}}^{\sharp}\sqb{\texttt{x},\texttt{A}}}}{\bot:P_{\infty}}{br:P_{br}}\mid P\in\mathcal{P}\}
\label{eq:Post:abstract:assignment}\\[0.75ex]
\textsf{\textup{Post}}^\sharp\sqb{\texttt{x = [$a$, $b$]}}^{\sharp}\mathcal{P}
&=&
\{\triple{e:P_{+}\mathbin{\fatsemi^{\sharp}}{\textsf{\textup{rassign}}^{\sharp}\sqb{\texttt{x},a,b}}}{\bot:P_{\infty}}{br:P_{br}}\mid P\in\mathcal{P}\}
\label{eq:Post:abstract:random:assignment}\\[0.75ex]
\textsf{\textup{Post}}^\sharp\sqb{\texttt{skip}}\mathcal{P}
&=&
\{\triple{e:P_{+}\mathbin{\fatsemi^{\sharp}}{\textsf{\textup{skip}}^{\sharp}}}{\bot:P_{\infty}}{br:P_{br}}\mid P\in\mathcal{P}\}
\label{eq:Post:abstract:skip}\\[0.75ex]
\textsf{\textup{Post}}^\sharp\sqb{\texttt{B}}^{\sharp}\mathcal{P}
&=&
\{\triple{e:P_{+}\mathbin{\fatsemi^{\sharp}}{\textsf{\textup{test}}^{\sharp}}\sqb{\texttt{B}}}{\bot:P_{\infty}}{br:P_{br}}\mid P\in\mathcal{P}\}
\label{eq:Post:abstract:B}\\[0.75ex]
\textsf{\textup{Post}}^\sharp\sqb{\texttt{break}}^{\sharp}\mathcal{P}
&=&
\{\triple{e:{\bot_{+}^{\sharp}}}{\bot:P_{\infty}}{br:P_{br}\mathbin{{\sqcup}^{\sharp}_{+}}(P_{e}\mathbin{\fatsemi^{\sharp}}{\textsf{\textup{break}}^{\sharp}})}\mid P\in\mathcal{P}\}
\label{eq:Post:abstract:break}\\[0.75ex]
{\textsf{\textup{Post}}^\sharp\sqb{\texttt{S}_1\texttt{;}\texttt{S}_2}^{\sharp}\mathcal{P}}
&=&
{\textsf{\textup{Post}}^\sharp\sqb{\texttt{S$_2$}}^{\sharp}(\textsf{\textup{Post}}^\sharp\sqb{\texttt{S$_1$}}^{\sharp}\mathcal{P})}
\label{eq:Post:abstract:seq}\\[0.75ex]
{\textsf{\textup{Post}}^\sharp\sqb{\texttt{if(B) S$_1$ else S$_2$}}^{\sharp}\mathcal{P}}
&=&
{(\textsf{\textup{Post}}^\sharp\sqb{\texttt{B;S}_1}^{\sharp}
\mathbin{\,\dot{\sqcup}^{\sharp}\,}
\textsf{\textup{Post}}^\sharp\sqb{\neg\texttt{B;S}_2}^{\sharp})\mathcal{P}}
\label{eq:Post:abstract:if}\\[0.75ex]
{\breve{\vec{F}}_{pe}^{\sharp}}&\triangleq&\LAMBDA{P}\LAMBDA{X}{\textsf{\textup{Post}}^\sharp({\textsf{\textup{init}}^{\sharp}})\{P\} \mathbin{\breve{\sqcup}_{+}^{\sharp}} \textsf{\textup{Post}}^\sharp(\sqb{\texttt{B;S}}_{e}^{\sharp})(X)}
\label{eq:def:F-pe-sharp-breve}
\\
{\breve{F}_{p\bot}^{\sharp}}&\triangleq&
\LAMBDA{X}{\bigcup\{\textsf{\textup{Post}}^\sharp(S)(\sqb{\texttt{B;S}}_{e}^{\sharp})\mid S\in X\}}
\label{eq:def:F-p-bot-sharp-breve}
\\
{\textsf{\textup{Post}}^\sharp\sqb{\texttt{while(B) S}}^{\sharp}\mathcal{P}}
&=&
\{\triple{e:Q_e}{\bot:Q_{\bot\ell}\mathbin{{\sqcup}_{\infty}^{\sharp}}{}
Q_{\bot b}}
{br:P_{br}}\mid {}\label{eq:Post:abstract:while}\\
&&\quad Q_e\in\textsf{\textup{Post}}^\sharp(\sqb{\neg\texttt{B}}_{e}^{\sharp}\mathbin{\sqcup_{e}^{\sharp}}\sqb{\texttt{B;S}}_{b}^{\sharp})(\Lfp{\breve{\sqsubseteq}_{+}^{\sharp}}{\breve{\vec{F}}_{pe}^{\sharp}}(P))\wedge{}\nonumber\\
&&\qquad Q_{\bot\ell}\in \textsf{\textup{Post}}^\sharp(\sqb{\texttt{B;S}}_{\bot}^{\sharp})(\Lfp{\breve{\sqsubseteq}_{+}^{\sharp}}({\breve{\vec{F}}_{pe}^{\sharp}}(P)))\wedge{}\nonumber\\
&&\qquad\quad
\exists Q_{\bot b}\mathrel{.}
Q_{\bot b}\in\textsf{\textup{Post}}^\sharp(Q_{p\bot})\{P\}\wedge Q_{p\bot}\in\Gfp{{\breve{\sqsubseteq}}_{\infty}^{\sharp}}{\breve{F}_{p\bot}^{\sharp}}\wedge
P\in\mathcal{P}\}
\nonumber
\end{eqntabular}
\textup{(}where $S_1
\mathbin{\dot{\msqcup{x}}}
S_2$ is defined in \textup{(\ref{eq:Post:abstract:dot-sqcup-sharp}))} is sound and complete.
\end{theorem}
\begin{toappendix}
\begin{proof}[Proof of theorem \ref{th:upper-abstract-hyper-properties}]
\begin{calculus}[=\ \ ]We need two preliminary results.\\
\hyphen{5}\formula{{\breve{\vec{F}}_{pe}^{\sharp}}(P)\{X\}}\\
=
\formulaexplanation{{\textsf{\textup{Post}}^\sharp({\textsf{\textup{init}}^{\sharp}})\{P\} \mathbin{\breve{\sqcup}_{+}^{\sharp}} \textsf{\textup{Post}}^\sharp(\sqb{\texttt{B;S}}_{e}^{\sharp})\{X\}}}{(\ref{eq:def:F-pe-sharp-breve})}\\
=
\formulaexplanation{{\{\textsf{post}^\sharp({\textsf{\textup{init}}^{\sharp}})P'\mid P'\in\{P\}}\}\mathbin{\breve{\sqcup}_{+}^{\sharp}} \{\textsf{\textup{post}}^\sharp(\sqb{\texttt{B;S}}_{e}^{\sharp})X\}}{(\ref{eq:def:Post}) and (\ref{eq:Post::post})}\\
=
\formulaexplanation{\{\textsf{post}^\sharp({\textsf{\textup{init}}^{\sharp}})P\}\mathbin{\breve{\sqcup}_{+}^{\sharp}} \{\textsf{\textup{post}}^\sharp(\sqb{\texttt{B;S}}_{e}^{\sharp})X\}}{def.\ $\in$}\\
=
\formulaexplanation{\{\textsf{post}^\sharp({\textsf{\textup{init}}^{\sharp}})P\mathbin{{\sqcup}_{+}^{\sharp}}\textsf{\textup{post}}^\sharp(\sqb{\texttt{B;S}}_{e}^{\sharp})X\}}{def.\ $\mathbin{\breve{\sqcup}_{+}^{\sharp}}$ in proposition \ref{prop:prop:singletonization}}\\
=
\numberedformulaexplanation{\{{\vec{F}_{pe}^{\sharp}}(P)X\}}{(\ref{eq:def:F-pe-sharp})}
\label{eq:def:F-pe-sharp-breve-commutation}\\[1ex]
\hyphen{5}\formula{{\breve{\vec{F}}_{p\bot}^{\sharp}}(\{X\})}\\
=
\formulaexplanation{\bigcup\{\textsf{\textup{Post}}^\sharp(S)(\sqb{\texttt{B;S}}_{e}^{\sharp})\mid S\in \{X\}\}}{def.\ (\ref{eq:def:F-p-bot-sharp-breve}) of ${\breve{\vec{F}}_{p\bot}^{\sharp}}$}\\
=
\formulaexplanation{\bigcup\{\textsf{\textup{Post}}^\sharp(X)(\sqb{\texttt{B;S}}_{e}^{\sharp})\}}{def.\ $\in$}\\
=
\formulaexplanation{\bigcup\{\{\textsf{\textup{post}}^\sharp(X)(\sqb{\texttt{B;S}}_{e}^{\sharp})\}\}}{(\ref{eq:Post::post})}\\
=
\formulaexplanation{\{\textsf{\textup{post}}^\sharp (X)
(\sqb{\texttt{B;S}}_{e}^{\sharp})\}}{def.\ $\bigcup$}\\
=
\numberedformulaexplanation{\{{\vec{F}_{p\bot}^{\sharp}}(X)\}}{(\ref{eq:def:F-p-bot-sharp})}\label{eq:def:F-p-bot-sharp-breve-commutation}
\end{calculus}

\smallskip

\noindent   The proof is by structural induction on the statement syntax.\par
\begin{calculus}[=\ \ ]
\hyphen{5}\formula{\textsf{Post}^\sharp\sqb{\texttt{x = A}}^{\sharp}\mathcal{P}}\\
=
\formulaexplanation{\{\textsf{post}^\sharp\sqb{\texttt{x = A}}^{\sharp}P\mid P\in\mathcal{P}\}}{def.\ (\ref{eq:def:Post}) of $\textsf{Post}^\sharp$}\\
=
\formulaexplanation{\{P\mathbin{\fatsemi^{\sharp}}\sqb{\texttt{x = A}}^\sharp\mid P\in\mathcal{P}\}}{def.\ (\ref{eq:def:Post}) of $\textsf{post}^\sharp$}\\
=
\formulaexplanation{\{P\mathbin{\fatsemi^{\sharp}}\triple{e:{\textsf{assign}^{\sharp}\sqb{\texttt{x},\texttt{A}}}}{\bot:{\bot_{\infty}^{\sharp}}}{br:P_{br}}\mid P\in\mathcal{P}\}}{(\ref{eq:def:Abstract-Semantic-Domain-Semantics}) and (\ref{eq:def:sem:abstract:basis})}\\
=
\formulaexplanation{\{\triple{e:P_{+}\mathbin{\fatsemi^{\sharp}}{\textsf{assign}^{\sharp}\sqb{\texttt{x},\texttt{A}}}}{\bot:P_{\infty}\mathbin{\fatsemi^{\sharp}}{\bot_{\infty}^{\sharp}}}{br:P_{br}}\mid P\in\mathcal{P}\}}
{def.\ (\ref{eq:def:abstract:sem:group:seq}) of $\mathbin{\fatsemi^{\sharp}}$}\\
=
\formula{\{\triple{e:P_{+}\mathbin{\fatsemi^{\sharp}}{\textsf{assign}^{\sharp}\sqb{\texttt{x},\texttt{A}}}}{\bot:P_{\infty}\mathbin{\fatsemi^{\sharp}}{\bot_{\infty}^{\sharp}}}{br:P_{br}}\mid P\in\mathcal{P}\}}\\[-0.5ex]
\rightexplanation{${\bot_{+}^{\sharp}}$ absorbent by definition \ref{def:abstract:domain:well:def:init:neutral}}\\
=
\formulaexplanation{\{\triple{e:P_{+}\mathbin{\fatsemi^{\sharp}}{\textsf{assign}^{\sharp}\sqb{\texttt{x},\texttt{A}}}}{\bot:P_{\infty}}{br:P_{br}}\mid P\in\mathcal{P}\}}{$P_{\infty}$ absorbent by definition \ref{def:abstract:domain:well:def:oo:absorbent}}\\[1ex]

\hyphen{5}\discussion{The $\textsf{Post}^\sharp$ characterizations (\ref{eq:Post:abstract:random:assignment}) for \texttt{x = [$a$, $b$]}, (\ref{eq:Post:abstract:skip}) for \texttt{x = skip}, and (\ref{eq:Post:abstract:B}) for \texttt{B} are similar.}\\[1em]

\hyphen{5}\formula{\textsf{Post}^\sharp\sqb{\texttt{break}}^{\sharp}\mathcal{P}}\\
=
\formulaexplanation{\{\textsf{post}^\sharp\sqb{\texttt{break}}^{\sharp}P\mid P\in\mathcal{P}\}}{def.\ (\ref{eq:def:Post}) of $\textsf{Post}^\sharp$}\\
=
\formulaexplanation{\{P\mathbin{\fatsemi^{\sharp}}\sqb{\texttt{break}}^\sharp\mid P\in\mathcal{P}\}}{def.\ (\ref{eq:def:Post}) of $\textsf{post}^\sharp$}\\
=
\formulaexplanation{\{P\mathbin{\fatsemi^{\sharp}}
\triple{e:{\bot_{+}^{\sharp}}}{\bot:{\bot_{\infty}^{\sharp}}}{br:{\textsf{break}^{\sharp}}}\mid P\in\mathcal{P}\}}{(\ref{eq:def:Abstract-Semantic-Domain-Semantics}) and (\ref{eq:def:sem:abstract:basis})}\\
=
\formulaexplanation{\{
\triple{e:P_{+}\mathbin{\fatsemi^{\sharp}}{\bot_{+}^{\sharp}}}{\bot:P_{\infty}\mathbin{\fatsemi^{\sharp}}{\bot_{\infty}^{\sharp}}}{br:P_{br}\mathbin{{\sqcup}^{\sharp}_{+}}(P_{e}\mathbin{\fatsemi^{\sharp}}{\textsf{\textup{break}}^{\sharp}})}\mid P\in\mathcal{P}\}}
{def.\ (\ref{eq:def:abstract:sem:group:seq}) of $\mathbin{\fatsemi^{\sharp}}$}\\
=
\formulaexplanation{\{
\triple{e:{\bot_{+}^{\sharp}}}{\bot:P_{\infty}}{br:P_{br}\mathbin{{\sqcup}^{\sharp}_{+}}(P_{e}\mathbin{\fatsemi^{\sharp}}{\textsf{\textup{break}}^{\sharp}})}\mid P\in\mathcal{P}\}}{definitions \ref{def:abstract:domain:well:def:oo:absorbent} and  \ref{def:abstract:domain:well:def:init:neutral}}\\[1ex]

\hyphen{5}\formula{\textsf{Post}^\sharp\sqb{\texttt{S}_1\texttt{;}\texttt{S}_2}^{\sharp}\mathcal{P}}\\
=
\formulaexplanation{\{\textsf{post}^\sharp\sqb{\texttt{S}_1\texttt{;}\texttt{S}_2}^{\sharp}P\mid P\in\mathcal{P}\}}{def.\ (\ref{eq:def:Post}) of $\textsf{Post}^\sharp$}\\
=
\formulaexplanation{\{P\mathbin{\fatsemi^{\sharp}}(\sqb{\texttt{S}_1\texttt{;}\texttt{S}_2}^\sharp)\mid P\in\mathcal{P}\}}{def.\ (\ref{eq:def:Post}) of $\textsf{post}^\sharp$}\\
=
\formulaexplanation{\{P\mathbin{\fatsemi^{\sharp}}(\sqb{\texttt{S}_1}^\sharp\mathbin{\fatsemi^{\sharp}}\sqb{\texttt{S}_2}^\sharp)\mid P\in\mathcal{P}\}}{def.\ (\ref{eq:def:abstract:sem:group:seq}) of $\mathbin{\fatsemi^{\sharp}}$}\\
=
\formulaexplanation{\{(P\mathbin{\fatsemi^{\sharp}}\sqb{\texttt{S}_1}^\sharp)\mathbin{\fatsemi^{\sharp}}\sqb{\texttt{S}_2}^\sharp\mid P\in\mathcal{P}\}}{$\mathbin{\fatsemi^{\sharp}}$ associative by definition \ref{def:abstract:domain:well:def}\ref{def:abstract:domain:well:def:operators}}\\
=
\formulaexplanation{\{\textsf{post}^\sharp\sqb{\texttt{S}_2}^{\sharp}(P\mathbin{\fatsemi^{\sharp}}\sqb{\texttt{S}_1}^\sharp)\mid P\in\mathcal{P}\}}{def.\ (\ref{eq:def:Post}) of $\textsf{post}^\sharp\sqb{\texttt{S}_2}^{\sharp}Q\triangleq Q\mathbin{\fatsemi^{\sharp}}\sqb{\texttt{S}_2}^\sharp$}\\
=
\formulaexplanation{\{\textsf{post}^\sharp\sqb{\texttt{S}_2}^{\sharp}(\textsf{post}^\sharp\sqb{\texttt{S}_1}^{\sharp}P)\mid P\in\mathcal{P}\}}{def.\ (\ref{eq:def:Post}) of $\textsf{post}^\sharp\sqb{\texttt{S}_1}^{\sharp}P\triangleq P\mathbin{\fatsemi^{\sharp}}\sqb{\texttt{S}_1}^\sharp$}\\
=
\formulaexplanation{\{\textsf{post}^\sharp\sqb{\texttt{S}_2}^{\sharp}Q\mid Q\in\{\textsf{post}^\sharp\sqb{\texttt{S}_1}^{\sharp}P\mid P\in\mathcal{P}\}\}}{def.\ $\in$}\\
=
\formulaexplanation{\textsf{Post}^\sharp\sqb{\texttt{S$_2$}}^{\sharp}(\textsf{Post}^\sharp\sqb{\texttt{S$_1$}}^{\sharp}\mathcal{P})}{def.\ (\ref{eq:def:Post}) of $\textsf{Post}^\sharp$}\\[1ex]

\hyphen{5}\formula{\textsf{Post}^\sharp\sqb{\texttt{if (B) S$_1$ else S$_2$}}^{\sharp}\mathcal{P}}\\
=
\formulaexplanation{\{Q_1\mathbin{\,\sqcup^{\sharp}\,}Q_2\mid  
Q_1\in\textsf{Post}^\sharp\sqb{\texttt{B;S}_1}^{\sharp}\{P\}
\wedge
Q_2\in\textsf{Post}^\sharp\sqb{\neg\texttt{B;S}_2}^{\sharp}\{P\} 
\wedge
P\in\mathcal{P}\}}{as shown above}\\
=
\formulaexplanation{(\textsf{Post}^\sharp\sqb{\texttt{B;S}_1}^{\sharp}
\mathbin{\,\dot{\sqcup}^{\sharp}\,}
\textsf{Post}^\sharp\sqb{\neg\texttt{B;S}_2}^{\sharp})\mathcal{P}}{by def.\ (\ref{eq:Post:abstract:dot-sqcup-sharp}) of $\dot{\sqcup}^{\sharp}$}\\[1ex]

\hyphen{5}\formula{\textsf{Post}^\sharp\sqb{\texttt{while (B) S}}^{\sharp}\mathcal{P}}\\
=
\formulaexplanation{\{\textsf{post}^\sharp\sqb{\texttt{while (B) S}}^{\sharp}P\mid P\in\mathcal{P}\}}{def.\ (\ref{eq:def:Post}) of $\textsf{Post}^\sharp$}\\
=
\formulaexplanation{\{ \langle ok:\langle{e:{\textsf{post}^\sharp(\sqb{\neg\texttt{B}}_{e}^{\sharp}\mathbin{\sqcup_{e}^{\sharp}}\sqb{\texttt{B;S}}_{b}^{\sharp})(\Lfp{{\sqsubseteq}_{+}^{\sharp}}({\vec{F}_{pe}^{\sharp}}(P)))}},\,\bot:{\textsf{post}^\sharp(\sqb{\texttt{B;S}}_{\bot}^{\sharp})(\Lfp{{\sqsubseteq}_{+}^{\sharp}}({\vec{F}_{pe}^{\sharp}}(P)))}
\mathbin{{\sqcup}_{\infty}^{\sharp}}{}
\textsf{\textup{post}}^\sharp(\Gfp{{\sqsubseteq}_{\infty}^{\sharp}}{F_{p\bot}^{\sharp}})P\rangle,
{br:P_{br}}\rangle\mid P\in\mathcal{P}\}}{(\ref{eq:post:abstract:while})}\\
=
\formula{\{\triple{e:Q_e}{\bot:Q_{\bot\ell}\mathbin{{\sqcup}_{\infty}^{\sharp}}{}
Q_{\bot b}}
{br:P_{br}}\mid 
Q_e\in\{{\textsf{post}^\sharp(\sqb{\neg\texttt{B}}_{e}^{\sharp}\mathbin{\sqcup_{e}^{\sharp}}\sqb{\texttt{B;S}}_{b}^{\sharp})(\Lfp{{\sqsubseteq}_{+}^{\sharp}}({\vec{F}_{pe}^{\sharp}}(P)))}\}\wedge
Q_{\bot\ell}\in\{\textsf{post}^\sharp(\sqb{\texttt{B;S}}_{\bot}^{\sharp})(\Lfp{{\sqsubseteq}_{+}^{\sharp}}({\vec{F}_{pe}^{\sharp}}(P)))\}\wedge\exists Q_{p\bot}\mathrel{.}
Q_{\bot b}\in\{\textsf{\textup{post}}^\sharp(Q_{p\bot})P\}\wedge Q_{p\bot}\in\{\Gfp{{\sqsubseteq}_{\infty}^{\sharp}}{F_{p\bot}^{\sharp}}\}\wedge
P\in\mathcal{P}\}}\\
\rightexplanation{def.\ singleton and $\in$}\\[1ex]
=
\formula{\{\triple{e:Q_e}{\bot:Q_{\bot\ell}\mathbin{{\sqcup}_{\infty}^{\sharp}}{}
Q_{\bot b}}
{br:P_{br}}\mid 
Q_e\in{\textsf{Post}^\sharp(\sqb{\neg\texttt{B}}_{e}^{\sharp}\mathbin{\sqcup_{e}^{\sharp}}\sqb{\texttt{B;S}}_{b}^{\sharp})\{\Lfp{{\sqsubseteq}_{+}^{\sharp}}({\vec{F}_{pe}^{\sharp}}(P)})\}\wedge
Q_{\bot\ell}\in \textsf{Post}^\sharp(\sqb{\texttt{B;S}}_{\bot}^{\sharp})\{\Lfp{{\sqsubseteq}_{+}^{\sharp}}({\vec{F}_{pe}^{\sharp}}(P))\}\wedge\exists Q_{\bot b}\mathrel{.}
Q_{\bot b}\in\textsf{Post}^\sharp(Q_{p\bot})\{P\}\wedge Q_{p\bot}\in\{\Gfp{{\sqsubseteq}_{\infty}^{\sharp}}{F_{p\bot}^{\sharp}}\}\wedge
P\in\mathcal{P}\}}\\[-0.5ex]\rightexplanation{(\ref{eq:Post::post})}\\[1ex]
=
\formula{\{\triple{e:Q_e}{\bot:Q_{\bot\ell}\mathbin{{\sqcup}_{\infty}^{\sharp}}{}
Q_{\bot b}}
{br:P_{br}}\mid 
Q_e\in\textsf{Post}^\sharp(\sqb{\neg\texttt{B}}_{e}^{\sharp}\mathbin{\sqcup_{e}^{\sharp}}\sqb{\texttt{B;S}}_{b}^{\sharp})(\Lfp{\breve{\sqsubseteq}_{+}^{\sharp}}{\breve{\vec{F}}_{pe}^{\sharp}}(P))\wedge
Q_{\bot\ell}\in \textsf{Post}^\sharp(\sqb{\texttt{B;S}}_{\bot}^{\sharp})(\Lfp{\breve{\sqsubseteq}_{+}^{\sharp}}{\breve{\vec{F}}_{pe}^{\sharp}}(P))\wedge\exists Q_{\bot b}\mathrel{.}
Q_{\bot b}\in\textsf{Post}^\sharp(Q_{p\bot})\{P\}\wedge Q_{p\bot}\in\{\Gfp{{\sqsubseteq}_{\infty}^{\sharp}}{F_{p\bot}^{\sharp}}\}\wedge
P\in\mathcal{P}\}}\\\rightexplanation{since $\{\Lfp{{\sqsubseteq}_{+}^{\sharp}}{\vec{F}_{pe}^{\sharp}}(P)\}
=
\Lfp{\breve{\sqsubseteq}_{+}^{\sharp}}{\breve{\vec{F}}_{pe}^{\sharp}}(P)
$ by (\ref{eq:def:F-pe-sharp-breve-commutation}) and proposition \ref{prop:prop:singletonization}}\\[1ex]
=
\formula{\{\triple{e:Q_e}{\bot:Q_{\bot\ell}\mathbin{{\sqcup}_{\infty}^{\sharp}}{}
Q_{\bot b}}
{br:P_{br}}\mid 
Q_e\in\textsf{Post}^\sharp(\sqb{\neg\texttt{B}}_{e}^{\sharp}\mathbin{\sqcup_{e}^{\sharp}}\sqb{\texttt{B;S}}_{b}^{\sharp})(\Lfp{\breve{\sqsubseteq}_{+}^{\sharp}}{\breve{\vec{F}}_{pe}^{\sharp}}(P))\wedge
Q_{\bot\ell}\in \textsf{Post}^\sharp(\sqb{\texttt{B;S}}_{\bot}^{\sharp})(\Lfp{\breve{\sqsubseteq}_{+}^{\sharp}}{\breve{\vec{F}}_{pe}^{\sharp}}(P))\wedge\exists Q_{\bot b}\mathrel{.}
Q_{\bot b}\in\textsf{Post}^\sharp(Q_{p\bot})\{P\}\wedge Q_{p\bot}\in\Gfp{{\breve{\sqsubseteq}}_{\infty}^{\sharp}}{\breve{F}_{p\bot}^{\sharp}}\wedge
P\in\mathcal{P}\}}\\
\lastrightexplanation{since $\{\Gfp{{\sqsubseteq}_{\infty}^{\sharp}}{F_{p\bot}^{\sharp}}\}=
\Gfp{{\breve{\sqsubseteq}}_{\infty}^{\sharp}}{\breve{F}_{p\bot}^{\sharp}}$ by (\ref{eq:def:F-p-bot-sharp-breve-commutation}) and proposition \ref{prop:prop:singletonization}}{\mbox{\qed}}
\end{calculus}
\let\qed\relax
\end{proof}
\end{toappendix}
\begin{example}[Finitary powerset calculational domain]\label{ex:powerset-deterministic-domain-post-continued}Continuing example \ref{ex:powerset-deterministic-domain-post} ignoring \texttt{break}s and nontermination, the hypercollecting semantics of \cite[p.\ 877]{DBLP:conf/popl/AssafNSTT17} is\par
\begin{eqntabular}{rl}
&{\textsf{\textup{Post}}^\sharp(\sqb{\neg\texttt{B}}_{e}^{\sharp})(\Lfp{\subseteq}\LAMBDA{X}{\mathcal{P} \cup \textsf{\textup{Post}}^\sharp(\sqb{\texttt{\texttt{if (B) S else skip}}}_{e}^{\sharp})(X)})}\label{eq:AssafNSTT17-unsound}\\
=&
{\{\textsf{\textup{Post}}^\sharp(\sqb{\neg\texttt{B}}_{e}^{\sharp})(\textsf{\textup{Post}}^\sharp(\sqb{\texttt{if (B) S else skip}}_{e}^{\sharp})^n\mathcal{P})\mid n\in\mathbb{N}\}}\nonumber\\
=&
{\{\textsf{\textup{Post}}^\sharp(\sqb{\neg\texttt{B}}_{e}^{\sharp})(\textsf{\textup{Post}}^\sharp(\sqb{\texttt{if (B) S else skip}}_{e}^{\sharp})^n\{P\}) \mid n\in\mathbb{N}\wedge P\in\mathcal{P}\}}\nonumber\\
\neq&
{\bigcup\{\textsf{\textup{Post}}^\sharp(\sqb{\neg\texttt{B}}_{e}^{\sharp})(\Lfp{{\subseteq}}{\breve{\vec{F}}_{pe}^{\sharp}}(P))\mid P\in\mathcal{P}\}}\nonumber
\end{eqntabular}
\noindent\ustrut By remark \ref{rem:Post-non-preservation}, this is different from (\ref{eq:Post:abstract:while}) (even when ignoring nontermination and \texttt{break}s) so that \cite[p.\ 877]{DBLP:conf/popl/AssafNSTT17} is incomplete and cannot be used as a hypercollecting semantics for general hyperproperties, as further discussed in sect.\@ \ref{sec:LogicRuleChainLimitOrderIdealAbstractSemanticProperties}. Moreover (\ref{eq:AssafNSTT17-unsound}) is unsound, invalidating  \cite[th.\@ 1]{DBLP:conf/popl/AssafNSTT17}. This will be fixed by the weak hypercollecting semantics defined in (\ref{eq:def:overline:Post:while}).
\end{example}

\section{Abstract Logic of Semantic (Hyper) Properties}\label{sec:Abstract-Logic-Semantic-Properties}
\subsection{Definition of the Upper and Lower Abstract Logics}
The upper (respectively lower) logic $\overline{\textsf{L}}^\sharp$ (resp. $\underline{\textsf{L}}^\sharp$) maps the semantics $S$ of a statement into a pair of a precondition and postcondition that is $\overline{\textsf{L}}^\sharp,\underline{\textsf{L}}^\sharp\in {{\mathbb{L}^{\sharp}}\functionto(\wp({\mathbb{L}^{\sharp}})\times\wp({\mathbb{L}^{\sharp}}))}$ ordered pointwise by $\dot{\subseteq}$ (the larger the precondition, the larger is the postcondition). We have
\begin{eqntabular}{rcl}
\overline{\textsf{L}}^\sharp(S)&\triangleq&\{\pair{\mathcal{P}}{\mathcal{Q}}\mid\textsf{Post}^\sharp(S)\mathcal{P}\subseteq\mathcal{Q}\}
\label{eq:def:upper-logic}
\end{eqntabular}
where $\pair{\mathcal{P}}{\mathcal{Q}}\in\overline{\textsf{L}}^\sharp\sqb{\texttt{S}}^{\sharp}$ is traditionally written $\overline{\llbrace}\,\mathcal{P}\,\overline{\rrbrace}\,\texttt{S}\,\overline{\llbrace}\,\mathcal{Q}\,\overline{\rrbrace}$. The $\subseteq$-dual holds for the lower abstract logic.
As was the case in sect.\@ \ref{sec:Algebraic-Logics-Program-Execution-Properties} for execution properties, this is an abstraction ${\maccent{\alpha}{\filledtriangleup}}(\textsf{P})\triangleq\LAMBDA{S}\{\pair{\mathcal{P}}{\mathcal{Q}}\mid \textsf{P}(S)\mathcal{P}\subseteq\mathcal{Q}\}$
\bgroup\belowdisplayskip0pt\begin{eqntabular}{c}
\pair{{\mathbb{L}^{\sharp}}\functionto\wp({\mathbb{L}^{\sharp}})\increasingfunctionto\wp({\mathbb{L}^{\sharp}})}{\ddot{\subseteq}
}\galoiS{\ustrut{\maccent{\alpha}{\filledtriangleup}}}{\lstrut{\maccent{\gamma}{\filledtriangleup}}}\pair{{\mathbb{L}^{\sharp}}\functionto(\wp({\mathbb{L}^{\sharp}})\times\wp({\mathbb{L}^{\sharp}}))}{\dot{\subseteq}}
\end{eqntabular}\egroup
where $\overline{\textsf{L}}^\sharp(S)={\maccent{\alpha}{\filledtriangleup}}(\textsf{Post}^\sharp)S$.

Defining the upper and lower logic triples
\bgroup\arraycolsep0.65\arraycolsep
\begin{eqntabular}[fl]{rclclcl}
\overline{\llbrace}\,\mathcal{P}\,\overline{\rrbrace}\,\texttt{S}\,\overline{\llbrace}\,\mathcal{Q}\,\overline{\rrbrace}
&\triangleq&
\pair{\mathcal{P}}{\mathcal{Q}}\in\overline{\textsf{L}}^\sharp\sqb{\texttt{S}}^{\sharp}
&=&
\textsf{Post}^\sharp\sqb{\texttt{S}}^\sharp\mathcal{P}\subseteq \mathcal{Q}
&=&
\forall  P\in\mathcal{P}\!\mathrel{.}\textsf{post}^\sharp\sqb{\texttt{S}}^\sharp P\in\mathcal{Q}
\label{eq:def:abstract:logical:triples}\\
\underline{\llbrace}\,\mathcal{P}\,\underline{\rrbrace}\,\texttt{S}\,\underline{\llbrace}\,\mathcal{Q}\,\underline{\rrbrace}
&\triangleq&
\pair{\mathcal{P}}{\mathcal{Q}}\in\underline{\textsf{L}}^\sharp\sqb{\texttt{S}}^{\sharp}
&=&
\mathcal{Q}\subseteq\textsf{Post}^\sharp\sqb{\texttt{S}}^\sharp\mathcal{P}
&=&
\rlap{$\forall Q\in\mathcal{Q}\mathrel{.}\exists P\in\mathcal{P}\!\mathrel{.}\textsf{post}^\sharp\sqb{\texttt{S}}^\sharp P=Q$}\nonumber
\end{eqntabular}\egroup
(where for symmetry, we can write $\overline{\llbrace}\,\mathcal{P}\,\overline{\rrbrace}\,\texttt{S}\,\overline{\llbrace}\,\mathcal{Q}\,\overline{\rrbrace}\triangleq\forall  P\in\mathcal{P}\!\mathrel{.}\exists Q\in\mathcal{Q}\mathrel{.}
\textsf{post}^\sharp(S)P=Q$.)
We get generalizations of Hoare logic \cite{DBLP:journals/cacm/Hoare69} and incorrectness logic \cite{DBLP:conf/sefm/VriesK11,DBLP:journals/pacmpl/OHearn20} from execution to semantic properties. 
\begin{example}[Finitary powerset nondeterministic calculational domain]\label{ex:powerset-nondeterministic-domain-post}In \cite{DBLP:journals/afp/Dardinier23a,DBLP:conf/pldi/DardinierM24}, the relational semantics is identical to that of \cite{DBLP:conf/popl/AssafNSTT17} in example \ref{ex:powerset-deterministic-domain-post} but for a nondeterministic language. Nontermination is abstracted away. The extended semantics \cite[Definition 4]{DBLP:journals/afp/Dardinier23a,DBLP:conf/pldi/DardinierM24} is $\textsf{post}^\sharp(S)P$ $=$ $\{\pair{\sigma}{\sigma''}\mid\exists\sigma'\in\Sigma\mathrel{.}\pair{\sigma}{\sigma'}\in P\wedge\pair{\sigma'}{\sigma''}\in S\}$, the same as in example \ref{ex:powerset-deterministic-domain-post}. Hyper-triples $\overline{\llbrace}\,\mathcal{P}\,\overline{\rrbrace}\,\texttt{S}\,\overline{\llbrace}\,\mathcal{Q}\,\overline{\rrbrace}$ are defined in \cite[Definition 5]{DBLP:journals/afp/Dardinier23a,DBLP:conf/pldi/DardinierM24} to be the powerset instance of (\ref{eq:def:abstract:logical:triples}), the same instance used in example \ref{ex:powerset-deterministic-domain-post}.
\end{example}

The upper and lower abstract logics can always be expressed in terms of singleton (although the equivalent formula is not part of the logic).
\begin{lemma}\label{lem:join:preserving:logics}\abovedisplayskip0pt\proofinapx
\abovedisplayskip-1em\begin{eqntabular}{@{\qquad}rcl}
{\overline{\llbrace}\,\mathcal{P}\,\overline{\rrbrace}\,\texttt{S}\,\overline{\llbrace}\,\mathcal{Q}\,\overline{\rrbrace}}
&\Leftrightarrow&
\forall P\in\mathcal{P}\mathrel{.}\exists Q\in\mathcal{Q}\mathrel{.} {\overline{\llbrace}\,\{P\}\,\overline{\rrbrace}\,\texttt{S}\,\overline{\llbrace}\,\{Q\}\,\overline{\rrbrace}}\renumber{\textup{(a)}}\\
{\underline{\llbrace}\,\mathcal{P}\,\underline{\rrbrace}\,\texttt{S}\,\underline{\llbrace}\,\mathcal{Q}\,\underline{\rrbrace}}
&\Leftrightarrow&
\forall Q\in\mathcal{Q}\mathrel{.}\exists  P\in\mathcal{P}\mathrel{.}{\underline{\llbrace}\,\{P\}\,\underline{\rrbrace}\,\texttt{S}\,\underline{\llbrace}\,\{Q\}\,\underline{\rrbrace}}\renumber{\textup{(b)}}
\end{eqntabular}
\end{lemma}
\begin{toappendix}
\begin{proof}[Proof of lemma \ref{lem:join:preserving:logics}]
\begin{calculus}[=\ \ ]
\hyphen{5}\formula{\overline{\llbrace}\,\mathcal{P}\,\overline{\rrbrace}\,\texttt{S}\,\overline{\llbrace}\,\mathcal{Q}\,\overline{\rrbrace}}\\
=
\formulaexplanation{\textsf{Post}^\sharp\sqb{\texttt{S}}^\sharp\mathcal{P}\subseteq \mathcal{Q}}{def.\ (\ref{eq:def:abstract:logical:triples}) of the logic triples}\\
=
\formulaexplanation{\{\textsf{post}^\sharp(S)P\mid P\in\mathcal{P}\}\subseteq \mathcal{Q}}{def.\ (\ref{eq:def:Post}) of $\textsf{Post}^\sharp$}\\
=
\formulaexplanation{\forall P\in\mathcal{P}\mathrel{.}\textsf{post}^\sharp(S)P\in\mathcal{Q}}{def.\ $\subseteq$}\\
=
\formulaexplanation{\forall P\in\mathcal{P}\mathrel{.}\exists Q\in\mathcal{Q}\mathrel{.} \textsf{post}^\sharp(S)P=Q}{def.\ $\exists$}\\
=
\formulaexplanation{\forall P\in\mathcal{P}\mathrel{.}\exists Q\in\mathcal{Q}\mathrel{.} \{\textsf{post}^\sharp(S)P\}\subseteq \{Q\}}{def.\ $\subseteq$}\\
=
\formulaexplanation{\forall P\in\mathcal{P}\mathrel{.}\exists Q\in\mathcal{Q}\mathrel{.} \{\textsf{post}^\sharp(S)P'\mid P'\in\{P\}\}\subseteq \{Q\}}{def.\ $\in$}\\
=
\formulaexplanation{\forall P\in\mathcal{P}\mathrel{.}\exists Q\in\mathcal{Q}\mathrel{.} \textsf{Post}^\sharp\sqb{\texttt{S}}^\sharp\{P\}\subseteq \{Q\}}{def.\ (\ref{eq:def:Post}) of $\textsf{Post}^\sharp$}\\
=
\formulaexplanation{\forall P\in\mathcal{P}\mathrel{.}\exists Q\in\mathcal{Q}\mathrel{.} {\overline{\llbrace}\,\{P\}\,\overline{\rrbrace}\,\texttt{S}\,\overline{\llbrace}\,\{Q\}\,\overline{\rrbrace}}}{def.\ (\ref{eq:def:abstract:logical:triples}) of the logic triples}\\[1ex]
\hyphen{5}\lastdiscussion{(b) is the $\subseteq$-dual of (a).}{\mbox{\qed}}
\end{calculus}
\let\qed\relax
\end{proof}
\end{toappendix}
\begin{corollary}\label{cor:underapproximation:one:element}\proofinapx\quad
$({\exists P\in\mathcal{P}\mathrel{.}\underline{\llbrace}\,\{P\}\,\underline{\rrbrace}\,\texttt{S}\,\underline{\llbrace}\,\{Q\}\,\underline{\rrbrace}})$
$\Leftrightarrow$
${\underline{\llbrace}\,\mathcal{P}\,\underline{\rrbrace}\,\texttt{S}\,\underline{\llbrace}\,\{Q\}\,\underline{\rrbrace}}$.
\end{corollary}
\begin{toappendix}
\begin{proof}[Proof of corollary \ref{cor:underapproximation:one:element}]
\begin{calculus}[$\Leftrightarrow$\ \ ]
\formula{\underline{\llbrace}\,\mathcal{P}\,\underline{\rrbrace}\,\texttt{S}\,\underline{\llbrace}\,\{Q\}\,\underline{\rrbrace}}\\
$\Leftrightarrow$
\formulaexplanation{\forall Q'\in\{Q\}\mathrel{.}\exists  P\in\mathcal{P}\mathrel{.}{\underline{\llbrace}\,\{P\}\,\underline{\rrbrace}\,\texttt{S}\,\underline{\llbrace}\,\{Q'\}\,\underline{\rrbrace}}}{lemma \ref{lem:join:preserving:logics}.b}\\
$\Leftrightarrow$
\lastformulaexplanation{\exists P\in\mathcal{P}\mathrel{.}\underline{\llbrace}\,\{P\}\,\underline{\rrbrace}\,\texttt{S}\,\underline{\llbrace}\,\{Q\}\,\underline{\rrbrace}}{def. $\in$}{\mbox{\qed}}
\end{calculus}
\let\qed\relax
\end{proof}
\end{toappendix}
\noindent For singletons, the two logics are equivalent.
\begin{lemma}\label{singleton:under=over}\proofinapx\quad For all $P,Q\in{\mathbb{L}^{\sharp}}$,
$\overline{\llbrace}\,\{P\}\,\overline{\rrbrace}\,\texttt{S}\,\overline{\llbrace}\,\{Q\}\,\overline{\rrbrace}$
=
$\underline{\llbrace}\,\{P\}\,\underline{\rrbrace}\,\texttt{S}\,\underline{\llbrace}\,\{Q\}\,\underline{\rrbrace}$.
\end{lemma}
\begin{toappendix}
\begin{proof}[Proof of lemma \ref{singleton:under=over}]
\begin{calculus}[=\ \ ]
\formula{\overline{\llbrace}\,\{P\}\,\overline{\rrbrace}\,\texttt{S}\,\overline{\llbrace}\,\{Q\}\,\overline{\rrbrace}}\\
=
\formulaexplanation{\textsf{Post}^\sharp\sqb{\texttt{S}}^\sharp\{P\}\subseteq \{Q\}}{def.\ (\ref{eq:def:abstract:logical:triples}) of logic triples}\\
=
\formulaexplanation{\{\textsf{post}^\sharp(S)P'\mid P'\in\{P\}\}\subseteq \{Q\}}{def.\ (\ref{eq:def:Post}) of $\textsf{Post}^\sharp$}\\
=
\formulaexplanation{\{\textsf{post}^\sharp(S)P\}\subseteq \{Q\}}{def.\ $\in$}\\
=
\formulaexplanation{\textsf{post}^\sharp(S)P=Q}{def.\ $\subseteq$}\\
=
\formulaexplanation{\{Q\}\subseteq\{\textsf{post}^\sharp(S)P\}} {def.\ $\subseteq$}\\
=
\formulaexplanation{\{Q\}\subseteq\{\textsf{post}^\sharp(S)P'\mid P'\in\{P\}\}} {def.\ $\in$}\\
=
\formulaexplanation{\{Q\}\subseteq\textsf{Post}^\sharp\sqb{\texttt{S}}^\sharp\{P\} }{def.\ (\ref{eq:def:Post}) of $\textsf{Post}^\sharp$}\\
=
\lastformulaexplanation{\underline{\llbrace}\,\{P\}\,\underline{\rrbrace}\,\texttt{S}\,\underline{\llbrace}\,\{Q\}\,\underline{\rrbrace}}{def.\ (\ref{eq:def:abstract:logical:triples}) of logic triples}{\mbox{\qed}}
\end{calculus}
\let \qed\relax
\end{proof}
\end{toappendix}

\subsection{The Proof Systems of the Upper and Lower Abstract Logics}\label{sec:Upper-Lower-Abstract-Logics}
Since the definition (\ref{eq:Post:abstract:assignment})---(\ref{eq:Post:abstract:while}) of $\textsf{Post}^\sharp\sqb{\texttt{S}}^{\sharp}$ by a Hilbert proof system is structural, it is the same for the logics. Following \cite{DBLP:journals/pacmpl/Cousot24}, this is obtained by Aczel correspondance between
set-based fixpoints and proof rules \cite{Aczel:1977:inductive-definitions}. For iteration fixpoint, over-approximation is provided by \cite[th.\@ II.3.4]{DBLP:journals/pacmpl/Cousot24} generalizing Park fixpoint induction \cite{Park69-MI5}, whereas under-approximation can be handled by \cite[th.\@ II.3.6]{DBLP:journals/pacmpl/Cousot24} generalizing Scott's induction or \cite[th.\@ II.3.8]{DBLP:journals/pacmpl/Cousot24} generalizing Turing/Floyd variant functions. 

Therefore the sound and complete Hilbert deductive system can be designed calculationally to be the following (where $\mathcal{P},\mathcal{Q}\in\wp({\mathbb{L}^{\sharp}})$, $\bowtie$ and ${{\llbrace}\,\mathcal{P}\,{\rrbrace}\,\texttt{S}\,{\llbrace}\,\mathcal{Q}\,{\rrbrace}}$ are respectively $\subseteq$ and ${\overline{\llbrace}\,\mathcal{P}\,\overline{\rrbrace}\,\texttt{S}\,\overline{\llbrace}\,\mathcal{Q}\,\overline{\rrbrace}}$ for the Upper Abstract Logic and $\supseteq$ and ${\underline{\llbrace}\,\mathcal{P}\,\underline{\rrbrace}\,\texttt{S}\,\underline{\llbrace}\,\mathcal{Q}\,\underline{\rrbrace}}$ for the Lower Abstract Logic and the calculational design proving theorem \ref{th:upper-abstract-logic-proof-system} follows in sect.\@ \ref{sec:ProofSystemUpperAbstractLogic}).
\begin{theorem}[Upper abstract logic proof system]\label{th:upper-abstract-logic-proof-system}
If\/ $\mathbb{D}^{\sharp}$ is a well-defined increasing and decreasing chain-complete join semilattice with right upper continuous sequential composition $\mathbin{{\fatsemi}^{\sharp}}$ then
\bgroup\belowdisplayskip0pt\jot=7pt
\begin{eqntabular}{c}
\frac{\;\{\triple{e:P_{+}\mathbin{\fatsemi^{\sharp}}{\textsf{\textup{assign}}^{\sharp}\sqb{\texttt{x},\texttt{A}}}}{\bot:P_{\infty}}{br:P_{br}}\mid P\in\mathcal{P}\}\bowtie \mathcal{Q}\;}{{\llbrace}\,\mathcal{P}\,{\rrbrace}\,\texttt{x = A}\,{\llbrace}\,\mathcal{Q}\,{\rrbrace}}
\label{eq:logic:abstract:assignment}
\\
\frac{\;\{\triple{e:P_{+}\mathbin{\fatsemi^{\sharp}}{\textsf{\textup{rassign}}^{\sharp}\sqb{\texttt{x},a,b}}}{\bot:P_{\infty}}{br:P_{br}}\mid P\in\mathcal{P}\}\bowtie \mathcal{Q}\;}{{\llbrace}\,\mathcal{P}\,{\rrbrace}\,\texttt{x = [$a$, $b$]}\,{\llbrace}\,Q\,{\rrbrace}}
\label{eq:logic:abstract:random:assignment}
\\
\frac{\;\{\triple{e:P_{+}\mathbin{\fatsemi^{\sharp}}{\textsf{\textup{skip}}^{\sharp}}}{\bot:P_{\infty}}{br:P_{br}}\mid P\in\mathcal{P}\}\bowtie \mathcal{Q}\;}{{\llbrace}\,\mathcal{P}\,{\rrbrace}\,\texttt{skip}\,{\llbrace}\,\mathcal{Q}\,{\rrbrace}}
\label{eq:logic:abstract:skip}
\\
\frac{\;\{\triple{e:P_{+}\mathbin{\fatsemi^{\sharp}}{\textsf{\textup{test}}^{\sharp}}\sqb{\texttt{B}}}{\bot:P_{\infty}}{br:P_{br}}\mid P\in\mathcal{P}\}\bowtie \mathcal{Q}\;}{{\llbrace}\,\mathcal{P}\,{\rrbrace}\,\texttt{B}\,{\llbrace}\,\mathcal{Q}\,{\rrbrace}}
\label{eq:logic:abstract:B}
\\
\frac{\;
{\{
\triple{e:{\bot_{+}^{\sharp}}}{\bot:P_{\infty}}{br:P_{br}\mathbin{{\sqcup}^{\sharp}_{+}}(P_{e}\mathbin{\fatsemi^{\sharp}}{\textsf{\textup{break}}^{\sharp}})}\mid P\in\mathcal{P}\}}
\bowtie \mathcal{Q}\;}{{\llbrace}\,\mathcal{P}\,{\rrbrace}\,\texttt{break}\,{\llbrace}\,\mathcal{Q}\,{\rrbrace}}
\label{eq:logic:abstract:break}
\\
\frac{\:{\llbrace}\,\mathcal{P}\,{\rrbrace}\,\texttt{S}_1\,{\llbrace}\,\mathcal{Q}\,{\rrbrace},\quad{\llbrace}\,\mathcal{Q}\,{\rrbrace}\,\texttt{S}_2\,{\llbrace}\,\mathcal{R}\,{\rrbrace}\:}{{\llbrace}\,\mathcal{P}\,{\rrbrace}\,\texttt{S}_1\texttt{;}\texttt{S}_2\,{\llbrace}\,\mathcal{R}\,{\rrbrace}}\label{eq:logic:abstract:seq}
\\
\frac{\:{\forall P\in\mathcal{P},\quad(\overline{\llbrace}\,\{P\}\,\overline{\rrbrace}\,{\texttt{B;S}_1}\,\overline{\llbrace}\,\{Q_1\}\,\overline{\rrbrace}
\wedge
\overline{\llbrace}\,\{P\}\,\overline{\rrbrace}\,{\neg\texttt{B;S}_2}\,\overline{\llbrace}\,\{Q_2\}\,\overline{\rrbrace})\Rightarrow(Q_1\mathbin{\,\sqcup^{\sharp}\,}Q_2\in \mathcal{Q})
}\:}{\overline{\llbrace}\,\mathcal{P}\,\overline{\rrbrace}\,\texttt{if (B) S}_1\texttt{ else S}_2\,\overline{\llbrace}\,\mathcal{Q}\,\overline{\rrbrace}}\label{eq:logic:abstract:if:upper}
\\
\frac{\:\parbox{0.7\textwidth}{\centering$\bigl(P_e=\Lfp{{\sqsubseteq}_{+}^{\sharp}}{\vec{F}_{pe}^{\sharp}}(P'){}\wedge{}$
$\overline{\llbrace}\,\{P_e\}\,\overline{\rrbrace}\,\neg\texttt{B}\,\overline{\llbrace}\,\{Q_e\}\,\overline{\rrbrace}
{}\wedge{}$ 
$\overline{\llbrace}\,\{P_e\}\,\overline{\rrbrace}\,\texttt{B;S}\,\overline{\llbrace}\,\{Q_b\}\,\overline{\rrbrace}
{}\wedge{}$
$\overline{\llbrace}\,\{P_e\}\,\overline{\rrbrace}\,\texttt{{B;S}}\,\overline{\llbrace}\,\{Q_{\bot\ell}\}\,\overline{\rrbrace}
{}\wedge{}$
$Q_{\bot b}=\Gfp{{\sqsubseteq}_{\infty}^{\sharp}}{F_{p\bot}^{\sharp}}
{}\wedge{}$
$P'\in\mathcal{P}\bigr)$
${}\Rightarrow{}$\\
$\bigl(\triple{e:Q_e\mathbin{\,\sqcup^{\sharp}_{e}\,}Q_b}{\bot:Q_{\bot\ell}\mathbin{{\sqcup}_{\infty}^{\sharp}}{}Q_{\bot b}}{br:P_{br}}\in\mathcal{Q}\bigr)$}\:}{\overline{\llbrace}\,\mathcal{I}\,\overline{\rrbrace}\,\texttt{while (B) S}\,\overline{\llbrace}\,\mathcal{Q}\,\overline{\rrbrace}}\label{eq:logic:abstract:while:upper}
\end{eqntabular}\egroup
is sound and complete.
\end{theorem}
Remarkably in (\ref{eq:logic:abstract:if:upper}) and (\ref{eq:logic:abstract:while:upper}), we have to consider all possible over approximations, and in (\ref{eq:logic:abstract:while:upper}) $P_e$ and $Q_{\bot b}$ must be exact fixpoints. This is because, for completeness and in full generality, hyperlogics cannot make any approximation of the program semantics defined by $\textsf{post}^\sharp$ in (\ref{eq:def:Post}) hence prohibiting approximations in (\ref{eq:def:abstract:logical:triples}).

Notice that no consequence rule is required for completeness, although they are sound\proofinapx.
\begin{eqntabular}{c@{\qquad}c}
\frac{\:\mathcal{P}\subseteq\mathcal{P}',\quad \overline{\llbrace}\,\mathcal{P}'\,\overline{\rrbrace}\,\texttt{S}\,\overline{\llbrace}\,\mathcal{Q}'\,\overline{\rrbrace},\quad \mathcal{Q}'\subseteq\mathcal{Q} \:}
{\overline{\llbrace}\,\mathcal{P}\,\overline{\rrbrace}\,\texttt{S}\,\overline{\llbrace}\,\mathcal{Q}\,\overline{\rrbrace}}
&
\frac{\:\mathcal{P}'\subseteq\mathcal{P},\quad \underline{\llbrace}\,\mathcal{P}'\,\underline{\rrbrace}\,\texttt{S}\,\underline{\llbrace}\,\mathcal{Q}'\,\underline{\rrbrace},\quad \mathcal{Q}\subseteq\mathcal{Q}'\:}
{\underline{\llbrace}\,\mathcal{P}\,\underline{\rrbrace}\,\texttt{S}\,\underline{\llbrace}\,\mathcal{Q}\,\underline{\rrbrace}}
\label{eq:logic:abstract:consequence}
\end{eqntabular}
\begin{toappendix}
\begin{proof}[Proof of (\ref{eq:logic:abstract:consequence})]
\begin{calculus}[$\Rightarrow$\ \ ]
\formula{\mathcal{P}\subseteq\mathcal{P}'\wedge \overline{\llbrace}\,\mathcal{P}'\,\overline{\rrbrace}\,\texttt{S}\,\overline{\llbrace}\,\mathcal{Q}'\,\overline{\rrbrace}\wedge \mathcal{Q}'\subseteq\mathcal{Q}}\\
$\Rightarrow$
\formulaexplanation{\mathcal{P}\subseteq\mathcal{P}'\wedge \textsf{Post}^\sharp\sqb{\texttt{S}}^\sharp\mathcal{P}'\subseteq \mathcal{Q}'\wedge\mathcal{Q}'\subseteq\mathcal{Q}}{def.\ (\ref{eq:def:abstract:logical:triples}) of the logic triples}\\
$\Rightarrow$
\formulaexplanation{\mathcal{P}\subseteq\mathcal{P}'\wedge \textsf{Post}^\sharp\sqb{\texttt{S}}^\sharp\mathcal{P}'\subseteq \mathcal{Q}}{$\subseteq$ transitive}\\
$\Rightarrow$
\formulaexplanation{\textsf{Post}^\sharp\sqb{\texttt{S}}^\sharp\mathcal{P}\subseteq\textsf{Post}^\sharp\sqb{\texttt{S}}^\sharp\mathcal{P}'\wedge \textsf{Post}^\sharp\sqb{\texttt{S}}^\sharp\mathcal{P}'\subseteq \mathcal{Q}}{$\textsf{Post}^\sharp\sqb{\texttt{S}}^\sharp$ increasing by (\ref{eq:GC:abstract:transformers})}\\
$\Rightarrow$
\formulaexplanation{\textsf{Post}^\sharp\sqb{\texttt{S}}^\sharp\mathcal{P}\subseteq \mathcal{Q}}{$\subseteq$ transitive}\\
$\Rightarrow$
\formulaexplanation{\overline{\llbrace}\,\mathcal{P}\,\overline{\rrbrace}\,\texttt{S}\,\overline{\llbrace}\,\mathcal{Q}\,\overline{\rrbrace}}{def.\ (\ref{eq:def:abstract:logical:triples}) of the logic triples}
\end{calculus}

\smallskip

\noindent The converse follows immediately by choosing $\mathcal{P}=\mathcal{P}'$ and $\mathcal{Q}'=\mathcal{Q}$ since $\subseteq$ is reflexive. The consequence rule for the lower abstract logic is $\subseteq$-dual.
\end{proof}
\end{toappendix}
\begin{example}[Choice]\label{ex:choice}
Let us define the choice $\texttt{S}_1+\texttt{S}_2\triangleq\texttt{c = [0,1]; if (c) S$_1$ else S$_2$}$ where auxiliary variable \texttt{c} does not appear in $\texttt{S}_1$ nor in $\texttt{S}_2$. The proof rule can be derived as follows
\begin{calculus}[$\Leftrightarrow$\ \ ]
\formula{\overline{\llbrace}\,\mathcal{P}\,\overline{\rrbrace}\,\texttt{S}_1+\texttt{S}_2\,\overline{\llbrace}\,\mathcal{Q}\,\overline{\rrbrace}}\\
$\Leftrightarrow$
\formulaexplanation{\overline{\llbrace}\,\mathcal{P}\,\overline{\rrbrace}\,\texttt{c = [0,1]; if (c) S$_1$ else S$_2$}\,\overline{\llbrace}\,\mathcal{Q}\,\overline{\rrbrace}}{def.\ choice $+$}\\
$\Leftrightarrow$
\formulaexplanation{\exists \mathcal{R}\mathrel{.}\overline{\llbrace}\,\mathcal{P}\,\overline{\rrbrace}\,\texttt{c = [0,1]}\,\overline{\llbrace}\,\mathcal{R}\,\overline{\rrbrace}\wedge \overline{\llbrace}\,\mathcal{R}\,\overline{\rrbrace}\,\texttt{if (c) S$_1$ else S$_2$}\,\overline{\llbrace}\,\mathcal{Q}\,\overline{\rrbrace}}{sequential composition (\ref{eq:logic:abstract:seq})}\\
$\Leftrightarrow$
\formulaexplanation{\exists \mathcal{R}\mathrel{.}\{P\mathbin{\fatsemi^{\sharp}}\textsf{\textup{rassign}}^{\sharp}\sqb{\texttt{c},$0$,$1$}\mid P\in\mathcal{P}\}\subseteq\mathcal{R}
\wedge \overline{\llbrace}\,\mathcal{R}\,\overline{\rrbrace}\,\texttt{if (c) S$_1$ else S$_2$}\,\overline{\llbrace}\,\mathcal{Q}\,\overline{\rrbrace}
}{(\ref{eq:logic:abstract:random:assignment})}\\
$\Leftrightarrow$
\formula{\overline{\llbrace}\,\{P\mathbin{\fatsemi^{\sharp}}\textsf{\textup{rassign}}^{\sharp}\sqb{\texttt{c},$0$,$1$}\mid P\in\mathcal{P}\}\,\overline{\rrbrace}\,\texttt{if (c) S$_1$ else S$_2$}\,\overline{\llbrace}\,\mathcal{Q}\,\overline{\rrbrace}
}\\[-0.5ex]\rightexplanation{taking $\mathcal{R} = \{P\mathbin{\fatsemi^{\sharp}}\textsf{\textup{rassign}}^{\sharp}\sqb{\texttt{c},$0$,$1$}\mid P\in\mathcal{P}\}$}\\
$\Leftrightarrow$
\formulaexplanation{\forall P\in\{P'\mathbin{\fatsemi^{\sharp}}\textsf{\textup{rassign}}^{\sharp}\sqb{\texttt{c},$0$,$1$}\mid P'\in\mathcal{P}\},Q_1,Q_2\mathrel{.}(\underline{\llbrace}\,\{P\}\,\underline{\rrbrace}\,{\texttt{B;S}_1}\,\underline{\llbrace}\,\{Q_1\}\,\underline{\rrbrace}
\wedge
\underline{\llbrace}\,\{P\}\,\underline{\rrbrace}\,{\neg\texttt{B;S}_2}\,\underline{\llbrace}\,\{Q_2\}\,\underline{\rrbrace})\Rightarrow(Q_1\mathbin{\,\sqcup^{\sharp}\,}Q_2\in \mathcal{Q})}{(\ref{eq:logic:abstract:if:upper})}\\[0.75ex]
$\Leftrightarrow$
\numberedformula{\forall P\in\mathcal{P},Q_1,Q_2\mathrel{.}(\underline{\llbrace}\,\{P\}\,\underline{\rrbrace}\,{\texttt{S}_1}\,\underline{\llbrace}\,\{Q_1\}\,\underline{\rrbrace}
\wedge
\underline{\llbrace}\,\{P\}\,\underline{\rrbrace}\,{\texttt{S}_2}\,\underline{\llbrace}\,\{Q_2\}\,\underline{\rrbrace})\Rightarrow(Q_1\mathbin{\,\sqcup^{\sharp}\,}Q_2\in \mathcal{Q})}\label{cal:choice:elementwise}\\
\explanation{assuming states where \texttt{c} is assigned $0$ or $1$, \texttt{B} is true for $0$ and $\neg\texttt{B}$ is true for $1$ (or conversely)}
\end{calculus}
\smallskip

\noindent so that we get the sound and complete rule
\begin{eqntabular}{c}
\frac{\:\forall P\in\mathcal{P},Q_1,Q_2\mathrel{.}(\underline{\llbrace}\,\{P\}\,\underline{\rrbrace}\,{\texttt{S}_1}\,\underline{\llbrace}\,\{Q_1\}\,\underline{\rrbrace}
\wedge
\underline{\llbrace}\,\{P\}\,\underline{\rrbrace}\,{\texttt{S}_2}\,\underline{\llbrace}\,\{Q_2\}\,\underline{\rrbrace})\Rightarrow(Q_1\mathbin{\,\sqcup^{\sharp}\,}Q_2\in \mathcal{Q})\:}{\overline{\llbrace}\,\mathcal{P}\,\overline{\rrbrace}\,\texttt{S}_1+\texttt{S}_2\,\overline{\llbrace}\,\mathcal{Q}\,\overline{\rrbrace}}
\label{eq:rule:choice:sqcup}
\end{eqntabular}
Let us now consider the particular case $\textsf{post}^\sharp(S)P$ $=$ $\{\pair{\sigma}{\sigma''}\mid\exists\sigma'\in\Sigma\mathrel{.}\pair{\sigma}{\sigma'}\in P\wedge\pair{\sigma'}{\sigma''}\in S\}$ as in 
example \ref{ex:powerset-deterministic-domain-post} (but this time with unbounded nondeterminism) so that $\sqcup^{\sharp}$ is $\cup$ in (\ref{eq:rule:choice:sqcup}). Then (\ref{eq:rule:choice:sqcup}) is implied, but not conversely, by the proof rule 
\begin{eqntabular*}{c}
\frac{\:{\overline{\llbrace}\,\mathcal{P}\,\overline{\rrbrace}\,\texttt{S}_1\,\overline{\llbrace}\,\mathcal{Q}_1\,\overline{\rrbrace}},\quad{\overline{\llbrace}\,\mathcal{P}\,\overline{\rrbrace}\,\texttt{S}_2\,\overline{\llbrace}\,\mathcal{Q}_2\,\overline{\rrbrace}}}{\:\overline{\llbrace}\,\mathcal{P}\,\overline{\rrbrace}\,\texttt{S}_1+\texttt{S}_2\,\overline{\llbrace}\,\{Q_1\cup Q_2\mid Q_1\in\mathcal{Q}_1\wedge Q_2\in\mathcal{Q}_2\} \overline{\rrbrace}\:}\quad(\textit{Choice})
\end{eqntabular*}
of \cite{DBLP:journals/afp/Dardinier23a}, which is sound but incomplete. For completeness,  \cite[p.\ 207:9]{DBLP:journals/afp/Dardinier23a} has to introduce an (\textit{Exist}) proof rule which amounts to the case by case analysis of rule (\ref{eq:rule:choice:sqcup}).
\end{example}
\begin{example}[Finitary powerset nondeterministic calculational domain]\label{ex:powerset-nondeterministic-domain-HL}Continuing example \ref{ex:powerset-nondeterministic-domain-post}, the iteration rule postulated in \cite[Fig.\ 2]{DBLP:journals/afp/Dardinier23a,DBLP:conf/pldi/DardinierM24} is (\ref{eq:logic:abstract:while:upper}), ignoring nontermination and breaks,
and applying proposition \ref{prop:Tarski:constructive} to reason on the fixpoint iterates.
\end{example}

\subsection{Calculational Design of the Proof System of the Upper Abstract Logic}\label{sec:ProofSystemUpperAbstractLogic}

\begin{proof}[Proof of (\ref{eq:logic:abstract:assignment}) --- (\ref{eq:logic:abstract:while:upper})] The proof of soundness and completeness  is by structural induction. We show the calculational design for the 
iteration (\ref{eq:logic:abstract:while:upper}). The other cases are in the appendix\proofinapx.
\medskip

\begin{toappendix}

\hyphen{5}\quad\textsc{Proof of (\ref{eq:logic:abstract:assignment}), (\ref{eq:logic:abstract:random:assignment}), (\ref{eq:logic:abstract:skip}), and (\ref{eq:logic:abstract:B})}.\quad The characterization (\ref{eq:Post:abstract:assignment}) of $\textsf{Post}^\sharp\sqb{\texttt{x = A}}^{\sharp}\mathcal{P}$ yields, by (\ref{eq:def:abstract:logical:triples}), the axiom (\ref{eq:logic:abstract:assignment}) for 
$\overline{\llbrace}\,\mathcal{P}\,\overline{\rrbrace}\,\texttt{x = A}\,\overline{\llbrace}\,\mathcal{Q}\,\overline{\rrbrace}$ (where the side condition is written as
a premiss). The rule for $\underline{\llbrace}\,\mathcal{P}\,\underline{\rrbrace}\,\texttt{x = A}\,\underline{\llbrace}\,\mathcal{Q}\,\underline{\rrbrace}$ is $\subseteq$-order dual. The rule (\ref{eq:logic:abstract:random:assignment}) for \texttt{x = [$a$, $b$]}, (\ref{eq:logic:abstract:skip}) for \texttt{x = skip}, and (\ref{eq:logic:abstract:B}) for \texttt{B} and their duals are similar.

\medskip

\hyphen{5}\quad\textsc{Proof of (\ref{eq:logic:abstract:break})}.\hfill The characterization (\ref{eq:def:abstract:logical:triples}) of $\textsf{Post}^\sharp\sqb{\texttt{break}}^{\sharp}\mathcal{P}$ yields the axiom (\ref{eq:logic:abstract:break}) for 
$\overline{\llbrace}\,\mathcal{P}\,\overline{\rrbrace}\,\texttt{break}\,\overline{\llbrace}\,\mathcal{Q}\,\overline{\rrbrace}$ (where the side condition is written as
a premiss). The rule for $\underline{\llbrace}\,\mathcal{P}\,\underline{\rrbrace}\,\texttt{break}\,\underline{\llbrace}\,\mathcal{Q}\,\underline{\rrbrace}$ is $\subseteq$-order dual. 

\medskip

\hyphen{5}\quad\textsc{Proof of (\ref{eq:logic:abstract:seq})}.\quad For sequential composition, we have
\begin{calculus}[$\leftrightarrow$\ \ ]
\formula{\overline{\llbrace}\,\mathcal{P}\,\overline{\rrbrace}\,\texttt{S}_1\texttt{;}\texttt{S}_2\,\overline{\llbrace}\,\mathcal{R}\,\overline{\rrbrace}}\\
$\Leftrightarrow$
\formulaexplanation{\textsf{Post}^\sharp\sqb{\texttt{S}_1\texttt{;}\texttt{S}_2}^{\sharp}\mathcal{P}\subseteq \mathcal{R}}{def.\ (\ref{eq:def:abstract:logical:triples}) of the logic triples}\\
$\Leftrightarrow$
\formulaexplanation{{\textsf{Post}^\sharp\sqb{\texttt{S$_2$}}^{\sharp}(\textsf{Post}^\sharp\sqb{\texttt{S$_1$}}^{\sharp}\mathcal{P})}
\subseteq \mathcal{R}}{(\ref{eq:Post:abstract:seq})}\\
$\Leftrightarrow$
\formula{\exists\mathcal{Q}\mathrel{.}\textsf{Post}^\sharp\sqb{\texttt{S$_1$}}^{\sharp}\mathcal{P}\subseteq \mathcal{Q}\wedge {\textsf{Post}^\sharp\sqb{\texttt{S$_2$}}^{\sharp}\mathcal{Q}}
\subseteq \mathcal{R}}\\
\explanation{(soundness, $\Rightarrow$) By (\ref{eq:def:Post}), $\textsf{Post}^\sharp\sqb{\texttt{S}}^\sharp$ is $\subseteq$-increasing so $\textsf{Post}^\sharp\sqb{\texttt{S$_1$}}^{\sharp}\mathcal{P}\subseteq \mathcal{Q}$ implies $\textsf{Post}^\sharp\sqb{\texttt{S$_2$}}^{\sharp}(\textsf{Post}^\sharp\sqb{\texttt{S$_1$}}^{\sharp}\mathcal{P})\subseteq \textsf{Post}^\sharp\sqb{\texttt{S$_2$}}^{\sharp}\mathcal{Q}$ and $\subseteq$ is transitive;
\\
(completeness, $\Leftarrow$) take $\mathcal{Q}=\textsf{Post}^\sharp\sqb{\texttt{S$_1$}}^{\sharp}\mathcal{P}$ and reflexivity}\\
$\Leftrightarrow$
\formula{\exists\mathcal{Q}\mathrel{.}\overline{\llbrace}\,\mathcal{P}\,\overline{\rrbrace}\,\texttt{S}_1\,\overline{\llbrace}\,\mathcal{Q}\,\overline{\rrbrace}\wedge\overline{\llbrace}\,\mathcal{Q}\,\overline{\rrbrace}\,\texttt{S}_2\,\overline{\llbrace}\,\mathcal{R}\,\overline{\rrbrace}}\\[-0.5ex]
\rightexplanation{def.\ (\ref{eq:def:abstract:logical:triples}) of the logic triples and dually for under approximation}
\end{calculus}

\medskip

\hyphen{5}\quad\textsc{Proof of (\ref{eq:logic:abstract:if:upper})}.\quad For the conditional, we have
\begin{calculus}[$\leftrightarrow$\ \ ]
\formula{\overline{\llbrace}\,\mathcal{P}\,\overline{\rrbrace}\,\texttt{if (B) S}_1\texttt{ else S}_2\,\overline{\llbrace}\,\mathcal{R}\,\overline{\rrbrace}}\\
$\Leftrightarrow$
\formulaexplanation{\textsf{Post}^\sharp\sqb{\texttt{if (B) S}_1\texttt{ else S}_2}^{\sharp}\mathcal{P}\subseteq \mathcal{R}}{def.\ (\ref{eq:def:abstract:logical:triples}) of the logic triples}\\
$\Leftrightarrow$
\formulaexplanation{{(\textsf{Post}^\sharp\sqb{\texttt{B;S}_1}^{\sharp}
\mathbin{\,\dot{\sqcup}^{\sharp}\,}
\textsf{Post}^\sharp\sqb{\neg\texttt{B;S}_2}^{\sharp})\mathcal{P}}\subseteq \mathcal{R}}{(\ref{eq:Post:abstract:if})}\\
$\Leftrightarrow$
\formulaexplanation{\{Q_1\mathbin{\,\sqcup^{\sharp}\,}Q_2\mid  
Q_1\in \textsf{Post}^\sharp\sqb{\texttt{B;S}_1}^{\sharp}\{P\}
\wedge
Q_2\in \textsf{Post}^\sharp\sqb{\neg\texttt{B;S}_2}^{\sharp}\{P\} 
\wedge
P\in\mathcal{P}\}\subseteq \mathcal{R}}{def.\ (\ref{eq:Post:abstract:dot-sqcup-sharp}) of $\mathbin{\,\dot{\sqcup}^{\sharp}\,}$}\\
$\Leftrightarrow$
\formula{\forall P,Q_1,Q_2\mathrel{.}(P\in\mathcal{P}\wedge Q_1\in \textsf{Post}^\sharp\sqb{\texttt{B;S}_1}^{\sharp}\{P\}
\wedge
Q_2\in \textsf{Post}^\sharp\sqb{\neg\texttt{B;S}_2}^{\sharp}\{P\} )\Rightarrow(Q_1\mathbin{\,\sqcup^{\sharp}\,}Q_2\in \mathcal{R})}\\[-0.5ex]\rightexplanation{def.\ $\subseteq$, $\wedge$ commutative}\\
$\Leftrightarrow$
\formula{\forall P,Q_1,Q_2\mathrel{.}(P\in\mathcal{P}\wedge \{Q_1\}\subseteq\textsf{Post}^\sharp\sqb{\texttt{B;S}_1}^{\sharp}\{P\}
\wedge
\{Q_2\}\subseteq\textsf{Post}^\sharp\sqb{\neg\texttt{B;S}_2}^{\sharp}\{P\} )\Rightarrow(Q_1\mathbin{\,\sqcup^{\sharp}\,}Q_2\in \mathcal{R})}\\[-0.5ex]\rightexplanation{def.\ $\subseteq$}\\
$\Leftrightarrow$
\formula{\forall P,Q_1,Q_2\mathrel{.}(P\in\mathcal{P}\wedge 
\underline{\llbrace}\,\{P\}\,\underline{\rrbrace}\,{\texttt{B;S}_1}\,\underline{\llbrace}\,\{Q_1\}\,\underline{\rrbrace}
\wedge
\underline{\llbrace}\,\{P\}\,\underline{\rrbrace}\,{\neg\texttt{B;S}_2}\,\underline{\llbrace}\,\{Q_2\}\,\underline{\rrbrace})\Rightarrow(Q_1\mathbin{\,\sqcup^{\sharp}\,}Q_2\in \mathcal{R})
}\\[-0.5ex]\rightexplanation{def.\ (\ref{eq:def:abstract:logical:triples}) of the lower abstract logic}
\end{calculus}

\medskip

\end{toappendix}
\begin{calculus}[$\leftrightarrow$\ \ ]
\formula{\overline{\llbrace}\,\mathcal{P}\,\overline{\rrbrace}\,\texttt{while (B) S}\,\overline{\llbrace}\,\mathcal{R}\,\overline{\rrbrace}}\\
$\Leftrightarrow$
\formulaexplanation{\textsf{Post}^\sharp\sqb{\texttt{while (B) S}}^{\sharp}\mathcal{P}\subseteq \mathcal{R}}{def.\ (\ref{eq:def:abstract:logical:triples}) of the logic triples}\\
$\Leftrightarrow$
\formula{\{\triple{e:Q_e}{\bot:Q_{\bot\ell}\mathbin{{\sqcup}_{\infty}^{\sharp}}{}Q_{\bot b}}{br:P_{br}}\mid 
Q_e\in\textsf{\textup{Post}}^\sharp(\sqb{\neg\texttt{B}}_{e}^{\sharp}\mathbin{\sqcup_{e}^{\sharp}}\sqb{\texttt{B;S}}_{b}^{\sharp})(\Lfp{\breve{\sqsubseteq}_{+}^{\sharp}}{\breve{\vec{F}}_{pe}^{\sharp}}(P))
{}\wedge{} 
Q_{\bot\ell}\in \textsf{\textup{Post}}^\sharp(\sqb{\texttt{B;S}}_{\bot}^{\sharp})(\Lfp{\breve{\sqsubseteq}_{+}^{\sharp}}{\breve{\vec{F}}_{pe}^{\sharp}}(P))
{}\wedge{}
\exists Q_{\bot b}\mathrel{.}Q_{\bot b}\in\textsf{\textup{Post}}^\sharp(Q_{p\bot})\{P\}{}\wedge{}
Q_{p\bot}\in\Gfp{{\breve{\sqsubseteq}}_{\infty}^{\sharp}}{\breve{F}_{p\bot}^{\sharp}}{}\wedge{}
P\in\mathcal{P}\}\subseteq \mathcal{R}}\\[-0.5ex]
\rightexplanation{(\ref{eq:Post:abstract:while})}\\
$\Leftrightarrow$
\formula{\{\triple{e:Q_e}{\bot:Q_{\bot\ell}\mathbin{{\sqcup}_{\infty}^{\sharp}}{}
Q_{\bot b}}{br:P_{br}}\mid {}
\exists I_e\mathrel{.}I_e\subseteq \Lfp{\breve{\sqsubseteq}_{+}^{\sharp}}{\breve{\vec{F}}_{pe}^{\sharp}}(P)\wedge
Q_e\in\textsf{Post}^\sharp(\sqb{\neg\texttt{B}}_{e}^{\sharp}\mathbin{\sqcup_{e}^{\sharp}}\sqb{\texttt{B;S}}_{b}^{\sharp})I_e\wedge{} 
Q_{\bot\ell}\in \textsf{Post}^\sharp(\sqb{\texttt{B;S}}_{\bot}^{\sharp})(I_e)
\wedge{}
\exists I_{\bot}\mathrel{.} I_{\bot}\subseteq\Gfp{{\breve{\sqsubseteq}}_{\infty}^{\sharp}}{\breve{F}_{p\bot}^{\sharp}}\wedge{}
Q_{\bot b}\in I_{\bot}\wedge
P\in\mathcal{P}\}\subseteq \mathcal{R}}\\
\explanation{($\Rightarrow$) Take $I_e=\Lfp{\breve{\sqsubseteq}_{+}^{\sharp}}{\breve{\vec{F}}_{pe}^{\sharp}}(P)$,  $I_{\bot}=\Gfp{{\breve{\sqsubseteq}}_{\infty}^{\sharp}}{\breve{F}_{p\bot}^{\sharp}}$, and $\subseteq$ reflexive\\
($\Leftarrow$) by (\ref{eq:GC:abstract:transformers}), $\textsf{Post}^\sharp(S)$ is $\subseteq$-increasing}\\
$\Leftrightarrow$
\formula{\{\triple{e:Q_e}{\bot:Q_{\bot\ell}\mathbin{{\sqcup}_{\infty}^{\sharp}}{}
Q_{\bot b}}{br:P_{br}}\mid {}
\exists P_e\mathrel{.}\{P_e\}=\Lfp{\breve{\sqsubseteq}_{+}^{\sharp}}{\breve{\vec{F}}_{pe}^{\sharp}}(P)\wedge
Q_e\in\textsf{Post}^\sharp(\sqb{\neg\texttt{B}}_{e}^{\sharp}\mathbin{\sqcup_{e}^{\sharp}}\sqb{\texttt{B;S}}_{b}^{\sharp})\{P_e\}\wedge{} 
Q_{\bot\ell}\in \textsf{Post}^\sharp(\sqb{\texttt{B;S}}_{\bot}^{\sharp})\{P_e\}
\wedge{}
\exists P_{\bot}\mathrel{.} \{P_{\bot}\}=\Gfp{{\breve{\sqsubseteq}}_{\infty}^{\sharp}}{\breve{F}_{p\bot}^{\sharp}}\wedge{}
Q_{\bot b}\in \{P_{\bot}\}\wedge
P\in\mathcal{P}\}\subseteq \mathcal{R}}\\
\explanation{If $I_e$ is empty then $\textsf{Post}^\sharp(\sqb{\neg\texttt{B}}_{e}^{\sharp}\mathbin{\sqcup_{e}^{\sharp}}\sqb{\texttt{B;S}}_{b}^{\sharp})I_e$ is empty by (\ref{eq:def:Post}), contrary to $Q_e\in\textsf{Post}^\sharp(\sqb{\neg\texttt{B}}_{e}^{\sharp}\mathbin{\sqcup_{e}^{\sharp}}\sqb{\texttt{B;S}}_{b}^{\sharp})I_e$ proving that $I_e$ cannot be empty. By (\ref{eq:def:F-pe-sharp-breve}), $\Lfp{\breve{\sqsubseteq}_{+}^{\sharp}}{\breve{\vec{F}}_{pe}^{\sharp}}(P)$ is a singleton, say $\{P_e\}$. For $I_e$ to be non-empty and included in a singleton, it must be equal to that singleton so $I_e=\{P_e\}$. The reasoning is the same for $I_{\bot}=\{P_{\bot}\}$}\\
$\Leftrightarrow$
\formula{\{\triple{e:Q_e}{\bot:Q_{\bot\ell}\mathbin{{\sqcup}_{\infty}^{\sharp}}{}
Q_{\bot b}}{br:P_{br}}\mid {}
\exists P_e\mathrel{.}\{P_e\}=\Lfp{\breve{\sqsubseteq}_{+}^{\sharp}}{\breve{\vec{F}}_{pe}^{\sharp}}(P)\wedge
Q_e\in\textsf{Post}^\sharp(\sqb{\neg\texttt{B}}_{e}^{\sharp}\mathbin{\sqcup_{e}^{\sharp}}\sqb{\texttt{B;S}}_{b}^{\sharp})\{P_e\}\wedge{} 
Q_{\bot\ell}\in \textsf{Post}^\sharp(\sqb{\texttt{B;S}}_{\bot}^{\sharp})\{P_e\}
\wedge{}
 \{Q_{\bot b}\}=\Gfp{{\breve{\sqsubseteq}}_{\infty}^{\sharp}}{\breve{F}_{p\bot}^{\sharp}}\wedge{}
P\in\mathcal{P}\}\subseteq \mathcal{R}}\\\rightexplanation{$Q_{\bot b}\in \{P_{\bot}\}$ if and only if $Q_{\bot b}=P_{\bot b}$}\\
$\Leftrightarrow$
\formula{\{\triple{e:Q_e}{\bot:Q_{\bot\ell}\mathbin{{\sqcup}_{\infty}^{\sharp}}{}
Q_{\bot b}}{br:P_{br}}\mid {}
\exists P_e\mathrel{.}\{P_e\}=\{\Lfp{{\sqsubseteq}_{+}^{\sharp}}{\vec{F}_{pe}^{\sharp}}(P)\}
\wedge
Q_e\in\textsf{Post}^\sharp(\sqb{\neg\texttt{B}}_{e}^{\sharp}\mathbin{\sqcup_{e}^{\sharp}}\sqb{\texttt{B;S}}_{b}^{\sharp})\{P_e\}\wedge{} 
Q_{\bot\ell}\in \textsf{Post}^\sharp(\sqb{\texttt{B;S}}_{\bot}^{\sharp})\{P_e\}
\wedge{}
 \{Q_{\bot b}\}=\{\Gfp{{\sqsubseteq}_{\infty}^{\sharp}}{F_{p\bot}^{\sharp}}\}
 \wedge{}
P\in\mathcal{P}\}\subseteq \mathcal{R}}\\
\explanation{since $\Lfp{\breve{\sqsubseteq}_{+}^{\sharp}}{\breve{\vec{F}}_{pe}^{\sharp}}(P)$
=
$\{\Lfp{{\sqsubseteq}_{+}^{\sharp}}{\vec{F}_{pe}^{\sharp}}(P)\}$
by (\ref{eq:def:F-pe-sharp-breve}), proposition \ref{prop:prop:singletonization},
and
$\Gfp{{\breve{\sqsubseteq}}_{\infty}^{\sharp}}({\breve{F}_{p\bot}^{\sharp}})$ 
=
$\{\Gfp{{\sqsubseteq}_{\infty}^{\sharp}}{F_{p\bot}^{\sharp}}\}$
by (\ref{eq:def:F-p-bot-sharp}) and proposition \ref{prop:prop:singletonization}}\\
$\Leftrightarrow$
\formulaexplanation{\{\triple{e:Q_e}{\bot:Q_{\bot\ell}\mathbin{{\sqcup}_{\infty}^{\sharp}}{}
Q_{\bot b}}{br:P_{br}}\mid {}
\exists P_e\mathrel{.}P_e=\Lfp{{\sqsubseteq}_{+}^{\sharp}}{\vec{F}_{pe}^{\sharp}}(P)
\wedge
Q_e\in\textsf{Post}^\sharp(\sqb{\neg\texttt{B}}_{e}^{\sharp}\mathbin{\sqcup_{e}^{\sharp}}\sqb{\texttt{B;S}}_{b}^{\sharp})\{P_e\}\wedge{} 
Q_{\bot\ell}\in \textsf{Post}^\sharp(\sqb{\texttt{B;S}}_{\bot}^{\sharp})\{P_e\}
\wedge{}
Q_{\bot b}=\Gfp{{\sqsubseteq}_{\infty}^{\sharp}}{F_{p\bot}^{\sharp}}
 \wedge{}
P\in\mathcal{P}\}\subseteq \mathcal{R}}{def.\ set equality}\\
$\Leftrightarrow$
\formulaexplanation{\{\triple{e:Q_e}{\bot:Q_{\bot\ell}\mathbin{{\sqcup}_{\infty}^{\sharp}}{}
Q_{\bot b}}{br:P_{br}}\mid {}
\exists P_e\mathrel{.}P_e=\Lfp{{\sqsubseteq}_{+}^{\sharp}}{\vec{F}_{pe}^{\sharp}}(P)
\wedge
\{Q_e\}\subseteq\textsf{Post}^\sharp(\sqb{\neg\texttt{B}}_{e}^{\sharp}\mathbin{\sqcup_{e}^{\sharp}}\sqb{\texttt{B;S}}_{b}^{\sharp})\{P_e\}\wedge{} 
\{Q_{\bot\ell}\}\subseteq\textsf{Post}^\sharp(\sqb{\texttt{B;S}}_{\bot}^{\sharp})\{P_e\}
\wedge{}
Q_{\bot b}=\Gfp{{\sqsubseteq}_{\infty}^{\sharp}}{F_{p\bot}^{\sharp}}
 \wedge{}
P\in\mathcal{P}\}\subseteq \mathcal{R}}{def.\ $\in$ and $\subseteq$}\\
$\Leftrightarrow$
\formula{\{\triple{e:Q_e}{\bot:Q_{\bot\ell}\mathbin{{\sqcup}_{\infty}^{\sharp}}{}
Q_{\bot b}}{br:P_{br}}\mid {}
\exists P_e\mathrel{.}P_e=\Lfp{{\sqsubseteq}_{+}^{\sharp}}{\vec{F}_{pe}^{\sharp}}(P)\wedge
\{Q_e\}\subseteq\textsf{Post}^\sharp(\sqb{\neg\texttt{B}}_{e}^{\sharp}\mathbin{\sqcup_{e}^{\sharp}}\sqb{\texttt{B;S}}_{b}^{\sharp})\{P_e\}\wedge{} 
\underline{\llbrace}\,\{P_e\}\,\underline{\rrbrace}\,\texttt{{B;S}}\,\underline{\llbrace}\,\{Q_{\bot\ell}\}\,\underline{\rrbrace}
\wedge{}
Q_{\bot b}=\Gfp{{\sqsubseteq}_{\infty}^{\sharp}}{F_{p\bot}^{\sharp}}\wedge{}
P\in\mathcal{P}\}\subseteq \mathcal{R}}\\[-0.5ex]
\rightexplanation{def.\ (\ref{eq:def:abstract:logical:triples}) of $\underline{\llbrace}\,\mathcal{P}\,\underline{\rrbrace}\,\texttt{S}\,\underline{\llbrace}\,\mathcal{Q}\,\underline{\rrbrace}$ $\triangleq$ ($\mathcal{Q}\subseteq\textsf{Post}^\sharp\sqb{\texttt{S}}^\sharp\mathcal{P}$)}\\
$\Leftrightarrow$
\formulaexplanation{\{\triple{e:Q_e}{\bot:Q_{\bot\ell}\mathbin{{\sqcup}_{\infty}^{\sharp}}{}
Q_{\bot b}}{br:P_{br}}\mid {}
\exists P_e\mathrel{.}P_e=\Lfp{{\sqsubseteq}_{+}^{\sharp}}{\vec{F}_{pe}^{\sharp}}(P)\wedge
\{Q_e\}\subseteq\textsf{Post}^\sharp(\sqb{\neg\texttt{B}}_{e}^{\sharp})\{P_e\}\mathbin{\dot{\sqcup}_{e}^{\sharp}}\textsf{Post}^\sharp(\sqb{\texttt{B;S}}_{b}^{\sharp})\{P_e\}\wedge{} 
\underline{\llbrace}\,\{P_e\}\,\underline{\rrbrace}\,\texttt{{B;S}}\,\underline{\llbrace}\,\{Q_{\bot\ell}\}\,\underline{\rrbrace}
\wedge{}
Q_{\bot b}=\Gfp{{\sqsubseteq}_{\infty}^{\sharp}}{F_{p\bot}^{\sharp}}\wedge{}
P\in\mathcal{P}\}\subseteq \mathcal{R}}{(\ref{eq:Post:abstract:dot-sqcup-sharp})}\\
$\Leftrightarrow$
\formulaexplanation{\{\triple{e:Q'_e}{\bot:Q_{\bot\ell}\mathbin{{\sqcup}_{\infty}^{\sharp}}{}
Q_{\bot b}}{br:P_{br}}\mid {}
\exists P_e\mathrel{.}P_e=\Lfp{{\sqsubseteq}_{+}^{\sharp}}{\vec{F}_{pe}^{\sharp}}(P)\wedge
\{Q'_e\}\subseteq\{Q_e\mathbin{\,\sqcup^{\sharp}_{e}\,}Q_b\mid \{Q_e\}\subseteq\textsf{Post}^\sharp(\sqb{\neg\texttt{B}}_{e}^{\sharp})\{P\}
\wedge \{Q_b\}\subseteq\textsf{Post}^\sharp(\sqb{\texttt{B;S}}_{b}^{\sharp})\{P\} \wedge P\in \{P_e\}\}\wedge{} 
\underline{\llbrace}\,\{P_e\}\,\underline{\rrbrace}\,\texttt{{B;S}}\,\underline{\llbrace}\,\{Q_{\bot\ell}\}\,\underline{\rrbrace}
\wedge{}
Q_{\bot b}=\Gfp{{\sqsubseteq}_{\infty}^{\sharp}}{F_{p\bot}^{\sharp}}\wedge{}
P\in\mathcal{P}\}\subseteq \mathcal{R}}{def.\ (\ref{eq:Post:abstract:dot-sqcup-sharp}) of $\mathbin{\,\dot{\sqcup}^{\sharp}_{e}\,}$}\\
$\Leftrightarrow$
\formulaexplanation{\{\triple{e:Q'_e}{\bot:Q_{\bot\ell}\mathbin{{\sqcup}_{\infty}^{\sharp}}{}
Q_{\bot b}}{br:P_{br}}\mid {}
\exists P_e\mathrel{.}P_e=\Lfp{{\sqsubseteq}_{+}^{\sharp}}{\vec{F}_{pe}^{\sharp}}(P')\wedge
\exists Q_e, Q_b, P\mathrel{.}
Q'_e=Q_e\mathbin{\,\sqcup^{\sharp}_{e}\,}Q_b\wedge \{Q_e\}\subseteq\textsf{Post}^\sharp(\sqb{\neg\texttt{B}}_{e}^{\sharp})\{P\}
\wedge \{Q_b\}\subseteq\textsf{Post}^\sharp(\sqb{\texttt{B;S}}_{b}^{\sharp})\{P\} \wedge P\in \{P_e\} \wedge{} 
\underline{\llbrace}\,\{P_e\}\,\underline{\rrbrace}\,\texttt{{B;S}}\,\underline{\llbrace}\,\{Q_{\bot\ell}\}\,\underline{\rrbrace}
\wedge{}
Q_{\bot b}=\Gfp{{\sqsubseteq}_{\infty}^{\sharp}}{F_{p\bot}^{\sharp}}\wedge{}
P'\in\mathcal{P}\}\subseteq \mathcal{R}}{def.\ singleton and $\subseteq$, renaming}\\
$\Leftrightarrow$
\formulaexplanation{\{\triple{e:Q_e\mathbin{\,\sqcup^{\sharp}_{e}\,}Q_b}{\bot:Q_{\bot\ell}\mathbin{{\sqcup}_{\infty}^{\sharp}}{}Q_{\bot b}}{br:P_{br}}\mid {}
\exists P_e\mathrel{.}P_e=\Lfp{{\sqsubseteq}_{+}^{\sharp}}{\vec{F}_{pe}^{\sharp}}(P')\wedge
\exists P\mathrel{.}
\{Q_e\}\subseteq\textsf{Post}^\sharp(\sqb{\neg\texttt{B}}_{e}^{\sharp})\{P\}
\wedge \{Q_b\}\subseteq\textsf{Post}^\sharp(\sqb{\texttt{B;S}}_{b}^{\sharp})\{P\} \wedge P\in \{P_e\} \wedge{} 
\underline{\llbrace}\,\{P_e\}\,\underline{\rrbrace}\,\texttt{{B;S}}\,\underline{\llbrace}\,\{Q_{\bot\ell}\}\,\underline{\rrbrace}
\wedge{}
Q_{\bot b}=\Gfp{{\sqsubseteq}_{\infty}^{\sharp}}{F_{p\bot}^{\sharp}}\wedge{}
P'\in\mathcal{P}\}\subseteq \mathcal{R}}{replacing $Q'_e$ by its value}\\
$\Leftrightarrow$
\formula{\{\triple{e:Q_e\mathbin{\,\sqcup^{\sharp}_{e}\,}Q_b}{\bot:Q_{\bot\ell}\mathbin{{\sqcup}_{\infty}^{\sharp}}{}Q_{\bot b}}{br:P_{br}}\mid {}
\exists P_e\mathrel{.}P_e=\Lfp{{\sqsubseteq}_{+}^{\sharp}}{\vec{F}_{pe}^{\sharp}}(P')\wedge
\{Q_e\}\subseteq\textsf{Post}^\sharp(\sqb{\neg\texttt{B}}_{e}^{\sharp})\{P_e\}
\wedge \{Q_b\}\subseteq\textsf{Post}^\sharp(\sqb{\texttt{B;S}}_{b}^{\sharp})\{P_e\} \wedge{} 
\underline{\llbrace}\,\{P_e\}\,\underline{\rrbrace}\,\texttt{{B;S}}\,\underline{\llbrace}\,\{Q_{\bot\ell}\}\,\underline{\rrbrace}
\wedge{}
Q_{\bot b}=\Gfp{{\sqsubseteq}_{\infty}^{\sharp}}{F_{p\bot}^{\sharp}}\wedge{}
P'\in\mathcal{P}\}\subseteq \mathcal{R}}\\[-0.5ex]\rightexplanation{corollary \ref{cor:underapproximation:one:element}}\\
$\Leftrightarrow$
\formula{\{\triple{e:Q_e\mathbin{\,\sqcup^{\sharp}_{e}\,}Q_b}{\bot:Q_{\bot\ell}\mathbin{{\sqcup}_{\infty}^{\sharp}}{}Q_{\bot b}}{br:P_{br}}\mid {}
\exists P_e\mathrel{.}P_e=\Lfp{{\sqsubseteq}_{+}^{\sharp}}{\vec{F}_{pe}^{\sharp}}(P')\wedge
\underline{\llbrace}\,\{P_e\}\,\underline{\rrbrace}\,\neg\texttt{B}\,\underline{\llbrace}\,\{Q_e\}\,\underline{\rrbrace}
\wedge 
\underline{\llbrace}\,\{P_e\}\,\underline{\rrbrace}\,\texttt{B;S}\,\underline{\llbrace}\,\{Q_b\}\,\underline{\rrbrace}
\wedge{} 
\underline{\llbrace}\,\{P_e\}\,\underline{\rrbrace}\,\texttt{{B;S}}\,\underline{\llbrace}\,\{Q_{\bot\ell}\}\,\underline{\rrbrace}
\wedge{}
Q_{\bot b}=\Gfp{{\sqsubseteq}_{\infty}^{\sharp}}{F_{p\bot}^{\sharp}}\wedge{}
P'\in\mathcal{P}\}\subseteq \mathcal{R}}\\[-0.25ex]
\rightexplanation{def.\ (\ref{eq:def:abstract:logical:triples}) of $\underline{\llbrace}\,\mathcal{P}\,\underline{\rrbrace}\,\texttt{S}\,\underline{\llbrace}\,\mathcal{Q}\,\underline{\rrbrace}$ $\triangleq$ ($\mathcal{Q}\subseteq\textsf{Post}^\sharp\sqb{\texttt{S}}^\sharp\mathcal{P}$)}\\
$\Leftrightarrow$
\formulaexplanation{(P_e=\Lfp{{\sqsubseteq}_{+}^{\sharp}}{\vec{F}_{pe}^{\sharp}}(P')\wedge
\underline{\llbrace}\,\{P_e\}\,\underline{\rrbrace}\,\neg\texttt{B}\,\underline{\llbrace}\,\{Q_e\}\,\underline{\rrbrace}
\wedge 
\underline{\llbrace}\,\{P_e\}\,\underline{\rrbrace}\,\texttt{B;S}\,\underline{\llbrace}\,\{Q_b\}\,\underline{\rrbrace}
\wedge{} 
\underline{\llbrace}\,\{P_e\}\,\underline{\rrbrace}\,\texttt{{B;S}}\,\underline{\llbrace}\,\{Q_{\bot\ell}\}\,\underline{\rrbrace}
\wedge{}
Q_{\bot b}=\Gfp{{\sqsubseteq}_{\infty}^{\sharp}}{F_{p\bot}^{\sharp}}\wedge{}
P'\in\mathcal{P})
\Rightarrow 
\triple{e:Q_e\mathbin{\,\sqcup^{\sharp}_{e}\,}Q_b}{\bot:Q_{\bot\ell}\mathbin{{\sqcup}_{\infty}^{\sharp}}{}Q_{\bot b}}{br:P_{br}}\in\mathcal{R}}{def.\ $\subseteq$}\\
$\Leftrightarrow$
\lastformulaexplanation{(P_e=\Lfp{{\sqsubseteq}_{+}^{\sharp}}{\vec{F}_{pe}^{\sharp}}(P')\wedge
\overline{\llbrace}\,\{P_e\}\,\overline{\rrbrace}\,\neg\texttt{B}\,\overline{\llbrace}\,\{Q_e\}\,\overline{\rrbrace}
\wedge 
\overline{\llbrace}\,\{P_e\}\,\overline{\rrbrace}\,\texttt{B;S}\,\overline{\llbrace}\,\{Q_b\}\,\overline{\rrbrace}
\wedge{} 
\overline{\llbrace}\,\{P_e\}\,\overline{\rrbrace}\,\texttt{{B;S}}\,\overline{\llbrace}\,\{Q_{\bot\ell}\}\,\overline{\rrbrace}
\wedge{}
Q_{\bot b}=\Gfp{{\sqsubseteq}_{\infty}^{\sharp}}{F_{p\bot}^{\sharp}}\wedge{}
P'\in\mathcal{P})
\Rightarrow 
\triple{e:Q_e\mathbin{\,\sqcup^{\sharp}_{e}\,}Q_b}{\bot:Q_{\bot\ell}\mathbin{{\sqcup}_{\infty}^{\sharp}}{}Q_{\bot b}}{br:P_{br}}\in\mathcal{R}}{lemma \ref{singleton:under=over}}{\mbox{\qed}}
\end{calculus}
\let\qed\relax
\end{proof}
\noindent Propositions \ref{prop:Tarski} and \ref{prop:Tarski:constructive} can be used to characterize the
fixpoints of increasing functions in (\ref{eq:logic:abstract:while:upper}).

\subsection{Calculational Design of the Proof System of the Lower Abstract Logic}

Apart from (\ref{eq:logic:abstract:assignment})---(\ref{eq:logic:abstract:seq}), the sound and complete induction rules for the lower abstract logic are constructed by calculational design as follows.
\begin{theorem}[Lower abstract logic proof system]\label{th:lower-abstract-logic-proof-system}\proofinapx\quad
If\/ $\mathbb{D}^{\sharp}$ is a well-defined increasing and decreasing chain-complete join semilattice with right upper continuous sequential composition $\mathbin{{\fatsemi}^{\sharp}}$ then
\begin{eqntabular}{c@{\qquad}}
\frac{\:\parbox{0.875\textwidth}{$\forall Q\in\mathcal{Q}\mathrel{.}
\exists P\in\mathcal{P}, Q_1, Q_2\mathrel{.}
\underline{\llbrace}\,\{P\}\,\underline{\rrbrace}\,{\texttt{B;S}_1}\,\underline{\llbrace}\,\{Q_1\}\,\underline{\rrbrace}
\wedge
\underline{\llbrace}\,\{P\}\,\underline{\rrbrace}\,{\neg\texttt{B;S}_2}\,\underline{\llbrace}\,\{Q_2\}\,\underline{\rrbrace}
\wedge
Q=Q_1\mathbin{\,\sqcup^{\sharp}\,}Q_2$}\:}{\underline{\llbrace}\,\mathcal{P}\,\underline{\rrbrace}\,\texttt{if (B) S}_1\texttt{ else S}_2\,\underline{\llbrace}\,\mathcal{Q}\,\underline{\rrbrace}}
\label{eq:logic:abstract:if:lower}
\\[1ex]
\frac{\:\parbox{0.875\textwidth}{\centering$
\forall \triple{e:Q'_e}{\bot:Q'_{\bot}}{br:{Q'_{br}}}\in\mathcal{Q}\mathrel{.}
\exists Q_e ,Q_b, Q_{\bot\ell}, Q_{\bot b}, P_e\mathrel{.}
{Q'_e}={Q_e\mathbin{\,\sqcup^{\sharp}_{e}\,}Q_b}
\wedge
{Q'_{\bot}}={Q_{\bot\ell}\mathbin{{\sqcup}_{\infty}^{\sharp}}{}Q_{\bot b}}
\wedge
{Q'_{br}}={P'_{br}}
\wedge
P_e=\Lfp{{\sqsubseteq}_{+}^{\sharp}}{\vec{F}_{pe}^{\sharp}}(P')
\wedge
\underline{\llbrace}\,\{P_e\}\,\underline{\rrbrace}\,\neg\texttt{B}\,\underline{\llbrace}\,\{Q_e\}\,\underline{\rrbrace}
\wedge 
\underline{\llbrace}\,\{P_e\}\,\underline{\rrbrace}\,\texttt{B;S}\,\underline{\llbrace}\,\{Q_b\}\,\underline{\rrbrace}
\wedge{} 
\underline{\llbrace}\,\{P_e\}\,\underline{\rrbrace}\,\texttt{{B;S}}\,\underline{\llbrace}\,\{Q_{\bot\ell}\}\,\underline{\rrbrace}
\wedge{}
Q_{\bot b}=\Gfp{{\sqsubseteq}_{\infty}^{\sharp}}{F_{p\bot}^{\sharp}}
\wedge{}
P'\in\mathcal{P}$}\:}{\underline{\llbrace}\,\mathcal{P}\,\underline{\rrbrace}\,\texttt{while (B) S}\,\underline{\llbrace}\,\mathcal{Q}\,\underline{\rrbrace}}
\label{eq:logic:abstract:while:lower}
\end{eqntabular}
\end{theorem}
\begin{toappendix}
\begin{proof}[Proof of theorem \ref{th:lower-abstract-logic-proof-system}]
\begin{calculus}[$\Leftrightarrow$\ \ ]
\hyphen{5}\formula{\underline{\llbrace}\,\mathcal{P}\,\underline{\rrbrace}\,\texttt{if (B) S}_1\texttt{ else S}_2\,\underline{\llbrace}\,\mathcal{Q}\,\underline{\rrbrace}}\\
$\Leftrightarrow$
\formulaexplanation{\mathcal{Q}\subseteq\textsf{Post}^\sharp\sqb{\texttt{\texttt{if (B) S}$_1$\texttt{ else S}$_2$}}^\sharp\mathcal{P}}{def.\ (\ref{eq:def:abstract:logical:triples}) of the logic triples}\\
$\Leftrightarrow$
\formulaexplanation{\mathcal{Q}\subseteq(\textsf{Post}^\sharp\sqb{\texttt{B;S}_1}^{\sharp}
\mathbin{\,\dot{\sqcup}^{\sharp}\,}
\textsf{Post}^\sharp\sqb{\neg\texttt{B;S}_2}^{\sharp})\mathcal{P}}{(\ref{eq:Post:abstract:if})}\\
$\Leftrightarrow$
\formulaexplanation{\mathcal{Q}\subseteq \{Q_1\mathbin{\,\sqcup^{\sharp}\,}Q_2\mid  
Q_1\in \textsf{Post}^\sharp\sqb{\texttt{B;S}_1}^{\sharp}\{P\}
\wedge
Q_2\in \textsf{Post}^\sharp\sqb{\neg\texttt{B;S}_2}^{\sharp}\{P\} 
\wedge
P\in\mathcal{P}\}}{def.\ (\ref{eq:Post:abstract:dot-sqcup-sharp}) of $\mathbin{\,\dot{\sqcup}^{\sharp}\,}$}\\
$\Leftrightarrow$
\formula{\forall Q\in\mathcal{Q}\mathrel{.}
\exists Q_1, Q_2,P\mathrel{.}
Q_1\in \textsf{Post}^\sharp\sqb{\texttt{B;S}_1}^{\sharp}\{P\}
\wedge
Q_2\in \textsf{Post}^\sharp\sqb{\neg\texttt{B;S}_2}^{\sharp}\{P\} 
\wedge
P\in\mathcal{P}
\wedge
Q=Q_1\mathbin{\,\sqcup^{\sharp}\,}Q_2}\\[-0.5ex]
\rightexplanation{def.\ $\subseteq$}\\
$\Leftrightarrow$
\formulaexplanation{\forall Q\in\mathcal{Q}\mathrel{.}
\exists Q_1, Q_2,P\mathrel{.}
\{Q_1\}\subseteq\textsf{Post}^\sharp\sqb{\texttt{B;S}_1}^{\sharp}\{P\}
\wedge
\{Q_2\}\subseteq\textsf{Post}^\sharp\sqb{\neg\texttt{B;S}_2}^{\sharp}\{P\} 
\wedge
P\in\mathcal{P}
\wedge
Q=Q_1\mathbin{\,\sqcup^{\sharp}\,}Q_2}{def.\ $\subseteq$ for singleton}\\
$\Leftrightarrow$
\formula{\forall Q\in\mathcal{Q}\mathrel{.}
\exists Q_1, Q_2,P\mathrel{.}
\underline{\llbrace}\,\{P\}\,\underline{\rrbrace}\,{\texttt{B;S}_1}\,\underline{\llbrace}\,\{Q_1\}\,\underline{\rrbrace}
\wedge
\underline{\llbrace}\,\{P\}\,\underline{\rrbrace}\,{\neg\texttt{B;S}_2}\,\underline{\llbrace}\,\{Q_2\}\,\underline{\rrbrace}
\wedge
P\in\mathcal{P}
\wedge
Q=Q_1\mathbin{\,\sqcup^{\sharp}\,}Q_2}\\[-0.5ex]
\rightexplanation{def.\ (\ref{eq:def:abstract:logical:triples}) of the logic triples}\\[1em]

\hyphen{5}\formula{\underline{\llbrace}\,\mathcal{P}\,\underline{\rrbrace}\,\texttt{while (B) S}\,\underline{\llbrace}\,\mathcal{Q}\,\underline{\rrbrace}}\\
$\Leftrightarrow$
\formulaexplanation{\mathcal{Q}\subseteq\textsf{Post}^\sharp\sqb{\texttt{while (B) S}}^\sharp\mathcal{P}}{def.\ (\ref{eq:def:abstract:logical:triples}) of the logic triples}\\
$\Leftrightarrow$
\formulaexplanation{\mathcal{Q}\subseteq\{\triple{e:Q_e}{\bot:Q_{\bot\ell}\mathbin{{\sqcup}_{\infty}^{\sharp}}{}
Q_{\bot b}}{br:P_{br}}\mid {}
Q_e\in\textsf{Post}^\sharp(\sqb{\neg\texttt{B}}_{e}^{\sharp}\mathbin{\sqcup_{e}^{\sharp}}\sqb{\texttt{B;S}}_{b}^{\sharp})(\Lfp{\breve{\sqsubseteq}_{+}^{\sharp}}{\breve{\vec{F}}_{pe}^{\sharp}}(P))\wedge{}
\ Q_{\bot\ell}\in \textsf{Post}^\sharp(\sqb{\texttt{B;S}}_{\bot}^{\sharp})(\Lfp{\breve{\sqsubseteq}_{+}^{\sharp}}{\breve{\vec{F}}_{pe}^{\sharp}}(P))\wedge{}
Q_{\bot b}\in\Gfp{{\breve{\sqsubseteq}}_{\infty}^{\sharp}}({\breve{F}_{p\bot}^{\sharp}})\wedge
P\in\mathcal{P}\}}{(\ref{eq:Post:abstract:while})}\\
$\Leftrightarrow$
\formula{\mathcal{Q}\subseteq\{\triple{e:Q_e\mathbin{\,\sqcup^{\sharp}_{e}\,}Q_b}{\bot:Q_{\bot\ell}\mathbin{{\sqcup}_{\infty}^{\sharp}}{}Q_{\bot b}}{br:P_{br}}\mid {}
\exists P_e\mathrel{.}P_e=\Lfp{{\sqsubseteq}_{+}^{\sharp}}{\vec{F}_{pe}^{\sharp}}(P')\wedge
\underline{\llbrace}\,\{P_e\}\,\underline{\rrbrace}\,\neg\texttt{B}\,\underline{\llbrace}\,\{Q_e\}\,\underline{\rrbrace}
\wedge 
\underline{\llbrace}\,\{P_e\}\,\underline{\rrbrace}\,\texttt{B;S}\,\underline{\llbrace}\,\{Q_b\}\,\underline{\rrbrace}
\wedge{} 
\underline{\llbrace}\,\{P_e\}\,\underline{\rrbrace}\,\texttt{{B;S}}\,\underline{\llbrace}\,\{Q_{\bot\ell}\}\,\underline{\rrbrace}
\wedge{}
Q_{\bot b}=\Gfp{{\sqsubseteq}_{\infty}^{\sharp}}{F_{p\bot}^{\sharp}}\wedge{}
P'\in\mathcal{P}\}}\\[-0.25ex]
\rightexplanation{following the same development as for the previous proof of (\ref{eq:logic:abstract:while:upper})}\\
$\Leftrightarrow$
\lastformulaexplanation{
\forall \triple{e:Q'_e}{\bot:Q'_{\bot}}{br:{Q'_{br}}}\in\mathcal{Q}\mathrel{.}
\exists Q_e ,Q_b, Q_{\bot\ell}, Q_{\bot b}, P_e\mathrel{.}
{Q'_e}={Q_e\mathbin{\,\sqcup^{\sharp}_{e}\,}Q_b}
\wedge
{Q'_{\bot}}={Q_{\bot\ell}\mathbin{{\sqcup}_{\infty}^{\sharp}}{}Q_{\bot b}}
\wedge
{Q'_{br}}={P'_{br}}
\wedge
P_e=\Lfp{{\sqsubseteq}_{+}^{\sharp}}{\vec{F}_{pe}^{\sharp}}(P')
\wedge
\underline{\llbrace}\,\{P_e\}\,\underline{\rrbrace}\,\neg\texttt{B}\,\underline{\llbrace}\,\{Q_e\}\,\underline{\rrbrace}
\wedge 
\underline{\llbrace}\,\{P_e\}\,\underline{\rrbrace}\,\texttt{B;S}\,\underline{\llbrace}\,\{Q_b\}\,\underline{\rrbrace}
\wedge{} 
\underline{\llbrace}\,\{P_e\}\,\underline{\rrbrace}\,\texttt{{B;S}}\,\underline{\llbrace}\,\{Q_{\bot\ell}\}\,\underline{\rrbrace}
\wedge{}
Q_{\bot b}=\Gfp{{\sqsubseteq}_{\infty}^{\sharp}}{F_{p\bot}^{\sharp}}
\wedge{}
P'\in\mathcal{P}}{def.\ $\subseteq$}{\mbox{\qed}}
\end{calculus}
\let \qed\relax
\end{proof}
\end{toappendix}

\begin{flushleft}
\bigskip
\hypertarget{PARTII}{\Large\textbf{\textsc{Part II:\ Abstraction of Semantics, Execution Properties, Semantic}}}\\
\phantom{\Large\textbf{\textsc{Part II:\ }}}{\Large\textbf{\textsc{(Hyper) Properties, Calculi, and Logics}}}
\end{flushleft}

Since hyperlogics deal with properties of semantics, there are four levels at which an abstraction can be applied.
\bgroup\abovedisplayskip0pt\belowdisplayskip0pt
\begin{eqntabular}[fl]{Pc{0.925\textwidth}}
\begin{enumerate}[topsep=0pt,partopsep=0pt,parsep=0pt,itemsep=3pt,leftmargin=1.5em,labelwidth=1.5em,label=(\arabic*),ref=(\ref{eq:semantic-abstraction}.\arabic*)]
\item \label{enum:semantic-abstraction} The first level is that of the program semantics considered in appendix sect.\@ \ref{sec:AbstractionAbstractSemantics} and illustrated by the relational semantics in example \ref{ex:relational-semantics} abstracting the trace semantics of sect.\@ \ref{sect:Trace-Semantics}. This abstraction is common in transformational logics \cite{DBLP:journals/pacmpl/Cousot24} such as Hoare logic \cite{DBLP:journals/cacm/Hoare69} but also in hyperlogics \cite{DBLP:journals/afp/Dardinier23a,DBLP:conf/pldi/DardinierM24};
\item \label{enum:property-abstraction}The second level is that of program properties of sect.\@ \ref{sect:Execution-Properties};
\item \label{enum:hyper-property-abstraction}The third level is that of program hyperproperties of sect.\@ \ref{sec:Calculus:Hyper:Properties};
\item \label{enum:hyper-logic-abstraction}The fourth level is that of the abstract logics of sect.\@ \ref{sec:Abstract-Logic-Semantic-Properties}.
\end{enumerate}\label{eq:semantic-abstraction}
\end{eqntabular}
\egroup
Because logics are required to be sound and complete, abstractions should be exact so that any proof of abstract properties in the concrete should be doable in the abstract. This relies on Galois retractions in sect.\@ \ref{sec:Galois:Connections:Retractions}. The main result is that the abstraction of a
logic of semantic (hyper) properties of sect.\@ \ref{sec:Abstract-Logic-Semantic-Properties} is a a
logic of semantic (hyper) properties.

\section{Abstraction of the Abstract Semantics}\label{sec:AbstractionAbstractSemantics}
We show that the abstraction of an instance of the abstract semantics is itself an instance of the abstract semantics.

\begin{definition}[Semantic abstraction]\label{def:exact:abstraction}We say that ${\bar{\mathbb{D}}^{{\sharp}}}\triangleq\pair{\bar{\mathbb{D}}^{\sharp}_{+}}{\bar{\mathbb{D}}^{\sharp}_{\infty}}$ is an exact (respectively approximate) 
abstraction of an abstract domain ${\mathbb{D}^{\sharp}}\triangleq\pair{\mathbb{D}^{\sharp}_{+}}{\mathbb{D}^{\sharp}_{\infty}}$ if and only if
\begin{enumerate}[leftmargin=*,itemsep=3pt,label={\Alph*}.,ref=\ref{def:exact:abstraction}.{\Alph*},labelsep=0.75em]

\item \label{def:abstraction:finite:GC}
There exists a Galois retraction $\pair{\mathbb{L}^{\sharp}_{+}}{\sqsubseteq_{+}^{\sharp}}\galoiS{\alpha_{+}}{\gamma_{+}}\pair{\bar{\mathbb{L}}^{\sharp}_{+}}{\bar{\sqsubseteq}_{+}^{\sharp}}$;

\item  \label{def:abstraction:operators}
$\alpha_{+}({\textsf{init}^{\sharp}})={\overline{\textsf{init}}^{\sharp}}$,
$\alpha_{+}\comp{\textsf{assign}^{\sharp}\sqb{\texttt{x},\texttt{A}}}={\overline{\textsf{assign}}^{\sharp}\sqb{\texttt{x},\texttt{A}}}\comp\alpha_{+}$,
$\alpha_{+}\comp{\textsf{rassign}^{\sharp}\sqb{\texttt{x},a,b}}={\overline{\textsf{rassign}}^{\sharp}\sqb{\texttt{x},a,b}}$ $\comp$ $\alpha_{+}$,
$\alpha_{+}\comp{\textsf{test}^{\sharp}\sqb{\texttt{B}}}={\overline{\textsf{test}}^{\sharp}\sqb{\texttt{B}}}\comp\alpha_{+}$, $\alpha_{+}({\textsf{break}^{\sharp}})={\overline{\textsf{break}}^{\sharp}}$, and
$\alpha_{+}({\textsf{skip}^{\sharp}})={\overline{\textsf{skip}}^{\sharp}}$;

\item \label{def:abstraction:infinite:GC}
There exists a Galois retraction $\pair{\mathbb{L}^{\sharp}_{\infty}}{\sqsupseteq_{\infty}^{\sharp}}\galoiS{\alpha_{{\infty}}}{\gamma_{{\infty}}}\pair{\bar{\mathbb{L}}^{\sharp}_{{\infty}}}{\bar{\sqsupseteq}_{{\infty}}^{\sharp}}$ (i.e. ${\alpha_{{\infty}}}$ preserves existing ${\sqcap_{\infty}^{\sharp}}$);

\item \label{def:abstraction:seq}
For $S\in{\mathbb{L}^{\sharp}_{+}}$,
$\alpha_{+}(S\mathbin{{\fatsemi}^{\sharp}}S')=
\alpha_{+}(S)\mathbin{\bar{\fatsemi}^{\sharp}}\alpha_{+}(S')$ when $S'\in{\mathbb{L}^{\sharp}_{+}}$ and 
$\alpha_{\infty}(S\mathbin{{\fatsemi}^{\sharp}}S')=
\alpha_{\infty}(S)\mathbin{\bar{\fatsemi}^{\sharp}}\alpha_{\infty}(S')$ when $S'\in{\mathbb{L}^{\sharp}_{\infty}}$.
\end{enumerate}
(respectively ``${\bar{\sqsubseteq}_{+}^{\sharp}}$'' or ``${\bar{\sqsupseteq}_{{\infty}}^{\sharp}}$'' instead of ``='' and $\galois{}{}$ instead of $\galoiS{}{}$ for approximate abstractions);
\end{definition}
Following (\ref{eq:def:Abstract-Semantic-Domain-Semantics}), the abstraction of the semantic domain and semantics are
\begin{eqntabular}{rcl}
{\bar{\mathbb{L}}^{\sharp}}&\triangleq&(e:\bar{{\mathbb{L}}}^{\sharp}_{+}\times\bot:\bar{{\mathbb{L}}}^{\sharp}_{\infty}\times{br:\bar{\mathbb{L}}^{\sharp}_{+}})\label{eq:exact:abstraction}\\
\alpha(\triple{e:S_{+}}{\bot:S_{\infty}}{br:S_{br}})&\triangleq&\triple{e:\alpha_{+}(S_{+})}{\bot:\alpha_{\infty}(S_{\infty})}{br:\alpha_{+}(S_{br})}\nonumber
\end{eqntabular}
are well-defined such that 
\bgroup\abovedisplayskip3pt\begin{eqntabular}{rcl}
\pair{\mathbb{L}^{\sharp}}{{\sqsubseteq}^{\sharp}}\galoiS{\alpha}{\gamma}\pair{\bar{\mathbb{L}}^{\sharp}}{\bar{\sqsubseteq}^{\sharp}}.
\label{eq:GC:semantic-abstraction}
\end{eqntabular}\egroup
\begin{lemma}\label{lem:exact:abstraction:well:defined}\proofinapx\quad An exact abstraction ${\bar{\mathbb{D}}^{{\sharp}}}\triangleq\pair{\bar{\mathbb{D}}^{\sharp}_{+}}{\bar{\mathbb{D}}^{\sharp}_{\infty}}$ of a well-defined concrete domain
${\mathbb{D}^{\sharp}}\triangleq\pair{\mathbb{D}^{\sharp}_{+}}{\mathbb{D}^{\sharp}_{\infty}}$ satisfying any one of the
hypotheses \ref{def:abstract:domain:well:def:init:neutral} to \ref{def:abstract:domain:well:def:increasing} to \ref{def:abstract:domain:well:limit-preserving} of definition \ref{def:abstract:domain:well:def} is a well-defined abstract  domain of the same nature.
\end{lemma}
\begin{toappendix}
\begin{proof}[Proof of lemma \ref{lem:exact:abstraction:well:defined}]
Lemma \ref{lem:exact:abstraction:well:defined} follows from the fact that in Galois connections the abstraction preserves existing joins \cite[lemma 11.38]{Cousot-PAI-2021}. and in Galois retractions $\alpha\comp\gamma$ is the identity \cite[exercise 11.50]{Cousot-PAI-2021}.
\end{proof}
\end{toappendix}
\begin{theorem}\label{th:Abstraction:Abstract:Semantics}\proofinapx\quad 
If\/ $\bar{\mathbb{D}}^{\sharp}$ is an exact (respectively approximate)
abstraction of\/ $\mathbb{D}^{\sharp}$ then $\forall \texttt{S}\in\mathbb{S}\mathrel{.}\osqb{\texttt{S}}^\sharp=\alpha(\sqb{\texttt{S}}^{\sharp})$ (respectively ``\,${\bar{\sqsubseteq}^{\sharp}}\!$'' instead of ``='' for approximate abstractions).\ulstrut
\end{theorem}
\begin{toappendix}
\begin{proof}[Proof of theorem \ref{th:Abstraction:Abstract:Semantics}]
The proof of theorem \ref{th:Abstraction:Abstract:Semantics} is an easy generalization of that of theorem 27.8 and corollary 27.20 of \cite{Cousot-PAI-2021}.
\end{proof}
\end{toappendix}
\begin{example}[Relational semantics]\label{ex:relational-semantics} The relational semantics $\sqb{\texttt{S}}^{\varrho}$ of \cite{DBLP:journals/pacmpl/Cousot24} is the following abstraction of the trace semantics $\sqb{\texttt{S}}^{\pi}$.
\begin{eqntabular}{rcl@{\quad\qquad}rcl@{\qquad}}
\alpha_{+}(S)&\triangleq&\{\pair{\sigma}{\sigma'}\mid\exists\pi\mathrel{.}{\sigma}\pi{\sigma'}\in S\cap\Sigma^{+}\}
&
\alpha_{\infty}(S)&\triangleq&\{\pair{\sigma}{\bot}\mid\exists\pi\mathrel{.}{\sigma}\pi\in S\cap\Sigma^{\infty}\}\nonumber
\end{eqntabular}
It follows, by theorem \ref{th:Abstraction:Abstract:Semantics}, that $\forall \texttt{S}\in\mathbb{S}\mathrel{.}\sqb{\texttt{S}}^\varrho=\alpha(\sqb{\texttt{S}}^{\pi})$ and by a classic calculational design, we would get the relational semantics of \cite[sect.\@ I.1]{DBLP:journals/pacmpl/Cousot24} (recalled in sect.\@ \ref{sect:RelationalSemantics} as a specific instance of the algebraic semantics of sect.\@ \ref{sect:Algebraic-Semantics}).
\end{example}

\section{Induced Abstraction of the Execution Transformer}\label{sec:AbstractionExecutionTransformer}
We have defined properties of program executions as
program semantics in ${\mathbb{L}^{\sharp}}$ (\ref{eq:def:Abstract-Semantic-Domain-Semantics}). This formalizes the observation that program semantics specify exactly the properties of all possible executions of any program of the language. An abstraction (\ref{eq:exact:abstraction}) of the semantics in definition \ref{def:exact:abstraction} induces an execution transformer $\overline{\textsf{post}}^\sharp\in{\bar{\mathbb{L}}^{\sharp}}\increasingfunctionto{\bar{\mathbb{L}}^{\sharp}}\increasingfunctionto{\bar{\mathbb{L}}^{\sharp}}$ (\ref{eq:def:abstract:transformer:post}) for this abstract semantics\proofinapx
\begin{eqntabular}{rclclcl}
\bar{\alpha}(\textsf{p})&\triangleq&\rlap{$\LAMBDA{\bar{S}}\LAMBDA{\bar{P}}{\alpha(\textsf{p}(\gamma(\bar{S}))\gamma(\bar{P}))}$}\nonumber
\\
\overline{\textsf{post}}^\sharp(\bar{S})\bar{P}
&\triangleq&
\bar{\alpha}({\textsf{post}}^\sharp)(\bar{S})\bar{P}
&=&
\alpha({\textsf{post}}^\sharp(\gamma(\bar{S}))\gamma(\bar{P}))
&=&
{\bar{P}\mathbin{\bar{\fatsemi}^{\sharp}}\bar{S}}
\label{eq:def:abstract:transformer:post:abstraction}
\end{eqntabular}
\begin{toappendix}
\begin{proof}[Proof of (\ref{eq:def:abstract:transformer:post:abstraction})]
\begin{calculus}[=\ \ ]
\formula{\alpha({\textsf{post}}^\sharp(\gamma(\bar{S}))\gamma(\bar{P}))}\\
=
\formulaexplanation{\alpha(\gamma(\bar{P})\mathbin{\fatsemi^{\sharp}}\gamma(\bar{S}))}{def.\ (\ref{eq:def:abstract:transformer:post}) of ${\textsf{post}}^\sharp$}\\
=
\formulaexplanation{\alpha(\gamma(\bar{P}))\mathbin{\bar{\fatsemi}^{\sharp}}\alpha(\gamma(\bar{S}))}{case \ref{def:abstraction:seq} of definition \ref{def:exact:abstraction} applied component wise}\\
=
\lastformulaexplanation{\bar{P}\mathbin{\bar{\fatsemi}^{\sharp}}\bar{S}}{component wise Galois retraction}{\mbox{\qed}}
\end{calculus}
\let\qed\relax
\end{proof}
\end{toappendix}
Notice that defining 
$\bar{\gamma}(\bar{\textsf{p}})\triangleq\LAMBDA{{S}}\LAMBDA{{P}}{\gamma(\bar{\textsf{p}}(\alpha({S}))\alpha({P}))}$, we have a Galois retraction\proofinapx
\begin{eqntabular}{c}
\pair{{{\mathbb{L}}^{\sharp}}\increasingfunctionto{{\mathbb{L}}^{\sharp}}\increasingfunctionto{{\mathbb{L}}^{\sharp}}}{\mathrel{\ddot{\sqsubseteq}^{\sharp}}}\galoiS{\bar{\alpha}}{\bar{\gamma}}\pair{{\bar{\mathbb{L}}^{\sharp}}\increasingfunctionto{\bar{\mathbb{L}}^{\sharp}}\increasingfunctionto{\bar{\mathbb{L}}^{\sharp}}}{\mathrel{\ddot{\bar{\sqsubseteq}}^{\sharp}}}
\label{eq:GC:post-post-bar}
\end{eqntabular}
\begin{toappendix}
\begin{proof}[Proof of (\ref{eq:GC:post-post-bar})]By \cite[th.\@ 11.78]{Cousot-PAI-2021}, we have a Galois connection. The retraction follows from
\begin{calculus}[=\ \ ]
\formula{\bar{\alpha}(\bar{\gamma}(\bar{\textsf{p}}))}\\
=
\formulaexplanation{\LAMBDA{\bar{S}}\LAMBDA{\bar{P}}{\alpha(\bar{\gamma}(\bar{\textsf{p}})(\gamma(\bar{S}))\gamma(\bar{P}))}}{def.\ (\ref{eq:def:abstract:transformer:post:abstraction}) of $\bar{\alpha}$}\\
=
\formulaexplanation{\LAMBDA{\bar{S}}\LAMBDA{\bar{P}}{\alpha({\gamma(\bar{\textsf{p}}(\alpha(\gamma(\bar{S})))\alpha(\gamma(\bar{P})))})}}{def.\ $\bar{\gamma}(\bar{\textsf{p}})\triangleq\LAMBDA{{S}}\LAMBDA{{P}}{\gamma(\bar{\textsf{p}}(\alpha({S}))\alpha({P}))}$}\\
=
\formulaexplanation{\LAMBDA{\bar{S}}\LAMBDA{\bar{P}}{{\bar{\textsf{p}}(\bar{S})\bar{P}}}}{Galois retraction (\ref{eq:GC:semantic-abstraction}) and \cite[exercise 11.50]{Cousot-PAI-2021}}\\
=
\lastformulaexplanation{\bar{\textsf{p}}}{def.\ lambda-notation and \cite[exercise 11.50]{Cousot-PAI-2021}}{\mbox{\qed}}
\end{calculus}\let\qed\relax
\end{proof}
\end{toappendix}
such that $\overline{\textsf{post}}^\sharp=\bar{\alpha}({\textsf{post}})$ in (\ref{eq:def:abstract:transformer:post:abstraction}). 
Observe that if an abstraction ${\bar{\mathbb{D}}^{{\sharp}}}\triangleq\pair{\bar{\mathbb{D}}^{\sharp}_{+}}{\bar{\mathbb{D}}^{\sharp}_{\infty}}$  of an abstract domain ${\mathbb{D}^{\sharp}}\triangleq\pair{\mathbb{D}^{\sharp}_{+}}{\mathbb{D}^{\sharp}_{\infty}}$ is commuting (\ref{eq:lem:commuting:abstraction}) then\proofinapx
\begin{eqntabular}{rcl}
\alpha({\textsf{post}}^\sharp(\gamma(\bar{S}))P)&=&\overline{\textsf{post}}^\sharp(\bar{S})(\alpha(P))
\label{eq:post:commuting:abstraction}
\end{eqntabular}
\begin{toappendix}
\begin{proof}[Proof of (\ref{eq:post:commuting:abstraction})]
\begin{calculus}[=\ \ ]
\formula{\alpha({\textsf{post}}^\sharp(\gamma(\bar{S}))P)}\\
=
\formulaexplanation{\alpha(P\mathbin{\fatsemi^{\sharp}}\gamma(\bar{S}))}{def.\ (\ref{eq:def:abstract:transformer:post}) of ${\textsf{post}}^\sharp$}\\
=
\formulaexplanation{\alpha(P)\mathbin{\bar{\fatsemi}^{\sharp}}\bar{S}}{commutation (\ref{eq:lem:commuting:abstraction})}\\
=
\lastformulaexplanation{\overline{\textsf{post}}^\sharp(\bar{S})(\alpha(P))}{characterization (\ref{eq:def:abstract:transformer:post:abstraction}) of $\overline{\textsf{post}}^\sharp$}{\mbox{\qed}}
\end{calculus}
\let\qed\relax
\end{proof}
\end{toappendix}
\begin{lemma}[Commutation]\label{lem:commuting:abstraction}\proofinapx\quad
If the abstraction ${\bar{\mathbb{D}}^{{\sharp}}}\triangleq\pair{\bar{\mathbb{D}}^{\sharp}_{+}}{\bar{\mathbb{D}}^{\sharp}_{\infty}}$  of an abstract domain ${\mathbb{D}^{\sharp}}\triangleq\pair{\mathbb{D}^{\sharp}_{+}}{\mathbb{D}^{\sharp}_{\infty}}$ is exact then
\begin{eqntabular}{rcl@{\qquad and \qquad}rcl}
\alpha(P\mathbin{{\fatsemi}^{\sharp}}\gamma(\bar{S}))&=&\alpha(P)\mathbin{\bar{\fatsemi}^{\sharp}}\bar{S}
&
\alpha(\textsf{\textup{post}}(\gamma(\bar{S}))P)
&=&
\overline{\textsf{\textup{post}}}(\bar{S})(\alpha(P))
\label{eq:lem:commuting:abstraction}
\end{eqntabular}
\end{lemma}
Lemma \ref{lem:commuting:abstraction} shows that doing the computation in the concrete and then abstracting is equivalent to doing the computation in the abstract. Relative to the abstraction, no information is lost.
\begin{toappendix}
\begin{proof}[Proof of lemma \ref{lem:commuting:abstraction}]
\begin{calculus}[=\ \ ]
\hyphen{5}\formula{\alpha(P\mathbin{{\fatsemi}^{\sharp}}\gamma(\bar{S}))}\\
=
\formulaexplanation{\alpha(P)\mathbin{\bar{\fatsemi}^{\sharp}}\alpha(\gamma(\bar{S}))}{commutation \ref{def:abstraction:seq}}\\
=
\formulaexplanation{\alpha(P)\mathbin{\bar{\fatsemi}^{\sharp}}\bar{S}}{Galois retraction (\ref{eq:GC:semantic-abstraction}). Q.E.D.}\\[1ex]
\hyphen{5}\formula{\alpha(\textsf{\textup{post}}(\gamma(\bar{S}))P)}\\
=
\formulaexplanation{\alpha(P\mathbin{\fatsemi^{\sharp}}(\gamma(\bar{S})))}{def.\ (\ref{eq:def:abstract:transformer:post}) of $\textsf{post}^\sharp$}\\
=
\formulaexplanation{\alpha(P)\mathbin{\bar{\fatsemi}^{\sharp}}\bar{S}}{as previously shown}\\
=
\lastformulaexplanation{\overline{\textsf{\textup{post}}}(\bar{S})(\alpha(P))}{(\ref{eq:def:abstract:transformer:post:abstraction}}{\mbox{\qed}}
\end{calculus}
\end{proof}
\end{toappendix}
Moreover, instead of deriving the Galois connection (\ref{eq:GC:post-post-bar}) from that  (\ref{eq:GC:semantic-abstraction}), we can start directly from an abstraction of \textsf{post} given by (\ref{eq:GC:post-post-bar}).
The abstract semantics is then $\bar{S}=\overline{\textsf{post}}^\sharp(\bar{S})\textsf{skip}$ proving the equivalence of \ref{enum:semantic-abstraction} and \ref{enum:property-abstraction}.

\section{Induced Abstraction of the Semantic Transformer}\label{sec:AbstractionSemanticTransformer}
The semantics transformer $\overline{\textsf{Post}}^\sharp\in{\bar{\mathbb{L}}^{\sharp}}\functionto\wp(\bar{\mathbb{L}}^{\sharp})\functionto\wp(\bar{\mathbb{L}}^{\sharp})$ for this abstract semantics is\proofinapx
\begin{eqntabular}{rcl}
\bar{\bar{\alpha}}(\textsf{P})
&\triangleq&
\LAMBDA{\bar{S}}\LAMBDA{\bar{\mathcal{P}}}
\{\alpha(R)\mid R\in\textsf{P}(\gamma(\bar{S}))(\{\gamma(\bar{P})\mid\bar{P}\in{\bar{\mathcal{P}}}\})\}
\label{eq:def:alpha-bar-bar}
\\
\overline{\textsf{Post}}^\sharp(\bar{S})\bar{\mathcal{P}}
&\triangleq&
{\bar{\bar{\alpha}}(\textsf{Post}^{\sharp})({\bar{S}}){\bar{\mathcal{P}}}}
\colsep{=}
\{\overline{\textsf{post}}^\sharp(\bar{S})\bar{P}\mid \bar{P}\in\bar{\mathcal{P}}\}
\label{eq:def:Post:abstraction}
\end{eqntabular}
\begin{toappendix}
\begin{proof}[Proof of (\ref{eq:def:Post:abstraction})]
\begin{calculus}[=\ \ ]
\formula{\overline{\textsf{Post}}^\sharp(\bar{S})\bar{\mathcal{P}}}\\
$\triangleq$
\formulaexplanation{\bar{\bar{\alpha}}(\textsf{Post}^{\sharp})({\bar{S}}){\bar{\mathcal{P}}}}{def.\ (\ref{eq:def:Post:abstraction}) of $\overline{\textsf{Post}}^\sharp$}\\
=
\formulaexplanation{\{\alpha(R)\mid R\in\textsf{Post}^{\sharp}(\gamma(\bar{S}))(\{\gamma(\bar{P})\mid\bar{P}\in{\bar{\mathcal{P}}}\})\}}{def.\ (\ref{eq:def:alpha-bar-bar}) of $\bar{\bar{\alpha}}$}\\
=
\formulaexplanation{\{\alpha(R)\mid R\in\{\textsf{post}^\sharp(\gamma(\bar{S}))P\mid P\in(\{\gamma(\bar{P})\mid\bar{P}\in{\bar{\mathcal{P}}}\})\}}{def.\ (\ref{eq:def:Post}) of $\textsf{Post}^\sharp$}\\
=
\formulaexplanation{\{\alpha(\textsf{post}^\sharp(\gamma(\bar{S}))(\gamma(\bar{P})))\mid \bar{P}\in{\bar{\mathcal{P}}}\}}{def.\ $\in$}\\
=
\lastformulaexplanation{\{\overline{\textsf{post}}^\sharp(\bar{S})\bar{P}\mid \bar{P}\in\bar{\mathcal{P}}\}}{(\ref{eq:def:abstract:transformer:post:abstraction})}{\mbox{\qed}}
\end{calculus}
\let\qed\relax
\end{proof}
\end{toappendix}
\begin{example}[Transformers for the relational semantics]\label{ex:relational-semantics-transformers}
For the relational semantics of example \ref{ex:relational-semantics}, the composition is
$S\mathbin{\bar{\fatsemi}^{\varrho}}S'$ $=$ $(S\cap(\Sigma\times\{\bot\}))\cup (S\cap(\Sigma\times\Sigma)\comp S')$ (intuitively \texttt{S$_1$;S$_2$} does not terminate if \texttt{S$_1$} does not terminate or \texttt{S$_1$} terminates but \texttt{S$_2$} doesn't and terminates if both \texttt{S$_1$} and \texttt{S$_2$} terminate with the composition of their effects). Then $\overline{\textsf{Post}}^\varrho\sqb{\texttt{S}}^\varrho\mathcal{P}=\{P\mathbin{\bar{\fatsemi}^{\varrho}}\sqb{\texttt{S}}^\varrho\mid P\in\mathcal{P}\}$ so that if $\mathcal{P}$ is a precondition relating the initial states of the command \texttt{S} to those of the program then $\overline{\textsf{Post}}^\varrho\sqb{\texttt{S}}^\varrho$ relates the final states of the command \texttt{S} or nontermination to the initial states of the program.
\end{example}
We have the Galois retraction\proofinapx
\begin{eqntabular}{c}
\pair{\mathbb{L}^\sharp\functionto\wp(\mathbb{L}^\sharp)\increasingfunctionto\wp(\mathbb{L}^\sharp)}{\ddot{\subseteq}}\galoiS{\bar{\bar{\alpha}}}{\bar{\bar{\gamma}}}\pair{\bar{\mathbb{L}}^\sharp\functionto\wp(\bar{\mathbb{L}}^\sharp)\increasingfunctionto\wp(\bar{\mathbb{L}}^\sharp)}{\ddot{\subseteq}}
\label{eq:GC:alpha-bar-bar}
\end{eqntabular}
\begin{toappendix}
\begin{proof}[Proof of (\ref{eq:GC:alpha-bar-bar})]
$\bar{\bar{\alpha}}$ preserves arbitrary point wise union $\ddot{\cup}$.
\end{proof}
\end{toappendix}
Observe that instead of deriving (\ref{eq:GC:alpha-bar-bar}) from (\ref{eq:GC:post-post-bar}), it is equivalent to start from a Galois retraction (\ref{eq:GC:alpha-bar-bar}) since we can recover \textsf{post} from \textsf{Post}
by  (\ref{eq:Post::post}). 

\section{Induced Abstraction of the Abstract Logics}\label{sec:AbstractionAbstractLogics}
Writing $f(X)\triangleq{\{f(x)\mid x\in X\}}$, the abstract logic $\overline{\textsf{L}}^\sharp\in{{\bar{\mathbb{L}}^{\sharp}}\functionto(\wp({\bar{\mathbb{L}}^{\sharp}})\times\wp({\bar{\mathbb{L}}^{\sharp}}))}$ is
\begin{eqntabular}{rcl}
\bar{\bar{\bar{\alpha}}}(\textsf{L})
&\triangleq&
\LAMBDA{\bar{S}}\{\pair{\bar{\mathcal{P}}}{\bar{\mathcal{Q}}}\mid
\alpha(\bigcap\{\mathcal{Q}\mid\pair{\gamma(\bar{\mathcal{P}})}{{\mathcal{Q}}}\in\textsf{L}(\gamma(\bar{S}))\})\subseteq{\bar{\mathcal{Q}}}\}
\label{eq:def:alpha-bar-bar-bar}
\\
\overline{\overline{\textsf{L}}}^\sharp(\bar{S})
&\triangleq&
{\bar{\bar{\bar{\alpha}}}(\overline{\textsf{L}}^{\sharp})({\bar{S}})}
\qquad\qquad
\overline{\underline{\textsf{L}}}^\sharp(\bar{S})
\colsep{\triangleq}
{\bar{\bar{\bar{\alpha}}}(\underline{\textsf{L}}^{\sharp})({\bar{S}})}
\label{eq:def:upper-logic:abstraction}
\end{eqntabular}
\begin{theorem}\label{th:upper-logic:abstraction}\proofinapx\quad
If\/ $\bar{\mathbb{D}}^{\sharp}$ is an exact abstraction of\/ $\mathbb{D}^{\sharp}$ then\/ $\overline{\overline{\textsf{L}}}^\sharp(\bar{S})
=
\{\pair{\bar{\mathcal{P}}}{\bar{\mathcal{Q}}}\mid\overline{\textsf{Post}}^\sharp(\bar{S})\bar{\mathcal{P}}\subseteq\bar{\mathcal{Q}}\}$ \textup{(}and $\overline{\underline{\textsf{L}}}^\sharp(\bar{S})
=\{\pair{\bar{\mathcal{P}}}{\bar{\mathcal{Q}}}\mid\bar{\mathcal{Q}}\subseteq\overline{\textsf{Post}}^\sharp(\bar{S})\bar{\mathcal{P}}\}$\textup{)}.
\end{theorem}
\begin{toappendix}
\begin{proof}[Proof of theorem \ref{th:upper-logic:abstraction}]
\begin{calculus}[=\ \ ]
\formulaexplanation{\overline{\overline{\textsf{L}}}^\sharp(\bar{S})\colsep{=}{\bar{\bar{\bar{\alpha}}}(\overline{\textsf{L}}^{\sharp})({\bar{S}})}}{(\ref{eq:def:upper-logic:abstraction})}\\
=
\formulaexplanation{\{\pair{\bar{\mathcal{P}}}{\bar{\mathcal{Q}}}\mid\alpha(\bigcap\{\mathcal{Q}\mid\pair{\gamma(\bar{\mathcal{P}})}{{\mathcal{Q}}}\in\textsf{L}(\gamma(\bar{S}))\})\subseteq{\bar{\mathcal{Q}}}\}}{(\ref{eq:def:alpha-bar-bar-bar})}\\
=
\formulaexplanation{\{\pair{\bar{\mathcal{P}}}{\bar{\mathcal{Q}}}\mid\alpha(\bigcap\{\mathcal{Q}\mid\pair{\gamma(\bar{\mathcal{P}})}{{\mathcal{Q}}}\in\{\pair{\mathcal{P}}{\mathcal{Q}}\mid\textsf{Post}^\sharp(\gamma(\bar{S}))\mathcal{P}\subseteq\mathcal{Q}\}\})\subseteq{\bar{\mathcal{Q}}}\}}{def.\ (\ref{eq:def:upper-logic}) of $\overline{\textsf{L}}^\sharp$}\\
=
\formulaexplanation{\{\pair{\bar{\mathcal{P}}}{\bar{\mathcal{Q}}}\mid\alpha(\bigcap\{\mathcal{Q}\mid\textsf{Post}^\sharp(\gamma(\bar{S}))\gamma(\bar{\mathcal{P}})\subseteq\mathcal{Q}\})\subseteq{\bar{\mathcal{Q}}}\}}{def.\ $\in$}\\
=
\formulaexplanation{\{\pair{\bar{\mathcal{P}}}{\bar{\mathcal{Q}}}\mid\alpha(\{\textsf{Post}^\sharp(\gamma(\bar{S}))\gamma(\bar{\mathcal{P}})\})\subseteq{\bar{\mathcal{Q}}}\}}{def.\ $\bigcap$ and $\subseteq$}\\
=
\formulaexplanation{\{\pair{\bar{\mathcal{P}}}{\bar{\mathcal{Q}}}\mid\alpha(\{\textsf{post}^\sharp(\gamma(\bar{S}))P\mid P\in\gamma(\bar{\mathcal{P}})\})\subseteq{\bar{\mathcal{Q}}}\}}{def.\ (\ref{eq:def:Post}) of $\textsf{Post}^\sharp$}\\
=
\formulaexplanation{\{\pair{\bar{\mathcal{P}}}{\bar{\mathcal{Q}}}\mid\{\alpha(\textsf{post}^\sharp(\gamma(\bar{S}))P)\mid P\in\gamma(\bar{\mathcal{P}})\}\subseteq{\bar{\mathcal{Q}}}\}}{def.\ image}\\
=
\formulaexplanation{\{\pair{\bar{\mathcal{P}}}{\bar{\mathcal{Q}}}\mid\forall\bar{P}\in\bar{\mathcal{P}}\mathrel{.}
\alpha({\textsf{post}}^\sharp(\gamma(\bar{S}))\gamma(\bar{P}))\in\bar{\mathcal{Q}}\}}{def.\ $\subseteq$}\\
=
\formulaexplanation{\{\pair{\bar{\mathcal{P}}}{\bar{\mathcal{Q}}}\mid\forall\bar{P}\in\bar{\mathcal{P}}\mathrel{.}
\overline{\textsf{post}}^\sharp(\bar{S})\bar{P}\in\bar{\mathcal{Q}}\}}{def.\ (\ref{eq:def:abstract:transformer:post:abstraction}) of $\overline{\textsf{post}}^\sharp$}\\
=
\formulaexplanation{\{\pair{\bar{\mathcal{P}}}{\bar{\mathcal{Q}}}\mid\{
\overline{\textsf{post}}^\sharp(\bar{S})\bar{P}\mid\bar{P}\in\bar{\mathcal{P}}\}\subseteq\bar{\mathcal{Q}}\}}{def.\ $\subseteq$}\\
=
\formulaexplanation{\{\pair{\bar{\mathcal{P}}}{\bar{\mathcal{Q}}}\mid 
\overline{\textsf{Post}}^\sharp(\bar{S})\bar{\mathcal{P}}\subseteq\bar{\mathcal{Q}}\}}{def.\ (\ref{eq:def:Post:abstraction}) of $\overline{\textsf{Post}}^\sharp$}
\end{calculus}
{\tiny\ustrut} The proof for $\overline{\underline{\textsf{L}}}^\sharp$ is $\subseteq$-dual.
\end{proof}
\end{toappendix}
It follows from theorem \ref{th:upper-logic:abstraction} that the logic proof system of theorem
\ref{th:upper-abstract-logic-proof-system} is applicable to the upper abstract logic $\overline{\overline{\textsf{L}}}^\sharp(\bar{S})$ (and dually theorem \ref{th:lower-abstract-logic-proof-system} for the lower abstract logic).

In conclusion of this \hyperlink{PARTII}{part II}, although the abstractions of the semantics, \textsf{post}, \textsf{Post}, and logics have been shown to be equally expressible for exact abstractions, they do not really solve the problem of the complexity of the resulting logic (although hyperproperties may be simpler). The logics  still have to handle exactly the (abstract) semantics occurring in the (hyper) properties. So our proposed proof system has rules (\ref{eq:logic:abstract:assignment})---(\ref{eq:logic:abstract:while:upper}) plus simplified rules applicable to less general classes of properties defined by the abstractions studied in the following \hyperlink{PARTIII}{part III}.

\newcommand{\uphyperlogic}[3] {{\overline{\llbrace}\,#1\,\overline{\rrbrace}\,\texttt{#2}\,\overline{\llbrace}\,#3\,\overline{\rrbrace}}}%
\newcommand{\uphoarelogic}[3] {{\overline{\{}\,#1\,\overline{\}}\,\texttt{#2}\,\overline{\{}\,#3\,\overline{\}}}}%
\newcommand{\downhoarelogic}[3] {{\overline{\{}\,#1\,\underline{\}}\,\texttt{#2}\,\underline{\{}\,#3\,\underline{\}}}}%

\begin{flushleft}
\bigskip
\hypertarget{PARTIII}{\Large\textbf{\textsc{Part III:\ Abstractions for Semantic (Hyper) Logics}}}
\end{flushleft}

The problem with (hyper) logics studied in \hyperlink{PARTI}{part I} (and their abstractions in \hyperlink{PARTII}{part II}) is that for a program to satisfy a semantic (hyper) property, its semantics must exactly occur in this (hyper) property and therefore the proof must exactly characterize the program semantics. So, contrary to Hoare logic or its dual, (hyper) proof rules cannot make over or under approximations of the program semantics in semantic properties. In this \hyperlink{PARTIII}{part III}, we study abstractions of semantic properties that yield simpler sound and complete proof rules for the less general semantic (hyper) properties defined by the abstraction. Such abstractions can also provide representations of abstract semantic  (hyper)  properties\footnote{Another example is the possible representation of semantic properties satisfying the decreasing chain condition by join irreducibles \cite[theorem 4.8]{Blyth-Lattices-2005}.}.

\section{Semantic to Execution Property Abstraction}\label{sec:SemanticExecutionPropertyAbstraction}
\subsection{Join Abstraction}
\ifshort The join abstraction $\alpha_{{\cup}}(\mathcal{P}) \triangleq{\bigcup}\mathcal{P}$ is classic to 
abstract set-based semantics (hyper) properties $\mathcal{P}$ into execution properties $\alpha_{{\cup}}(\mathcal{P})$
\cite[section 8.6]{Cousot-PAI-2021}. It is relegated to the appendix \proofinapx.\fi
\begin{toappendix}
\subsubsection{Definition of the Join Abstraction}
In a complete lattice, the abstraction $\alpha_{{\sqcup}}(\mathcal{P}) \triangleq{\bigsqcup}\mathcal{P}$ and 
$\gamma_{{\sqcup}}(Q)\triangleq \{P\mid P \sqsubseteq Q\}$ yields a Galois retraction.
\bgroup\abovedisplayskip3pt\belowdisplayskip3pt
\begin{eqntabular}{c@{\qquad and so\qquad}c}
\pair{\wp({\mathbb{L}})}{\subseteq}\galoiS{\alpha_{{\sqcup}}}{\gamma_{{\sqcup}}}\pair{{\mathbb{L}}}{{\sqsubseteq}}
\label{eq:def:Abstract:Execution:Properties}
&
\pair{\wp({\mathbb{L}})}{\subseteq}\galoiS{\gamma_{{\sqcup}}\;\comp\;\alpha_{{\sqcup}}}{\mathbb{1}}\pair{{\gamma_{{\sqcup}}\comp\alpha_{{\sqcup}}}(\wp({\mathbb{L}}))}{\subseteq}
\end{eqntabular}\egroup
\begin{proof}[Proof of (\ref{eq:def:Abstract:Execution:Properties})]
\begin{calculus}[$\Leftrightarrow$\ \ ]
\formula{\alpha_{{\sqcup}}(\mathcal{P}) \mathrel{{\sqsubseteq}} Q}\\
$\Leftrightarrow$
\formulaexplanation{\mathop{{\bigsqcup}}\mathcal{P}\mathrel{{\sqsubseteq}} Q}{def.\ $\alpha_{{\sqcup}}$}\\
$\Leftrightarrow$
\formulaexplanation{\forall P\in\mathcal{P}\mathrel{.} P\mathrel{{\sqsubseteq}} Q}{def.\ least upper bound}\\
$\Leftrightarrow$
\formulaexplanation{\mathcal{P}\subseteq\{P\mid P\mathrel{{\sqsubseteq}} Q\}}{def.\ $\subseteq$}\\
$\Leftrightarrow$
\formulaexplanation{\mathcal{P}\subseteq\gamma_{{\sqcup}}(Q)}{def.\ $\gamma_{{\sqcup}}$}
\end{calculus}
{\tiny\ustrut}It follows that ${\gamma_{{\sqcup}}\comp\alpha_{{\sqcup}}}$ is an upper closure operator hence the second Galois retraction.
\end{proof}
The properties in ${{\gamma_{{\sqcup}}\comp\alpha_{{\sqcup}}}(\wp({\mathbb{L}}))}$ are called execution properties as opposed to semantic (hyper) properties in $\wp({\mathbb{L}})$.
If the abstract domains ${\mathbb{D}^{\sharp}}$ of definition \ref{def:abstract:domain:well:def} or
their abstractions by definition \ref{def:exact:abstraction} are complete lattices, this abstraction approximates 
abstracts semantic properties in $\wp({\mathbb{L}^{\sharp}})$ into executions in ${\mathbb{L}^{\sharp}}$. 
\begin{example}[Trace property abstraction]The trace hyperproperties in $\wp(\wp(\Sigma^{+\infty}))$ can be abstracted to trace properties in $\wp(\Sigma^{+\infty})$ by $\uLstrut\pair{\wp(\wp(\Sigma^{+\infty})}{\subseteq}\galoiS{\alpha_{\cup}}{\gamma_{\cup}}\pair{\wp(\Sigma^{+\infty}))}{\subseteq}$ with $\alpha_{\cup}(P)=\bigcup P$ and $\gamma_{\cup}(Q)=\wp(Q)$ as done e.g.\ in \cite[section 5, p.\ 246]{DBLP:conf/popl/CousotC12} which is the starting point of \cite{DBLP:journals/pacmpl/Cousot24} to recover Hoare logic and its variants. ${\gamma_{\cup}}(P)$ is called the lift of trace property $P\in\wp(\Sigma^{+\infty})$ in \cite[page 1162]{DBLP:journals/jcs/ClarksonS10}.
\end{example}
\begin{example}[Hyperlogic to execution logic abstraction]
Applied to $\textsf{Post}^\sharp(S)$ in (\ref{eq:def:Post}) this join abstraction yields
$\textsf{post}^\sharp(S)$ in (\ref{eq:def:abstract:transformer:post}), so that the
hyperproperty calculus of theorem \ref{th:upper-abstract-hyper-properties} is 
abstracted into the execution property calculus of theorem \ref{th:Program:execution:properties:calculus} 
and therefore the hyperlogic of theorem \ref{th:upper-abstract-logic-proof-system} is abstracted in the 
classic program logic of execution properties (as considered in \cite{DBLP:journals/pacmpl/Cousot24}, after appropriate  generalization to the algebraic semantics of section \ref{sect:Algebraic-Semantics}).
\end{example}

\subsubsection{Proof Rule Simplification}By correspondence (\ref{eq:def:Abstract:Execution:Properties}), the abstract logical ordering (abstracting the implication $\subseteq$) is also the computational ordering in lemma \ref{lem:abstract-domain-sharp} whereas, in general, for the generic algebraic abstract semantics the computational ordering ${{\sqsubseteq}^{\sharp}}$ and the logical ordering and $\subseteq$ are not directly related, which is at the origin of complications in proofs. 
Therefore, the \texttt{while} rule (\ref{eq:logic:abstract:while:upper}) can be simplified since fixpoints can be over approximated (or under approximated) hence handled by fixpoint induction such as Park induction \cite[theorem II.3.1]{DBLP:journals/pacmpl/Cousot24}
or  Scott-Kleene induction \cite[theorem II.3.6]{DBLP:journals/pacmpl/Cousot24}.
\end{toappendix}

\section{Homomorphic Semantic Abstraction}\label{sec:HomomorphicSemanticAbstraction}
\ifshort The homomorphic abstraction $\alpha(S)\triangleq\{h(x)\mid x\in S\}$ is also well known \cite[exercise 11.6]{DBLP:journals/pacmpl/Cousot24} and can be used e.g.\ to define partial hypercorrectness, trace safety hyperproperties, etc.\ \proofinapx.\fi
\begin{toappendix}
Given an execution property abstraction $\alpha\in{\mathbb{L}^{\sharp}}\functionto{\mathbb{A}}$, it can be extended elementwise
to $\pair{\wp(\mathbb{L}^{\sharp})}{{\subseteq}}\galois{\dot{\alpha}}{\dot{\gamma}}\pair{\wp(\mathbb{A})}{\subseteq}$ by $\dot{\alpha}(\mathcal{P})\triangleq\{\alpha(P)\mid P\in\mathcal{P}\}$ and
 $\dot{\gamma}(\mathcal{Q})\triangleq\{P\mid \alpha(P)\in\mathcal{Q}\}$.

\begin{example}[Partial hypercorrectness]\label{ex:Partial-hyper-correctness}Partial hypercorrectness consists in ignoring one component $\mathbb{D}^{\sharp}_{+}$ or $\mathbb{D}^{\sharp}_{\infty}$ of the abstract domain and preserving only the other, that is $\alpha^+(\triple{e:P_{+}}{\bot:P_{\infty}}{{br:P_{b}}})\triangleq \pair{ok:P_{+}}{{br:P_{b}}}$ or $\alpha^\infty(\triple{e:P_{+}}{\bot:P_{\infty}}{{br:P_{b}}})\triangleq P_{\infty}$ in (\ref{eq:def:Abstract-Semantic-Domain-Semantics}). This execution property abstraction $\alpha$ is extended to semantic properties by the homomorphic abstraction $\alpha(\mathcal{P})\triangleq\{\alpha(P)\mid P\in \mathcal{P}\}$. This yields a Galois retraction $\pair{\wp(\mathbb{L}^{\sharp})}{\subseteq}\galoiS{\alpha}{\gamma}\pair{\alpha(\wp(\mathbb{L}^{\sharp}))}{\subseteq}$ hence a closure $\pair{\wp(\mathbb{L}^{\sharp})}{\subseteq}\galoiS{\gamma\;\comp\;\alpha}{\mathbb{1}}\pair{\gamma\comp\alpha(\wp(\mathbb{L}^{\sharp}))}{\subseteq}$. This is an extension of partial correctness or termination to semantic (hyper) properties. The \texttt{while} rule
(\ref{eq:logic:abstract:while:upper}) can be simplified by ignoring  one of the two fixpoints. However, the other fixpoint must still be calculated exactly. 
\end{example}
\begin{example}[Trace safety hyperproperties]\label{ex:Trace-Safety-semantic-properties}The safety abstraction $\alpha$ by prefix and limit abstraction of trace properties \cite[section 6.1]{DBLP:conf/popl/CousotC12} can be applied to the trace semantic (hyper) properties of section 
\ref{sect:Trace-Semantics} so that $\dot{\alpha}(\wp(ok:\wp(\Sigma^{+\infty})\times br:\wp(\Sigma^{+})))$ yields safety semantic (hyper) properties of \cite{DBLP:journals/jcs/ClarksonS10}. This consists in replacing 
each semantics in the semantic property by its safety approximation by prefixes (in $\mathbb{L}^{\sharp}_{+}$) and limits (in $\mathbb{L}^{\sharp}_{\infty}$).
\end{example}
\begin{example}[Algebraic safety hyperproperties]\label{ex:algebraic-Safety-semantic-properties}The trace safety hyperproperties of example \ref{ex:Trace-Safety-semantic-properties} can be generalized to the algebraic semantics by requiring that, under the hypotheses of lemmas \ref{lem:abstract:Lfp-subseteq-F-e} and \ref{lem:Fbotsharp-welldefined}, algebraic safety properties $\mathcal{P}$ do satisfy that $\sqb{\texttt{S$_1$;S$_2$}}^{\sharp}\in \mathcal{P}$ implies $\sqb{\texttt{S$_1$}}_{e}^{\sharp}\cup\sqb{\texttt{S$_1$}}_{br}^{\sharp}\in \mathcal{P}$ (prefix closure), that $\forall \delta\in\mathbb{O}\mathrel{.}(\sqb{\texttt{B;S}}_{e}^{\sharp})^{\delta}\in \mathcal{P}$ implies $\Lfp{\sqsubseteq_{+}^{\sharp}}{\Cev{F}_{e}^{\sharp}}\in \mathcal{P}$ (limit closure for finite executions), and that $\forall \delta\in\mathbb{O}\mathrel{.}((\sqb{\texttt{B;S}}_{e}^{\sharp})^{\delta} \mathbin{{\fatsemi}^{\sharp}} {\top_{\infty}^{\sharp}})\in \mathcal{P}$ implies 
$\Gfp{\sqsubseteq_{\infty}^{\sharp}}{F_{\bot}^{\sharp}}\in \mathcal{P}$ (limit closure for infinite executions). Then  $\dot{\alpha}_{\textsf{safety}}(\mathcal{P})\triangleq\{\alpha_{\textsf{safety}}(P)\mid P\in\mathcal{P}\}$ thus generalizing the classic definition of safety property $\alpha_{\textsf{safety}}(P) = P$ as prefix closed and limit closed sets of traces  \cite[Definition  14.11]{DBLP:journals/pacmpl/Cousot24}. Then the proof rule (\ref{eq:logic:abstract:while:upper}) can be simplified since the passage to the limit need not be checked since it is guaranteed by the safety hypothesis.
\end{example}
\end{toappendix}

\section{Execution Property Elimination}\label{sec:ExecutionPropertyElimination}
Given a set $\mathbb{I}\in\wp(\wp(\mathbb{L}^{\sharp}))$ of semantic properties of interest, the Galois retraction $$\pair{\wp(\mathbb{L}^{\sharp})}{\subseteq}\galoiS{\LAMBDA{\mathcal{P}}\mathcal{P}\,\cap\,\mathbb{I}}{\LAMBDA{\mathcal{Q}}\mathcal{Q}\,\cup\,\mathbb{I}}\pair{\mathbb{I}}{\subseteq}$$ \cite[exercise 11.5]{Cousot-PAI-2021} eliminates the semantics of no interest. \ifshort It can be used e.g.\ to handle $k$-semantic properties \proofinapx.\fi
\begin{toappendix}
 We have used this abstraction ${\LAMBDA{\mathcal{P}}\mathcal{P}\,\cap\,\mathbb{I}}$ implicitly in examples \ref{ex:powerset-deterministic-domain-post} and 
\ref{ex:powerset-deterministic-domain-post-continued} when saying that we ignored nontermination. The logics of section \ref{sec:Upper-Lower-Abstract-Logics} are simplified by intersection with $\mathbb{I}$ but this still requires the restricted fixpoints in the \texttt{while} rules (\ref{eq:logic:abstract:while:upper}) and (\ref{eq:logic:abstract:while:lower}) to be computed exactly, which, mechanically, does not scale up.
\begin{example}[$k$-semantic properties]If $\mathbb{L}=\wp(L)$ is a powerset (which is the case for the trace semantics of section \ref{sec:Biinductive-Trace-Semantic-Domain}), $\mathbb{I} \triangleq\{\mathcal{P}\in\wp(\wp(L))\mid |\mathcal{P}|\leqslant k\}$, $k\geqslant 1$, where $|S|$ is the cardinality of set $S$, restricts the trace properties to be considered in the semantic properties to those of cardinality at most $k$. An instance of this abstraction is the $k$-hypersafety of 
\cite[page 1170]{DBLP:journals/jcs/ClarksonS10}.
\end{example}
\end{toappendix}

\section{Principal Order Ideal Abstraction}\label{sec:Principal-Ideal-Abstraction}
\subsection{Definition of the Principal Order Ideal Abstraction}
 Subject to the existence of the least upper bound, the principal ideal abstraction is
\bgroup\abovedisplayskip3pt\belowdisplayskip1pt
\begin{eqntabular}{rcl}
\alpha^{\curlywedgedownarrow}(\mathcal{P}) &\triangleq&\{P\mid P\sqsubseteq \bigsqcup\mathcal{P}\}
\label{eq:def:alpha:principal}
\end{eqntabular}\egroup
\begin{lemma}\label{lem:alpha:principal}\proofinapx\quad
$\alpha^{\curlywedgedownarrow}$ is an upper closure operator and $\sextuple{\alpha^{\curlywedgedownarrow}(\wp(\mathbb{L}))}{\subseteq}{\{\bot\}}{\mathbb{L}}{\LAMBDA{X}\alpha^{\curlywedgedownarrow}(\cup X)}{\cap}$  is a complete lattice. 
\end{lemma}
\begin{toappendix}
\begin{proof}[Proof of lemma \ref{lem:alpha:principal}]
By definition, $\alpha^{\curlywedgedownarrow}$ is increasing and extensive. For idempotence, we have
\begin{calculus}[=\ \ ]
\formula{\alpha^{\curlywedgedownarrow}(\alpha^{\curlywedgedownarrow}(\mathcal{P}))}\\
=
\formulaexplanation{\alpha^{\curlywedgedownarrow}(\{P\mid P\sqsubseteq \bigsqcup\mathcal{P}\})}{def.\ (\ref{eq:def:alpha:principal}) of $\alpha^{\curlywedgedownarrow}$}\\
=
\formulaexplanation{\{P\mid P\sqsubseteq \bigsqcup\{P'\mid P'\sqsubseteq \bigsqcup\mathcal{P}\}\}}{def.\ (\ref{eq:def:alpha:principal}) of $\alpha^{\curlywedgedownarrow}$}\\
=
\formulaexplanation{\{P\mid P\sqsubseteq \bigsqcup\mathcal{P}\}}{def.\ lub $\bigsqcup$}\\
=
\formulaexplanation{\alpha^{\curlywedgedownarrow}(\mathcal{P})}{def.\ (\ref{eq:def:alpha:principal}) of $\alpha^{\curlywedgedownarrow}$}
\end{calculus}

\smallskip

\noindent By Morgan Ward's \cite[theorem 4.1]{Ward42},  $\sextuple{\alpha^{\curlywedgedownarrow}(\wp(\mathbb{L}))}{\subseteq}{\{\bot\}}{\mathbb{L}}{\LAMBDA{X}\alpha^{\curlywedgedownarrow}(\cup X)}{\cap}$  is a complete lattice. 
\end{proof}
\end{toappendix}

\subsection{Proof Rule Simplification}If $\pair{\mathbb{L}}{\sqsubseteq}$ is a complete lattice and the composition  preserves arbitrary existing limits in definition \ref{def:abstract:domain:well:def:join:additive}  then proofs in the upper abstract semantic logic can be based on the classic upper abstract execution property logic of section \ref{sec:Algebraic-Logics-Program-Execution-Properties} for principal ideal closed properties and their dual~\proofinapx.
\begin{eqntabular}{c@{\qquad\quad}c}
\frac{\:{\overline{\{}\bigsqcup\mathcal{P}\overline{\}}\,\texttt{S}\,\overline{\{}\bigsqcup\mathcal{Q}\overline{\}}}\:}{{\overline{\llbrace}\,\mathcal{P}\,\overline{\rrbrace}\,\texttt{S}\,\overline{\llbrace}\,\mathcal{Q}\,\overline{\rrbrace}}}\raisebox{0.5ex}{,}\quad \alpha^{\curlywedgedownarrow}(\mathcal{Q})=\mathcal{Q}
&
\frac{\:\forall P\in\mathcal{P}\mathrel{.}{\underline{\{}P\underline{\}}\,\texttt{S}\,\underline{\{}\bigsqcap\mathcal{Q}\underline{\}}}\:}{{\overline{\llbrace}\,\mathcal{P}\,\overline{\rrbrace}\,\texttt{S}\,\overline{\llbrace}\,\mathcal{Q}\,\overline{\rrbrace}}}\raisebox{0.5ex}{,}\quad \alpha^{\curlyveeuparrow}(\mathcal{Q})=\mathcal{Q}
\label{eq:alpha:principal:rule}
\end{eqntabular}
\begin{toappendix}
\begin{proof}[Soundness and completeness proof of rule (\ref{eq:alpha:principal:rule})]
\begin{calculus}[$\Leftrightarrow$\ \ ]
\hyphen{5}\formula{{\overline{\llbrace}\,\mathcal{P}\,\overline{\rrbrace}\,\texttt{S}\,\overline{\llbrace}\,\mathcal{Q}\,\overline{\rrbrace}}}\\
$\Leftrightarrow$
\formulaexplanation{\textsf{Post}^\sharp(S)\mathcal{P}\subseteq\mathcal{Q}}{def.\ (\ref{eq:def:upper-logic}) of $\overline{\llbrace}\,\mathcal{P}\,\overline{\rrbrace}\,\texttt{S}\,\overline{\llbrace}\,\mathcal{Q}\,\overline{\rrbrace}$}\\
$\Leftrightarrow$
\formulaexplanation{\{\textsf{post}^\sharp(S)P\mid P\in\mathcal{P}\}\subseteq\mathcal{Q}}{def.\ (\ref{eq:def:Post}) of $\textsf{Post}^\sharp$}\\
$\Leftrightarrow$
\formulaexplanation{\{\textsf{post}^\sharp(S)P\mid P\in\mathcal{P}\}\subseteq\{P'\mid P'\sqsubseteq \bigsqcup\mathcal{Q}\}}{hypothesis $\alpha^{\curlywedgedownarrow}(\mathcal{Q})=\mathcal{Q}$ and def.\ (\ref{eq:def:alpha:principal}) of $\alpha^{\curlywedgedownarrow}$}\\
$\Leftrightarrow$
\formulaexplanation{\forall  P\in\mathcal{P}\mathrel{.} \textsf{post}^\sharp(S)P\sqsubseteq \bigsqcup\mathcal{Q}}{def.\ $\subseteq$}\\
$\Leftrightarrow$
\formulaexplanation{\bigsqcup_{P\in\mathcal{P}}\textsf{post}^\sharp(S)P\sqsubseteq \bigsqcup\mathcal{Q}}{def.\ lub $\sqcup$}\\
$\Leftrightarrow$
\formula{\textsf{post}^\sharp(S)\bigl(\bigsqcup_{P\in\mathcal{P}}P\bigr)\sqsubseteq \bigsqcup\mathcal{Q}}\\
\explanation{by hypothesis,  the composition  preserves arbitrary existing limits in definition \ref{def:abstract:domain:well:def:join:additive} and (\ref{eq:GC:post-S})}\\
$\Leftrightarrow$
\formulaexplanation{\textsf{post}^\sharp\sqb{\texttt{S}}^\sharp(\bigsqcup\mathcal{P})\mathrel{{\sqsubseteq^{\sharp}}} \bigsqcup\mathcal{Q}}{def.\ $\bigsqcup$}\\
=
\formulaexplanation{\overline{\{}\bigsqcup\mathcal{P}\overline{\}}\,\texttt{S}\,\overline{\{}\bigsqcup\mathcal{Q}\overline{\}}}{def.\ $\overline{\{}\,P\,\overline{\}}\,\texttt{S}\,\overline{\{}\,Q\,\overline{\}}$ in section \ref{sec:Algebraic-Logics-Program-Execution-Properties}}\\[1em]

\hyphen{5}\formula{{\overline{\llbrace}\,\mathcal{P}\,\overline{\rrbrace}\,\texttt{S}\,\overline{\llbrace}\,\mathcal{Q}\,\overline{\rrbrace}}}\\
$\Leftrightarrow$
\formulaexplanation{\textsf{Post}^\sharp(S)\mathcal{P}\subseteq\mathcal{Q}}{def.\ (\ref{eq:def:upper-logic}) of $\overline{\llbrace}\,\mathcal{P}\,\overline{\rrbrace}\,\texttt{S}\,\overline{\llbrace}\,\mathcal{Q}\,\overline{\rrbrace}$}\\
$\Leftrightarrow$
\formulaexplanation{\{\textsf{post}^\sharp(S)P\mid P\in\mathcal{P}\}\subseteq\mathcal{Q}}{def.\ (\ref{eq:def:Post}) of $\textsf{Post}^\sharp$}\\
$\Leftrightarrow$
\formulaexplanation{\{\textsf{post}^\sharp(S)P\mid P\in\mathcal{P}\}\subseteq\{P'\mid\bigsqcap\mathcal{Q}\sqsubseteq P'\}}{hypothesis $\alpha^{\curlyveeuparrow}(\mathcal{Q})=\mathcal{Q}$ and dual def.\ (\ref{eq:def:alpha:principal}) of $\alpha^{\curlyveeuparrow}$}\\
$\Leftrightarrow$
\formulaexplanation{\forall P\in\mathcal{P}\mathrel{.}\textsf{post}^\sharp(S)P \in\{P'\mid\bigsqcap\mathcal{Q}\sqsubseteq P'\}}{def.\ $\subseteq$}\\
$\Leftrightarrow$
\formulaexplanation{\forall P\in\mathcal{P}\mathrel{.}\bigsqcap\mathcal{Q}\sqsubseteq \textsf{post}^\sharp(S)P}{def.\ $\in$}\\
$\Leftrightarrow$
\lastformulaexplanation{\forall P\in\mathcal{P}\mathrel{.}\underline{\{} P\underline{\}}\,\texttt{S}\,\underline{\{}\bigsqcap\mathcal{Q}\underline{\}}}{def.\ $\underline{\{}\,P\,\underline{\}}\,\texttt{S}\,\underline{\{}\,Q\,\underline{\}}$ in section \ref{sec:Algebraic-Logics-Program-Execution-Properties}}{\mbox{\qed}}
\end{calculus}
\let\qed\relax
\end{proof}
\end{toappendix}

\begin{example}[Proof reduction for principal ideal hyperproperties] Consider the instantiation for the natural relational semantics in section $\ref{sect:RelationalSemantics}$ with no break. Define the assertional execution  postcondition $Q_1\triangleq\{\sigma\in\Sigma \mid \sigma(x) \leq 10\}$ with
relational equivalent $Q_2\triangleq\Sigma \times Q_1$ and hyperproperty 
$\mathcal{Q} \triangleq \alpha^{\curlywedgedownarrow} (Q_2)=\alpha^{\curlywedgedownarrow} (\Sigma \times \{\sigma\in\Sigma \mid \sigma(x) \leq 10\})$ and similarly $\mathcal{P} \triangleq \{ (\Sigma \times \{\sigma\in\Sigma \mid\ \sigma(x) =n  \}) \mid  n\in\mathbb{N}\wedge n > 10 \}$. To prove the following hyperlogic triple $\uphyperlogic{\mathcal{P}}{while(x>10) {x=x-1}}{\mathcal{Q}}$, it is equivalent to prove the following. 
\begin{calculus}[$\Leftrightarrow$ ]
    \phantom{$\Leftrightarrow$ } \formula{\uphyperlogic{\mathcal{P}}{while(x>10) {x=x-1}}{\mathcal{Q}}} \\
    $\Leftrightarrow$ \formulaexplanation{\uphoarelogic{\bigcup\mathcal{P}}{while(x>10) {x=x-1}}{\bigcup\mathcal{Q}}}{By rule of (${\ref{eq:alpha:principal:rule}}$)}\\
    $\Leftrightarrow$ \formula{\uphoarelogic{\Sigma \times \{\sigma\in \Sigma \mid \sigma(x) > 10\}  }{while(x>10) {x=x-1}}{\Sigma \times \{\sigma\in\Sigma \mid \sigma(x) \leq 10\}}}
\end{calculus}

\smallskip

\noindent Then one can use the over-approximation logic with termination proof in \cite{DBLP:journals/pacmpl/Cousot24Zenodo}.
\end{example}

\section{Order Ideal Abstraction}\label{sec:Order-Ideal-Abstraction}
\subsection{Definition of the Order Ideal Abstraction}
The order ideal abstraction on $\pair{\wp(\mathbb{L})}{\subseteq}$ is
\bgroup\abovedisplayskip4pt\belowdisplayskip2pt
\begin{eqntabular}{rcl@{\qquad}c}
\alpha^{\sqsubseteq}(\mathcal{P}) 
&\triangleq&
\{P'\in \mathbb{L}\mid \exists P\in\mathcal{P}\mathrel{.}P'\sqsubseteq P\}
&
\pair{\wp(\mathbb{L})}{\subseteq}\galoiS{\alpha^{\sqsubseteq}}{\mathbb{1}}\pair{\alpha^{\sqsubseteq}(\wp(\mathbb{L}))}{\subseteq}
\label{eq:GC:alpha:sqsubseteq}
\end{eqntabular}\egroup
$\alpha^{\sqsubseteq}$ is an upper closure operator and $\sextuple{\alpha^{\sqsubseteq}(\wp(\mathbb{L}))}{\subseteq}{\emptyset}{\mathbb{L}}{\LAMBDA{X}\alpha^{\sqsubseteq}(\cup X)}{\cap}$  is a complete lattice \cite[theorem 4.1]{Ward42}. The order filter abstraction $\alpha^{\sqsupseteq}$ is defined dually. Note that $\alpha^{\curlywedgedownarrow}(\mathcal{P}) =\alpha^{\sqsubseteq}(\{\bigsqcup\mathcal{P}\})$.
As  observed by \cite[page 239]{DBLP:conf/sas/MastroeniP17} for subset-closed hyperproperties, all execution properties are order-ideal closed for trace properties (where $\sqsubseteq$ is $\subseteq$), but not conversely, citing observational determinism \cite{DBLP:conf/csfw/ZdancewicM03} as a counterexample.

\subsection{Proof Rule Simplification}The main interest of the order ideal/filter abstraction is the substantial simplification of the \texttt{while} rules
(\ref{eq:logic:abstract:while:upper}) and (\ref{eq:logic:abstract:while:lower}). To show this consider
properties in $\alpha^{{\sqsupseteq}^{\sharp}}(\wp(\mathbb{L}^{\sharp})$ where ${\sqsupseteq}^{\sharp}$
is defined component wise on $\mathbb{L}^{\sharp}$ in (\ref{eq:def:Abstract-Semantic-Domain-Semantics})
with ${\sqsupseteq_{+}^{\sharp}}$ on the exit and break components and ${\sqsubseteq_{\infty}^{\sharp}}$
on the infinite component. We abstract $\textsf{Post}^\sharp$ in (\ref{eq:def:Post}) to $\textsf{Post}^{{\sqsupseteq}^{\sharp}}\in{\mathbb{L}^{\sharp}}\functionto\alpha^{{\sqsupseteq}^{\sharp}}(\wp({\mathbb{L}^{\sharp}}))\increasingfunctionto\alpha^{{\sqsupseteq}^{\sharp}}(\wp({\mathbb{L}^{\sharp}}))$ by $(\mathcal{P}\in\alpha^{{\sqsupseteq}^{\sharp}}(\wp({\mathbb{L}^{\sharp}}))$)
\begin{calculus}[$\textsf{Post}^{{\sqsupseteq}^{\sharp}}(S)\mathcal{P}$\ = \ ]
$\textsf{Post}^{{\sqsupseteq}^{\sharp}}(S)\mathcal{P}$\ $\triangleq$
\formulaexplanation{\alpha^{{\sqsupseteq}^{\sharp}}(\textsf{Post}^\sharp(S)\mathcal{P})\colsep{=}\{P'\in \mathbb{L}^{\sharp}\mid \exists P\in\textsf{Post}^\sharp(S)\mathcal{P}\mathrel{.}P'\mathrel{{\sqsupseteq}^{\sharp}} P\}}{def.\ (\ref{eq:GC:alpha:sqsubseteq}) of $\alpha^{{\sqsupseteq}^{\sharp}}$}\\
\phantom{$\textsf{Post}^{{\sqsupseteq}^{\sharp}}(S)\mathcal{P}$\ }=
\formulaexplanation{\{P'\in \mathbb{L}^{\sharp}\mid \exists P\in\{\textsf{post}^\sharp(S)P\mid P\in\mathcal{P}\}\mathrel{.}P\mathrel{{\sqsubseteq}^{\sharp}} P'\}}{def.\ (\ref{eq:def:Post}) of \textsf{Post} and inversion of ${{\sqsupseteq}^{\sharp}}$}\\
\phantom{$\textsf{Post}^{{\sqsupseteq}^{\sharp}}(S)\mathcal{P}$\ }=
\formulaexplanation{\{P'\in \mathbb{L}^{\sharp}\mid \exists P\in\mathcal{P}\mathrel{.}\textsf{post}^\sharp(S)P\in\mathcal{P}\mathrel{{\sqsubseteq}^{\sharp}} P'\}}{def.\ $\in$}
\end{calculus}
The consequence is that the \texttt{while} loop verification condition (\ref{eq:logic:abstract:while:upper}) simplifies to 
$\Lfp{{\sqsubseteq}_{+}^{\sharp}}{\vec{F}_{pe}^{\sharp}}(P'){}\mathrel{{\sqsubseteq}_{+}^{\sharp}}P_e$
and
$\Gfp{{\sqsubseteq}_{\infty}^{\sharp}}{F_{p\bot}^{\sharp}}\mathrel{{\sqsubseteq}_{\infty}^{\sharp}}Q_{\bot b}$ which can respectively be handled
by Park induction \cite[theorem II.3.1]{DBLP:journals/pacmpl/Cousot24} and greatest fixpoint over apppoximation by transfinite iterates using the dual of \cite[theorem II.3.6]{DBLP:journals/pacmpl/Cousot24} as is the case, for classic execution properties, in Hoare logic and termination proofs. The reasoning is dual for (\ref{eq:logic:abstract:while:lower}).

\begin{example}[Proof reduction for the order ideal abstraction: bounded nondeterminism]
\label{ex:order-ideal-proof-reduction}
Let us consider proofs of programs with bounded nondeterminism, assuming that the value of variables could only be integers. Consider the instantiation of relational natural semantics in section $\ref{sect:RelationalSemantics}$ with no break and no nontermination where  $\mathbb{V} = \mathbb{Z}$. Let $|S|$ be the cardinality of a set $S$ and consider the semantic (hyper) property $\mathcal{F}\triangleq\wp_{\textsf{fin}}(\mathbb{L})\triangleq\{P\in\wp(\mathbb{L})\mid |P| \in\mathbb{N}\}$ to be the set of finite execution semantics i.e.\ programs satisfying $\mathcal{F}$ cannot have infinitely many different executions although $\mathbb{L}$ has an infinite cardinality.

    Now, suppose we want prove that $\hyperlogicup{\mathcal{F}}{S}{\mathcal{F}}$, where $\texttt{S}\triangleq \texttt{x =  [$0$, $\infty$]; while(x>0) x=x-1}$. Since $\mathcal{F}$ is an order ideal abstraction (subset-closed), we need to find a function $\mathcal{I}\in \mathcal{F}\rightarrow \mathcal{F}$ such that for arbitrary $P\in\mathcal{P}$, we have $ \textsf{\textup{post}}\sqb{\texttt{S}}\subseteq \mathcal{I}(P)$, and, at the same time, the image of $\mathcal{I}$ is a subset of $\mathcal{F}$. Let $m$ and $n$ to be any integer such that $m<0<n$, we can set this $\mathcal{I}$ to be 
\begin{eqntabular}{c}
        \mathcal{I} = \LAMBDA{P}{\{\pair{\sigma}{\sigma'}\in \Sigma\times\Sigma \mid  m\,{<}\,\sigma'(x) \,{\leq}\, n \land \exists \pair{\sigma_1}{\sigma'_1}\in P \mathrel{.} (\sigma_1 \,{=}\, \sigma \land \forall v \in \mathbb{V}\mathrel{.} v\neq x \Rightarrow\sigma_1'(x) \,{=}\, \sigma'(x))\}} \nonumber
\end{eqntabular}
We notice that this program component eventually assigns the value $0$ to $x$ while keeping the value of the other variables unchanged. As a result, for arbitrary $P\in\mathcal{P}$
\begin{eqntabular*}{rcl}
\textsf{\textup{post}}\sqb{\texttt{S}}(P)&=&\{\pair{\sigma}{\sigma'}\in \Sigma\times\Sigma \mid\begin{array}[t]{l} \sigma'(x) = 0 \land \exists \pair{\sigma_1}{\sigma'_1}\in P \mathrel{.} (\sigma_1 = \sigma \land \forall v \in \mathbb{V}\mathrel{.} v\neq x \Rightarrow{}\\\sigma_1'(x) = \sigma'(x))\}\subseteq \mathcal{I}(P)\end{array}\end{eqntabular*}
For the cardinality of $\mathcal{I}(P)$, we let the sequence $\pair{X^i}{n<i\leq m}$ such that $X^i = \{\pair{\sigma}{\sigma'}\in\Sigma\times\Sigma \mid \sigma'(x) = i \land \exists \pair{\sigma_1}{\sigma'_1}\in P \mathrel{.} (\sigma_1 = \sigma \land \forall v \in \mathbb{V}\mathrel{.} v\neq x \Rightarrow\sigma_1'(x) = \sigma'(x))\}$. The cardinality of $X^i$ in this case will be smaller than that of $P$, meaning $|X^i|\in\mathbb{N}$. Thus, the finite union of $X^i$, \smash{$\displaystyle\bigcup_{m<i\leq n }X^i$} also has finite cardinality.
\end{example}

\section{Frontiers Abstractions}\label{sec:FrontierAbstraction}
Another solution to represent order ideal abstractions as proposed by \cite[proposition 1]{DBLP:conf/sas/MastroeniP17} is to consider the maximal elements of the order ideal closed semantic (hyper) property only. Unfortunately, this is not the same abstraction. 
\begin{counterexample}\label{cex:order-ideal-frontier}Consider the hyperproperty $\mathcal{F}\triangleq\wp_{\textsf{fin}}(\mathbb{L})\triangleq\{P\in\wp(\mathbb{L})\mid |P| \in\mathbb{N}\}$ in example \ref{ex:order-ideal-proof-reduction} i.e.\ programs satisfying $\mathcal{F}$ cannot have infinitely many different executions although $\mathbb{L}$ has an infinite cardinality. Then the order ideal abstraction is
$\alpha^{\sqsubseteq}(\mathcal{F})=\mathcal{F}$ which has no maximal elements so the maximal elements abstraction of this order ideal abstraction 
$\alpha^{\sqsubseteq}(\mathcal{F})=\mathcal{F}$ is the empty set which is definitely different from this order ideal abstraction 
$\alpha^{\sqsubseteq}(\mathcal{F})=\mathcal{F}$.
\end{counterexample}
Let us study this abstraction in more detail.

\subsection{Lower Frontier Abstraction}
The lower frontier abstraction abstracts a subset of a poset to its mimimal elements
\bgroup\abovedisplayskip3pt\belowdisplayskip3pt
\begin{eqntabular}{rcl}
{\alpha}^{\underline{F\!}\mskip3mu}(\mathcal{P}) &\triangleq&\{P\in\mathcal{P}\mid\forall P'\in\mathcal{P}\mathrel{.}P'\sqsubseteq P\Rightarrow P'=P\}
\label{eq:def:alpha:underline:F}
\end{eqntabular}\egroup
${\alpha}^{\underline{F\!}\mskip3mu}$ is reductive and idempotent by not necessarily increasing (and so does not necessarily preserve existing joins)
hence may not be the lower adjoint of a Galois connection.

\noindent \begin{minipage}[t]{0.85\textwidth}
\begin{counterexample}\label{cex:alpha:underline:F}Consider the complete lattice $\{\bot, 0, 1, \top\}$ with $\bot\sqsubseteq\bot\sqsubseteq 0\sqsubseteq 0\sqsubseteq\top\sqsubseteq\top$ and $\bot\sqsubseteq1\sqsubseteq1\sqsubseteq\top$. We have $\mathcal{P}_1=\{\top\}\subseteq \{0,1,\top\}=\mathcal{P}_2$ but ${\alpha}^{\underline{F\!}\mskip3mu}(\mathcal{P}_1)=\{\top\}\nsubseteq\{0,1\}={\alpha}^{\underline{F\!}\mskip3mu}(\mathcal{P}_2)$ proving that ${\alpha}^{\overline{F}}$ is not increasing hence does not preserve existing joins hence is not the lower adjoint of a Galois connection. By duality, neither is ${\alpha}^{\overline{F}}$.\let\qef\relax
\end{counterexample}
\end{minipage}\hskip1em
\raisebox{-15.25mm}[0pt][0pt]{\includegraphics[width=0.1\textwidth]{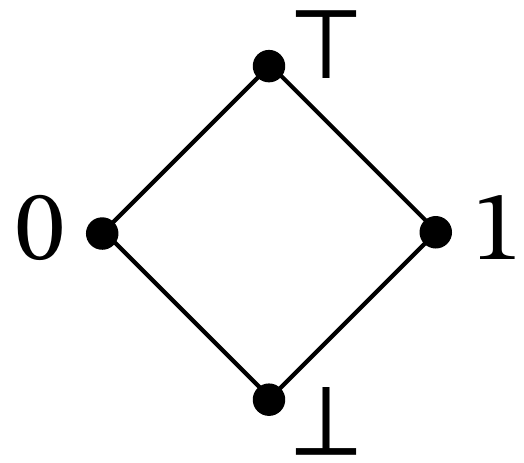}}\hskip0.35em\raisebox{-17.25mm}[0pt][0pt]{\mbox{\qef}}

\subsection{Frontier Order Ideal Abstraction}
The frontier order ideal abstraction
\bgroup\abovedisplayskip4pt\belowdisplayskip2pt
\begin{eqntabular}{rcl}
\alpha^{{\sqsupseteq}\underline{F}}&\triangleq&\alpha^{{\sqsupseteq}}\comp\alpha^{\underline{F}}
\label{eq:def:alpha:sqsubseteq:underline:F}
\end{eqntabular}\egroup
closes the frontier by its over approximations, as shown by the following
\begin{lemma}\label{lem:alpha:sqsubseteq:underline:F}\proofinapx\quad
$\alpha^{{\sqsupseteq}\underline{F}}(\mathcal{P})$ = $\{P\in \mathbb{L}\mid \exists F\in\alpha^{\underline{F}}(\mathcal{P})\mathrel{.}F\sqsubseteq P\}$ =
$\{P\in \mathbb{L}\mid \exists F\in\mathcal{P}\mathrel{.}\forall P'\in\mathcal{P}\mathrel{.}P'\sqsubseteq F\Rightarrow P'=F\wedge F\sqsubseteq P\}$.
\end{lemma}
\begin{toappendix}
\begin{proof}[Proof of lemma \ref{lem:alpha:sqsubseteq:underline:F}]
\begin{calculus}[=\ \ ]
\formula{\alpha^{{\sqsupseteq}}\comp\alpha^{\underline{F}}(\mathcal{P})}\\
=
\formulaexplanation{\{P\in \mathbb{L}\mid \exists F\in\alpha^{\underline{F}}(\mathcal{P})\mathrel{.}F\sqsubseteq P\}}{def.\ function composition $\comp$ and (\ref{eq:GC:alpha:sqsubseteq}) of the dual $\alpha^{{\sqsubseteq}}$}\\
=
\formulaexplanation{\{P\in \mathbb{L}\mid \exists F\in\{P\in\mathcal{P}\mid\forall P'\in\mathcal{P}\mathrel{.}P'\sqsubseteq P\Rightarrow P'=P\}\mathrel{.}F\sqsubseteq P\}}{def.\ (\ref{eq:def:alpha:underline:F}) of ${\alpha}^{\underline{F\!}\mskip3mu}$}\\
=
\lastformulaexplanation{\{P\in \mathbb{L}\mid \exists F\in\mathcal{P}\mathrel{.}\forall P'\in\mathcal{P}\mathrel{.}P'\sqsubseteq F\Rightarrow P'=F\wedge F\sqsubseteq P\}}{def.\ $\in$}{\mbox{\qed}}
\end{calculus}
\let\qed\relax
\end{proof}
\end{toappendix}
Observe that $\alpha^{{\sqsupseteq}\underline{F}}$ is idempotent but not necessarily increasing or extensive.
\begin{counterexample}Consider $\mathbb{L}=\{\pair{a}{n}\mid n\in\mathbb{N}\}\cup \{\pair{b}{m}\mid m\in\mathbb{N}\}$ with $\pair{x}{n}\sqsubseteq\pair{y}{m}\triangleq x=y\wedge n\geqslant m$ be two incomparable infinite decreasing chains. $\mathbb{L}\not\subseteq\alpha^{{\sqsupseteq}\underline{F}}(\mathbb{L})=\emptyset$. Take $\mathcal{P}=\{\pair{a}{n}\mid n\in\mathbb{N}\}\cup \{\pair{b}{0}\}$ so that $\mathcal{P}\subseteq\mathbb{L}$
but $\alpha^{{\sqsupseteq}\underline{F}}(\mathcal{P})=\{\pair{b}{m}\mid m\in\mathbb{N}\}\not\subseteq\alpha^{{\sqsupseteq}\underline{F}}(\mathbb{L})=\emptyset$.
\end{counterexample}
${\alpha^{{\sqsupseteq}\underline{F}}(\wp(\mathbb{L}))}$ is not closed by intersection.

\noindent \begin{minipage}[t]{0.67\textwidth}
\begin{counterexample}\label{cex:alpha:sqsupseteq:underline:F:non:cap}Consider the lattice on the right.
Let $\mathcal{P}_1=\{Z^i\mid i\in\mathbb{N}_{\ast}\}\cup \{X^i\mid i\in\mathbb{N}_{\ast}\}$ with frontier $\mathcal{F}_1=\{X^i\mid i\in\mathbb{N}_{\ast}\}$ and $\mathcal{P}_2=\{Z^i\mid i\in\mathbb{N}_{\ast}\}\cup\{Y^i\mid i\in\mathbb{N}_{\ast}\}$ with frontier $\mathcal{F}_2=\{Y^i\mid i\in\mathbb{N}_{\ast}\}$. There is  no largest set smaller than $\mathcal{P}_1$ and $\mathcal{P}_2$ with an existing frontier.
\end{counterexample}
\end{minipage}\hfill
\raisebox{-20mm}[0pt][0pt]{\includegraphics[width=0.3\textwidth]{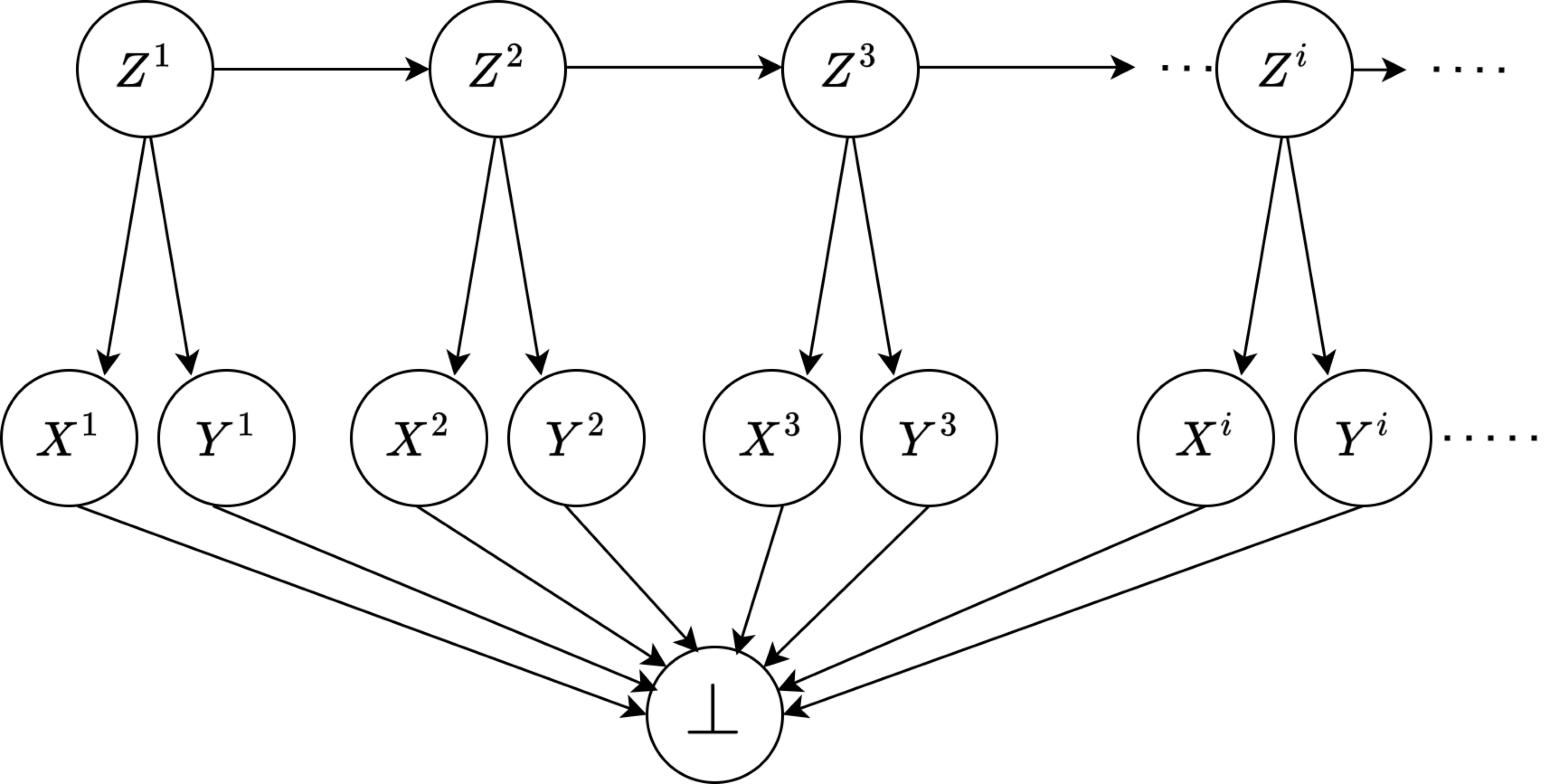}}

\begin{lemma}\label{lem:lower:frontier:closed:lattice}\proofinapx\quad$\quintuple{\alpha^{{\sqsupseteq}\underline{F}}(\wp(\mathbb{L}))}{\subseteq}{\emptyset}{\mathbb{L}}{\cup}$ is a join semilattice. 
\end{lemma}
\begin{toappendix}
\begin{proof}[Proof of lemma \ref{lem:lower:frontier:closed:lattice}]Given $\mathcal{P}_1$, $\mathcal{P}_2\in{\alpha^{{\sqsupseteq}\underline{F}}(\wp(\mathbb{L}))}$, we have to prove that $\mathcal{P}_1\cup\mathcal{P}_2\in{\alpha^{{\sqsupseteq}\underline{F}}(\wp(\mathbb{L}))}$ that is the existence of a frontier $\mathcal{F}\in\alpha^{\underline{F}}(\wp(\mathbb{L}))$ such that
$\mathcal{P}_1\cup\mathcal{P}_2=\alpha^{{\sqsupseteq}}(\mathcal{F})$. Let $\mathcal{F}_1$, $\mathcal{F}_2$ be the frontiers such that $\mathcal{P}_1=\alpha^{{\sqsupseteq}}(\mathcal{F}_1)$ and $\mathcal{P}_2=\alpha^{{\sqsupseteq}}(\mathcal{F}_2)$. Define the frontier
\begin{eqntabular}{rclcl}
\mathcal{F}&\triangleq&\alpha^{\underline{F}}(\mathcal{F}_1\cup\mathcal{F}_2)&=&\{P\in\mathcal{F}_1\cup\mathcal{F}_2\mid\forall P'\in\mathcal{F}_1\cup\mathcal{F}_2\mathrel{.}P'\sqsubseteq P\Rightarrow P'=P\}
\label{eq:def:alphaF:F1UF2}
\end{eqntabular}
\hyphen{5}\quad To prove $\mathcal{P}_1\cup\mathcal{P}_2\subseteq\alpha^{{\sqsupseteq}}(\mathcal{F})$, given any $X\in\mathcal{P}_1\cup\mathcal{P}_2$, let us show the existence of $F\in\mathcal{F}$ such that $X\in\alpha^{{\sqsupseteq}}(P)$ that is $F\sqsubseteq X$. There are two cases.
\begin{enumerate}[leftmargin=*]
\item \label{case:lower:frontier:closed:lattice:1} If $X\in\mathcal{P}_1$ and $X\not\in\mathcal{P}_2$ then $\exists F_1\in\mathcal{F}_1\mathrel{.}F_1\sqsubseteq X$ and $\forall F_2\in\mathcal{F}_2\mathrel{.}F_2\not\sqsubseteq X$ so taking $P=F_1$ in (\ref{eq:def:alphaF:F1UF2}), we have $P=F_1\in\mathcal{F}_1\cup\mathcal{F}_2$ and $\forall P'\in\mathcal{F}_1\cup\mathcal{F}_2$, if $P'\sqsubseteq P =F_1$ then $P'\sqsubseteq X$ by transitivity so $P'\not\in \mathcal{F}_2$ proving
$P'\in\mathcal{F}_1$ and so $P'=F_1=P$ by $F_1\in \alpha^{\underline{F}}(\mathcal{F}_1)$;
\item \label{case:lower:frontier:closed:lattice:2} The case $X\not\in\mathcal{P}_1$ and $X\in\mathcal{P}_2$ is symmetric;
\item \label{case:lower:frontier:closed:lattice:3} Otherwise $X\in\mathcal{P}_1\cap\mathcal{P}_2$. In that case $\exists F_1\in\mathcal{F}_1\mathrel{.}F_1\sqsubseteq X$. Let $\mathcal{M}=
\mathcal{F}_2\cap\alpha^{\sqsupseteq}(X)$. There are two subcases.
\begin{enumerate}[leftmargin=*]
\item $\forall F_2\in\mathcal{M}\mathrel{.}F_2\nsqsubset F_1$. This is similar to case \ref{case:lower:frontier:closed:lattice:1};
\item $\exists F_2\in\mathcal{M}\mathrel{.}F_2\sqsubset F_1$. No element $F'_1$ of $\mathcal{F}_1\setminus\{F_1\}$ is comparable to $F_2$ since
otherwise $F'_1\sqsubseteq F_2\sqsubset F_1$ would contradict that $F_1$ is in the frontier of $\mathcal{P}_1$. Therefore, taking
$P=F_2$, we have $P=F_2\in \mathcal{F}_1\cup\mathcal{F}_2$ and if $P'\in\mathcal{F}_1\cup\mathcal{F}_2$ then $P'\in \mathcal{F}_2$ is impossible so
$P'\in\mathcal{F}_1$ so $P'=F_1=P$ by $F_1\in \alpha^{\underline{F}}(\mathcal{F}_1)$;
\end{enumerate}
\end{enumerate}
\hyphen{5}\quad Conversely, to prove $\mathcal{P}_1\cup\mathcal{P}_2\supseteq\alpha^{{\sqsupseteq}}(\mathcal{F})$, assume $X\in\alpha^{{\sqsupseteq}}(\mathcal{F})$ so that there exists $F\in \alpha^{\underline{F}}(\mathcal{F}_1\cup\mathcal{F}_2)$ such that 
$F\sqsubseteq X$. By (\ref{eq:def:alphaF:F1UF2}), either $F\in \mathcal{F}_1$ and $X\in \mathcal{P}_1$ or $F\in \mathcal{F}_2$ and $X\in \mathcal{P}_2$ proving $X\in \mathcal{P}_1\cup\mathcal{P}_2$.
\end{proof}
\end{toappendix}

\subsection{A Frontier Characterization of the Order Ideal Abstraction}
\begin{lemma}\label{lem:def-alpha-upper-frontier-domain}\proofinapx\quad There is a Galois isomorphism
$\pair{\alpha^{{\sqsubseteq}\overline{F}}(\wp(\mathbb{L}))}{\subseteq}
\GaloiS{{\alpha}^{\overline{F}}}{\alpha^{\sqsubseteq}} 
\pair{{\alpha}^{\overline{F}}(\wp(\mathbb{L}))}{\mathrel{{\preceq}^{\overline{F}}}}$ and $\triple{{\alpha}^{\overline{F}}(\wp(\mathbb{L}))}{\mathrel{{\preceq}^{\overline{F}}}}{\mathbin{{\curlyvee}^{\overline{F}}}}$ is a join semi lattice with $P\mathrel{{\preceq}^{\overline{F}}}Q\triangleq (\alpha^{\sqsubseteq}(P)\subseteq \alpha^{\sqsubseteq}(Q))$ and $P\mathbin{{\curlyvee}^{\overline{F}}}Q\triangleq{{\alpha}^{\overline{F}}}(\alpha^{\sqsubseteq}(P)\cup\alpha^{\sqsubseteq}(Q))$.
\end{lemma}
\begin{toappendix}
\begin{proof}[Proof of lemma \ref{lem:def-alpha-upper-frontier-domain}]
\hyphen{5}\quad We first show that ${\alpha}^{\overline{F}}\comp \alpha^{\sqsubseteq}\comp{\alpha}^{\overline{F}}= {\alpha}^{\overline{F}}$. Consider $\mathcal{P}\in{\alpha}^{\overline{F}}(\wp(\mathbb{L}))$. Then
\begin{calculus}[=\ \ ]
\formula{{\alpha}^{\overline{F}}\comp \alpha^{\sqsubseteq}(\mathcal{P})}\\
= 
\formulaexplanation{\{P\in\alpha^{\sqsubseteq}(\mathcal{P})\mid \forall P'\in \alpha^{\sqsubseteq}(\mathcal{P}) \mathrel{.}P\sqsubseteq P' \Rightarrow P= P'\}}{By def.\ (\ref{eq:def:alpha:underline:F}) of ${\alpha}^{\overline{F}}$}\\
= 
\formula{\{P\in\{P'\in \mathbb{L}\mid \exists F\in\mathcal{P}\mathrel{.}P'\sqsubseteq F\}\mid \forall P'\in \{P'\in \mathbb{L}\mid \exists F'\in\mathcal{P}\mathrel{.}P'\sqsubseteq F'\} \mathrel{.}P\sqsubseteq P' \Rightarrow P= P'\}}\\[-0.5ex]\rightexplanation{by def.\ (\ref{eq:GC:alpha:sqsubseteq}) of $\alpha^{\sqsubseteq}$}\\
=
\formulaexplanation{\{P\mid \exists F\in\mathcal{P}\mathrel{.}P\sqsubseteq F\wedge\forall P'\in \{P'\in \mathbb{L}\mid \exists F'\in\mathcal{P}\mathrel{.}P'\sqsubseteq F'\} \mathrel{.}P\sqsubseteq P' \Rightarrow P= P'\}}{def.\ $\in$}\\
=
\formulaexplanation{\{P\mid \exists F\in\mathcal{P}\mathrel{.}P\sqsubseteq F\wedge\forall P'\in\mathbb{L}\mathrel{.} (\exists F'\in\mathcal{P}\mathrel{.}P'\sqsubseteq F')\Rightarrow (P\sqsubseteq P' \Rightarrow P= P')\}}{def.\ $\in$}\\
= 
\formulaexplanation{\{P\in\mathbb{L}\mid \exists F\in\mathcal{P}\mathrel{.}P\sqsubseteq F \wedge\forall P'\in\mathbb{L} \mathrel{.} (\exists F'\in \mathcal{P} \mathrel{.} P\sqsubseteq P'\sqsubseteq F') \Rightarrow P= P'\}}{def.\ $\Rightarrow$ and transitivity}\\
= \formula{\{P\in\mathbb{L}\mid P\in\mathcal{P}\}}\\
    \explanation{$(\subseteq)$ Let $P' = F$, so that $\exists F'\in \mathcal{P} \mathrel{.} P\sqsubseteq P'\sqsubseteq F'$ holds by choosing $F'=F$ which implies $P=P'=F\in \mathcal{P}$ so   $P\in \mathcal{P}$ ;\\
    $(\supseteq)$ Let $P\in\mathcal{P}$ and choose $F=P$ so that $P\sqsubseteq F$. Consider any $P'\in\mathbb{L}$.
    Then, by choosing $F'=P'$, $(\exists F'\in \mathcal{P} \mathrel{.} P\sqsubseteq P'\sqsubseteq F')$ if and only if $P\sqsubseteq P'$. But $P=F\in\mathcal{P}$ and $\mathcal{P}\in{\alpha}^{\overline{F}}(\wp(\mathbb{L}))$ is a frontier so $P=P'$}\\
=
\formulaexplanation{\mathcal{P}}{def.\ set in extension}
\end{calculus}

\smallskip

\hyphen{5}\quad  If $\mathcal{P}\in{\alpha^{{\sqsubseteq}\overline{F}}(\wp(\mathbb{L}))}$ then there exists $\mathcal{P}'\in\wp(\mathbb{L}))$ such that $\mathcal{P}=\alpha^{{\sqsubseteq}\overline{F}}(\mathcal{P}')$ and then
\begin{calculus}[=\ \ ]
\formula{\alpha^{\sqsubseteq}\comp\alpha^{\overline{F}}(\mathcal{P})}\\
=
\formulaexplanation{\alpha^{\sqsubseteq}\comp\alpha^{\overline{F}}\comp\alpha^{{\sqsubseteq}\overline{F}}(\mathcal{P}')}{$\mathcal{P}=\alpha^{{\sqsubseteq}\overline{F}}(\mathcal{P}')$}\\
=
\formulaexplanation{\alpha^{\sqsubseteq}\comp\alpha^{\overline{F}}\comp\alpha^{\sqsubseteq}\comp\alpha^{\overline{F}}(\mathcal{P}')}{dual def.\ (\ref{eq:def:alpha:sqsubseteq:underline:F}) of $\alpha^{{\sqsubseteq}\overline{F}}$}\\
=
\formulaexplanation{\alpha^{\sqsubseteq}\comp\alpha^{\overline{F}}(\mathcal{P}')}{since ${\alpha}^{\overline{F}}\comp \alpha^{\sqsubseteq}\comp{\alpha}^{\overline{F}}= {\alpha}^{\overline{F}}$}\\
=
\formulaexplanation{\alpha^{{\sqsubseteq}\overline{F}}(\mathcal{P}')}{dual def.\ (\ref{eq:def:alpha:sqsubseteq:underline:F}) of $\alpha^{{\sqsubseteq}\overline{F}}$}\\
=
\formulaexplanation{\mathcal{P}}{by definition  $\mathcal{P}=\alpha^{{\sqsubseteq}\overline{F}}(\mathcal{P}')$}
\end{calculus}

\smallskip

\hyphen{5}\quad  If $\mathcal{Q}\in{\alpha^{\overline{F}}(\wp(\mathbb{L}))}$ then there exists $\mathcal{Q}'\in\wp(\mathbb{L}))$ such that $\mathcal{Q}=\alpha^{\overline{F}}(\mathcal{Q}')$ and then
\begin{calculus}[=\ \ ]
\formula{\alpha^{\overline{F}}\comp\alpha^{\sqsubseteq}(\mathcal{Q})}\\
=
\formulaexplanation{\alpha^{\overline{F}}\comp\alpha^{\sqsubseteq}\comp\alpha^{\overline{F}}(\mathcal{Q}')}{$\mathcal{Q}=\alpha^{\overline{F}}(\mathcal{Q}')$}\\
=
\formulaexplanation{\alpha^{\overline{F}}(\mathcal{Q}')}{since ${\alpha}^{\overline{F}}\comp \alpha^{\sqsubseteq}\comp{\alpha}^{\overline{F}}= {\alpha}^{\overline{F}}$}\\
=
\formulaexplanation{\mathcal{Q}}{$\mathcal{Q}=\alpha^{\overline{F}}(\mathcal{Q}')$}\
\end{calculus}

\smallskip

\hyphen{5}\quad  It follows that there is a bijection ${{\alpha}^{\overline{F}}}$ with inverse ${\alpha^{\sqsubseteq}}$ between ${\alpha^{{\sqsubseteq}\overline{F}}(\wp(\mathbb{L}))}$ and 
${{\alpha}^{\overline{F}}(\wp(\mathbb{L}))}$. 

\smallskip

\hyphen{5}\quad Defining $P\mathrel{{\preceq}^{\overline{F}}}Q\triangleq (\alpha^{\sqsubseteq}(P)\subseteq \alpha^{\sqsubseteq}(Q))$ this yields the Galois retraction
$\pair{\alpha^{{\sqsubseteq}\overline{F}}(\wp(\mathbb{L}))}{\subseteq}
\GaloiS{{\alpha}^{\overline{F}}}{\alpha^{\sqsubseteq}} 
\pair{{\alpha}^{\overline{F}}(\wp(\mathbb{L}))}{\mathrel{{\preceq}^{\overline{F}}}}$. By the dual of lemma 
\ref{lem:lower:frontier:closed:lattice}, $\quintuple{\alpha^{{\sqsubseteq}\overline{F\!}\mskip3mu}(\wp(\mathbb{L}))}{\subseteq}{\emptyset}{\mathbb{L}}{\cup}$ is a join semilattice. Therefore the finite joins are preserved by the Galois connection so that $\triple{{\alpha}^{\overline{F}}(\wp(\mathbb{L}))}{\mathrel{{\preceq}^{\overline{F}}}}{\mathbin{{\curlyvee}^{\overline{F}}}}$ is a join semilattice with $P\mathrel{{\preceq}^{\overline{F}}}Q\triangleq \alpha^{\sqsubseteq}(P)\subseteq \alpha^{\sqsubseteq}(Q)$ and $P\mathbin{{\curlyvee}^{\overline{F}}}Q\triangleq{{\alpha}^{\overline{F}}}(\alpha^{\sqsubseteq}(P)\cup\alpha^{\sqsubseteq}(Q))$.
\end{proof}
\end{toappendix}
Define the principal ideal ${\downarrow^{\sqsubseteq}}(P) \triangleq\{P'\in\mathbb{L}\mid P'\sqsubseteq P\}$. The following lemma \ref{lem:FrontierCharacterizationOrderIdealAbstraction} is a characterization of ${{\alpha}^{{\sqsubseteq}\overline{F}}(\wp(\mathbb{L}))}$ that corrects and generalizes \cite[Proposition 1]{DBLP:conf/sas/MastroeniP17}.
\begin{lemma}\label{lem:FrontierCharacterizationOrderIdealAbstraction}\proofinapx\quad If\/ $\mathcal{P}\in{{\alpha}^{{\sqsubseteq}\overline{F}}(\wp(\mathbb{L}))}$ then
$\displaystyle\mathcal{P}=\smash{\bigcup_{P\,\in\, {\alpha}^{\overline{F}}(\mathcal{P})}{\downarrow^{\sqsubseteq}}(P)}$.
\end{lemma}
\begin{toappendix}
\begin{proof}[Proof of lemma \ref{lem:FrontierCharacterizationOrderIdealAbstraction}]
\begin{calculus}[=\ \ ]
\formula{\mathcal{P}}\\
=
\formulaexplanation{{{\alpha}^{{\sqsubseteq}\overline{F}}(\mathcal{P})}}{$\mathcal{P}\in{{\alpha}^{{\sqsubseteq}\overline{F}}(\wp(\mathbb{L}))}$ and  lemma \ref{lem:def-alpha-upper-frontier-domain}}\\
=
\formulaexplanation{\alpha^{{\sqsubseteq}}(\alpha^{\underline{F}}(\mathcal{P}))}{dual def.\ (\ref{eq:def:alpha:sqsubseteq:underline:F}) of $\alpha^{{\sqsubseteq}\underline{F}}$ and composition $\comp$}\\
=
\formulaexplanation{\{P'\in \mathbb{L}\mid \exists P\in\alpha^{\underline{F}}(\mathcal{P})\mathrel{.}P'\sqsubseteq P\}}{def.\ (\ref{eq:GC:alpha:sqsubseteq}) of $\alpha^{\sqsubseteq}$}\\
=
\formulaexplanation{\bigcup_{P\,\in\,\alpha^{\underline{F}}(\mathcal{P})}\{P'\in \mathbb{L}\mid P'\sqsubseteq P\}}{def.\ $\cup$}\\
=
\lastformulaexplanation{\bigcup_{P\,\in\,\alpha^{\underline{F}}(\mathcal{P})}{\downarrow^{\sqsubseteq}}(P)}{def.\ ${\downarrow^{\sqsubseteq}}(P) \triangleq\{P'\in\mathbb{L}\mid P'\sqsubseteq P\}$}{\mbox{\qed}}
\end{calculus}
\let\qed\relax
\end{proof}
\end{toappendix}

\section{Chain Limit Abstraction}\label{ChainLimitAbstraction}
\subsection{Chain Limit Abstraction Definition and Properties}Another possible representation of order ideal abstractions would be by limits of chains.
Define
\bgroup\abovedisplayskip4pt\belowdisplayskip2pt
\begin{eqntabular}{rcl}
\alpha^{\downarrow}(\mathcal{P}) &\triangleq&\{\bigsqcap_{i\in\mathbb{N}}P_i\mid\pair{P_i}{i\in\mathbb{N}}\in\mathcal{P}\textrm{\ is a decreasing chain with existing glb}\}
\label{eq:def:alpha:downarrow}
\end{eqntabular}\egroup
$\alpha^{\downarrow}$ is $\subseteq$ increasing and extensive but not necessarily idempotent as shown by counter example \ref{cex:non-idempotent-closed-infinite-union} below. The iteration of $\alpha^{\downarrow}$ (possibly transfinitely)
\bgroup\abovedisplayskip4pt\belowdisplayskip2pt
\begin{eqntabular}{rcl}
{\maccent{\alpha}{\ast}}^{\downarrow}(\mathcal{P}) &\triangleq&\Lfp{\subseteq}\LAMBDA{X}\mathcal{P}\cup\alpha^{\downarrow}(X)
\label{eq:def:hat:alpha:downarrow}
\end{eqntabular}\egroup
yields an upper closure operator \cite[lemma 29.1]{Cousot-PAI-2021}.

\medskip

\noindent \begin{minipage}[t]{0.70\textwidth}
\begin{counterexample}\label{cex:non-idempotent-closed-infinite-union}Consider the complete lattice $\mathbb{L}$ on the right. Let $\mathcal{P}=\{X^{ij}\mid i,j>0\}$. 
We have $\alpha^{\downarrow}(\mathcal{P})=\{X^{ij}\mid i,j>0\}\cup\{Y^i\mid i>0\}$. We have $\bigsqcap\{Y^i\mid i>0\}=\bot$ so  ${\alpha^{\downarrow}}({\alpha^{\downarrow}}(\mathcal{P}))=\{X^{ij}\mid i,j>0\}\cup\{Y^i\mid i>0\}\cup\{\bot\}\neq \alpha^{\downarrow}(\mathcal{P})$. 

Moreover ${\maccent{\alpha}{\ast}}^{\downarrow}(\mathcal{Q}_i)\in  {\maccent{\alpha}{\ast}}^{\downarrow}(\wp(\mathbb{L}))$, $i>0$ but $\bigcup_{i>0}\maccent{\alpha}{\ast}^{\downarrow}(\mathcal{Q}_i)\not\in {\maccent{\alpha}{\ast}}^{\downarrow}(\wp(\mathbb{L}))$.
\end{counterexample}
\end{minipage}\hfill
\raisebox{-20mm}[0pt][0pt]{\includegraphics[width=0.25\textwidth]{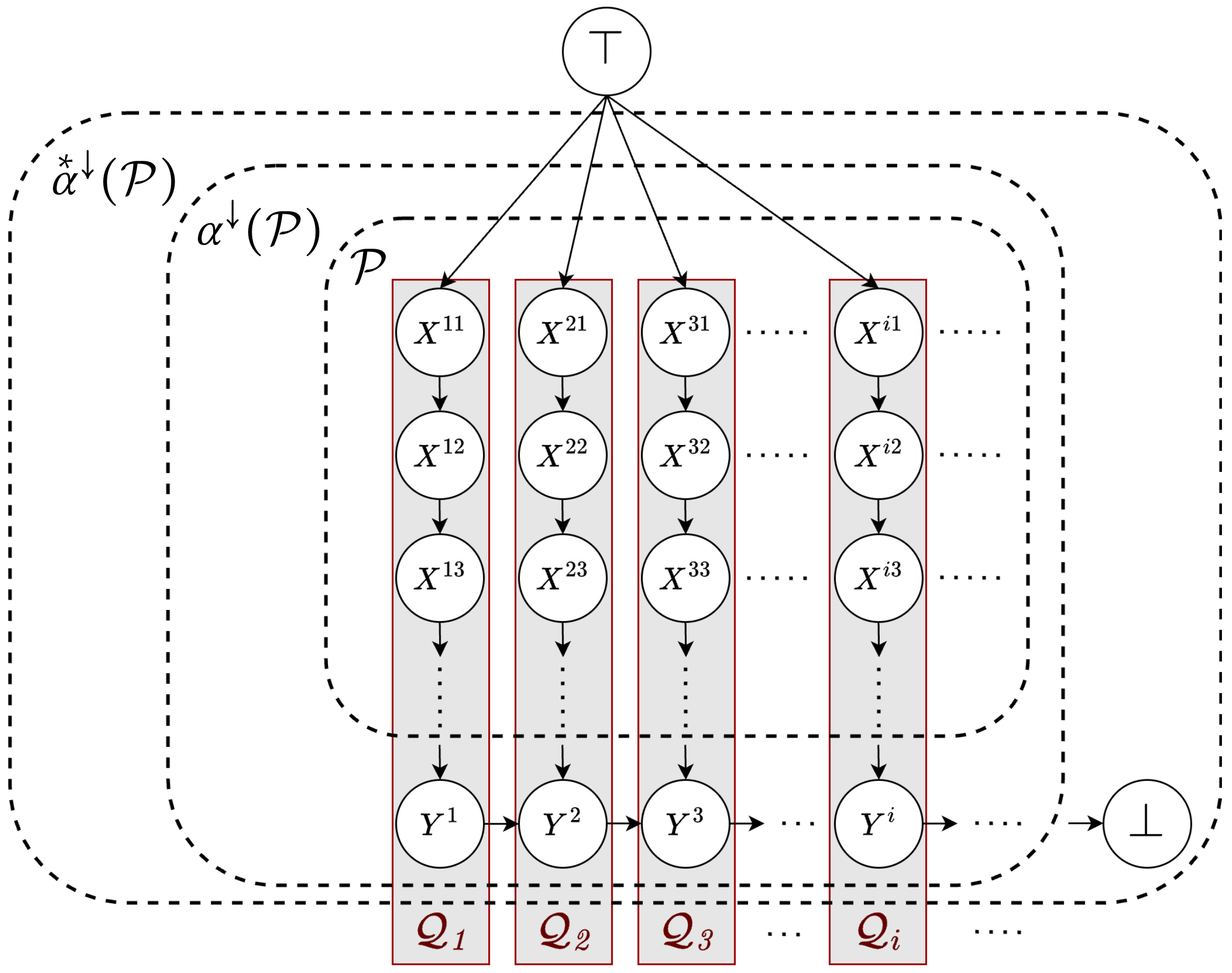}}
\begin{lemma}\label{lem:GC:hat:alpha:downarrow}\proofinapx\quad\smash{$\pair{\wp(\mathbb{L})}{\subseteq}\galoiS{{\maccent{\alpha}{\ast}}^{\downarrow}}{\mathbb{1}}\pair{{\maccent{\alpha}{\ast}}^{\downarrow}(\wp(\mathbb{L}))}{\subseteq}$} and
$\sextuple{{\maccent{\alpha}{\ast}}^{\downarrow}(\wp(\mathbb{L}))}{\subseteq}{\emptyset}{\mathbb{L}}{\LAMBDA{X}{\maccent{\alpha}{\ast}}^{\downarrow}(\bigcup X)}{\bigcap}$  is a complete lattice.
\end{lemma}
\begin{toappendix}
\begin{proof}[Proof of lemma \ref{lem:GC:hat:alpha:downarrow}]
By  \cite[lemma 29.1]{Cousot-PAI-2021}, ${{\maccent{\alpha}{\ast}}^{\downarrow}}$
is the smallest upper closure operator pointwise greater than or equal to ${\alpha}^{\downarrow}$. By Morgan Ward's \cite[theorem 4.1]{Ward42}, 
$\pair{{\maccent{\alpha}{\ast}}^{\downarrow}(\wp(\mathbb{L}))}{\subseteq}$ is a complete lattice with infimum ${\maccent{\alpha}{\ast}}^{\downarrow}(\{\bot\}) =\{\bot\}$ and join
${\LAMBDA{X}{\maccent{\alpha}{\ast}}^{\downarrow}(\bigcup X)}$.
\end{proof}
\end{toappendix}
\begin{lemma}\label{lem:alpha:downarrow:fixpoint}\proofinapx\quad
$\forall \mathcal{P}\in\wp(\mathbb{L})\mathrel{.}{\alpha}^{\downarrow}({\maccent{\alpha}{\ast}}^{\downarrow}(\mathcal{P}))={\maccent{\alpha}{\ast}}^{\downarrow}(\mathcal{P})$.
\end{lemma}
\begin{toappendix}
\begin{proof}[Proof of lemma \ref{lem:alpha:downarrow:fixpoint}]
By the fixpoint definition (\ref{eq:def:hat:alpha:downarrow}) of ${\maccent{\alpha}{\ast}}^{\downarrow}$, we have
${\maccent{\alpha}{\ast}}^{\downarrow}(\mathcal{P})=\Lfp{\subseteq}\LAMBDA{X}\mathcal{P}\cup\alpha^{\downarrow}(X)$ so
${\maccent{\alpha}{\ast}}^{\downarrow}(\mathcal{P})=\mathcal{P}\cup\alpha^{\downarrow}({\maccent{\alpha}{\ast}}^{\downarrow}(\mathcal{P}))$. 
Since ${\alpha}^{\downarrow}$ and ${\maccent{\alpha}{\ast}}^{\downarrow}$ are extensive, we have $\mathcal{P}\subseteq{\maccent{\alpha}{\ast}}^{\downarrow}(\mathcal{P})\subseteq{\alpha}^{\downarrow}({\maccent{\alpha}{\ast}}^{\downarrow}(\mathcal{P}))$ so $\mathcal{P}\cup{\alpha}^{\downarrow}({\maccent{\alpha}{\ast}}^{\downarrow}(\mathcal{P}))$ = ${\alpha}^{\downarrow}({\maccent{\alpha}{\ast}}^{\downarrow}(\mathcal{P}))$ proving ${\maccent{\alpha}{\ast}}^{\downarrow}(\mathcal{P})$ = ${\alpha}^{\downarrow}({\maccent{\alpha}{\ast}}^{\downarrow}(\mathcal{P}))$ by transitivity.
\end{proof}
\end{toappendix}
\begin{lemma}\label{lem:alpha:downarrow=*downarrow}\proofinapx\quad
For all $\mathcal{P}\in\wp(\mathbb{L})$, ${\alpha}^{\downarrow}(\mathcal{P})=\mathcal{P}$ implies
${\maccent{\alpha}{\ast}}^{\downarrow}(\mathcal{P})=\mathcal{P}$.
\end{lemma}
\begin{toappendix}
\begin{proof}[Proof of lemma \ref{lem:alpha:downarrow=*downarrow}]
Consider the iterates of $\Lfp{\subseteq}\LAMBDA{X}\mathcal{P}\cup\alpha^{\downarrow}(X)$ from $X^0=\emptyset$. $X^1=\mathcal{P}\cup\alpha^{\downarrow}(X^0)=\mathcal{P}\cup\alpha^{\downarrow}(\emptyset)=\mathcal{P}$ since $\alpha^{\downarrow}(\emptyset)=\emptyset$ by definition (\ref{eq:def:alpha:downarrow}). We have $X^2=\mathcal{P}\cup\alpha^{\downarrow}(X^1)=\mathcal{P}\cup\alpha^{\downarrow}(\mathcal{P})=\mathcal{P}\cup\mathcal{P}=\mathcal{P}=X^1$ by hypothesis ${\alpha}^{\downarrow}(\mathcal{P})=\mathcal{P}$. By (\ref{eq:def:hat:alpha:downarrow}), we conclude that
${\maccent{\alpha}{\ast}}^{\downarrow}(\mathcal{P})$ = $\Lfp{\subseteq}\LAMBDA{X}\mathcal{P}\cup\alpha^{\downarrow}(X)$ = $\mathcal{P}$.
\end{proof}
\end{toappendix}

$\alpha^{\uparrow}$ is defined $\sqsubseteq$ dually, and ${\maccent{\alpha}{\ast}}^{\uparrow}(\mathcal{P})\triangleq\Lfp{\subseteq}\LAMBDA{X}\mathcal{P}\cup\alpha^{\uparrow}(X)$ is an upper closure operator.

\subsection{Forall Exists Hyperproperties}\label{sec:GeneralizedForallExists}
Assuming that $\pair{\mathbb{L}}{\sqsubseteq}=\pair{\wp(\Pi)}{\subseteq}$ (where e.g.\ $\Pi=\Sigma^{+\infty}$ is a set of traces) $\forall\exists$ hyperproperties have the form 
\bgroup\abovedisplayskip4pt\belowdisplayskip2pt\begin{eqntabular}{rcl}
\mathcal{A\mskip1muEH}&\triangleq&
\{\{P\in\wp(\Pi)\mid\forall\pi_1\in P\mathrel{.}\exists \pi_2\in P\mathrel{.}\pair{\pi_1}{\pi_2}\in A\}\mid 
A\in\wp(\Pi\times\Pi)\}
\label{def:AEH}
\end{eqntabular}\egroup
 (this easily generalizes to $\forall\pi_1,\ldots,\pi_n\in P\mathrel{.}\exists \pi'_1,\ldots,\pi'_m\in P\mathrel{.}\sextuple{\pi_1}{\ldots}{\pi_n}{\pi'_1}{\ldots}{\pi'_m}\in A$ \cite{DBLP:conf/concur/FinkbeinerH16}).
\begin{example}[Generalized non-interference]A typical forall exists hyperproperty is generalized non interference \cite{DBLP:conf/sp/McCullough87,DBLP:journals/tse/McLean96,DBLP:conf/aplas/DickersonYZD22} for the trace semantics of appendix \ref{sect:Trace-Semantics}. Let $\mathtt{L}\in \mathbb{X}$ be a low variable and $\mathtt{H}\in \mathbb{X}$ be a high variable, we have
\bgroup\abovedisplayskip3pt\belowdisplayskip0pt\begin{eqntabular}{rcl@{\qquad}}
\textit{GNI}&\triangleq&\{P\in\wp(\Sigma^{+})\mid\forall \sigma_1\pi_1\sigma'_1,\sigma_2\pi_2\sigma'_2\in P\mathrel{.}\exists \sigma_3\pi_3\sigma'_3\in P\mathrel{.}(\sigma_1(\mathtt{L})=\sigma_2(\mathtt{L}))\Rightarrow{}\label{eq:def:GNI}\\
&&\qquad\qquad\qquad\qquad\qquad(\sigma_3(\mathtt{L})=\sigma_1(\mathtt{L})\wedge\sigma_3(\mathtt{H})=\sigma_2(\mathtt{H})\wedge\sigma'_3(\mathtt{L})=\sigma'_1(\mathtt{L}))\}\renumber{\qef}
\end{eqntabular}\egroup
\let\qef\relax
\end{example}
Assuming chain-complete lattices in \ref{def:abstract:domain:well:def:finite:domain} and \ref{def:abstract:domain:well:def:infinite:domain}, chain limit closed semantic properties in ${\maccent{\alpha}{\ast}}^{\uparrow}(\wp(\wp(\Pi)))$ subsume $\forall\exists$ hyperproperties in $\mathcal{A\mskip1muEH}$ in that\proofinapx
\bgroup\abovedisplayskip3pt\belowdisplayskip3pt\begin{eqntabular}{rcl}
\mathcal{A\mskip1muEH}&\subseteq&{{\maccent{\alpha}{\ast}}^{\uparrow}}(\wp(\wp(\Pi)))\label{eq:AEH:alpha:sqsubseteq:uparrow:closed}
\end{eqntabular}\egroup
\begin{toappendix}
\begin{proof}[Proof of (\ref{eq:AEH:alpha:sqsubseteq:uparrow:closed})]We must prove that $\forall\mathcal{P}\in\mathcal{A\mskip1muEH}\mathrel{.}{{\maccent{\alpha}{\ast}}^{\uparrow}}(\mathcal{P})\in\mathcal{A\mskip1muEH}$. By the dual of lemma \ref{lem:alpha:downarrow=*downarrow}, it is sufficient to assume that $\mathcal{P}\in\mathcal{A\mskip1muEH}$ and prove that
${\alpha}^{\uparrow}(\mathcal{P})\in\mathcal{A\mskip1muEH}$.
\begin{calculus}[=\ \ ]
\formula{{\alpha}^{\uparrow}(\mathcal{P})}\\
=
\formulaexplanation{\{\bigcup_{i\in\mathbb{N}}P_i\mid\pair{P_i}{i\in\mathbb{N}}\in\mathcal{P}\textrm{\ is an increasing chain with existing lub}\}}{dual def.\ (\ref{eq:def:alpha:downarrow}) of ${\alpha}^{\uparrow}$}\\
=
\formulaexplanation{\{\bigcup_{i\in\mathbb{N}}P_i\mid\pair{P_i}{i\in\mathbb{N}}\in\mathcal{P}\textrm{\ is an increasing chain}\}}{chain completeness hypothesis}\\
=
\formula{\{\bigcup_{i\in\mathbb{N}}P_i\mid\pair{P_i}{i\in\mathbb{N}}\in\mathcal{P}\textrm{\ is an increasing chain}\wedge\forall i\in\mathbb{N} \mathrel{.}\forall\pi_1\in P_i\mathrel{.}\exists \pi_2\in P_i\mathrel{.}\pair{\pi_1}{\pi_2}\in A\}}\\[-1ex]\rightexplanation{$\mathcal{P}\in\mathcal{A\mskip1muEH}$ and def.\ $\mathcal{A\mskip1muEH}$}\\[-2ex]
=\formula{\{P\in\mathcal{P}\mid\forall\pi_1\in P\mathrel{.}\exists \pi_2\in P\mathrel{.}\pair{\pi_1}{\pi_2}\in A\}}\\
\explanation{($\subseteq$)\quad if $\pi_1\in \bigcup_{i\in\mathbb{N}}P_i$ then there exists $i\in\mathbb{N}$ such that $\pi_1\in P_i$ so that, by hypothesis, $\exists \pi_2\in P_i\mathrel{.}\pair{\pi_1}{\pi_2}\in A$, proving $\exists \pi_2\in \bigcup_{i\in\mathbb{N}}P_i\mathrel{.}\pair{\pi_1}{\pi_2}\in A$; \\
($\supseteq$)\quad conversely, consider the chain $\pair{P}{i\in\mathbb{N}}$ so that $\bigcup_{i\in\mathbb{N}}P=P$.}\\
=
\lastformulaexplanation{\mathcal{P}}{since $\mathcal{P}\in\mathcal{A\mskip1muEH}$ so that by (\ref{def:AEH}) the condition holds for all elements of  $\mathcal{P}$}{\mbox{\qed}}
\end{calculus}
\let\qed\relax
\end{proof}
\end{toappendix}
 
 \section{Chain Limit Order Ideal Abstraction}\label{sec:ChainLimitOrderIdealAbstraction}
\subsection{Chain Limit Order Ideal Abstraction Definition and Properties}
Define
\bgroup\abovedisplayskip3pt\belowdisplayskip3pt\begin{eqntabular}{rcl@{\qquad\qquad and \qquad\qquad}rcl}
\alpha^{{\sqsubseteq}{\uparrow}}&\triangleq&\alpha^{{\sqsubseteq}}\comp\alpha^{\uparrow}
&
{\maccent{\alpha}{\ast}}^{{\sqsubseteq}{\uparrow}}(\mathcal{P})&\triangleq&\Lfp{\subseteq}\LAMBDA{X}\mathcal{P}\cup\alpha^{{\sqsubseteq}{\uparrow}}(X)
\label{eq:def:hat:alpha:sqsubseteq:uparrow}
\end{eqntabular}\egroup
to get an upper closure operator (since $\alpha^{{\sqsubseteq}{\uparrow}}$ is increasing and expansive although not idempotent).
\begin{counterexample}\label{cex:hat:alpha:sqsubseteq:uparrow}Define $\pair{\mathbb{L}}{\sqsubseteq}=\pair{\wp(\mathbb{N})}{\subseteq}$ and $\mathcal{N}\triangleq\{\mathbb{N}\setminus\{n\}\mid n\in\mathbb{N}\}\in\wp(\mathbb{N})$ to be the set of all sets $\mathbb{N}$ with one missing element. Since any two different elements of $\mathcal{N}$ are $\subseteq$- incomparable, $\mathcal{N}$ is both a lower and upper frontier so chains are reduced to one element. Therefore ${\maccent{\alpha}{\ast}}^{\downarrow}(\mathcal{N}) ={\maccent{\alpha}{\ast}}^{\uparrow}(\mathcal{N}) =\mathcal{N}$. By (\ref{eq:def:hat:alpha:sqsubseteq:uparrow}), it follows that ${{\alpha}}^{{\sqsubseteq}{\uparrow}}(\mathcal{N})=\alpha^{{\sqsubseteq}}(\mathcal{N})=\wp(\mathbb{N})\setminus\{\mathbb{N}\}$. Consider the increasing chain $\mathcal{C}=\pair{\{i\mid i<j\}}{j\in\mathbb{N}}$ of elements of ${\maccent{\alpha}{\ast}}^{\uparrow}(\mathcal{N})$. Its limit is $\bigcup_{j\in\mathbb{N}}{\{i\mid i<j\}}=\mathbb{N}\not\in{{\alpha}}^{{\sqsubseteq}{\uparrow}}(\mathcal{N})=\wp(\mathbb{N})\setminus\{\mathbb{N}\}$ proving that ${{\alpha}}^{{\sqsubseteq}{\uparrow}}$ is not idempotent.
\end{counterexample}
\begin{lemma}\label{lem:GC:hat:alpha:sqsubseteq:uparrow}\proofinapx\quad\smash{$\pair{\wp(\mathbb{L})}{\subseteq}\galoiS{{\maccent{\alpha}{\ast}}^{{\sqsubseteq}{\uparrow}}}{\mathbb{1}}\pair{{\maccent{\alpha}{\ast}}^{{\sqsubseteq}{\uparrow}}(\wp(\mathbb{L}))}{\subseteq}$} and
$\sextuple{{\maccent{\alpha}{\ast}}^{{\sqsubseteq}{\uparrow}}(\wp(\mathbb{L}))}{\subseteq}{\emptyset}{\mathbb{L}}{\LAMBDA{X}{\maccent{\alpha}{\ast}}^{{\sqsubseteq}{\uparrow}}(\bigcup X)}{\bigcap}$  is a complete lattice.
\end{lemma}
\begin{toappendix}
\begin{proof}[Proof of lemma \ref{lem:GC:hat:alpha:sqsubseteq:uparrow}]
By  \cite[lemma 29.1]{Cousot-PAI-2021}, ${\maccent{\alpha}{\ast}}^{{\sqsubseteq}{\uparrow}}$
is the smallest upper closure operator pointwise greater than or equal to $\alpha^{{\sqsubseteq}{\uparrow}}$. By Morgan Ward's \cite[theorem 4.1]{Ward42}, 
$\pair{{\maccent{\alpha}{\ast}}^{{\sqsubseteq}{\uparrow}}(\wp(\mathbb{L}))}{\subseteq}$ is a complete lattice with infimum ${\maccent{\alpha}{\ast}}^{\uparrow}(\{\bot\}) =\{\bot\}$ and join
${\LAMBDA{X}{\maccent{\alpha}{\ast}}^{{\sqsubseteq}{\uparrow}}(\bigcup X)}$.
\end{proof}
\end{toappendix}

\subsection{Forall Hyperproperties}\label{sec:generalized:forall:forall}
$\forall$ hyperproperties are usually
defined in the context of trace semantics of section \ref{sect:Trace-Semantics}, for which, in absence of \texttt{break}s, $\pair{\mathbb{L}}{\sqsubseteq}=\pair{\wp(\Sigma^{+\infty})}{\subseteq})$ as in section \ref{sec:Biinductive-Trace-Semantic-Domain}. In this case, by definition of $\subseteq$, we get
\bgroup\abovedisplayskip4pt\belowdisplayskip2pt\begin{eqntabular}{rcl}
\mathcal{AAH}&\triangleq&\{\{P\in \wp(\Sigma^{+\infty})\mid\forall \pi_1,\pi_2\in P\mathrel{.}\pair{\pi_1}{\pi_2}\in A\}\mid A\in\wp(\Sigma^{+\infty}\times\Sigma^{+\infty})\}
\label{eq:def:AAH:pi}
\end{eqntabular}\egroup

\begin{example}[Non-interference]
A typical forall hyperproperty is non interference $\textit{NI}\in\mathcal{AAH}$ for the trace semantics of section \ref{sect:Trace-Semantics} \cite{DBLP:conf/sosp/Cohen77,DBLP:conf/sp/GoguenM82a,DBLP:conf/sp/GoguenM84}. Let $\mathtt{L}\in \mathbb{X}$ be a low variable, we have
\begin{eqntabular}{rcl@{\qquad}}
\textit{NI}&\triangleq&\{P\in\wp(\Sigma^{+})\mid\forall {\sigma_1\pi_1\sigma'_1},{\sigma_2\pi_2\sigma'_2}\in P\mathrel{.}(\sigma_1(\mathtt{L})=\sigma_2(\mathtt{L}))\Rightarrow(\sigma'_1(\mathtt{L})=\sigma'_2(\mathtt{L}))\}\label{eq:def:NI}\
\end{eqntabular}
We have $\textit{NI}\in\mathcal{AAH}$ by defining
$A\triangleq\{\pair{\sigma_1\pi_1\sigma'_1}{\sigma_2\pi_2\sigma'_2}\mid(\sigma_1(\mathtt{L})=\sigma_2(\mathtt{L}))\Rightarrow(\sigma'_1(\mathtt{L})=\sigma'_2(\mathtt{L}))\}$.
\end{example}
\section{Logic Rule for Chain Limit Order Ideal Abstract Semantic Properties}\label{sec:LogicRuleChainLimitOrderIdealAbstractSemanticProperties}
\cite[sect.\ 5.3]{DBLP:conf/pldi/DardinierM24} have introduced a sound but incomplete logic for proving $\forall^{\ast}\exists^{\ast}$ hyperproperties. We generalize the rule in our algebraic lattice-theoretic framework for the chain limit abstract semantic properties in ${{\maccent{\alpha}{\ast}}^{{\uparrow}}(\wp(\mathbb{L}))}$.

\subsection{A Sound and Incomplete Rule}
\cite{DBLP:conf/pldi/DardinierM24} does not consider \texttt{break}s and nontermination so that the fields $\bot$ and $b$ of $\triple{e:F}{\bot:I}{b:B}$ in (\ref{eq:def:Abstract-Semantic-Domain-Semantics}) can be ignored and the tuple reduces to the value $F$ of the field $e$. In this section, \ref{def:abstract:domain:well:def:finite:domain} is a lattice which is increasing chain complete, \ref{def:abstract:domain:well:def:infinite:domain} and \ref{def:abstract:domain:well:def:oo:absorbent} are omitted, and limits of increasing chains are assumed to be preserved in \ref{def:abstract:domain:well:def:join:additive}. We also assume that 
$\sqb{\texttt{$\neg$B}}_{e}^{\sharp}\mathbin{\fatsemi^{\sharp}}\sqb{\texttt{$\neg$B}}_{e}^{\sharp}$
=
$\sqb{\texttt{$\neg$B}}_{e}^{\sharp}$,
$\sqb{\texttt{$\neg$B}}_{e}^{\sharp}\mathbin{\fatsemi^{\sharp}}\sqb{\texttt{B}}_{e}^{\sharp}$
=
$\sqb{\texttt{B}}_{e}^{\sharp}\mathbin{\fatsemi^{\sharp}}\sqb{\texttt{$\neg$B}}_{e}^{\sharp}$
=
${\bot}_{+}^{\sharp}$, and 
$\sqb{\texttt{skip}}_{e}^{\sharp}\triangleq{\textsf{skip}^{\sharp}}$
=
${\textsf{init}^{\sharp}}$ in (\ref{eq:def:sem:abstract:basis}), which does not hold for traces but holds e.g.\ for a relational semantics.

The rule of \cite{DBLP:conf/pldi/DardinierM24} generalizes to
\bgroup\abovedisplayskip5pt\belowdisplayskip5pt\begin{eqntabular}{c@{\qquad}}
\frac{\,\mathcal{P}\subseteq\mathcal{I},\quad\overline{\llbrace}\,\mathcal{I}\,\overline{\rrbrace}\,\texttt{if (B) else skip}\,\overline{\llbrace}\,\mathcal{I}\,\overline{\rrbrace},\quad\overline{\llbrace}\,\mathcal{I}\,\overline{\rrbrace}\,\texttt{$\neg$B}\,\overline{\llbrace}\,\mathcal{Q}\,\overline{\rrbrace}\,}{\,\overline{\llbrace}\,\mathcal{P}\,\overline{\rrbrace}\,\texttt{while (B) S}\,\overline{\llbrace}\,\mathcal{Q}\,\overline{\rrbrace}
\,},\quad\mathcal{Q}\in{\maccent{\alpha}{\ast}}^{{\uparrow}}(\wp(\mathbb{L}^{\sharp}))
\label{eq:rule:exists:forall:generalized}
\end{eqntabular}\egroup
The key idea to prove that for any $P\in\mathcal{P}\in\wp(\mathbb{L}^{\sharp}_{+})$, the exact postcondition $Q={\textsf{post}^\sharp\sqb{\texttt{while (B) S}}_{e}^{\sharp}\,P}$ will be in $\mathcal{Q}$ is to exhibit an increasing chain in $\mathcal{Q}$ with least upper bound $Q$, also in $\mathcal{Q}$ by the hypothesis that $\mathcal{Q}$ is a chain limit order ideal abstract semantic property. \ifshort Soundness follows from theorem \ref{th:rule:exists:forall:generalized:sound}  in the appendix\proofinapx.\fi
\begin{toappendix}
\subsection{A Soundness Proof of \textup{(\ref{eq:rule:exists:forall:generalized})}}
Let $P\in\wp(\mathbb{L}^{\sharp}_{+})$. The iterates $\pair{X^i}{i\in\mathbb{N}\cup\{\omega\}}$ of $\LAMBDA{X}P \mathbin{{\sqcup}^{\sharp}_{+}} \textsf{\textup{post}}^{\sharp}\sqb{\texttt{if(B) S else skip}}^{\sharp}_{e}(X)$ from ${\bot}^{\sharp}_{+}$ are defined as
\begin{eqntabular}{rcl}
X^0&\triangleq&P\nonumber\\[-0.5ex]
X^{n+1}&\triangleq&\textsf{post}^\sharp\sqb{\texttt{if (B) S else skip}}^{\sharp}_{e}\, X^{n}
\label{eq:def:ule:exists:forall:generalized:X:n}\\[-0.5ex]
X^{\omega}&\triangleq&\mathop{{\bigsqcup}^{\sharp}_{+}}\limits_{n\in\mathbb{N}}X^{n}\nonumber
\end{eqntabular}
Since the iterates are a function of $P$, we write $X^i(P)$, ${i\in\mathbb{N}\cup\{\omega\}}$ when this dependency must be made clear.
\begin{lemma}\label{lem:rule:exists:forall:value:X:n}
\begin{eqntabular}{rcl}
\forall n\in\mathbb{N}\mathrel{.}X^n&=&\bigl(\mathop{{\bigsqcup}^{\sharp}_{+}}\limits_{i=0}^{n-1}
\textsf{\textup{post}}^{\sharp}\sqb{\texttt{$\neg$B}}_{e}^{\sharp}((\textsf{\textup{post}}^{\sharp}\sqb{\texttt{B;S}}_{e}^{\sharp})^i P)
\mathbin{{\sqcup}^{\sharp}_{+}}
 \bigl((\textsf{\textup{post}}^{\sharp}\sqb{\texttt{B;S}}_{e}^{\sharp})^{n}P\bigr)
 \label{eq:rule:exists:forall:value:X:n}\\
 {X^{\omega}}
 &=&
 {\mathop{{\bigsqcup}^{\sharp}_{+}}\limits_{n\in\mathbb{N}}
\textsf{\textup{post}}^{\sharp}\sqb{\texttt{$\neg$B}}_{e}^{\sharp}((\textsf{\textup{post}}^{\sharp}\sqb{\texttt{B;S}}_{e}^{\sharp})^n P)
\mathbin{{\sqcup}^{\sharp}_{+}}
 \mathop{{\bigsqcup}^{\sharp}_{+}}\limits_{n\in\mathbb{N}}(\textsf{\textup{post}}^{\sharp}\sqb{\texttt{B;S}}_{e}^{\sharp})^{n}P}\nonumber
\end{eqntabular}
\end{lemma}
\begin{proof}[Proof of lemma \ref{lem:rule:exists:forall:value:X:n}]
The proof is by recurrence on $n$.

\hyphen{5}\quad For the basis, this is ${\bot}^{\sharp}_{+}\mathbin{{\sqcup}^{\sharp}_{+}} \bigl(P\mathbin{\fatsemi^{\sharp}}\textsf{\textup{skip}}\bigr)$ = $P\mathbin{\fatsemi^{\sharp}}\textsf{\textup{init}}$ = $P$ = $X^0$ by hypothesis $\sqb{\texttt{skip}}_{e}^{\sharp}={\textsf{init}^{\sharp}}$, \ref{def:abstract:domain:well:def:finite:domain}, and \ref{def:abstract:domain:well:def:init:neutral}.

\smallskip

\hyphen{5}\quad For the induction step, we observe that $\pair{X^i}{i\leqslant n}$ is a $\mathrel{{\sqsubseteq}^{\sharp}_{+}}$-increasing chain, by definition of the lub $\mathbin{{\sqcup}^{\sharp}_{+}}$. Then 
\begin{calculus}[=\ \ ]
\formula{X^{n+1}}\\
=
\formulaexplanation{\textsf{post}^\sharp\sqb{\texttt{if (B) S else skip}}^{\sharp}_{e}\, X^{n}}{def.\ ${X^{n+1}}$}\\
=
\formulaexplanation{\textsf{\textup{post}}^\sharp\sqb{\texttt{B;S}}^{\sharp}_{e}\,X^{n}\mathbin{\,\sqcup^{\sharp}_{+}\,}\textsf{\textup{post}}^\sharp\sqb{\neg\texttt{B;skip}}^{\sharp}_{e}\,X^{n}}{(\ref{eq:post:abstract:if})}\\
=
\formulaexplanation{\Bigl(X^{n}\mathbin{\fatsemi^{\sharp}}\sqb{\texttt{B;S}}^{\sharp}_{e}\Bigr)\mathbin{\,\sqcup^{\sharp}_{+}\,}\Bigl(X^{n}\mathbin{\fatsemi^{\sharp}}\sqb{\neg\texttt{B}}^{\sharp}_{e}\Bigr)}{by def.\ (\ref{eq:def:abstract:transformer:post}) of \textsf{post}, $\sqb{\texttt{skip}}_{e}^{\sharp}={\textsf{init}^{\sharp}}$, and \ref{def:abstract:domain:well:def:init:neutral}}\\
=
\formulaexplanation{\Bigl(\bigl(\mathop{{\bigsqcup}^{\sharp}_{+}}\limits_{i=0}^{n-1}P\mathbin{\fatsemi^{\sharp}}(\sqb{\texttt{B;S}}_{e}^{\sharp})^{i}\mathbin{\fatsemi^{\sharp}}\sqb{\texttt{$\neg$B}}_{e}^{\sharp}\bigr)
\mathbin{{\sqcup}^{\sharp}_{+}}
 \bigl(P\mathbin{\fatsemi^{\sharp}}(\sqb{\texttt{B;S}}_{e}^{\sharp})^{n}\bigr)\mathbin{\fatsemi^{\sharp}}\sqb{\texttt{B;S}}^{\sharp}_{e}\Bigr)
 \mathbin{\,\sqcup^{\sharp}_{+}\,}
 \Bigl(\bigl(\mathop{{\bigsqcup}^{\sharp}_{+}}\limits_{i=0}^{n-1}P\mathbin{\fatsemi^{\sharp}}(\sqb{\texttt{B;S}}_{e}^{\sharp})^{i}\mathbin{\fatsemi^{\sharp}}\sqb{\texttt{$\neg$B}}_{e}^{\sharp}\bigr)
\mathbin{{\sqcup}^{\sharp}_{+}}
 \bigl(P\mathbin{\fatsemi^{\sharp}}(\sqb{\texttt{B;S}}_{e}^{\sharp})^{n}\bigr)\mathbin{\fatsemi^{\sharp}}\sqb{\neg\texttt{B}}^{\sharp}_{e}\Bigr)}{induction hypothesis}\\
=
\formula{\Bigl(\bigl(\mathop{{\bigsqcup}^{\sharp}_{+}}\limits_{i=0}^{n-1}P\mathbin{\fatsemi^{\sharp}}(\sqb{\texttt{B;S}}_{e}^{\sharp})^{i}\mathbin{\fatsemi^{\sharp}}\sqb{\texttt{$\neg$B}}_{e}^{\sharp}\bigr)
\mathbin{{\sqcup}^{\sharp}_{+}}
 \bigl(P\mathbin{\fatsemi^{\sharp}}(\sqb{\texttt{B;S}}_{e}^{\sharp})^{n+1}\bigr)\Bigr)
 \mathbin{\,\sqcup^{\sharp}_{+}\,}
\bigl(\mathop{{\bigsqcup}^{\sharp}_{+}}\limits_{i=0}^{n}P\mathbin{\fatsemi^{\sharp}}(\sqb{\texttt{B;S}}_{e}^{\sharp})^{i}\mathbin{\fatsemi^{\sharp}}\sqb{\texttt{$\neg$B}}_{e}^{\sharp}\bigr)
}\\
\rightexplanation{integrating the term $P\mathbin{\fatsemi^{\sharp}}(\sqb{\texttt{B;S}}_{e}^{\sharp})^{n}\bigr)\mathbin{\fatsemi^{\sharp}}\sqb{\neg\texttt{B}}^{\sharp}_{e}$ in the join and $\sqb{\texttt{B;S}}^{\sharp}_{e}$ in $P\mathbin{\fatsemi^{\sharp}}(\sqb{\texttt{B;S}}_{e}^{\sharp})^{n}\bigr)$}\\
 =
\formulaexplanation{\bigl(\mathop{{\bigsqcup}^{\sharp}_{+}}\limits_{i=0}^{n}P\mathbin{\fatsemi^{\sharp}}(\sqb{\texttt{B;S}}_{e}^{\sharp})^{i}\mathbin{\fatsemi^{\sharp}}\sqb{\texttt{$\neg$B}}_{e}^{\sharp}\bigr)
\mathbin{{\sqcup}^{\sharp}_{+}}
 \bigl(P\mathbin{\fatsemi^{\sharp}}(\sqb{\texttt{B;S}}_{e}^{\sharp})^{n+1}\bigr)
}{idempotence of $\mathbin{{\sqcup}^{\sharp}_{+}}$}\\
 =
\formulaexplanation{\bigl(\mathop{{\bigsqcup}^{\sharp}_{+}}\limits_{i=0}^{n}
\textsf{\textup{post}}^{\sharp}\sqb{\texttt{$\neg$B}}_{e}^{\sharp}(\textsf{\textup{post}}^{\sharp}\sqb{\texttt{B;S}}_{e}^{\sharp})^i P)
\mathbin{{\sqcup}^{\sharp}_{+}}
 \bigl((\textsf{\textup{post}}^{\sharp}\sqb{\texttt{B;S}}_{e}^{\sharp})^{n+1}P\bigr)}{def.\ (\ref{eq:def:abstract:transformer:post}) of \textsf{post}}
\end{calculus}

\smallskip

\hyphen{5}\quad For the limit, we have
\begin{calculus}[=\ \ ]
\formula{X^{\omega}}\\
=
\formulaexplanation{\mathop{{\bigsqcup}^{\sharp}_{+}}\limits_{n\in\mathbb{N}}\Bigl(\bigl(\mathop{{\bigsqcup}^{\sharp}_{+}}\limits_{i=0}^{n-1}
\textsf{\textup{post}}^{\sharp}\sqb{\texttt{$\neg$B}}_{e}^{\sharp}(\textsf{\textup{post}}^{\sharp}\sqb{\texttt{B;S}}_{e}^{\sharp})^i P\bigr)
\mathbin{{\sqcup}^{\sharp}_{+}}
 \bigl((\textsf{\textup{post}}^{\sharp}\sqb{\texttt{B;S}}_{e}^{\sharp})^{n}P\bigr)\Bigr)}{(\ref{eq:def:ule:exists:forall:generalized:X:n}) and (\ref{eq:rule:exists:forall:value:X:n})}\\
=
\lastformulaexplanation{\mathop{{\bigsqcup}^{\sharp}_{+}}\limits_{n\in\mathbb{N}}
\textsf{\textup{post}}^{\sharp}\sqb{\texttt{$\neg$B}}_{e}^{\sharp}((\textsf{\textup{post}}^{\sharp}\sqb{\texttt{B;S}}_{e}^{\sharp})^n P)
\mathbin{{\sqcup}^{\sharp}_{+}}
 \mathop{{\bigsqcup}^{\sharp}_{+}}\limits_{n\in\mathbb{N}}(\textsf{\textup{post}}^{\sharp}\sqb{\texttt{B;S}}_{e}^{\sharp})^{n}P}{$\mathbin{{\sqcup}^{\sharp}_{+}}$ associative}{\mbox{\qed}}
\end{calculus}
\let\qed\relax
\end{proof}
\begin{lemma}\label{lem:rule:while:lfp} For all  $P\in\wp(\mathbb{L}^{\sharp}_{+})$,
\begin{eqntabular}{rcl}
\Lfp{\sqsubseteq}\LAMBDA{X}{P \mathbin{{\sqcup}^{\sharp}_{+}} \textsf{\textup{post}}^{\sharp}\sqb{\texttt{if(B) S else skip}}^{\sharp}_{e}(X)} 
&=&X^{\omega} \label{eq:rule:while:lfp}
\end{eqntabular}
\end{lemma}
\begin{proof}[Proof of lemma \ref{lem:rule:while:lfp}]
The iterates $\pair{X^i}{i\in\mathbb{N}\cup\{\omega\}}$ are characterized in lemma \ref{lem:rule:exists:forall:value:X:n}. Let use prove that $X^{\omega}=P \mathbin{{\sqcup}^{\sharp}_{+}} \textsf{post}^\sharp\sqb{\texttt{if (B) S else skip}}^{\sharp}_{e}\, X^{\omega}$ is a fixpoint.
\begin{calculus}[=\ \ ]
\formula{P \mathbin{{\sqcup}^{\sharp}_{+}} \textsf{post}^\sharp\sqb{\texttt{if (B) S else skip}}^{\sharp}_{e}\, X^{\omega}}\\
=
\formulaexplanation{P \mathbin{{\sqcup}^{\sharp}_{+}} \textsf{\textup{post}}^\sharp\sqb{\texttt{B;S}}^{\sharp}_{e}\,X^{\omega}\mathbin{\,\sqcup^{\sharp}_{+}\,}\textsf{\textup{post}}^\sharp\sqb{\neg\texttt{B;skip}}^{\sharp}_{e}\,X^{\omega}}{(\ref{eq:post:abstract:if})}\\
=
\formulaexplanation{P \mathbin{{\sqcup}^{\sharp}_{+}} \Bigl(X^{\omega}\mathbin{\fatsemi^{\sharp}}\sqb{\texttt{B;S}}^{\sharp}_{e}\Bigr)\mathbin{\,\sqcup^{\sharp}_{+}\,}\Bigl(X^{\omega}\mathbin{\fatsemi^{\sharp}}\sqb{\neg\texttt{B}}^{\sharp}_{e}\Bigr)}{by def.\ (\ref{eq:def:abstract:transformer:post}) of \textsf{post}, $\sqb{\texttt{skip}}_{e}^{\sharp}={\textsf{init}^{\sharp}}$, and \ref{def:abstract:domain:well:def:init:neutral}}\\
=
\formula{P \mathbin{{\sqcup}^{\sharp}_{+}} \Bigl(\Bigl({\mathop{{\bigsqcup}^{\sharp}_{+}}\limits_{n\in\mathbb{N}}
\textsf{\textup{post}}^{\sharp}\sqb{\texttt{$\neg$B}}_{e}^{\sharp}((\textsf{\textup{post}}^{\sharp}\sqb{\texttt{B;S}}_{e}^{\sharp})^n P)
\mathbin{{\sqcup}^{\sharp}_{+}}
 \mathop{{\bigsqcup}^{\sharp}_{+}}\limits_{n\in\mathbb{N}}(\textsf{\textup{post}}^{\sharp}\sqb{\texttt{B;S}}_{e}^{\sharp})^{n}P}\Bigr)\mathbin{\fatsemi^{\sharp}}\sqb{\texttt{B;S}}^{\sharp}_{e}\Bigr)\\
 \mathbin{\,\sqcup^{\sharp}_{+}\,}\\
 \Bigl(\Bigl({\mathop{{\bigsqcup}^{\sharp}_{+}}\limits_{n\in\mathbb{N}}
\textsf{\textup{post}}^{\sharp}\sqb{\texttt{$\neg$B}}_{e}^{\sharp}((\textsf{\textup{post}}^{\sharp}\sqb{\texttt{B;S}}_{e}^{\sharp})^n P)
\mathbin{{\sqcup}^{\sharp}_{+}}
 \mathop{{\bigsqcup}^{\sharp}_{+}}\limits_{n\in\mathbb{N}}(\textsf{\textup{post}}^{\sharp}\sqb{\texttt{B;S}}_{e}^{\sharp})^{n}P}\Bigr)\mathbin{\fatsemi^{\sharp}}\sqb{\neg\texttt{B}}^{\sharp}_{e}\Bigr)}\\[-1ex]
 \rightexplanation{characterization (\ref{eq:rule:exists:forall:value:X:n}) of $X^{\omega}$}\\
 =
\formulaexplanation{P \mathbin{{\sqcup}^{\sharp}_{+}} \Bigl(\textsf{\textup{post}}^{\sharp}\sqb{\texttt{B;S}}^{\sharp}_{e}\Bigl({\mathop{{\bigsqcup}^{\sharp}_{+}}\limits_{n\in\mathbb{N}}
\textsf{\textup{post}}^{\sharp}\sqb{\texttt{$\neg$B}}_{e}^{\sharp}((\textsf{\textup{post}}^{\sharp}\sqb{\texttt{B;S}}_{e}^{\sharp})^n P)
\mathbin{{\sqcup}^{\sharp}_{+}}
 \mathop{{\bigsqcup}^{\sharp}_{+}}\limits_{n\in\mathbb{N}}(\textsf{\textup{post}}^{\sharp}\sqb{\texttt{B;S}}_{e}^{\sharp})^{n}P}\Bigr)\Bigr)\\
 \mathbin{\,\sqcup^{\sharp}_{+}\,}\\
 \Bigl(\textsf{\textup{post}}^{\sharp}\sqb{\neg\texttt{B}}^{\sharp}_{e}\Bigl({\mathop{{\bigsqcup}^{\sharp}_{+}}\limits_{n\in\mathbb{N}}
\textsf{\textup{post}}^{\sharp}\sqb{\texttt{$\neg$B}}_{e}^{\sharp}((\textsf{\textup{post}}^{\sharp}\sqb{\texttt{B;S}}_{e}^{\sharp})^n P)
\mathbin{{\sqcup}^{\sharp}_{+}}
 \mathop{{\bigsqcup}^{\sharp}_{+}}\limits_{n\in\mathbb{N}}(\textsf{\textup{post}}^{\sharp}\sqb{\texttt{B;S}}_{e}^{\sharp})^{n}P}\Bigr)\Bigr)}{def.\ (\ref{eq:def:abstract:transformer:post}) of \textsf{post}}\\
 =
\formula{P \mathbin{{\sqcup}^{\sharp}_{+}} \Bigl({\mathop{{\bigsqcup}^{\sharp}_{+}}\limits_{n\in\mathbb{N}}
\textsf{\textup{post}}^{\sharp}\sqb{\texttt{B;S}}^{\sharp}_{e}(\textsf{\textup{post}}^{\sharp}\sqb{\texttt{$\neg$B}}_{e}^{\sharp}((\textsf{\textup{post}}^{\sharp}\sqb{\texttt{B;S}}_{e}^{\sharp})^n P))
\mathbin{{\sqcup}^{\sharp}_{+}}
 \mathop{{\bigsqcup}^{\sharp}_{+}}\limits_{n\in\mathbb{N}}\textsf{\textup{post}}^{\sharp}\sqb{\texttt{B;S}}^{\sharp}_{e}(\textsf{\textup{post}}^{\sharp}\sqb{\texttt{B;S}}_{e}^{\sharp})^{n})P}\Bigr)\\
 \mathbin{\,\sqcup^{\sharp}_{+}\,}\\
 \Bigl(={\mathop{{\bigsqcup}^{\sharp}_{+}}\limits_{n\in\mathbb{N}}
\textsf{\textup{post}}^{\sharp}\sqb{\neg\texttt{B}}^{\sharp}_{e}(\textsf{\textup{post}}^{\sharp}\sqb{\texttt{$\neg$B}}_{e}^{\sharp}((\textsf{\textup{post}}^{\sharp}\sqb{\texttt{B;S}}_{e}^{\sharp})^n P))
\mathbin{{\sqcup}^{\sharp}_{+}}
 \mathop{{\bigsqcup}^{\sharp}_{+}}\limits_{n\in\mathbb{N}}\textsf{\textup{post}}^{\sharp}\sqb{\neg\texttt{B}}^{\sharp}_{e}((\textsf{\textup{post}}^{\sharp}\sqb{\texttt{B;S}}_{e}^{\sharp})^{n}P)}\Bigr)}\\
 \rightexplanation{$\textsf{\textup{post}}^{\sharp}(S)$ preserve joins of increasing chains by (\ref{def:abstract:domain:well:def:join:additive})}\\
=
\formula{\Bigl(P \mathbin{{\sqcup}^{\sharp}_{+}} 
 \mathop{{\bigsqcup}^{\sharp}_{+}}\limits_{n\in\mathbb{N}}(\textsf{\textup{post}}^{\sharp}\sqb{\texttt{B;S}}_{e}^{\sharp})^{n+1}P\Bigr)
 \mathbin{\,\sqcup^{\sharp}_{+}\,}
\Bigl(\mathop{{\bigsqcup}^{\sharp}_{+}}\limits_{n\in\mathbb{N}}
\textsf{\textup{post}}^{\sharp}\sqb{\texttt{$\neg$B}}_{e}^{\sharp}((\textsf{\textup{post}}^{\sharp}\sqb{\texttt{B;S}}_{e}^{\sharp})^n P)\Bigr)}\\
\rightexplanation{def.\ (\ref{eq:def:abstract:transformer:post}) of \textsf{post} and hypotheses $\sqb{\texttt{$\neg$B}}_{e}^{\sharp}\mathbin{\fatsemi^{\sharp}}\sqb{\texttt{$\neg$B}}_{e}^{\sharp}$
=
$\sqb{\texttt{$\neg$B}}_{e}^{\sharp}$ and
$\sqb{\texttt{$\neg$B}}_{e}^{\sharp}\mathbin{\fatsemi^{\sharp}}\sqb{\texttt{B}}_{e}^{\sharp}$
=
$\sqb{\texttt{B}}_{e}^{\sharp}\mathbin{\fatsemi^{\sharp}}\sqb{\texttt{$\neg$B}}_{e}^{\sharp}$
}\\
=
\formulaexplanation{X^{\omega}}{integrating $P$ in the join with $\textsf{\textup{post}}^{\sharp}(S)^0=\mathbb{1}$ and commutativity}
\end{calculus}

\smallskip

\noindent
It follows that the sequence $\pair{X^i}{i\in\mathbb{N}\cup\{\omega\}}$ is increasing and stationary at $\omega$ which is therefore the least fixpoint (\ref{eq:rule:while:lfp}).
\end{proof}
\begin{lemma}\label{lem:charaterization:psot:while:e}
\begin{eqntabular}{rclcl}
\textsf{\textup{post}}^\sharp\sqb{\texttt{while (B) S}}^{\sharp}_{e}P
&=&{\textsf{\textup{post}}^\sharp\sqb{\texttt{$\neg$B}}^{\sharp}_{e}\, X^{\omega}}
&=&
\mathop{{\bigsqcup}^{\sharp}_{+}}\limits_{n\in\mathbb{N}}
\textsf{\textup{post}}^{\sharp}\sqb{\texttt{$\neg$B}}_{e}^{\sharp}((\textsf{\textup{post}}^{\sharp}\sqb{\texttt{B;S}}_{e}^{\sharp})^n P)
\label{eq:charaterization:psot:while:e}
\end{eqntabular}
\end{lemma}
\begin{proof}[Proof of lemma \ref{lem:charaterization:psot:while:e}]
\begin{calculus}[=\ \ ]
\formula{\textsf{\textup{post}}^\sharp\sqb{\texttt{while (B) S}}^{\sharp}_{e}P}\\
=
\formulaexplanation{\langle ok:\langle e:{\textsf{\textup{post}}^\sharp(\sqb{\neg\texttt{B}}_{e}^{\sharp}\mathbin{\sqcup_{e}^{\sharp}}\sqb{\texttt{B;S}}_{b}^{\sharp})(\Lfp{{\sqsubseteq}_{+}^{\sharp}}({\vec{F}_{pe}^{\sharp}}(P)))},
\bot:{\textsf{\textup{post}}^\sharp(\sqb{\texttt{B;S}}_{\bot}^{\sharp})(\Lfp{{\sqsubseteq}_{+}^{\sharp}}({\vec{F}_{pe}^{\sharp}}(P)))}\mathbin{{\sqcup}_{\infty}^{\sharp}}{}\textsf{\textup{post}}^\sharp(\Gfp{{\sqsubseteq}_{\infty}^{\sharp}}{F_{p\bot}^{\sharp}})P \rangle,
{br:P_{br}}\rangle_{e}}{by (\ref{eq:post:abstract:while})}\\
=
\formula{\textsf{\textup{post}}^\sharp(\sqb{\neg\texttt{B}}_{e}^{\sharp})(\Lfp{{\sqsubseteq}_{+}^{\sharp}}(\LAMBDA{X}{\textsf{\textup{post}}^\sharp({\textsf{\textup{init}}^{\sharp}})P \mathbin{\sqcup_{+}^{\sharp}} \textsf{\textup{post}}^\sharp(\sqb{\texttt{B;S}}_{e}^{\sharp})(X)}))}\\[-0.5ex]\rightexplanation{in absence of \texttt{break}s and ignoring non termination}\\
=
\formulaexplanation{\textsf{\textup{post}}^\sharp(\sqb{\neg\texttt{B}}_{e}^{\sharp})(\Lfp{{\sqsubseteq}_{+}^{\sharp}}(\LAMBDA{X}{\textsf{\textup{post}}^\sharp({\textsf{\textup{init}}^{\sharp}})P \mathbin{\sqcup_{+}^{\sharp}} \textsf{\textup{post}}^\sharp(\sqb{\texttt{B;S}}_{e}^{\sharp})(X)}))}{def.\ (\ref{eq:def:F-pe-sharp}) of ${\vec{F}_{pe}^{\sharp}}$}\\
=
\formulaexplanation{\textsf{\textup{post}}^\sharp(\sqb{\neg\texttt{B}}_{e}^{\sharp})(\Lfp{{\sqsubseteq}_{+}^{\sharp}}(\LAMBDA{X}{P \mathbin{\sqcup_{+}^{\sharp}} \textsf{\textup{post}}^\sharp(\sqb{\texttt{B;S}}_{e}^{\sharp})(X)}))}{def.\ (\ref{eq:def:abstract:transformer:post}) of $\textsf{\textup{post}}^\sharp$ and \ref{def:abstract:domain:well:def:init:neutral}}\\
=
\formulaexplanation{\textsf{\textup{post}}^\sharp(\sqb{\neg\texttt{B}}_{e}^{\sharp})(X^{\omega})}{lemma \ref{lem:rule:while:lfp}}\\
 =
\formulaexplanation{\textsf{\textup{post}}^\sharp(\sqb{\neg\texttt{B}}_{e}^{\sharp})(\mathop{{\bigsqcup}^{\sharp}_{+}}\limits_{n\in\mathbb{N}}
\textsf{\textup{post}}^{\sharp}\sqb{\texttt{$\neg$B}}_{e}^{\sharp}((\textsf{\textup{post}}^{\sharp}\sqb{\texttt{B;S}}_{e}^{\sharp})^n P)
\mathbin{{\sqcup}^{\sharp}_{+}}
 \mathop{{\bigsqcup}^{\sharp}_{+}}\limits_{n\in\mathbb{N}}(\textsf{\textup{post}}^{\sharp}\sqb{\texttt{B;S}}_{e}^{\sharp})^{n}P)}{lemma \ref{lem:rule:while:lfp}}\\
=
\formula{\mathop{{\bigsqcup}^{\sharp}_{+}}\limits_{n\in\mathbb{N}}
\textsf{\textup{post}}^\sharp\sqb{\neg\texttt{B}}_{e}^{\sharp}(\textsf{\textup{post}}^{\sharp}\sqb{\texttt{$\neg$B}}_{e}^{\sharp}((\textsf{\textup{post}}^{\sharp}\sqb{\texttt{B;S}}_{e}^{\sharp})^n P))
\mathbin{{\sqcup}^{\sharp}_{+}}
\mathop{{\bigsqcup}^{\sharp}_{+}}\limits_{n\in\mathbb{N}}\textsf{\textup{post}}^\sharp\sqb{\neg\texttt{B}}_{e}^{\sharp} (\textsf{\textup{post}}^{\sharp}\sqb{\texttt{B;S}}_{e}^{\sharp})^{n}P)}\\
 \rightexplanation{$\textsf{\textup{post}}^\sharp(S)$ preserves joins $\mathbin{{\sqcup}^{\sharp}_{+}}$ by  def.\ (\ref{eq:def:abstract:transformer:post}) of \textsf{post} and \ref{def:abstract:domain:well:def:join:additive}}\\
=
\formula{\mathop{{\bigsqcup}^{\sharp}_{+}}\limits_{n\in\mathbb{N}}
\textsf{\textup{post}}^{\sharp}\sqb{\texttt{$\neg$B}}_{e}^{\sharp}((\textsf{\textup{post}}^{\sharp}\sqb{\texttt{B;S}}_{e}^{\sharp})^n P)}\\
\lastexplanation{(\ref{eq:post:abstract:seq}), def.\ (\ref{eq:def:abstract:transformer:post}) of $\textsf{\textup{post}}^\sharp$, hypotheses $\sqb{\texttt{B}}_{e}^{\sharp}\mathbin{\fatsemi^{\sharp}}\sqb{\texttt{$\neg$B}}_{e}^{\sharp}$
=
${\bot}_{+}^{\sharp}$ and $\sqb{\texttt{$\neg$B}}_{e}^{\sharp}\mathbin{\fatsemi^{\sharp}}\sqb{\texttt{$\neg$B}}_{e}^{\sharp}$
=
$\sqb{\texttt{$\neg$B}}_{e}^{\sharp}$, def.\ function powers, and  def.\ lub}{\mbox{\qed}}
\end{calculus}
\let\qed\relax
\end{proof}

\begin{theorem}\label{th:rule:exists:forall:generalized:sound}Proof rule \textup{(\ref{eq:rule:exists:forall:generalized})} is sound.
\end{theorem}
\begin{proof}[Proof of theorem \ref{th:rule:exists:forall:generalized:sound}]
Observe that if $P\in\mathcal{P}$ then $X^0=P\in\mathcal{I}$ by $\mathcal{P}\subseteq\mathcal{I}$ by the premise of the rule (\ref{eq:rule:exists:forall:generalized}). Assume $X^n\in\mathcal{I}$ then, by (\ref{eq:def:abstract:logical:triples}),
$\overline{\llbrace}\,\mathcal{I}\,\overline{\rrbrace}\,\texttt{if (B) else skip}\,\overline{\llbrace}\,\mathcal{I}\,\overline{\rrbrace}$
if and only if
$\forall  P\in\mathcal{I}\!\mathrel{.}\textsf{post}^\sharp\sqb{\texttt{if (B) else skip}}^\sharp P\in\mathcal{I}$
if and only if
$\forall  P\in\mathcal{I}\!\mathrel{.}\textsf{post}^\sharp\sqb{\texttt{if (B) else skip}}^\sharp_{e} P\in\mathcal{I}$ since nontermination and \texttt{break}s are ignored. By (\ref {eq:def:ule:exists:forall:generalized:X:n}), this implies that $X^{n+1}\in\mathcal{I}$. By recurrence
$\forall n\in\mathbb{N}\mathrel{.}X^n\in \mathcal{I}$.

By (\ref{eq:def:abstract:transformer:post}) and \ref{def:abstract:domain:well:def:join:additive}, $\textsf{post}^\sharp\sqb{\texttt{$\neg$B}}^{\sharp}_{e}$ is increasing so that the sequence $\pair{\textsf{post}^\sharp\sqb{\texttt{$\neg$B}}^{\sharp}_{e}\, X^{n}}{n\in\mathbb{N}\cup\{\omega\}}$ is increasing.


By (\ref{eq:def:abstract:logical:triples}),
$\overline{\llbrace}\,\mathcal{I}\,\overline{\rrbrace}\,\texttt{$\neg$B}\,\overline{\llbrace}\,\mathcal{Q}\,\overline{\rrbrace}$
$\Leftrightarrow$
$\forall  P\in\mathcal{I}\!\mathrel{.}\textsf{post}^\sharp\sqb{\texttt{$\neg$B}}^\sharp P\in\mathcal{Q}$
$\Leftrightarrow$
$\forall  P\in\mathcal{I}\!\mathrel{.}\textsf{post}^\sharp\sqb{\texttt{$\neg$B}}^\sharp_{e} P\in\mathcal{Q}$ since nontermination and \texttt{break}s are ignored. Since $\forall n\in\mathbb{N}\mathrel{.}X^n\in\mathcal{I}$, this implies that $\forall n\in\mathbb{N}\mathrel{.}\textsf{post}^\sharp\sqb{\texttt{$\neg$B}}^{\sharp}_{e}\, X^{n}\in\mathcal{Q}$. By hypothesis, $\mathcal{Q}\in{\maccent{\alpha}{\ast}}^{{\sqsubseteq}{\uparrow}}(\wp(\mathbb{L}^{\sharp}_{+}))$ so that by the dual of (\ref{eq:def:alpha:downarrow}), ${\textsf{post}^\sharp\sqb{\texttt{$\neg$B}}^{\sharp}_{e}\, X^{\omega}}\in\mathcal{Q}$. It follows by (\ref{eq:charaterization:psot:while:e}) that $\textsf{\textup{post}}^\sharp\sqb{\texttt{while (B) S}}^{\sharp}_{e}P\in\mathcal{Q}$. 

\smallskip

 We conclude that $\forall P\in\mathcal{P}\mathrel{.}{\textsf{post}^\sharp\sqb{\texttt{while (B) S}}_{e}^{\sharp}\,P}={\textsf{post}^\sharp\sqb{\texttt{$\neg$B}}^{\sharp}_{e}\, X^{\omega}}\in\mathcal{Q}$ which, by (\ref{eq:def:abstract:logical:triples}), implies that $\overline{\llbrace}\,\mathcal{P}\,\overline{\rrbrace}\,\texttt{while (B) S}\,\overline{\llbrace}\,\mathcal{Q}\,\overline{\rrbrace}$, proving soundness of the rule (\ref{eq:rule:exists:forall:generalized}).
\end{proof}
\end{toappendix}
\ifshort A counter-example proving incompleteness is also provided by lemma \ref{lem:rule:exists:forall:generalized:incomplete} in the appendix\proofinapx. \fi
\begin{toappendix}
\begin{lemma}\label{lem:rule:exists:forall:generalized:incomplete}Proof rule \textup{(\ref{eq:rule:exists:forall:generalized})} is incomplete.
\end{lemma}
\begin{proof}[Proof of lemma (\ref{lem:rule:exists:forall:generalized:incomplete})]
Consider $\overline{\llbrace}\,\{P\}\,\overline{\rrbrace}\,\texttt{while (B) S}\,\overline{\llbrace}\,\{\textsf{post}^\sharp\sqb{\texttt{while (B) S}}_{e}^{\sharp}P\}\,\overline{\rrbrace}$ which holds by (\ref{eq:def:abstract:logical:triples}). 
Since ${\maccent{\alpha}{\ast}}^{{\sqsubseteq}{\uparrow}}(\{\textsf{post}^\sharp\sqb{\texttt{while (B) S}}_{e}^{\sharp}P\})=\{\textsf{post}^\sharp\sqb{\texttt{while (B) S}}_{e}^{\sharp}P\},$ we can apply proof rule (\ref{eq:rule:exists:forall:generalized}). By $\mathcal{P}\subseteq\mathcal{I}$, we should have
$P\in\mathcal{I}$ so $X^0(P)\in\mathcal{I}$. The second condition $\overline{\llbrace}\,\mathcal{I}\,\overline{\rrbrace}\,\texttt{if (B) else skip}\,\overline{\llbrace}\,\mathcal{I}\,\overline{\rrbrace}$ of the premise implies, by (\ref{eq:def:abstract:logical:triples}), that $\forall P\in \mathcal{I}\mathrel{.}\textsf{post}^\sharp\sqb{\texttt{if (B) S else skip}}^{\sharp}_{e}\, P$. Therefore by (\ref{eq:def:ule:exists:forall:generalized:X:n}) and recurrence, 
$\forall n\in\mathbb{N}\mathrel{.}X^{n}(P)\in\mathcal{I}$. Then the third condition of the premiss, requires that $\overline{\llbrace}\,\mathcal{I}\,\overline{\rrbrace}\,\texttt{$\neg$B}\,\overline{\llbrace}\,\mathcal{Q}\,\overline{\rrbrace}$, equivalently, by (\ref{eq:def:abstract:logical:triples}), 
$\forall  P\in\mathcal{I}\!\mathrel{.}\textsf{post}^\sharp\sqb{\texttt{$\neg$B}}^\sharp P\in\{\textsf{post}^\sharp\sqb{\texttt{while (B) S}}_{e}^{\sharp}P\}$ and therefore $\forall  P\in\mathcal{I}\!\mathrel{.}\textsf{post}^\sharp\sqb{\texttt{$\neg$B}}^\sharp P$ = $\textsf{post}^\sharp\sqb{\texttt{while (B) S}}_{e}^{\sharp}P$. In particular, we must have $\forall n\in\mathbb{N}\mathrel{.}\textsf{post}^\sharp\sqb{\texttt{$\neg$B}}^\sharp X^n(P)$ = $\textsf{post}^\sharp\sqb{\texttt{while (B) S}}_{e}^{\sharp}P$.
By the characterization (\ref{eq:charaterization:psot:while:e}) of $\textsf{\textup{post}}^\sharp\sqb{\texttt{while (B) S}}^{\sharp}_{e}P$, this means that $\forall n\in\mathbb{N}\mathrel{.}\textsf{post}^\sharp\sqb{\texttt{$\neg$B}}^\sharp X^n(P)=\mathop{{\bigsqcup}^{\sharp}_{+}}\limits_{n\in\mathbb{N}}
\textsf{\textup{post}}^{\sharp}\sqb{\texttt{$\neg$B}}_{e}^{\sharp}((\textsf{\textup{post}}^{\sharp}\sqb{\texttt{B;S}}_{e}^{\sharp})^n P)$. Otherwise stated the loop is never entered, which, apart from the cas where \texttt{B} is false, is not the case in general.
\end{proof}
\end{toappendix}

\subsection{Completeness Relative to an Abstract Hypercollecting Semantics}\ifshort Proof \else
Lemma \ref{lem:rule:exists:forall:generalized:incomplete} shows that proof \fi rule \textup{(\ref{eq:rule:exists:forall:generalized})} is incomplete relative to the hypercollecting semantics (\ref{eq:Post:abstract:while}) of section \ref{sec:Calculus:Hyper:Properties}. We show that the rule is complete relative to the following abstraction of the hypercollecting semantics (\ref{eq:Post:abstract:while}).
\begin{definition}[Weak structural hypercollecting semantics for iteration]
\bgroup\abovedisplayskip5pt\belowdisplayskip0pt
   \begin{eqntabular}{r@{\ \ }c@{\ \ }l@{\qquad}}
       \adhocHyperSemantics\sqb{\texttt{while(B) S}}^{\sharp}_{e}\mathcal{P} 
&\triangleq& 
\Hpost\sqb{\texttt{$\neg$}B}^\sharp_e(\Lfp{\subseteq}\LAMBDA{\mathcal{X}}{\mathcal{P} \cup \adhocHyperSemantics\sqb{\texttt{if(B) S else skip}}^{\sharp}_{e}(\mathcal{X})}) 
       \label{eq:def:overline:Post:while}
   \end{eqntabular}\egroup
   \end{definition}
\noindent $\adhocHyperSemantics\sqb{\texttt{while(B) S}}^{\sharp}_{e}$ is an algebraic form of the hypercollecting semantics postulated by \cite[p.\ 877]{DBLP:conf/popl/AssafNSTT17}.
\ifshort 
We characterize by theorem \ref{th:characterization:adhocHyperSemantics} in the appendix the executions satisfying (\ref{eq:def:overline:Post:while}) \proofinapx. \fi
\begin{toappendix}
\begin{theorem}[Characterization of the executions satisfying (\ref{eq:def:overline:Post:while})]\label{th:characterization:adhocHyperSemantics}
\begin{eqntabular}{rcl}
\Lfp{\subseteq}\LAMBDA{\mathcal{X}}{\mathcal{P} \cup \adhocHyperSemantics\sqb{\texttt{if(B) S else skip}}^{\sharp}_{e}(\mathcal{X})}
&=&
\{X^{n}(P)\mid P\in\mathcal{P}\wedge n\in\mathbb{N}\}
\label{eq:characterization:adhocHyperSemantics}
\end{eqntabular}
\end{theorem}
\begin{proof}[Proof of theorem \ref{th:characterization:adhocHyperSemantics}]
the iterates of $\LAMBDA{\mathcal{X}}{\mathcal{P} \cup \adhocHyperSemantics\sqb{\texttt{if(B) S else skip}}^{\sharp}_{e}(\mathcal{X})}$ are
   
\begin{calculus}[$\mathcal{X}^{n+1}$\quad]
   $\mathcal{X}^0$  = \formula{\emptyset}\\
   $\mathcal{X}^1$ = \formula{\mathcal{P}}\\
   $\mathcal{X}^2$ =  \formula{\mathcal{P} \cup\adhocHyperSemantics\sqb{\texttt{if (B) S else skip}}^{\sharp}_{e}( \mathcal{P})}\\
   ...\\
   $\mathcal{X}^n$ = \formulaexplanation{\bigcup\limits_{i=0}^{n} (                                                                              
           \adhocHyperSemantics\sqb{\texttt{if (B) S else skip}}^{\sharp}_{e})^i( \mathcal{P})}{induction hypothesis}\\
$\mathcal{X}^{n+1}$ = \formula{\mathcal{P} \cup \adhocHyperSemantics\sqb{\texttt{if(B) S else skip}}^{\sharp}_{e}(\mathcal{X}^n)}\\

\phantom{$\mathcal{X}^{n+1}$} = \formula{\mathcal{P} \cup \adhocHyperSemantics\sqb{\texttt{if(B) S else skip}}^{\sharp}_{e}({\bigcup\limits_{i=0}^{n} (                                                                              
           \adhocHyperSemantics\sqb{\texttt{if (B) S else skip}}^{\sharp}_{e})^i( \mathcal{P})})}\\
           
 \phantom{$\mathcal{X}^{n+1}$} = \formula{\mathcal{P} \cup \bigcup\limits_{i=1}^{n+1} (                                                                              
           \adhocHyperSemantics\sqb{\texttt{if (B) S else skip}}^{\sharp}_{e})^i( \mathcal{P})}\\
           
 \phantom{$\mathcal{X}^{n+1}$} = \formula{\bigcup\limits_{i=0}^{n+1} (                                                                              
           \adhocHyperSemantics\sqb{\texttt{if (B) S else skip}}^{\sharp}_{e})^i( \mathcal{P})}
\end{calculus}  

\smallskip

\noindent A similar calculation shows  that $\mathcal{X}^{\omega}$ = ${\bigcup\limits_{n\in\mathbb{N}} (                                                                              
           \adhocHyperSemantics\sqb{\texttt{if (B) S else skip}}^{\sharp}_{e})^n( \mathcal{P})}$ is stable so is the least fixpoint
$\Lfp{\sqsubseteq}\LAMBDA{\mathcal{X}}\mathcal{P} \cup \adhocHyperSemantics\llbracket\texttt{if (B) S }$ $\texttt{else skip}\rrbracket^{\sharp}_{e}(\mathcal{X})$
=
$\mathcal{X}^{\omega}$.

\smallskip

Observe that we have
\begin{calculus}[=\ \ ]
\hyphen{5}\quad
\formula{(\adhocHyperSemantics\llbracket\texttt{if (B) S }$ $\texttt{else skip}\rrbracket^{\sharp}_{e})^0( \mathcal{P})}\\ 
=
\formulaexplanation{\mathcal{P}}{def.\ function powers}\\
=
\formulaexplanation{\{X^0(P)\mid P\in\mathcal{P}\}}{def.\ function powers}\\
\hyphen{5}\quad\discussion{for induction}\\
\formula{(\adhocHyperSemantics\llbracket\texttt{if (B) S }$ $\texttt{else skip}\rrbracket^{\sharp}_{e})^{n+1}( \mathcal{P})}\\
=
\formulaexplanation{\adhocHyperSemantics\llbracket\texttt{if (B) S }$ $\texttt{else skip}\rrbracket^{\sharp}_{e}((\adhocHyperSemantics\llbracket\texttt{if (B) S }$ $\texttt{else skip}\rrbracket^{\sharp}_{e})^{n}( \mathcal{P}))}{def.\ powers}\\
=
\formulaexplanation{\adhocHyperSemantics\llbracket\texttt{if (B) S }$ $\texttt{else skip}\rrbracket^{\sharp}_{e}(\{X^n(P)\mid P\in\mathcal{P}\})}{induction hypothesis}\\
=
\formulaexplanation{\Hpost\llbracket\texttt{if (B) S }$ $\texttt{else skip}\rrbracket^{\sharp}_{e}(\{X^n(P)\mid P\in\mathcal{P}\})}{def.\ $\adhocHyperSemantics{}$ for conditional}\\
=
\formulaexplanation{\{\textsf{post}^\sharp\sqb{\texttt{if (B) S else skip}}^{\sharp}_{e}\,X^n(P)\mid P\in\mathcal{P}\}}{def.\ (\ref{eq:def:Post}) of $\textsf{Post}^\sharp$}\\
=
\formulaexplanation{\{X^{n+1}(P)\mid P\in\mathcal{P}\}}{def.\ (\ref{eq:def:ule:exists:forall:generalized:X:n}) of the iterates}
\end{calculus}

\smallskip

We conclude, by recurrence, that $\forall n\in\mathbb{N}\mathrel{.}(\adhocHyperSemantics\llbracket\texttt{if (B) S }$ $\texttt{else skip}\rrbracket^{\sharp}_{e})^{n}( \mathcal{P})$
=
$\{X^{n}(P)\mid P\in\mathcal{P}\}$. It follows that
\begin{calculus}[=\ \ ]
\formula{\Lfp{\sqsubseteq}\LAMBDA{\mathcal{X}}\mathcal{P} \cup \adhocHyperSemantics\llbracket\texttt{if (B) S }\texttt{else skip}\rrbracket^{\sharp}_{e}(\mathcal{X})}\\
=
\formula{\mathcal{X}^{\omega}}\\
= 
\formula{{\bigcup\limits_{n\in\mathbb{N}} (                                                                              
           \adhocHyperSemantics\sqb{\texttt{if (B) S else skip}}^{\sharp}_{e})^n( \mathcal{P})}}\\
= 
\formula{{\bigcup\limits_{n\in\mathbb{N}}\{X^{n}(P)\mid P\in\mathcal{P}\}}}\\
= 
\lastformula{\{X^{n}(P)\mid P\in\mathcal{P}\wedge n\in\mathbb{N}\}}{\mbox{\qed}}
\end{calculus}
\let\qed\relax
\end{proof}

The following lemma \ref{lem:hypercollecting-semantics-is-sound} shows the correspondance between $\adhocHyperSemantics\sqb{\texttt{while(B) S}}^{\sharp}_{e}$ and the hypercollecting semantics (\ref{eq:Post:abstract:while}). It shows that $\adhocHyperSemantics\sqb{\texttt{while(B) S}}^{\sharp}_{e}$ misses limits.
\begin{lemma}\label{lem:hypercollecting-semantics-is-sound}
     $   \forall\mathcal{P}\in\wp(\mathbb{L}^{\sharp})\mathrel{.}
     \Hpost\sqb{\texttt{while(B) S}}^{\sharp}_{e}\mathcal{P}\subseteq\alpha^{\uparrow}(\adhocHyperSemantics\sqb{\texttt{while(B) S}}^{\sharp}_{e}\mathcal{P})$.
   \end{lemma}
\begin{proof}[Proof of lemma \ref{lem:hypercollecting-semantics-is-sound}]
\begin{calculus}[$\subseteq$\ \ ]
\formula{\Hpost\sqb{\texttt{while(B) S}}^{\sharp}_{e}\mathcal{P}}\\
=
\formulaexplanation{\{\textsf{\textup{post}}^\sharp\sqb{\texttt{while (B) S}}^{\sharp}_{e}P\mid P\in\mathcal{P}\}}{(\ref{eq:def:Post})}\\
=
\formulaexplanation{\{{\textsf{\textup{post}}^\sharp\sqb{\texttt{$\neg$B}}^{\sharp}_{e}\, X^{\omega}(P)}\mid  P\in\mathcal{P}\}}{(\ref{eq:charaterization:psot:while:e})}\\
=
\formulaexplanation{\{\textsf{\textup{post}}^\sharp\sqb{\texttt{$\neg$B}}^{\sharp}_{e}\,(\mathop{{\bigsqcup}^{\sharp}_{+}}\limits_{n\in\mathbb{N}}X^{n}(P))\mid P\in\mathcal{P}\}}{def.\ (\ref{eq:def:ule:exists:forall:generalized:X:n}) of $X^{\omega}$}\\
=
\formulaexplanation{\{\mathop{{\bigsqcup}^{\sharp}_{+}}\limits_{n\in\mathbb{N}}\textsf{\textup{post}}^\sharp\sqb{\texttt{$\neg$B}}^{\sharp}_{e}\,X^{n}(P)\mid P\in\mathcal{P}\}}{join preservation \ref{def:abstract:domain:well:def:join:additive}}\\
=
\formulaexplanation{\{\mathop{{\bigsqcup}^{\sharp}_{+}}\{\textsf{\textup{post}}^\sharp\sqb{\texttt{$\neg$B}}^{\sharp}_{e}\,X^{n}(P)\mid {n\in\mathbb{N}}\}\mid P\in\mathcal{P}\}}{def.\ of $\mathop{{\bigsqcup}^{\sharp}_{+}}$}\\
$\subseteq$
\formula{\alpha^{\uparrow}(\{\textsf{post}^\sharp\sqb{\texttt{$\neg$}B}^\sharp_e\,X^{n}(P)\mid P\in\mathcal{P}\wedge n\in\mathbb{N}\})}\\
\explanation{Any $\mathop{{\bigsqcup}^{\sharp}_{+}}\{\textsf{\textup{post}}^\sharp\sqb{\texttt{$\neg$B}}^{\sharp}_{e}\,X^{n}(P)\mid {n\in\mathbb{N}}\}$ is the least upper bound of the increasing chain 
$\pair{\textsf{\textup{post}}^\sharp\sqb{\texttt{$\neg$B}}^{\sharp}_{e}\,X^{n}(P)}{{n\in\mathbb{N}}}$ of
$\{\textsf{post}^\sharp\sqb{\texttt{$\neg$}B}^\sharp_e\,X^{n}(P)\mid P\in\mathcal{P}\wedge n\in\mathbb{N}\}$ which, by the
dual def.\ (\ref{eq:def:alpha:downarrow}) of $\alpha^{\uparrow}$, belongs to $\alpha^{\uparrow}(\{\textsf{post}^\sharp\sqb{\texttt{$\neg$}B}^\sharp_e\,X^{n}(P)\mid P\in\mathcal{P}\wedge n\in\mathbb{N}\})$}\\[0.5ex]
=
\formulaexplanation{\alpha^{\uparrow}(\Hpost\sqb{\texttt{$\neg$}B}^\sharp_e(\{X^{n}(P)\mid P\in\mathcal{P}\wedge n\in\mathbb{N}\})}{(\ref{eq:def:Post})}\\
=
\formulaexplanation{\alpha^{\uparrow}(\Hpost\sqb{\texttt{$\neg$}B}^\sharp_e(\Lfp{\subseteq}\LAMBDA{\mathcal{X}}{\mathcal{P} \cup \adhocHyperSemantics\sqb{\texttt{if(B) S else skip}}^{\sharp}_{e}(\mathcal{X})}))}{(\ref{eq:characterization:adhocHyperSemantics}}\\
=
\lastformulaexplanation{\alpha^{\uparrow}(\adhocHyperSemantics\sqb{\texttt{while(B) S}}^{\sharp}_{e}\mathcal{P})}{(\ref{eq:def:overline:Post:while}}{\mbox{\qed}}
\end{calculus}
\let\qed\relax
\end{proof}
\end{toappendix}
\ifshort 

Therefore \else Notice that \fi  $\adhocHyperSemantics\sqb{\texttt{while(B) S}}^{\sharp}_{e}\mathcal{P}$ may contain chains
$\textsf{post}^\sharp\sqb{\texttt{$\neg$}B}^\sharp_e\,X^{n}(P_0)$ $\mathrel{{\sqsubseteq}^{\sharp}_{+}}$ $
\textsf{post}^\sharp\sqb{\texttt{$\neg$}B}^\sharp_e\,X^{n}(P_1)$ $\mathrel{{\sqsubseteq}^{\sharp}_{+}}\ldots
$ $\mathrel{{\sqsubseteq}^{\sharp}_{+}}$ $\textsf{post}^\sharp\sqb{\texttt{$\neg$}B}^\sharp_e\,X^{n}(P_k)$ $\mathrel{{\sqsubseteq}^{\sharp}_{+}}$ $\ldots$ which limit will be in
$\alpha^{\uparrow}(\adhocHyperSemantics\sqb{\texttt{while(B) S}}^{\sharp}_{e}\mathcal{P})$ but not necessarily in $ \Hpost\sqb{\texttt{while(B) S}}^{\sharp}_{e}\mathcal{P}$. It follows that $\adhocHyperSemantics\sqb{\texttt{while(B) S}}^{\sharp}_{e}$ may miss limits but also may introduce chains with irrelevant limits of infeasible executions (which invalidates \cite[theorem 1]{DBLP:conf/popl/AssafNSTT17} soundness claim).

The following theorem shows the soundness and  completeness of rule (\ref{eq:rule:exists:forall:generalized})
for the abstract hypercollecting semantics $\adhocHyperSemantics\sqb{\texttt{while(B) S}}^{\sharp}_{e}$\ifshort\else, which by (\ref{eq:characterization:adhocHyperSemantics}),\fi\ requires the consequent $\mathcal{Q}$ to contain the post condition of any number of iterations for any element $P$ of the antecedent $\mathcal{P}$.
\begin{theorem}
\label{th:loop-invariants-sound-and-complete-hypercollecting-semantics}\proofinapx\quad The proof rule
\textup{(\ref{eq:rule:exists:forall:generalized})} is sound and complete relative to \textup{(\ref{eq:def:overline:Post:while})}.
\end{theorem}
\begin{toappendix}
\begin{proof}[Proof of theorem \ref{th:loop-invariants-sound-and-complete-hypercollecting-semantics}]
\hyphen{5}\quad The proof of soundness is similar to that of theorem \ref{th:rule:exists:forall:generalized:sound}.

\hyphen{5}\quad For completeness, let 
\begin{eqntabular}{rcl}
\mathcal{I}&\triangleq&\{X^n(P)\mid P\in\mathcal{P}\wedge n\in\mathbb{N}\}
\label{eq:def:completeness:invariant}
\end{eqntabular}
\hyphen{6}\quad The condition $\mathcal{P}\subseteq\mathcal{I}$ of the premise holds for $n=0$;

\smallskip

\hyphen{6}\quad The second condition $\overline{\llbrace}\,\mathcal{I}\,\overline{\rrbrace}\,\texttt{if (B) else skip}\,\overline{\llbrace}\,\mathcal{I}\,\overline{\rrbrace}$ of the premise is equivalent, by (\ref{eq:def:abstract:logical:triples}), to
$\forall  I\in\mathcal{I}\mathrel{.}\textsf{post}^\sharp\sqb{\texttt{if (B) else skip}}^\sharp I\in\mathcal{I}$, which holds by definition (\ref{eq:def:ule:exists:forall:generalized:X:n}) of the iterates.

\smallskip

\hyphen{6}\quad The last condition $\overline{\llbrace}\,\mathcal{I}\,\overline{\rrbrace}\,\texttt{$\neg$B}\,\overline{\llbrace}\,\mathcal{Q}\,\overline{\rrbrace}$ of the premise follows from the hypothesis provided by the conclusion of the rule (\ref{eq:rule:exists:forall:generalized}).
\begin{calculus}[$\Rightarrow$\ \ ]
\formula{\adhocHyperSemantics\sqb{\texttt{while(B) S}}^{\sharp}_{e}(\mathcal{P})\subseteq\mathcal{Q}}\\
$\Rightarrow$
\formula{\Hpost\sqb{\texttt{$\neg$}B}^\sharp_e(\Lfp{\subseteq}\LAMBDA{\mathcal{X}}{\mathcal{P} \cup \adhocHyperSemantics\sqb{\texttt{if(B) S else skip}}^{\sharp}_{e}(\mathcal{X})}) 
\subseteq\mathcal{Q}}\\\rightexplanation{def.\ (\ref{eq:def:overline:Post:while}) of $\adhocHyperSemantics\sqb{\texttt{while(B) S}}^{\sharp}_{e}\mathcal{P} $}\\
$\Rightarrow$
\formulaexplanation{\Hpost\sqb{\texttt{$\neg$}B}^\sharp_e(\{X^{n}(P)\mid P\in\mathcal{P}\wedge n\in\mathbb{N}\}) 
\subseteq\mathcal{Q}}{lemma \ref{eq:characterization:adhocHyperSemantics}}\\
$\Rightarrow$
\formulaexplanation{\{\textsf{post}^{\sharp}\sqb{\texttt{$\neg$}B}^\sharp_e(X^{n}(P))\mid P\in\mathcal{P}\wedge n\in\mathbb{N}\} 
\subseteq\mathcal{Q}}{(\ref{eq:def:Post})}\\
$\Rightarrow$
\formulaexplanation{\forall P\in\mathcal{P},n\in\mathbb{N}\mathrel{.} \textsf{post}^{\sharp}\sqb{\texttt{$\neg$}B}^\sharp_e(X^{n}(P))\in\mathcal{Q}}{def.\ $\subseteq$}\\
$\Rightarrow$
\formulaexplanation{\forall  P\in\mathcal{I}\!\mathrel{.}\textsf{post}^\sharp\sqb{\texttt{$\neg$B}}^\sharp P\in\mathcal{Q}}{def.\ (\ref{eq:def:completeness:invariant}) of $\mathcal{I}$}\\
$\Rightarrow$
\lastformulaexplanation{\overline{\llbrace}\,\mathcal{I}\,\overline{\rrbrace}\,\texttt{$\neg$B}\,\overline{\llbrace}\,\mathcal{Q}\,\overline{\rrbrace}}{(\ref{eq:def:abstract:logical:triples})}{\mbox{\qed}}
\end{calculus}
\let\qed\relax
\end{proof}
\end{toappendix}
Theorem \ref{th:loop-invariants-sound-and-complete-hypercollecting-semantics} illustrates the importance of the proper choice of the collecting semantics since proof rule (\ref{eq:rule:exists:forall:generalized}) is unsound if $\mathcal{Q}\not\in{\maccent{\alpha}{\ast}}^{{\uparrow}}(\wp(\mathbb{L}^{\sharp}))$ and is complete for collecting semantics (\ref{eq:def:overline:Post:while}) but not with respect to collecting semantics (\ref{eq:Post:abstract:while}) hence not for the algebraic semantics of section \ref{sect:Algebraic-Semantics}. 

By deriving the collecting semantics \textsf{post} for execution properties and hypercollecting semantics \textsf{Post} for semantic properties by systematic abstraction of the algebraic semantics of section \ref{sect:Algebraic-Semantics}, we guarantee, by composition of successive abstractions satisfying definition \ref{def:exact:abstraction}, that the proof rules for these abstractions are sound with respect to any instance of the algebraic semantics satisfying definition \ref{def:abstract:domain:well:def}. Moreover, the proof rules are guaranteed to be complete with respect to these abstract properties, by construction.

\section{Sound and Complete Proof Rules for Generalized Exists Forall Hyperproperties}\label{GeneralizedExistsForall}
 In section \ref{sec:ConjunctiveAbstraction} of the appendix\proofinapx, we furthermore introduce conjunctive abstractions (i.e.\ conjunctions in logics or reduced products in static analysis).
\begin{toappendix}
\renewcommand{\Hpost}[1]{{\textup{\textsf{Post}}\sqb{\texttt{#1}}^{\sharp}}}

\subsection{Conjunctive Abstraction}\label{sec:ConjunctiveAbstraction}
In this \hyperlink{PARTIII}{part III}, we have introduced abstractions and their compositions.
We now consider their conjunctions by intersection. In static analysis with two different abstract domains this would correspond to a reduced product.

\subsubsection{Conjunctive Abstractions of Dual Operators}
We define the conjunction of abstractions introduced in previous sections of \hyperlink{PARTIII}{part III}.
\begin{definition}[Dual abstractions]
    \begin{eqntabular}{rcl}
        \mathbb{Op}^\sqsubseteq &=& \{\alpha^\sqsubseteq,\alpha^{\sqsubseteq\overline{F}}, {\maccent{\alpha}{\ast}}^{\sqsubseteq\uparrow},\alpha^{\curlywedgedownarrow}\}\\
        \mathbb{Op}^\sqsupseteq &=& \{\alpha^\sqsupseteq,\alpha^{\sqsupseteq\underline{F}}, {\maccent{\alpha}{\ast}}^{\sqsupseteq\downarrow},\alpha^{\curlyveeuparrow}\}
    \end{eqntabular}
\end{definition}

The conjunctive abstraction operator $\mathbf{R}$ takes two idempotent abstraction $\alpha_1 \in \mathbb{Op}^\sqsubseteq$ and $\alpha_2 \in \mathbb{Op}^\sqsupseteq$ and returns a new abstraction function that abstracts property $\mathcal{P}$ to the intersection of $\alpha_1(\mathcal{P})$ and $\alpha_2(\mathcal{P})$.
\begin{eqntabular}{rcl}
    \textbf{R}_{\pair{\alpha_1}{\alpha_2}}\triangleq \LAMBDA{\mathcal{P}}{\alpha_1(\mathcal{P}) \cap \alpha_2(\mathcal{P})}
\end{eqntabular}
\begin{lemma}[Properties of the well-defined conjunctive abstraction]
    For any $\alpha_1 \in \mathbb{Op}^\sqsubseteq$, $\alpha_2 \in \mathbb{Op}^\sqsupseteq$, and $\mathcal{P}\in\textbf{R}_{\pair{\alpha_1}{\alpha_2}}(\wp(\mathbb{L}))$,  we have \textup{(1)} $\alpha_1(\wp(\mathbb{L}))\subseteq \alpha^\sqsubseteq(\wp(\mathbb{L}))$ and $\alpha_2(\wp(\mathbb{L}))\subseteq \alpha^\sqsupseteq(\wp(\mathbb{L}))$; \textup{(2)} $\alpha^\sqsubseteq(\mathcal{P})\cap \alpha^\sqsupseteq(\mathcal{P}) = \mathcal{P}$; and \textup{(3)} if both $\alpha_1$ and $\alpha_2$ are upper-closures, then $\textbf{R}_{\pair{\alpha_1}{\alpha_2}}$ is also an upper-closure.
\label{def:well-defined-conjunction-abstraction}
\end{lemma}
\begin{proof}[Proof of lemma \ref{def:well-defined-conjunction-abstraction}]
    (1) directly follows from the definitions. Let us then prove (2).
    For an arbitrary hyperproperty $\mathcal{Q}\in \textbf{R}_{\pair{\alpha_1}{\alpha_2}}$, we have $\mathcal{Q} =  \textbf{R}_{\pair{\alpha_1}{\alpha_2}}(\mathcal{P})$ for some $\mathcal{P}\in\mathbb{L}$. It follows that
\begin{calculus}[=\ \  ]
   \phantom{$=$} $\alpha^\sqsubseteq(\mathcal{Q})\cap \alpha^\sqsupseteq(\mathcal{Q})$\\
    $=$\formulaexplanation{\alpha^\sqsubseteq(\alpha_1(\mathcal{P}) \cap \alpha_2(\mathcal{P})) \cap\alpha^\sqsupseteq(\alpha_1(\mathcal{P}) \cap \alpha_2(\mathcal{P})) }{def.\ of $\textbf{R}_{\pair{\alpha_1}{\alpha_2}}$}\\
        $\subseteq$\formulaexplanation{\alpha^\sqsubseteq\circ\alpha_1(\mathcal{P})\cap \alpha^\sqsubseteq\circ\alpha_2(\mathcal{P}) \cap\alpha^\sqsupseteq\circ\alpha_2(\mathcal{P}) \cap \alpha^\sqsupseteq\circ\alpha_1(\mathcal{P})}{def.\ of $\alpha^\sqsubseteq$ and $\alpha^\sqsupseteq$ that are increasing}\\
        $=$\formula{\alpha^\sqsubseteq\circ\alpha_1(\mathcal{P}) \cap \alpha^\sqsupseteq \circ\alpha_2(\mathcal{P})}\\ \explanation{$\alpha^\sqsupseteq\circ\alpha_1 (\mathcal{P})= \mathbb{L}$ for non-empty $\alpha_1 (\mathcal{P})$ since $\{\bot\} \in\alpha_1(\mathcal{P}) $. The equation trivially holds when $\alpha_1(\mathcal{P}) = \emptyset$}\\
        $=$\formulaexplanation{\alpha_1(\mathcal{P}) \cap \alpha_2(\mathcal{P})}{def.\ of $\alpha_1(\mathcal{P}) \in \alpha^\sqsubseteq(\wp(\mathbb{L}))$ and $\alpha_2(\mathcal{P}) \in \alpha^\sqsupseteq(\wp(\mathbb{L}))$}\\
        $=$ \formulaexplanation{\textbf{R}_{\pair{\alpha_1}{\alpha_2}}(\mathcal{P}) = \mathcal{Q}}{def.\ of $\textbf{R}_{\pair{\alpha_1}{\alpha_2}}$}
\end{calculus}
The inverse holds because $\alpha^\sqsubseteq \mathbin{\dot{\cap}} \alpha^\sqsupseteq$ is extensive.
Then we have $\alpha^\sqsubseteq(\mathcal{P})\cap \alpha^\sqsupseteq(\mathcal{P}) = \mathcal{P}$.

Now let us now prove (3). $\mathbf{R}_{\pair{\alpha_1}{\alpha_2}}$ is increasing and extensive by definition when $\alpha_1$ and $\alpha_2$ are increasing and extensive. We now prove that it is idempotent, which amounts to showing that $\textbf{R}_{\pair{\alpha_1}{\alpha_2}}(\mathcal{P}) \subseteq \textbf{R}_{\pair{\alpha_1}{\alpha_2}}\circ \textbf{R}_{\pair{\alpha_1}{\alpha_2}}(\mathcal{P})$ for any $\mathcal{P}\in\wp(\mathbb{L})$.
\begin{calculus}[=\ \ ]
\formula{\textbf{R}_{\pair{\alpha_1}{\alpha_2}}\circ \textbf{R}_{\pair{\alpha_1}{\alpha_2}}(\mathcal{P})}\\
=
\formulaexplanation{\textbf{R}_{\pair{\alpha_1}{\alpha_2}}(\alpha_1(\mathcal{P}) \cap \alpha_2(\mathcal{P} ))}{def.\ of $\textbf{R}_{\pair{\alpha_1}{\alpha_2}}$}\\
=
\formulaexplanation{\alpha_1(\alpha_1(\mathcal{P}) \cap \alpha_2(\mathcal{P})) \cap\alpha_2(\alpha_1(\mathcal{P}) \cap \alpha_2(\mathcal{P})) }{def.\ of $\textbf{R}_{\pair{\alpha_1}{\alpha_2}}$}\\
$\subseteq$
\formulaexplanation{\alpha_1\circ\alpha_1(\mathcal{P})\cap \alpha_1\circ\alpha_2(\mathcal{P}) \cap\alpha_2\circ\alpha_2(\mathcal{P}) \cap \alpha_2\circ\alpha_1(\mathcal{P})}{def.\ of $\alpha_1$ and $\alpha_2$ that are increasing}\\
=
\formula{\alpha_1\circ\alpha_1(\mathcal{P}) \cap \alpha_2\circ\alpha_2(\mathcal{P})}\\ 
\explanation{$\alpha_1\circ\alpha_1(\emptyset) = \emptyset$. For non-empty $\mathcal{P}$, $\alpha_2\circ\alpha_1 (\mathcal{P})= \mathbb{L}$, since $\bot\in \alpha_1 (\mathcal{P})$. The equation trivially holds when $\alpha_1(\mathcal{P}) = \emptyset$}\\
=
\formulaexplanation{\alpha_1(\mathcal{P}) \cap \alpha_2(\mathcal{P})}{def.\ of $\alpha_1$ and $\alpha_2$ that are idempotent}\\
=
\formulaexplanation{\textbf{R}_{\pair{\alpha_1}{\alpha_2}}(\mathcal{P}) }{def.\ of $\textbf{R}_{\pair{\alpha_1}{\alpha_2}}$}
\end{calculus}

\medskip

\noindent As a result,  $\textbf{R}_{\pair{\alpha_1}{\alpha_2}}(\mathcal{P})$ is idempotent since   $\textbf{R}_{\pair{\alpha_1}{\alpha_2}}(\mathcal{P})$ is extensive that implies $\textbf{R}_{\pair{\alpha_1}{\alpha_2}}(\mathcal{P}) \supseteq \textbf{R}_{\pair{\alpha_1}{\alpha_2}}\circ \textbf{R}_{\pair{\alpha_1}{\alpha_2}}(\mathcal{P})$.
\end{proof}

The domain conjunctive abstraction $\mathbf{R}_{\pair{\alpha_1}{\alpha_2}}$ is more expressive than both both $\alpha_1$ and $\alpha_2$. 

\begin{lemma} 
\label{lem:conjunction-abstraction-more-general}
For a well-defined conjunctive abstraction ${\mathbf{R}_{\pair{\alpha_1}{\alpha_2}}}$, we have the Galois retractions
\begin{centering}
    $ \pair{\mathbf{R}_{\pair{\alpha_1}{\alpha_2}}(\wp(\mathbb{L}))}{\subseteq} \galoiS{\alpha^\sqsubseteq}{\mathbb{1}} \pair{\alpha_1(\wp(\mathbb{L}))}{\subseteq}$ and $\pair{\mathbf{R}_{\pair{\alpha_1}{\alpha_2}}(\wp(\mathbb{L}))}{\subseteq}\galoiS{\alpha^\sqsupseteq}{\mathbb{1}} \pair{\alpha_2(\wp(\mathbb{L}))}{\subseteq}$.
\end{centering}
\end{lemma}

\begin{proof}[Proof of lemma \ref{lem:conjunction-abstraction-more-general}]
    Without any loss of generality, let us prove the first Galois connection. 

    We first show that for an arbitrary $\mathcal{P} \in\alpha_1(\wp(\mathbb{L}))$, $\mathbb{1}(\mathcal{P}) = \mathcal{P}$ is in $\mathbf{R}_{\pair{\alpha_1}{\alpha_2}}(\wp(\mathbb{L}))$. $\mathcal{P}$ can be express by $\mathcal{P} = \alpha_1(\mathcal{Q})$ for some $\mathcal{Q}\in\mathbb{L}$. If $\mathcal{P} = \emptyset$, then it's trivially in $\mathbf{R}_{\pair{\alpha_1}{\alpha_2}}$. If $\mathcal{P} \neq \emptyset$, then 
    \begin{calculus}[$\mathcal{P}$ $=$ ]
        $\mathcal{P}$ $=$ $\alpha_1(\mathcal{Q})$\\
        \phantom{$\mathcal{P}$} $=$ \formulaexplanation{\alpha_1\circ \alpha_1(\mathcal{Q})\cap \alpha_2\circ \alpha_1(\mathcal{Q})}{$\alpha_1$ is idempotent and $\alpha_2\circ \alpha_1(\mathcal{Q}) = \mathbb{L}$ for non-empty $ \alpha_1(\mathcal{Q})$}\\
        \phantom{$\mathcal{P}$} $=$ \formulaexplanation{\alpha_1(\mathcal{P})\cap \alpha_2(\mathcal{P})}{replace $\alpha_1(\mathcal{Q})$ by $\mathcal{P}$}
    \end{calculus}
    
    \medskip
    
    \noindent
    Thus, $\mathcal{P}$ is in $\mathbf{R}_{\pair{\alpha_1}{\alpha_2}}(\wp(\mathbb{L}))$. Since $\mathcal{P}\in\alpha^\sqsubseteq(\wp(\mathbb{L}))$ by (1) of lemma $\ref{def:well-defined-conjunction-abstraction}$, we know that $\alpha^\sqsubseteq\circ \mathbb{1}(\mathcal{P}) = \mathcal{P}$, proving the Galois retraction.
\end{proof}

\begin{lemma}
\label{lem:well-defined-conjunction-abstraction-closure}
    For a well-defined conjunctive abstraction operator $\mathbf{R}_{\pair{\alpha_1}{\alpha_2}}$, if $\alpha_1$ and $\alpha_2$ are upper closure operators, so is $\mathbf{R}_{\pair{\alpha_1}{\alpha_2}}$, and $\pair{\wp(\mathbb{L})}{\subseteq} \galoiS{\mathbf{R}_{\pair{\alpha_1}{\alpha_2}}}{\mathbb{1}} \pair{\mathbf{R}_{\pair{\alpha_1}{\alpha_2}}(\wp(\mathbb{L}))}{\subseteq}$. 
\end{lemma}
\begin{proof}
This follows from Lemma \ref{def:well-defined-conjunction-abstraction} implying  that $\mathbf{R}_{\pair{\alpha_1}{\alpha_2}}$ is an upper closure operator.
\end{proof}

\subsubsection{Proof Rule Simplification}

Applying the consequence rule $\frac{\hyperlogicup{\mathcal{P}}{S}{\mathcal{Q}} \quad \hyperlogicup{\mathcal{P}}{S}{\mathcal{R}} }{\hyperlogicup{\mathcal{P}}{S}{\mathcal{Q}\cap\mathcal{R} } }$,
we get the following sound and complete rule for the conjunctive abstraction.
\begin{eqntabular}{c}
    \frac{\hyperlogicup{\mathcal{P}}{S}{\alpha^\sqsubseteq(\mathcal{Q})} \quad \hyperlogicup{\mathcal{P}}{S}{\alpha^\sqsupseteq(\mathcal{Q})}}{\hyperlogicup{\mathcal{P}}{S}{\mathcal{Q}}},\quad \mathcal{Q}\in\textbf{R}_{\pair{\alpha_1}{\alpha_2}}(\wp(\mathbb{L}))
    \label{eq:rule:well-defined-conjunction-abstraction}
\end{eqntabular}
\begin{proof}[Proof of $(\ref{eq:rule:well-defined-conjunction-abstraction})$]
    \begin{calculus}[$\Leftrightarrow$\ \ ]
        \formula{\hyperlogicup{\mathcal{P}}{S}{\mathcal{Q}}}\\
        $\Leftrightarrow$ 
        \formulaexplanation{\Hpost{S}(\mathcal{P}) \subseteq \mathcal{Q}}{def.\ of $\hyperlogicup{\mathcal{P}}{S}{\mathcal{Q}}$}\\
        $\Leftrightarrow$  \formulaexplanation{\Hpost{S}(\mathcal{P}) \subseteq \alpha^\sqsubseteq(\mathcal{Q}) \cap \alpha^\sqsupseteq(\mathcal{Q}) }{By lemma $\ref{def:well-defined-conjunction-abstraction}$}\\
          $\Leftrightarrow$  \formulaexplanation{\Hpost{S}(\mathcal{P}) \subseteq \alpha^\sqsubseteq(\mathcal{Q}) \land  \Hpost{S}(\mathcal{P}) \subseteq\alpha^\sqsupseteq(\mathcal{Q}) }{By consequence rule}\\
          $\Leftrightarrow$ \lastformulaexplanation{\hyperlogicup{\mathcal{P}}{S}{\alpha^\sqsubseteq(\mathcal{Q})} \land \hyperlogicup{\mathcal{P}}{S}{\alpha^\sqsupseteq(\mathcal{Q})}}{def.\ of $\hyperlogicup{\mathcal{P}}{S}{\mathcal{Q}}$}{\mbox{\qed}}
    \end{calculus}
    \let\qed \relax
\end{proof}
Lemma  $\ref{lem:conjunction-abstraction-more-general}$ shows that
$\alpha^\sqsubseteq(\mathcal{Q}) \in \alpha_1(\wp(\mathbb{L}))$, and $\alpha^\sqsupseteq(\mathcal{Q}) \in \alpha_2(\wp(\mathbb{L}))$. Therefore we have similar rules for the case when the post-condition is in $\alpha_1(\wp(\mathbb{L}))$ and $\alpha_2(\wp(\mathbb{L}))$ respectively. An example is given in the next section.

\end{toappendix}
\ifshort
Such conjunctive abstractions are used in section \ref{sec:generalized:exists:forall:hyperprop} of the appendix to provide the following sound and complete proof rule for generalized $\exists\forall$-hyperproperties\proofinapx. Define $\varrho^{\sqsubseteq\underline{F}}({\mathcal{P}}) \triangleq {\bigcup_{F\,\in\, \alpha^{\underline{F}}(\mathcal{P})} \varphi^\sqsubseteq(F) \mathcal{P}}$ and $\varphi^\sqsubseteq(F) \triangleq \LAMBDA{\mathcal{X}}\{P\in \mathcal{X} \mid F\sqsubseteq P\wedge{}$ $\forall P'\in \mathbb{L}\mathrel{.} F\sqsubseteq P' \sqsubseteq P \rightarrow P'\in\mathcal{X} \}$, ${F\in\mathbb{L}}$ then, for $\mathcal{Q}\in\varrho^{\sqsubseteq\underline{F}}(\wp(\mathbb{L}^{\sharp}))$,\lstrut
\bgroup
\abovedisplayskip0pt
\belowdisplayskip0pt
\small
\begin{eqntabular*}{l}
    \frac{
        {\displaystyle\exists\mathcal{X}\in\alpha^{\underline{F}}(\mathcal{Q})\,{\rightarrow}\,\wp(\mathbb{L}^{\sharp})\,{.}\,\mathcal{P}\subseteq \!\!\!\!\!\!\bigcup_{F\in\alpha^{\underline{F}}(\mathcal{Q})}\!\!\!\!\!\!\!\mathcal{X}_F,\ \             (\forall F\in\alpha^{\underline{F}}(\mathcal{Q})\,{.}\,\forall P \in\mathcal{X}_F\,{.}\,\exists Q\in \varphi^\sqsubseteq(F)\mathcal{Q} \,{.}\, \overline{\{}P\overline{\}}{S}\overline{\{}Q\overline{\}}  \wedge
           \textstyle\underline{\{}P\underline{\}}{S}\underline{\{}F\underline{\}})
        }
    }{\hyperlogicup{\mathcal{P}}{S}{\mathcal{Q}}} \label{eq:rule:bottom-frontier-varrho-elimination}
\end{eqntabular*}
\vspace*{10pt}
\egroup
\fi
\begin{toappendix}

\subsection{Lower $\sqsubseteq$-closed and frontier elimination}\label{sec:generalized:exists:forall:hyperprop}
\renewcommand{\Hpost}[1]{{\textup{\textsf{Post}}\sqb{\texttt{#1}}^{\sharp}}}
\renewcommand{\postSemantics}[1]{{\textup{\textsf{post}}\sqb{\texttt{#1}}^{\sharp}}}

Let us define the $\sqsubseteq$-closed lower closure operator $\varrho^\sqsubseteq$
\begin{eqntabular}{rcl}
    \varrho^\sqsubseteq &\triangleq& \LAMBDA{\mathcal{P}}{\{P\in\mathcal{P}\mid \forall P'\in\mathbb{L}\mathrel{.} P'\sqsubseteq P\Rightarrow P'\in\mathcal{P} \} }
\end{eqntabular}
\begin{lemma}
    \label{lem:def-lower-closure-subseteq}
    $\varrho^\sqsubseteq$ is a lower-closure that is increasing, reductive and idempotent, and $\pair{\wp(\mathbb{L})}{\supseteq} \galoiS{\varrho^\sqsubseteq}{\mathbb{1}} \pair{\alpha^\sqsubseteq(\wp(\mathbb{L}))}{\supseteq}$
\end{lemma}
\begin{proof}[Proof of lemma \ref{lem:def-lower-closure-subseteq}]
    By definition of $\varrho^\sqsubseteq$, it is trivially increasing and reductive. Let us first prove that $\varrho^\sqsubseteq(\mathcal{P})\in\alpha^\sqsubseteq(\wp(\mathbb{L}))$ for arbitrary $\mathcal{P}\in\wp(\mathbb{P})$. We have
    \begin{calculus}[=\ \ ]
       \phantom{$=$} $\alpha^\sqsubseteq \circ \varrho^\sqsubseteq(\mathcal{P})$ \\
       $=$ \formulaexplanation{\{P\in\mathbb{L} \mid \exists P'\in \varrho^\sqsubseteq(\mathcal{P}) \mathrel{.} P\sqsubseteq P' \} }{def.\ of $\alpha^\sqsubseteq$}\\
       $=$ \formulaexplanation{\{P\in\mathbb{L} \mid \exists P'\in \mathcal{P} \mathrel{.}(\forall P''\in \mathbb{L}\mathrel{.} P''\sqsubseteq P' \Rightarrow P''\in\mathcal{P}) \land P\sqsubseteq P' \} }{def.\ of $\varrho^\sqsubseteq$}\\[-0.5ex]
        $=$ \formula{\{P\in\mathbb{L} \mid P\in\mathcal{P} \land (\forall P''\in \mathbb{L}\mathrel{.} P''\sqsubseteq P \Rightarrow P''\in\mathcal{P})  \} }\\
        \explanation{$(\subseteq)$\quad holds as $\alpha^\sqsubseteq$ is extensive;\\
        $(\supseteq)$\quad choose $P' = P$}\\
        $=$ \formulaexplanation{\varrho^\sqsubseteq(\mathcal{P})}{def.\ of $\varrho^\sqsubseteq$}
    \end{calculus}
    
    \smallskip

\noindent We then prove that $\varrho^\sqsubseteq\circ \alpha^\sqsubseteq(\mathcal{P}) = \alpha^\sqsubseteq(\mathcal{P})$
\begin{calculus}[=\ \ ]
    \phantom{$=$ } $\varrho^\sqsubseteq\circ \alpha^\sqsubseteq(\mathcal{P})$ \\
    $=$ \formulaexplanation{\{P\in \alpha^\sqsubseteq(\mathcal{P})\mid \forall P'\in\mathbb{L}\mathrel{.} P'\sqsubseteq P\Rightarrow P'\in \alpha^\sqsubseteq(\mathcal{P}) \}}{def.\ of $\varrho^\sqsubseteq$}\\
     $=$ \formulaexplanation{\{P\in \mathbb{L}\mid (\exists Q \in\mathcal{P}\mathrel{.} P\sqsubseteq Q)\land \forall P'\in\mathbb{L}\mathrel{.} P'\sqsubseteq P\Rightarrow (\exists Q'\in\mathcal{Q}\mathrel{.} P'\sqsubseteq Q' ) \}}{def.\ of $\alpha^\sqsubseteq$} \\
     $=$ \formula{\{P\in\mathbb{L} \mid \exists Q \in\mathcal{P}\mathrel{.} P\sqsubseteq Q\} }\\[-0.5ex]
     \explanation{$(\supseteq)$\quad as $\varrho^\sqsubseteq$ is reductive;\\
     $(\subseteq)$\quad for all $P'\sqsubseteq P$, simply let $Q' = Q$, then $P'\sqsubseteq P \sqsubseteq Q = Q'$ holds}\\
     $=$ \formulaexplanation{\alpha^\sqsubseteq(\mathcal{P})}{def.\ of $\alpha^\sqsubseteq$}
\end{calculus}

\smallskip

Thus, we have proved that $\varrho^\sqsubseteq(\wp(\mathbb{L})) = \alpha^\sqsubseteq(\wp(\mathbb{L}))$. For any $\mathcal{P}\in\mathbb{P}$, we have $\varrho^\sqsubseteq\circ\varrho^\sqsubseteq (\mathcal{P}) = \varrho^\sqsubseteq (\mathcal{P})$, since $\varrho^\sqsubseteq(\mathcal{P})$ is included in $\alpha^\sqsubseteq(\wp(\mathbb{L}))$. 
\end{proof}

\subsection{Frontier $\varrho$-Elimination Abstraction}
We define a new abstraction based on $\alpha^{\underline{F}}$ and $\varrho^\sqsubseteq$
\begin{eqntabular}{rcl}
    \varrho^{\sqsubseteq\underline{F}} &\triangleq& \LAMBDA{\mathcal{P}}{\bigcup_{F\,\in\, \alpha^{\underline{F}}(\mathcal{P})} \varphi^\sqsubseteq(F) \mathcal{P}}  \\
    \text{where } \varphi^\sqsubseteq &\triangleq& \LAMBDA{F\in\mathbb{L}}{\LAMBDA{\mathcal{X}\in \wp(\mathbb{L})}{\{P\in \mathcal{X} \mid F\sqsubseteq P\land \forall P'\in \mathbb{L}\mathrel{.} F\sqsubseteq P' \sqsubseteq P \Rightarrow P'\in\mathcal{X} \}}} \nonumber 
\end{eqntabular}
\begin{lemma}
\label{lem:varrho-frontier-elimination-idempotent}
    $\varrho^{\sqsubseteq\underline{F}}$ is reductive and idempotent
\end{lemma}
\begin{proof}[Proof of lemma \ref{lem:varrho-frontier-elimination-idempotent}]
    For any $\mathcal{P}\in\wp(\mathbb{L})$ and $P\in \varrho^{\sqsubseteq\underline{F}}(\mathcal{P})$, we have $P\in\varphi^\sqsubseteq(F)\mathcal{P}$ for some $F\in \alpha^{\underline{F}}(\mathcal{P})$, meaning it is in $\mathcal{P}$. Thus $\varrho^{\sqsubseteq\underline{F}}$ is reductive. 
    To prove  idempotency, let us first prove that $\varrho^{\sqsubseteq\underline{F}}$ preserve lower-frontiers, that is $\alpha^{\underline{F}}(\mathcal{P}) =\alpha^{\underline{F}}\circ\varrho^{\sqsubseteq\underline{F}}(\mathcal{P})$.
\begin{calculus}[=\ \ ]
\formula{\alpha^{\underline{F}}\circ\varrho^{\sqsubseteq\underline{F}}(\mathcal{P})}\\
=
\formulaexplanation{\{P\in \varrho^{\sqsubseteq\underline{F}}(\mathcal{P})\mid \forall P' \in \varrho^{\sqsubseteq\underline{F}}(\mathcal{P}) \mathrel{.} P'\sqsubseteq P \Rightarrow P=P'\}}{def.\ of $\alpha^{\underline{F}}$}\\
=
\formula{\{P \in\mathcal{P}\mid (\exists F\in \alpha^{\underline{F}}(\mathcal{P})\mathrel{.} F\sqsubseteq P \land \forall P'\in \mathbb{L}\mathrel{.} F\sqsubseteq P' \sqsubseteq P \Rightarrow P'\in\mathcal{P})\land \\  \forall P_1 \in\mathcal{P} \mathrel{.} (( \exists F_1\in \alpha^{\underline{F}}(\mathcal{P})\mathrel{.} F_1\sqsubseteq P_1 \land \forall P_1'\in \mathbb{L}\mathrel{.} F_1\sqsubseteq P_1' \sqsubseteq P_1 \Rightarrow P_1'\in\mathcal{P}) \land P_1\sqsubseteq P )\Rightarrow P=P_1\}}\\ \rightexplanation{def.\ of $\varrho^{\sqsubseteq\underline{F}}$}\\[-1ex]
=
\formula{\{P\in\mathbb{L}\mid \exists G\in\alpha^{\underline{F}}(P)\mathrel{.} G=P \} = \alpha^{\underline{F}}(\mathcal{P})}\\ 
\explanation{$(\supseteq)$\quad When $G=P$, then for all $P_1$ such that $P_1 \sqsubseteq P$, $P_1= P$ holds trivially;\\
$(\subseteq)$\quad Since $\exists F_1\in \alpha^{\underline{F}}(\mathcal{P})\mathrel{.} F_1\sqsubseteq P_1 \land \forall P_1'\in \mathbb{L}\mathrel{.} F_1\sqsubseteq P_1' \sqsubseteq P_1 \Rightarrow P_1'\in\mathcal{P}$ holds if $P_1$ is instantiated to $F$, then the equality $P=F$ holds, where $F$ is a lower-frontier. Thus we can simply let $G$ to be $F$.}
\end{calculus}

\smallskip

    We now prove idempotency. Since $\alpha^{\underline{F}}(\mathcal{P}) =\alpha^{\underline{F}}\circ\varrho^{\sqsubseteq\underline{F}}(\mathcal{P})$, it remains to prove that $\varphi^\sqsubseteq(F)\mathcal{P} = \varphi^\sqsubseteq(F)(\varrho^{\sqsubseteq\underline{F}}(\mathcal{P}))$ for any $F\in\alpha^{\underline{F}}(\mathcal{P})$.
\begin{calculus}[=\ \ ]
\formula{\varphi^\sqsubseteq(F)(\varrho^{\sqsubseteq\underline{F}}(\mathcal{P}))}\\
=
\formulaexplanation{\{P\in \varrho^{\sqsubseteq\underline{F}}(\mathcal{P}) \mid F\sqsubseteq P\land \forall P'\in \mathbb{L}\mathrel{.} F\sqsubseteq P' \sqsubseteq P \Rightarrow P'\in  \varrho^{\sqsubseteq\underline{F}}(\mathcal{P}) \}}{def.\ of $\varphi^\sqsubseteq$}\\
=
\formula{\{P\in \mathcal{P} \mid (\exists F_1\in \alpha^{\underline{F}}(\mathcal{P})\mathrel{.} F_1\sqsubseteq P \land \forall P'\in \mathbb{L}\mathrel{.} F_1\sqsubseteq P' \sqsubseteq P \Rightarrow P'\in\mathcal{P})\land  F\sqsubseteq P\land \\\forall P_2\in \mathbb{L}\mathrel{.} F\sqsubseteq P_2 \sqsubseteq P \Rightarrow (\exists F_2\in \alpha^{\underline{F}}(\mathcal{P})\mathrel{.} F_2\sqsubseteq P_2 \land (\forall P_2'\in \mathbb{L}\mathrel{.} F_2\sqsubseteq P_2' \sqsubseteq P_2 \Rightarrow P_2'\in\mathcal{P})) \}}
      \\[0.5ex] \rightexplanation{def.\ of $\varrho^{\sqsubseteq\underline{F}}$, replace $P'$ with $P_2$}\\[-1ex]
      $=$ \formula{\{P\in \mathcal{P} \mid F\sqsubseteq P\land \forall P'\in \mathbb{L}\mathrel{.} F\sqsubseteq P' \sqsubseteq P \Rightarrow P'\in  \mathcal{P} \}}\\
      \explanation{$(\supseteq)$\quad since we have assumed that $F$ is a lower frontier for $\mathcal{P}$, we can simply let $F_1=F_2=F$, and all the conditions do hold;\\
      $(\subseteq)$\quad To prove $\forall P'\in \mathbb{L}\mathrel{.} F\sqsubseteq P' \sqsubseteq P \Rightarrow P'\in  \mathcal{P}$. We are allowed to instantiate $P_2 = P'$ in the premise $\forall P_2\in\mathbb{L}\mathrel{.}F\sqsubseteq P_2\sqsubseteq P\Rightarrow\exists F_2\in \alpha^{\underline{F}}(\mathcal{P})\mathrel{.} F_2\sqsubseteq P_2 \land ( \forall P_2'\in \mathbb{L}\mathrel{.} F_2\sqsubseteq P_2' \sqsubseteq P_2 \Rightarrow P_2'\in\mathcal{P})$. Then we get $ \forall P_2'\in \mathbb{L}\mathrel{.} F_2\sqsubseteq P_2' \sqsubseteq P' \Rightarrow P_2'\in\mathcal{P}$ for some frontier $F_2$ where $F_2\sqsubseteq P'$. We  are then allowed to instantiate $P'_2$ to $P'$, which implies that $P' \in\mathcal{P}$ holds.}
    \end{calculus}
    
    \smallskip
    
    \noindent
Therefore, we proved idempotency.
\end{proof}
\subsection{Exist Forall Hyperproperties}
Assuming that $\pair{\mathbb{L}}{\sqsubseteq} \triangleq \pair{\wp(\Pi)}{\subseteq}$. $\exists\forall$ hyperproperties have the form 
\begin{eqntabular}{rcl}
    \mathcal{E\mskip-2muA\mskip-0.5muH} &\triangleq& \{\{P\in \wp(\Pi)\mid \exists \pi_1 \in P\mathrel{.} \forall \pi_2 \in P\mathrel{,} \pair{\pi_1}{\pi_2}\in A \} \mid A\in \wp(\Pi \times \Pi )\}
\end{eqntabular}
\begin{example}
The negation $\textit{GD}$ of the generalized non-interference properties $\textit{GNI}$ in (\ref{eq:def:GNI})
is a $\exists\forall$ hyperproperty expressing generalized dependency. A set of executions satisfies the generalized dependency when altering the initial values of high variables does change the set of possible final values of any low variable.
    \begin{eqntabular}{rcl}
        \textit{GD} &\triangleq&  \{{P}\in \mathbb{L}\mid \exists \sigma_1\pi_1\sigma_1',\sigma_2\pi\sigma_2'\in{P}\mathrel{.}\forall\sigma_3\pi\sigma_3'\in{P}\mathrel{.} \label{eq:def:GD}\\ 
        &&\quad\quad (\sigma_1(\mathtt{L}) = \sigma_2(\mathtt{L})= \sigma_3(\mathtt{L}))\Rightarrow (\sigma_3(\mathtt{H}) = \sigma_2(\mathtt{H}) \Rightarrow \sigma_3'(\mathtt{L}) \neq \sigma_1'(\mathtt{L}))\}\renumber{\qef}
    \end{eqntabular}
    \let\qef\relax
\end{example}
The hyperproperties with $\varrho^{\sqsubseteq\underline{F}}$ subsume $\exists\forall$ hyperproperties.
\begin{eqntabular}{rcl}
    \mathcal{E\mskip-2muA\mskip-0.5muH} &\subseteq& \varrho^{\sqsubseteq\underline{F}}(\wp(\wp(\Pi)))
    \label{eq:varrho-frontier-elimination-subsume-exist-forall}
\end{eqntabular}

\begin{proof}[Proof of $(\ref{eq:varrho-frontier-elimination-subsume-exist-forall})$]
    We prove that $\forall \mathcal{P}\in \mathcal{E\mskip-2muA\mskip-0.5muH} \mathrel{.} \mathcal{P} \in \varrho^{\sqsubseteq\underline{F}}(\wp(\wp(\Pi)))$. By Lemma $\ref{lem:varrho-frontier-elimination-idempotent}$, it is sufficient to prove that $ \mathcal{P} \subseteq \varrho^{\sqsubseteq\underline{F}}(\mathcal{P}) $ due to the fact that $\varrho^{\sqsubseteq\underline{F}}$ is reductive and idempotent. $\mathcal{P}$ is expressed as $\mathcal{P}\triangleq \{P\in \wp(\Pi)\mid \exists \pi_1 \in P\mathrel{.} \forall \pi_2 \in P\mathrel{,} \pair{\pi_1}{\pi_2}\in A \}$ for some $A$.
    \begin{calculus}[=\ \ ]
\formula{\varrho^{\sqsubseteq\underline{F}}(\mathcal{P})}\\
=
\formulaexplanation{\bigcup_{F\in \alpha^{\underline{F}}(\mathcal{P})} \varphi^\sqsubseteq(F) \mathcal{P}}{def.\ of $\varrho^{\sqsubseteq\underline{F}}$}\\
=
\formulaexplanation{\bigcup_{F\in \alpha^{\underline{F}}(\mathcal{P})} \{P\in \mathcal{P}\mid F\subseteq P\land \forall P'\in \mathbb{L}\mathrel{.} F\subseteq P' \subseteq P \Rightarrow P'\in\mathcal{P} \}}{def.\ of $\varphi^\sqsubseteq$}\\
=
\formulaexplanation{\{P\in \mathcal{P}\mid \exists F \in \alpha^{\underline{F}}(\mathcal{P})\mathrel{.} F\subseteq P\land \forall P'\in \mathbb{L}\mathrel{.} F\subseteq P' \subseteq P \Rightarrow P'\in\mathcal{P} \}}{def.\ of $\bigcup$}\\
$\supseteq$ 
\formula{ \{P\in \wp(\Pi)\mid \exists \pi_1 \in P\mathrel{.} \forall \pi_2 \in P\mathrel{.} \pair{\pi_1}{\pi_2}\in A \} = \mathcal{P}}\\ 
       \lastexplanation{For arbitrary $P\in\mathcal{P}$, there exists $\pi\in P$ where for all $\pair{\pi}{\pi'}\in A$ holds for all $\pi'\in P$. Let $F\triangleq \{\pi\}$, which would be in $\alpha^{\underline{F}}(\mathcal{P})$ as $\emptyset\notin \mathcal{P}$ by definition. Then $F\subseteq P$ holds trivially. For all $P'$ such that $F = \{\pi\} \subseteq P' \subseteq P$. Let $\pi$ be the existent $\pi_1$, then for all $\pi_2\in P'$, it is also in $P$. Thus we have $\pair{\pi}{\pi_2}\in A$, meaning that $P'\in\mathcal{P}$ }{\mbox{\qed}}
    \end{calculus}
    \let\qed\relax
\end{proof}
\subsection{Proof Rule Simplification}
Using the consequence rule, we introduce a sound and complete proof rule that splits an abstract frontier-$\varrho^\sqsubseteq$ eliminated abstract hyperproperties into a conjunctive abstraction. This requires manual efforts that partition the precondition $\mathcal{P}$ into frontier-indexed preconditions $\mathcal{X}$ where $\mathcal{X}_F\in\wp(\mathbb{L})$ for $F\in\alpha^{\underline{F}}(\mathcal{Q})$. Then we can further use the consequence rule to prove the triple for the correspondent conjunctive abstraction.
\begin{eqntabular}{c}
    \frac{
        \begin{array}{c}
    {\exists\mathcal{X}\in\alpha^{\underline{F}}(\mathcal{Q})\,{\rightarrow}\,\wp(\mathbb{L}^{\sharp})\,{.}\,\forall F\in\alpha^{\underline{F}}(\mathcal{Q})\mathrel{.} \hyperlogicup{\mathcal{X}_F}{S}{\varphi^\sqsubseteq(F)\mathcal{Q} }},\quad
    {\displaystyle \mathcal{P}\subseteq \!\!\!\!\!\bigcup_{F\in\alpha^{\underline{F}}(\mathcal{Q})}\!\!\!\!\!\mathcal{X}_F}
        \end{array}
    }{\hyperlogicup{\mathcal{P}}{S}{\mathcal{Q}}},\quad \mathcal{Q}\in \varrho^{\sqsubseteq\underline{F}}(\wp(\wp(\Pi))) \label{eq:rule:frontier-varrho-elimination}\stepcounter{equation}\renumber{\raisebox{-1.5em}[0pt][0pt]{(\ref{eq:rule:frontier-varrho-elimination})}}
\end{eqntabular}
\begin{proof}[Proof of $(\ref{eq:rule:frontier-varrho-elimination})$]
    \begin{calculus}[$\Leftrightarrow$\ \ ]
\formula{\hyperlogicup{\mathcal{P}}{S}{\mathcal{Q}}} \\
$\Leftrightarrow$ 
\formulaexplanation{\Hpost{S}(\mathcal{P}) \subseteq \mathcal{Q}}{def.\ of $\hyperlogicup{\mathcal{P}}{S}{\mathcal{Q}}$}\\
$\Leftrightarrow$ 
\formulaexplanation{\Hpost{S}(\mathcal{P}) \subseteq \bigcup_{F\in \alpha^{\underline{F}}(\mathcal{Q})}\varphi^\sqsubseteq(F)\mathcal{Q} }{lemma $\ref{lem:varrho-frontier-elimination-idempotent}$}\\
$\Leftrightarrow$ 
\formula{\exists\mathcal{X}\in\alpha^{\underline{F}}(\mathcal{Q})\,{\rightarrow}\,\wp(\mathbb{L}^{\sharp})\mathrel{.}\Hpost{S}(\bigcup_{F\in\alpha^{\underline{F}}(\mathcal{Q})} \mathcal{X}_F) \subseteq \bigcup_{F\in \alpha^{\underline{F}}(\mathcal{Q})}\varphi^\sqsubseteq(F)\mathcal{Q} \quad \land \quad  \mathcal{P}\subseteq \bigcup_{F\in\alpha^{\underline{F}}(\mathcal{Q})}\mathcal{X}_F }
        \\ 
        \rightexplanation{$(\Rightarrow )$\quad let $\mathcal{X}_F = \mathcal{P}$ for all $F$.\quad $(\Leftarrow )$\quad $\Hpost{S}(\mathcal{P})$ is increasing.}\\
$\Leftrightarrow$
\formula{\exists\mathcal{X}\in\alpha^{\underline{F}}(\mathcal{Q})\,{\rightarrow}\,\wp(\mathbb{L}^{\sharp})\mathrel{.}\bigcup_{F\in\alpha^{\underline{F}}(\mathcal{Q})}\Hpost{S}( \mathcal{X}_F) \subseteq \bigcup_{F\in \alpha^{\underline{F}}(\mathcal{Q})}\varphi^\sqsubseteq(F)\mathcal{Q} \quad \land \quad  \mathcal{P}\subseteq \bigcup_{F\in\alpha^{\underline{F}}(\mathcal{Q})}\mathcal{X}_F }\\
\rightexplanation{$\Hpost{S}{}$ is join preserving} \\
$\Leftrightarrow$
\formula{\exists\mathcal{X}\in\alpha^{\underline{F}}(\mathcal{Q})\,{\rightarrow}\,\wp(\mathbb{L}^{\sharp})\mathrel{.}\forall F\in\alpha^{\underline{F}}(\mathcal{Q}) \mathrel{.} \Hpost{S}( \mathcal{X}_F) \subseteq ^\sqsubseteq(F)\mathcal{Q} \quad \land \quad  \mathcal{P}\subseteq \bigcup_{F\in\alpha^{\underline{F}}(\mathcal{Q})}\mathcal{X}_F}\\
\rightexplanation{consequence rule}\\
         $\Leftrightarrow$ 
         \formula{\exists\mathcal{X}\in\alpha^{\underline{F}}(\mathcal{Q})\,{\rightarrow}\,\wp(\mathbb{L}^{\sharp})\mathrel{.}\forall F\in\alpha^{\underline{F}}(\mathcal{Q}) \mathrel{.} \hyperlogicup{\mathcal{X}_F}{S}{\varphi^\sqsubseteq(F)\mathcal{Q}} \quad \land \quad  \mathcal{P}\subseteq \bigcup_{F\in\alpha^{\underline{F}}(\mathcal{Q})}\mathcal{X}_F}\\
         \lastrightexplanation{def.\ of $\hyperlogicup{\mathcal{P}}{S}{\mathcal{Q}}$}{\mbox{\qed}}
    \end{calculus}
    \let\qed\relax
\end{proof}

Now the problem is reduced to proving the premise $\hyperlogicup{\mathcal{X}_F}{S}{\varphi^\sqsubseteq(F)\mathcal{Q}}$. Interestingly, we are able to apply the rule for conjunctive abstraction to $\varphi^\sqsubseteq(F)\mathcal{Q}$.
\begin{lemma}
\label{lem:frontier-varrho-elimination-conjunction-abstraction}
    For arbitrary $\mathcal{P} \in \wp(\mathbb{L})$, and $F\in\alpha^{\underline{F}}(\mathcal{P})$, $\varphi^\sqsubseteq(F)\mathcal{P}\in \textbf{R}_{\pair{\alpha^\sqsubseteq}{\alpha^\curlyveeuparrow}}(\wp(\mathbb{L}))$.
\end{lemma}
\begin{proof}
    By lemma $\ref{lem:well-defined-conjunction-abstraction-closure}$, it's  sufficient to prove that $\textbf{R}_{\pair{\alpha^\sqsubseteq}{\alpha^\curlyveeuparrow}}\circ \varphi^\sqsubseteq(F)\mathcal{P} = \varphi^\sqsubseteq(F)\mathcal{P}$
    \begin{calculus}[=\ \ ]
\formula{\textbf{R}_{\pair{\alpha^\sqsubseteq}{\alpha^\curlyveeuparrow}}\circ \varphi^\sqsubseteq(F)\mathcal{P}}\\
=
\formulaexplanation{\alpha^\sqsubseteq\circ \varphi^\sqsubseteq(F)\mathcal{P} \cap \alpha^\curlyveeuparrow\circ \varphi^\sqsubseteq(F)\mathcal{P}}{def.\ of $\textbf{R}_{\pair{\alpha^\sqsubseteq}{\alpha^\curlyveeuparrow}}$}\\
=
\formulaexplanation{\{P\in\mathbb{L}\mid \exists P'\in \varphi^\sqsubseteq(F)\mathcal{P}\mathrel{.} P\sqsubseteq P'\} \cap  \{P\in\mathbb{L}\mid P\sqsupseteq \bigsqcap \varphi^\sqsubseteq(F)\mathcal{P}\}}{def.\ of $\alpha^\sqsubseteq$ and $\alpha^\curlyveeuparrow$}\\
=
\formula{\{P\in\mathbb{L}\mid \exists P'\in\mathcal{P}\mathrel{.}(F\sqsubseteq P' \land \forall P''\in \mathbb{L} \mathrel{.}F\sqsubseteq P''\sqsubseteq P' \Rightarrow P''\in\mathcal{P}) \land P\sqsubseteq P'\} \cap  \{P\in\mathbb{L}\mid P\sqsupseteq \bigsqcap \varphi^\sqsubseteq(F)\mathcal{P}\}} \\ \rightexplanation{def.\ of $\varphi^\sqsubseteq$}\\
=
\formulaexplanation{\{P\in\mathbb{L}\mid \exists P'\in\mathcal{P}\mathrel{.}(F\sqsubseteq P' \land \forall P''\in \mathbb{L} \mathrel{.}F\sqsubseteq P''\sqsubseteq P' \Rightarrow P''\in\mathcal{P}) \land P\sqsubseteq P' \land F\sqsubseteq P\}}{$\bigsqcap \varphi^\sqsubseteq(F)\mathcal{P} = F$}\\
=
\formula{\{P\in\mathbb{L}\mid F\sqsubseteq P \land \forall P_1\in \mathbb{L} \mathrel{.}F\sqsubseteq P_1\sqsubseteq P \Rightarrow P_1\in\mathcal{P} \} = \varphi^\sqsubseteq(F)\mathcal{P} }\\
        \lastexplanation{$(\supseteq)$\quad let $P' = P$, then $F\sqsubseteq P' \land \forall P''\in \mathbb{L} \mathrel{.}F\sqsubseteq P''\sqsubseteq P' \Rightarrow P''\in\mathcal{P}$ holds by replacing $P''$ with $P_1$;\\
        $(\subseteq)$\quad For any $P_1$ such that $F\sqsubseteq P_1 \sqsubseteq P$ holds, we have $F\sqsubseteq P\sqsubseteq P'$ for some $P'$ so that  $F\sqsubseteq P_1 \sqsubseteq P \sqsubseteq P'$ also holds, which implies that $P_1\in \mathcal{P}$. By the premise, we have  $\forall P''\in \mathbb{L} \mathrel{.}F\sqsubseteq P''\sqsubseteq P' \Rightarrow P''\in\mathcal{P}$, we are allowed to instantiate $P''$ to $P_1$ and have $P_1\in\mathcal{P}$}{\mbox{\qed}}
    \end{calculus}
    \let\qed\relax
\end{proof}

\begin{lemma}
\label{eq:equivalent-order-ideal-rule}
We can equivalently rewrite the rule in $(\ref{eq:alpha:principal:rule})$ and $(\ref{sec:Order-Ideal-Abstraction})$ by the following.

\begin{eqntabular}{c}
    \frac{\forall P\in\mathcal{P}\mathrel{.} \hoarelogicdown{P}{S}{\bigsqcap\mathcal{Q}}}{\hyperlogicup{\mathcal{P}}{S}{\alpha^\sqsupseteq(\mathcal{Q})}},\; \alpha^\sqsubseteq(\mathcal{Q}) \in \alpha^\curlywedgedownarrow(\wp(\mathbb{L})) \quad
    \frac{\forall P\in\mathcal{P}\mathrel{.}\exists Q\in\mathcal{Q} \mathrel{.}\hoarelogicup{P}{S}{\mathcal{Q}}}{\hyperlogicup{\mathcal{P}}{S}{\alpha^\sqsubseteq(\mathcal{Q})}} \nonumber
\end{eqntabular}
\end{lemma}
\begin{proof}[Proof of lemma \ref{eq:equivalent-order-ideal-rule}]
    Let us prove the first one: by rule $(\ref{eq:alpha:principal:rule})$, it is sufficient to show that $\bigsqcap\mathcal{Q} = \bigsqcap\alpha^\sqsupseteq(\mathcal{Q})$. Since $\alpha^\sqsupseteq$ is extensive, then $\bigsqcap\mathcal{Q} \sqsupseteq \bigsqcap\alpha^\sqsupseteq(\mathcal{Q})$ holds trivially. For arbitrary $P$ in $\alpha^\sqsubseteq(\mathcal{Q})$, there exists $Q\in\mathcal{Q}$ such that $Q\sqsubseteq P$ and then $\bigsqcap\mathcal{Q}\sqsubseteq P$. Thus $\bigsqcap\mathcal{Q}$ is a lower bound of $\alpha^\sqsubseteq(\mathcal{Q})$ and is smaller than the greatest lower bound of it.
    Now let us prove the second one:
    \begin{calculus}[$\Leftrightarrow$\ \ ]
        \phantom{$\Leftrightarrow$} $\hyperlogicup{\mathcal{P}}{S}{\alpha^\sqsubseteq(\mathcal{Q})}$\\
        $\Leftrightarrow$ \formulaexplanation{\Hpost{S}(\mathcal{P})\subseteq \alpha^\sqsupseteq(\mathcal{Q}) }{def.\ of $\hyperlogicup{\mathcal{P}}{S}{\mathcal{Q}}$}\\
        $\Leftrightarrow$ \formulaexplanation{\forall P \in \mathcal{P}\mathrel{.} \postSemantics{S}(P) \in \alpha^\sqsupseteq(\mathcal{Q})}{def.\ of $\subseteq$}\\
        $\Leftrightarrow$ \formulaexplanation{\forall P \in \mathcal{P}\mathrel{.} \exists Q\in\mathcal{Q}\mathrel{.} \postSemantics{S}(P)\sqsubseteq Q }{def.\ of $\alpha^\sqsupseteq$}\\
        $\Leftrightarrow$ \formulaexplanation{\forall P \in \mathcal{P}\mathrel{.} \exists Q\in\mathcal{Q}\mathrel{.} \hoarelogicup{P}{S}{Q} }{def.\ of $\hoarelogicup{P}{S}{Q}$}
    \end{calculus}
\end{proof}

Lemma ${\ref{lem:frontier-varrho-elimination-conjunction-abstraction}}$  implies that we can simplify the proof rule $(\ref{eq:rule:frontier-varrho-elimination})$ by further applying $(\ref{eq:rule:well-defined-conjunction-abstraction})$, and then Lemma \ref{eq:equivalent-order-ideal-rule}, where its hypothesis is implied by ${\ref{lem:conjunction-abstraction-more-general}}$. Since we have proved that all the intermediate rules are sound and complete, rule $(\ref{eq:rule:bottom-frontier-varrho-elimination})$ is sound and complete for all postconditions $\mathcal{Q} \in \varrho^{\sqsubseteq\underline{F}}(\wp(\mathbb{L})$.

\begin{eqntabular*}{l}
\fontsize{8pt}{11pt}\selectfont
    \dfrac{
    \begin{array}{l}
    {\displaystyle\exists\mathcal{X}\in\alpha^{\underline{F}}(\mathcal{Q})\,{\rightarrow}\,\wp(\mathbb{L}^{\sharp})\,{.}\,\mathcal{P}\subseteq \bigcup_{F\in\alpha^{\underline{F}}(\mathcal{Q})}\mathcal{X}_F}\\
        {\forall F\in\alpha^{\underline{F}}(\mathcal{Q})\mathrel{.}  
        \dfrac{ 
            \dfrac{\forall P \in\mathcal{X}_F\mathrel{.}\exists Q\in \varphi^\sqsubseteq(F)\mathcal{Q} \mathrel{.} \hoarelogicup{P}{S}{Q} }{\hyperlogicup{\mathcal{X}_F}{S}{ \alpha^{\sqsubseteq}\comp\varphi^\sqsubseteq(F)\mathcal{Q}}}
            (\ref{eq:equivalent-order-ideal-rule}) \quad 
            \dfrac{\forall P\in\mathcal{X}_F\mathrel{.}\hoarelogicdown{P}{S}{\bigsqcap  \varphi^\sqsubseteq(F)\mathcal{Q}} }{\hyperlogicup{\mathcal{X}_F}{S}{\alpha^\sqsupseteq\circ \varphi^\sqsubseteq(F)\mathcal{Q}}}(\ref{eq:equivalent-order-ideal-rule}) }
        {\hyperlogicup{\mathcal{X}_F}{S} {\varphi^\sqsubseteq(F)\mathcal{Q}}} }(\ref{eq:rule:well-defined-conjunction-abstraction})
    \end{array}
    }{\hyperlogicup{\mathcal{P}}{S}{\mathcal{Q}}}\raisebox{-0.75em}[0pt][0pt]{\llap{(\ref{eq:rule:frontier-varrho-elimination})}}  \nonumber
\end{eqntabular*}

Removing the intermediate steps, the rule becomes
\bgroup\small\begin{eqntabular}{l}
    \frac{
        {\displaystyle\exists\mathcal{X}\in\alpha^{\underline{F}}(\mathcal{Q})\rightarrow\wp(\mathbb{L}^{\sharp})\,{.}\,\mathcal{P}\subseteq \!\!\!\!\!\!\bigcup_{F\in\alpha^{\underline{F}}(\mathcal{Q})}\!\!\!\!\!\!\mathcal{X}_F
        \wedge
            (\forall F\in\alpha^{\underline{F}}(\mathcal{Q})\,{.}\,\forall P \in\mathcal{X}_F\,{.}\,\exists Q\in \varphi^\sqsubseteq(F)\mathcal{Q} \,{.}\, \overline{\{}P\overline{\}}{S}\overline{\{}Q\overline{\}},\ \ 
           \textstyle\underline{\{}P\underline{\}}{S}\underline{\{}F\underline{\}})
        }
    }{\hyperlogicup{\mathcal{P}}{S}{\mathcal{Q}}}  \nonumber \\ \renumber{\llap{\raisebox{1em}[0pt][0pt]{(\ref{eq:rule:bottom-frontier-varrho-elimination})\ }}}
\end{eqntabular}\egroup
\end{toappendix}
\ifshort An example \proofinapx\ is provided in the appendix.\vspace*{-10pt}\fi
\begin{toappendix}
\begin{example}[Proof reduction for frontier $\varrho$-elimination abstraction: bounded output] Consider the reachability without break and nontermination. Let the hyperproperties $\mathcal{P}\triangleq \{P\in\wp(\Sigma^{\*}) \mid  \exists \sigma_{max}, \sigma_{min}\in P\mathrel{.} \forall \sigma\in P\mathrel{.} \sigma_{min}(x) \leq \sigma(x) \leq \sigma_{max}(x)\}$, and $\mathcal{Q}\triangleq \{P\in\wp(\Sigma^{\*}) \mid  \exists \sigma_{max}\in P\mathrel{.} \forall \sigma\in P\mathrel{.}\sigma(x) \leq \sigma_{max}(x)\}$, and we want to prove $\hyperlogicup{\mathcal{P}}{S}{\mathcal{Q}}$ where $\texttt{S}\triangleq \texttt{if(x>0) x=x else x=-x}$ using the rule ($\ref{eq:rule:bottom-frontier-varrho-elimination}$). In this case $\alpha^{\underline{F}}(\mathcal{Q})= \{\{\sigma\}\mid  \sigma\in \Sigma\}$ is a set of singleton states.
We let the partition variant $\mathcal{X}$ be
\bgroup\abovedisplayskip0pt\belowdisplayskip3pt\begin{eqntabular}{rcl}
    \mathcal{X} &\triangleq& \LAMBDA{\{\sigma\}}{\mathcal{X}_{\{\sigma\}} \cup \bar{\mathcal{X}}_{\{\sigma\}}}\nonumber\\ 
    \text{where } 
     \mathcal{X}_{\{\sigma\}} &\triangleq&\{ P\in\wp(\Sigma^{\*})\mid \sigma\in P\land \forall \sigma'\in P \mathrel{.} \sigma'(x) \leq \sigma(x) \land -\sigma'(x) \leq \sigma(x) \} 
     \nonumber\\
    \text{and } \bar{\mathcal{X}}_{\{\sigma\}} &\triangleq&\{ P\in\wp(\Sigma^{\*})\mid \bar{\sigma}\in P\land \forall \sigma'\in P \mathrel{.} \sigma'(x) \leq \sigma(x) \land -\sigma'(x)\leq \sigma(x) \} \nonumber
\end{eqntabular}\egroup
where $\bar{\sigma}$ is a shorthand for $\sigma[\texttt{x}\leftarrow-\sigma(\texttt{x})]$. Now let us prove the case of $\hyperlogicup{\mathcal{X}_{\{\sigma\}}}{S}{\varphi^\sqsubseteq(F)\mathcal{Q}}$ for arbitrary $\sigma$, as the case for $\bar{\mathcal{X}}_{\{\sigma\}}$ is symmetrical and they can be combined by the consequence rules. Then the rule application proof steps are the following (for an arbitrary $P\in\mathcal{X}_{\{\sigma\}}$)
\bgroup\small\begin{eqntabular}{c}
    \fontsize{9pt}{11pt}\selectfont\everymath{\scriptstyle}\everydisplay{\scriptstyle}
  \dfrac{\text{let }Q = \{\sigma'\in \Sigma\mid\sigma'(x) \leq \sigma(x) \} \mathrel{.}  
    \dfrac{
        \dfrac{\text{by def of $Q$}}{\{\sigma\}\in Q } \: \dfrac{\sigma''\in Q' \text{ implies } \sigma''\in Q}{\forall Q'\mathrel{.} \{\sigma\}\subseteq Q'\subseteq Q \Rightarrow Q'\in \mathcal{X}_{\{\sigma\}}}}{Q\in\varphi^\sqsubseteq(F)\mathcal{Q}}\quad
    \dfrac{ \dfrac{\text{by def of $\mathcal{X}_{\sigma}$ and Q}}{\forall \sigma''\in P\mathrel{.} \sigma''(x)\leq\sigma(\sigma)} }
    {\hoarelogicup{P}{S}{Q}} }
    {\exists Q\in \varphi^\sqsubseteq(F)\mathcal{Q} \quad  \mathrel{.} \hoarelogicup{P}{S}{Q}} 
    \nonumber \\ \textup{and} \nonumber\\ 
    \dfrac{\dfrac{\text{by def of of $\mathcal{X}_{\sigma}$ where $\sigma'=\sigma$}}{\forall \sigma'\in F=\{\sigma\}\mathrel{.}\sigma'\in \mathcal{X}_{\{\sigma\}}}}{\hoarelogicdown{P}{S}{F}}\nonumber
\end{eqntabular}\egroup
 
 Now it only remains to show that $\displaystyle\mathcal{P}\subseteq\bigcup_{\sigma\in\Sigma}\mathcal{X}_{\{\sigma\}}\cup \bar{\mathcal{X}}_{\{\sigma\}}$. For arbitrary $P\in\mathcal{P}$, there exists $\sigma_{min}$ and $\sigma_{max}$ in $P$ where $\sigma_{min}(x)\leq \sigma'(x) \sigma_{max}(x)$ for all $\sigma'$ in P with two possible cases: 
 \begin{enumerate}[leftmargin=*]
 \item $|\sigma_{min}(x)|\leq|\sigma_{max}(x)|$: then we know that $P$ is in $\mathcal{X}_{\{\sigma_{max}\}}$ by definition.
  \item $|\sigma_{min}(x)|>|\sigma_{max}(x)|$: then $\sigma_{min}(x)$ must be negative and $\sigma_{max}(x)< - \sigma_{min}(x)$. In this case, $P$ would be in $\bar{\mathcal{X}}_{\bar{\sigma}_{min}}$ because of the following: $\bar{\bar{\sigma}} = \sigma$ has implied that $\bar{\bar{\sigma}} \in P$. Moreover, for arbitrary $\sigma'$ in $P$, $\sigma'(x)\leq \sigma_{max}(x)< -\sigma_{min}(x) = \bar{\sigma}_{min}(x)$, so as $-\sigma'(x)\leq\sigma_{min}(x)$ holds as $\sigma_{min}(x)$ is the lower bound.\qef
\end{enumerate}\let\qef\relax
\end{example}
\end{toappendix}

\section{Hierarchy of hyperproperties abstractions}\label{sect:ComparingAbstractions}
To compare these abstractions, we first show that chain limit order ideal abstract properties have an equivalent frontier order ideal representation\proofinapx.
\bgroup\abovedisplayskip3pt\belowdisplayskip3pt\begin{eqntabular}{c}
\pair{\alpha^{{\sqsubseteq}\overline{F}}(\wp(\mathbb{L}))}{\subseteq}\galoiS{{\maccent{\alpha}{\ast}}^{{\sqsubseteq}{\uparrow}}}{\mathbb{1}}\pair{{\maccent{\alpha}{\ast}}^{{\sqsubseteq}{\uparrow}}(\wp(\mathbb{L}))}{\subseteq}
\label{eq:GC:alpha:sqsubseteq:oF:alpha:ast:sqsubseteq:uparrow}
\end{eqntabular}\egroup
\begin{toappendix}
\begin{proof}[Proof of (\ref{eq:GC:alpha:sqsubseteq:oF:alpha:ast:sqsubseteq:uparrow})]
Let $\mathcal{P}\in{\alpha^{{\sqsubseteq}\overline{F}}(\wp(\mathbb{L}))}$ so that there exists 
 $\mathcal{P}'$ such that $\mathcal{P}={\alpha^{{\sqsubseteq}\overline{F}}(\mathcal{P}')}$. Let us consider
\begin{calculus}[=\ \ ]
\formula{\alpha^{{\sqsubseteq}{\uparrow}}(\mathcal{P})}\\
=
\formulaexplanation{{\alpha^{{\sqsubseteq}{\uparrow}}}({\alpha^{{\sqsubseteq}\overline{F}}(\mathcal{P}')})}{def.\ $\mathcal{P}={\alpha^{{\sqsubseteq}\overline{F}}(\mathcal{P}')}$}\\
=
\formulaexplanation{\alpha^{{\sqsubseteq}}(\alpha^{\uparrow}(\alpha^{{\sqsubseteq}}(\alpha^{\overline{F}}(\mathcal{P}'))))}{def.\ (\ref{eq:def:hat:alpha:sqsubseteq:uparrow}
) of $\alpha^{{\sqsubseteq}{\uparrow}}$ and dual def.\ (\ref{eq:def:alpha:sqsubseteq:underline:F}) of $\alpha^{{\sqsubseteq}\overline{F}}$ and composition $\comp$}\\
=
\formulaexplanation{\{P'\in \mathbb{L}\mid \exists P\in\alpha^{\uparrow}(\alpha^{{\sqsubseteq}}(\alpha^{\overline{F}}(\mathcal{P}')))\mathrel{.}P'\sqsubseteq P\}}{def.\ (\ref{eq:GC:alpha:sqsubseteq}) of $\alpha^{\sqsubseteq}$}\\
=
\formulaexplanation{\{P'\in \mathbb{L}\mid \exists P\in\{\bigsqcup_{i\in\mathbb{N}}P_i\mid\pair{P_i}{i\in\mathbb{N}}\in\alpha^{{\sqsubseteq}}(\alpha^{\overline{F}}(\mathcal{P}'))\textrm{\ is an increasing chain with existing lub}\}\mathrel{.}P'\sqsubseteq P\}}{dual def.\ (\ref{eq:def:alpha:downarrow}) of $\alpha^{\uparrow}$}\\
=
\formulaexplanation{\{P'\in \mathbb{L}\mid \exists \textrm{\ an increasing chain $\pair{P_i}{i\in\mathbb{N}}$ with existing lub}\mathrel{.}\forall i\in\mathbb{N}\mathrel{.}{P_i}\in\alpha^{{\sqsubseteq}}(\alpha^{\overline{F}}(\mathcal{P}'))
\wedge P'\sqsubseteq\bigsqcup_{i\in\mathbb{N}}P_i\}}{def.\ $\in$}\\
=
\formulaexplanation{\{P'\in \mathbb{L}\mid \exists \textrm{\ an increasing chain $\pair{P_i}{i\in\mathbb{N}}$ with existing lub}\mathrel{.}\forall i\in\mathbb{N}\mathrel{.}{P_i}\in\{P'\in \mathbb{L}\mid \exists P''\in\alpha^{\overline{F}}(\mathcal{P}')\mathrel{.}P'\sqsubseteq P''\}
\wedge P'\sqsubseteq\bigsqcup_{i\in\mathbb{N}}P_i\}}{def.\ (\ref{eq:GC:alpha:sqsubseteq}) of $\alpha^{\sqsubseteq}$}\\
=
\formulaexplanation{\{P'\in \mathbb{L}\mid \exists \textrm{\ an increasing chain $\pair{P_i}{i\in\mathbb{N}}$ with existing lub}\mathrel{.}\forall i\in\mathbb{N}\mathrel{.}\exists P''\in\alpha^{\overline{F}}(\mathcal{P}')\mathrel{.}P_i\sqsubseteq P''
\wedge P'\sqsubseteq\bigsqcup_{i\in\mathbb{N}}P_i\}}{def.\ $\in$}\\
=
\formula{\{P'\in \mathbb{L}\mid \exists P''\in\alpha^{\overline{F}}(\mathcal{P}')\mathrel{.} P'\sqsubseteq P''\}}\\
\explanation{($\Rightarrow$) $\forall i\in\mathbb{N}\mathrel{.}P_i\sqsubseteq P''$ implies $\bigsqcup_{i\in\mathbb{N}}P_i\sqsubseteq P''$ by def.\ existing lub, so that $P'\sqsubseteq P''$ by transitivity;
\\
($\Leftarrow$) Conversely choose the constant hence increasing chain $\pair{P'}{i\in\mathbb{N}}$ with existing lub $P'$ so that $\forall i\in\mathbb{N}\mathrel{.}P_i=P\sqsubseteq P''
\wedge P'\sqsubseteq\bigsqcup_{i\in\mathbb{N}}P_i=P$}\\
=
\formulaexplanation{\alpha^{{\sqsubseteq}}(\alpha^{\overline{F}}(\mathcal{P}'))}{def.\ (\ref{eq:GC:alpha:sqsubseteq}) of $\alpha^{\sqsubseteq}$}\\
=
\formulaexplanation{\mathcal{P}}{def.\ $\mathcal{P}$}
\end{calculus}

\smallskip

\noindent It follows by the fixpoint definition (\ref{eq:def:hat:alpha:sqsubseteq:uparrow}) of ${\maccent{\alpha}{\ast}}^{{\sqsubseteq}{\uparrow}}(\mathcal{P})\triangleq\Lfp{\sqsubseteq}\LAMBDA{X}\mathcal{P}\cup\alpha^{{\sqsubseteq}{\uparrow}}(X)$ that ${\maccent{\alpha}{\ast}}^{{\sqsubseteq}{\uparrow}}(\mathcal{P})=\mathcal{P}$ so that the Galois retraction (\ref{eq:GC:alpha:sqsubseteq:oF:alpha:ast:sqsubseteq:uparrow}) follows immediately.
\end{proof}
\end{toappendix}
\begin{figure}[h]
     \centering
     \includegraphics[width=0.90\textwidth]{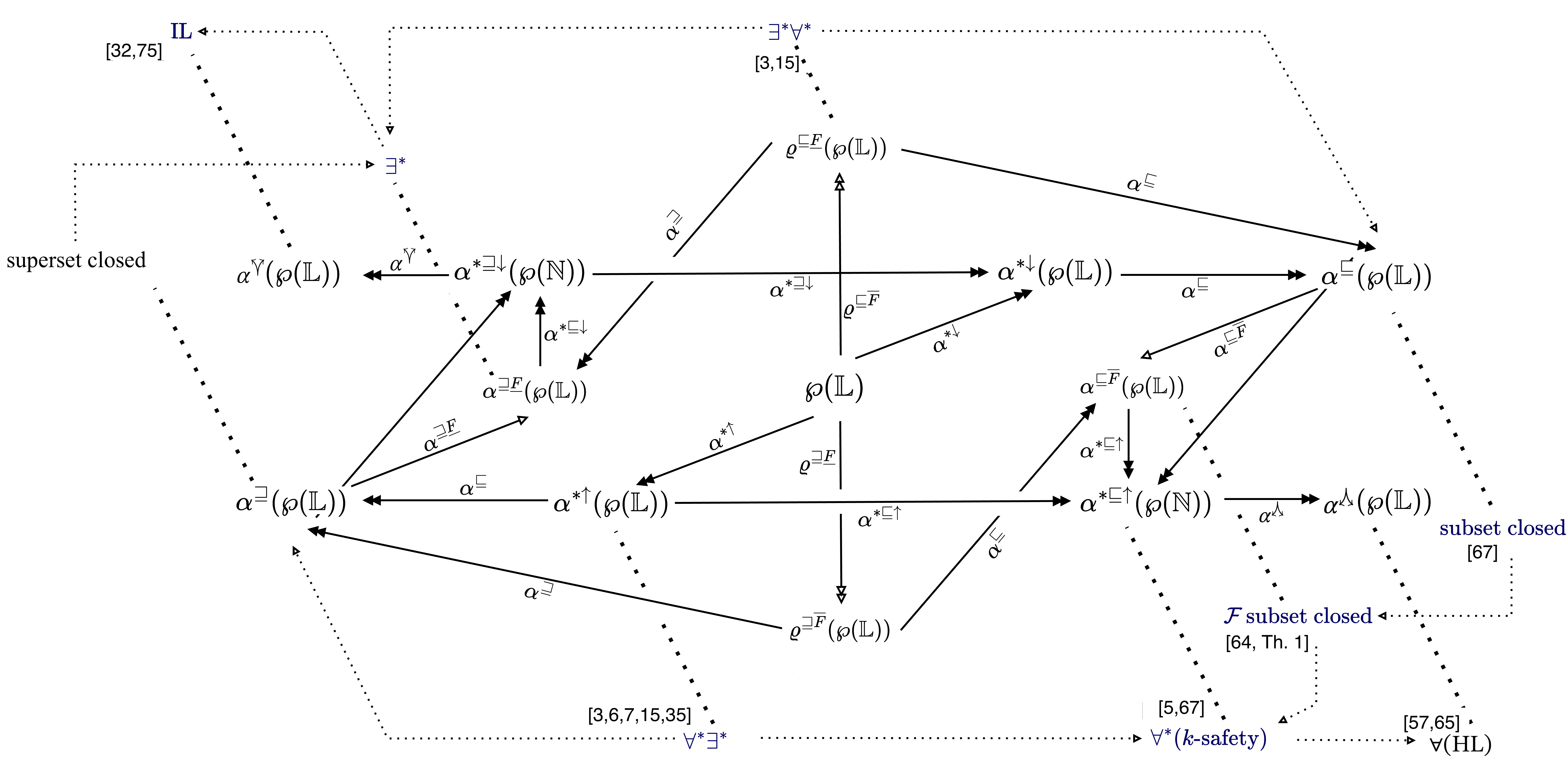}\vspace*{-1em}
     \caption{The hierarchy of hyperproperties by abstraction. The arrow is interpreted as ``more general than'' where the double arrow represents Galois surjection. Dotted line  indicated the hyperproperties subsumed by our abstract in the related works. \protect\proofinapx\label{fig:abstraction}} 
\vspace*{-15pt}
\end{figure}
Figure \ref{fig:abstraction} shows a lattice of hyperproperties derived by our abstractions as well as the related hyperproperties that they subsume.

\begin{toappendix}
\begin{proof}[Proof of figure.\ref{fig:abstraction}]
By (\ref{eq:GC:alpha:sqsubseteq:oF:alpha:ast:sqsubseteq:uparrow}), 
if $\mathcal{P}\in{{\maccent{\alpha}{\ast}}^{{\sqsubseteq}{\uparrow}}(\wp(\mathbb{L}))}$ then
$\mathbb{1}(\mathcal{P})=\mathcal{P}\in{\alpha^{{\sqsubseteq}\overline{F}}(\wp(\mathbb{L}))}$ proving
${\maccent{\alpha}{\ast}}^{{\sqsubseteq}{\uparrow}}(\wp(\mathbb{L}))$ $\subseteq$ ${\alpha^{{\sqsubseteq}\overline{F}}(\wp(\mathbb{L}))}$. 

If $\mathcal{P}\in{\alpha^{{\sqsubseteq}\overline{F}}(\wp(\mathbb{L}))}$ then $\exists\mathcal{Q}\in\wp(\mathbb{L})\mathrel{.}\mathcal{P}=\alpha^{{\sqsubseteq}\overline{F}}(\mathcal{Q})$ so that, by idempotency in (\ref{eq:GC:alpha:sqsubseteq}), $\alpha^{{\sqsubseteq}}(\alpha^{{\sqsubseteq}\overline{F}}(\mathcal{Q}))=\alpha^{{\sqsubseteq}\overline{F}}(\mathcal{Q})=\mathcal{P}$,
proving ${\alpha^{{\sqsubseteq}\overline{F}}(\wp(\mathbb{L}))}\subseteq \alpha^{\sqsubseteq}(\wp(\mathbb{L}))$.

For arbitrary non-empty $\mathcal{P}$ in $\alpha^{\curlywedgedownarrow}(\wp(\mathbb{L}))$ and consider then $\pair{\alpha^{\curlywedgedownarrow}(\wp(\mathbb{L}))}{\sqsubseteq}$ is a sublattice which is complete, meaning that it is chain-closed. Thus, $\alpha^{{\sqsubseteq}\overline{F}}(\wp(\mathbb{L}))(\mathcal{P}) = \mathcal{P}$, and so $\alpha^{{\sqsubseteq}\overline{F}}(\wp(\mathbb{L}))\subseteq \alpha^{\curlywedgedownarrow}(\wp(\mathbb{L}))$

For arbitrary non-empty $\mathcal{P}$ in $\alpha^{\sqsupseteq\underline{F}}$, $\mathcal{P} = \bigcup_{F\in {\alpha}^{\underline{F}}(\mathcal{P})}\{P\in\mathbb{L}\mathrel{.}F\sqsubseteq P\}$ by lemma $\ref{lem:FrontierCharacterizationOrderIdealAbstraction}$, then for arbitrary $P$ in $\mathcal{P}$, $F\sqsubseteq P$ for some $F$ in $\alpha^{\underline{F}}(\mathcal{P})$. $P$ would be in $\varphi^\sqsubseteq(F)\mathcal{P}$ as for all $P'$ such that $F\sqsubseteq P'\sqsubseteq P$, $P'\in \{P'\in\mathbb{L}\mid F\sqsubseteq P'\}$ trivially. This implies that $\mathcal{P}\sqsubseteq\varrho^{\sqsubseteq\underline{F}}(\mathcal{P})$, meaning that $\mathcal{P}=\varrho^{\sqsubseteq\underline{F}}(\mathcal{P})$ as the inverse holds by the fact that $\varrho^{\sqsubseteq\underline{F}}$ is reducive. This proves that $\alpha^{\sqsupseteq\underline{F}}(\wp(\mathbb{L}))\subseteq \varrho^{\sqsubseteq\underline{F}}(\wp(\mathbb{L}))$.

For arbitrary non-empty $\mathcal{P}$ in $\alpha^\sqsubseteq(\wp(\mathbb{L}))$, $\alpha^{\underline{F}} = \{\bot\}$. Then $\varrho^{\sqsupseteq\underline{F}}(\mathcal{P}) = \varrho^{\sqsupseteq}(\mathcal{P}) = \mathcal{P}$. The last equation holds because $\varrho^{\sqsupseteq}$ is closure operator on  $\alpha^\sqsubseteq(\wp(\mathbb{L}))$. This proves that $\alpha^{\subseteq}(\wp(\mathbb{L}))\sqsubseteq \varrho^{\sqsubseteq\underline{F}}(\wp(\mathbb{L}))$.
\end{proof}
\end{toappendix}
     
\ifshort\vspace*{-5pt}\fi
\section{Related Work}\label{sect:RelatedWork}
Algebraic semantics \cite{DBLP:journals/jacm/GoguenTWW77,DBLP:books/daglib/0084874,DBLP:conf/concur/Hoare14,DBLP:conf/RelMiCS/MollerOH21} is rooted in the previous concept of program schemes \cite{DBLP:journals/jcss/Goguen74,DBLP:conf/mfcs/GoguenM77,DBLP:conf/ershov/Ershov79,DBLP:conf/ifip/Nivat80,DBLP:journals/toplas/BroyWP87}. The idea of handling logics algebraically using an abstract domain goes back to \cite[section 5]{DBLP:conf/oopsla/CousotCLB12}. It requires a distinction between computational and logical orderings which first appeared in strictness analysis (using Scott partial order for computational ordering and inclusion for logical ordering \cite{DBLP:phd/ethos/Mycroft82}). It is not uncommon in abstract interpretation since then. The calculational methodology that we have used is based on \cite{DBLP:journals/pacmpl/Cousot24}. Following the introduction of trace hyperproperties \cite{DBLP:journals/jcs/ClarksonS10}, most semantics \cite{DBLP:conf/popl/AssafNSTT17,DBLP:conf/sas/MastroeniP17} and verification methods for semantic (hyper) properties have been on subclasses of hyperproperties \cite{DBLP:conf/atal/BeutnerFF024,DBLP:conf/tacas/Beutner24,DBLP:journals/lmcs/BeutnerF23,DBLP:conf/cav/BeutnerFFM23,DBLP:conf/post/ClarksonFKMRS14,DBLP:conf/cav/BeutnerF22,DBLP:conf/cav/CoenenFST19,DBLP:journals/afp/Dardinier23a,DBLP:conf/sas/MastroeniP18,DBLP:conf/pldi/DardinierM24}, further reviewed in extreme great detail in \cite[section\ 6]{DBLP:conf/pldi/DardinierM24}.

\section{Conclusion and Future Work}\label{sect:Conclusion}
Transformational (hyper) logics have traditionally been based on transformers themselves equivalent to an
operational semantics. When considering nontermination, other semantics like denotational semantics are
relevant, but the corresponding logics are in a separate world \cite{DBLP:journals/apal/Abramsky91,DBLP:journals/tcs/Heckmann93}. 

In an attempt to design (hyper) logics valid for various (abstract) semantics, we have defined an algebraic semantics (which can be instantiated to operational, denotational, or relational semantics, and is also useful for  deductive methods and static analysis). 

We have designed, by calculus, a structural fixpoint collecting semantics \textsf{post} for execution properties (e.g.\ sets of execution traces), its hypercollecting semantics \textsf{Post} for semantic properties (e.g sets of sets of traces), and the various over or under approximation logics corresponding to these transformers for correctness and incorrectness (\hyperlink{PARTIII}{part III} is for over approximation only, but the main reason to use the under approximation logic is to disprove over approximations which is expressible as $\neg{\,\overline{\llbrace}\,\mathcal{P}\,\overline{\rrbrace}\,\texttt{S}\,\overline{\llbrace}\,\mathcal{Q}\,\overline{\rrbrace}}$ $\Leftrightarrow$ $ \exists \emptyset\subsetneq\mathcal{P}'\subseteq\mathcal{P}\mathrel{.}{\,\overline{\llbrace}\,\mathcal{P}'\,\overline{\rrbrace}\,\texttt{S}\,\overline{\llbrace}\,\neg\mathcal{Q}\,\overline{\rrbrace}}$ \proofinapx).
\begin{toappendix}
\begin{proof}[Proof of $\neg{\,\overline{\llbrace}\,\mathcal{P}\,\overline{\rrbrace}\,\texttt{S}\,\overline{\llbrace}\,\mathcal{Q}\,\overline{\rrbrace}}$ $\Leftrightarrow$ $ \exists \emptyset\subsetneq\mathcal{P}'\subseteq\mathcal{P}\mathrel{.}{\,\overline{\llbrace}\,\mathcal{P}'\,\overline{\rrbrace}\,\texttt{S}\,\overline{\llbrace}\,\neg\mathcal{Q}\,\overline{\rrbrace}}$]
\begin{calculus}[$\Leftrightarrow$\ \ ]
\formula{\neg{\,\overline{\llbrace}\,\mathcal{P}\,\overline{\rrbrace}\,\texttt{S}\,\overline{\llbrace}\,\mathcal{Q}\,\overline{\rrbrace}}}\\
$\Leftrightarrow$
\formulaexplanation{\neg(\textsf{Post}^\sharp\sqb{\texttt{S}}^\sharp\mathcal{P}\subseteq \mathcal{Q})}{(\ref{eq:def:abstract:logical:triples})}\\
$\Leftrightarrow$
\formulaexplanation{\neg(\{\textsf{post}^\sharp(S)P\mid P\in\mathcal{P}\}\subseteq \mathcal{Q})}{(\ref{eq:def:Post})}\\
$\Leftrightarrow$
\formulaexplanation{\neg(\forall P\in\mathcal{P}\mathrel{.}\textsf{post}^\sharp(S)P\in\mathcal{Q})}{def.\ $\subseteq$}\\
$\Leftrightarrow$
\formulaexplanation{\exists P\in\mathcal{P}\mathrel{.}\textsf{post}^\sharp(S)P\in\neg\mathcal{Q}}{def.\ negation}\\
$\Leftrightarrow$
\formula{\exists \emptyset\subsetneq\mathcal{P}'\subseteq\mathcal{P}\mathrel{.}\{\textsf{post}^\sharp(S)P'\mid P'\in\mathcal{P}'\}\subseteq\neg\mathcal{Q}}\\
\explanation{($\Rightarrow$)\quad choose $\mathcal{P}'=\{P\}$ and def.\ $\subseteq$;\\
($\Leftarrow$)\quad since $\emptyset\subsetneq\mathcal{P}'\subseteq\mathcal{P}$ there exists a $P\in\mathcal{P}'$ such that $P\in \mathcal{P}$ and $\{\textsf{post}^\sharp(S)P'\mid P'\in\{P\}\}$
$\subseteq$
$\{\textsf{post}^\sharp(S)P'\mid P'\in\mathcal{P}'\}$
$\subseteq$
$\neg\mathcal{Q}$ proving $\textsf{post}^\sharp(S)P\in\neg\mathcal{Q}$.
}\\
$\Leftrightarrow$
\formulaexplanation{\exists \emptyset\subsetneq\mathcal{P}'\subseteq\mathcal{P}\mathrel{.}\textsf{Post}^\sharp\sqb{\texttt{S}}^\sharp\mathcal{P}'\subseteq \mathcal{Q}}{(\ref{eq:def:Post})}\\
$\Leftrightarrow$
\lastformulaexplanation{\exists \emptyset\subsetneq\mathcal{P}'\subseteq\mathcal{P}\mathrel{.}\overline{\llbrace}\,\mathcal{P}'\,\overline{\rrbrace}\,\texttt{S}\,\overline{\llbrace}\,\neg\mathcal{Q}\,\overline{\rrbrace}}{(\ref{eq:def:abstract:logical:triples})}{\mbox{\qed}}
\end{calculus}
\let\qed\relax
\end{proof}
\end{toappendix}

Since, and contrary to classic logics, proofs of general semantic (hyper) properties relative to a program semantics requires the exact characterization of this semantics in the proof, an extreme complication, we have considered abstractions of the semantic properties for which this constraint can be relaxed. This has yielded to new sound and complete simplified proof rules, including for algebraic generalizations of forall-forall, forall-exists, and exists-forall  semantic (hyper) properties.

The verification of semantic (hyper) properties is still in its infancy and far from reaching the simplicity observed in the verification of execution properties. Several compromises will be needed maybe by relaxing implication (e.g.\ using Egli-Milner order instead of inclusion), considering abstract properties (for classes of properties of practical interest), and possibly by preserving soundness but renouncing to completeness. However, in full generality, the sound and complete proof methods introduced in this paper, will ultimately be, up to equivalence, the only one applicable.
%

\ifshort 
\section*{Data Availability Statement}
\hypertarget{DataAvailabilityStatement}{The full version of this article is available with its appendix as auxiliary material and in a single file on Zenodo with clickable hyper references to the appendix \href{https://doi.org/10.5281/zenodo.14173477}{https://doi.org/10.5281/zenodo.14173477}.}
\fi
\bibliographystyle{ACM-Reference-Format}
\bibliography{bib}
\end{document}